\documentclass[aps,prb,notitlepage,nofootinbib,floatfix]{revtex4-2}
\usepackage{graphicx}
\usepackage[dvipsnames]{xcolor}
\usepackage{amsfonts}
\usepackage{amssymb}
\usepackage{amsmath}
\usepackage{mathtools}
\usepackage{afterpage}
\usepackage[mathscr]{euscript}
\usepackage{braket}
\usepackage{subfigure}
\usepackage{multirow}
\usepackage{float}
\usepackage{qcircuit}
\usepackage{caption}
\usepackage{enumitem}
\usepackage{color}   
\usepackage{hyperref}
\hypersetup{
    colorlinks=true, 
    linktoc=all,     
    linkcolor=blue,  
}

\usepackage[utf8]{inputenc} 
\usepackage{yfonts}
\usepackage{dsfont}
\usepackage{cancel}
\usepackage[scr=boondox]{mathalpha}

\usepackage[pdf]{pstricks}
\usepackage{leftidx}
\usepackage{pst-all}
\usepackage{booktabs}

\usepackage{slashed}
\usepackage{stmaryrd}

\usepackage[bottom]{footmisc} 

\usepackage{amsthm,mathtools,amsmath}
\usepackage{tikz}

\newcommand{\Z}{\mathbb{Z}}
\newcommand{\R}{\mathbb{R}}

\newcommand{\C}{\mathbb{C}}
\newcommand{\CP}{\mathbb{CP}}

\newcommand{\bP}{{\bf P}}

\newcommand{\tr}{\mathrm{tr}}
\newcommand{\mathbbm}[1]{\text{\usefont{U}{bbm}{m}{n}#1}} 

\usepackage{environ}
\NewEnviron{eqs}{%
\begin{equation}\begin{split}
    \BODY
\end{split}\end{equation}
}

\usepackage{hyperref}
\hypersetup{
    colorlinks=true,
    linkcolor=blue,
    filecolor=magenta,
    urlcolor=cyan,
}

\newtheorem{thm}{Theorem}
\newtheorem{lem}[thm]{Lemma}
\newtheorem{prop}[thm]{Proposition}
\newtheorem{cor}[thm]{Corollary}

\newcommand\restr[2]{{
  \left.\kern-\nulldelimiterspace 
  #1 
  \vphantom{\big|} 
  \right|_{#2} 
  }}

\newcommand{\floor}[1]{\left \lfloor #1 \right \rfloor}

\usepackage{MnSymbol}
 
\usepackage{scalerel,amssymb}

\def\gl{\mathfrak{gl}}

\begin{document}

\title{2D Fermions and Statistical Mechanics: Critical Dimers and Dirac Fermions in a background gauge field}

\author{Sri Tata$^1$}
\affiliation{$^1$Yale University, Department of Mathematics}

\begin{abstract}
In the limit of the lattice spacing going to zero, we consider the dimer model on isoradial graphs in the presence of singular $U(N,\C)$ gauge fields flat away from a set of punctures. 
We consider the cluster expansion of this twisted dimer partition function show it matches an analogous cluster expansion of the 2D Dirac partition function in the presence of this gauge field. The latter is often referred to as a tau function. This reproduces and generalizes various computations of Dubédat~\cite{Dubedat2018doubleDimersConformalLoopEnsemblesIsomonodromicDeformations}.
In particular, both sides' cluster expansion are matched up term-by-term and each term is shown to equal a sum of a particular holomorphic integral and its conjugate.
On the dimer side, we evaluate the terms in the expansion using various exact lattice-level identities of discrete exponential functions and the inverse Kasteleyn matrix. 
On the fermion side, the cluster expansion leads us to two novel series expansions of tau functions, one involving the Fuschian representation and one involving the monodromy representation.
\end{abstract}

\maketitle

\twocolumngrid
\tableofcontents
\onecolumngrid

\section{Introduction}

The dimer model is a statistical model that counts perfect matchings on a planar graph. The Dirac fermion in two dimensions is a simple yet important quantum field theory.
Both of these models occupy special places in the study of two-dimensional physics, particularly in the setting of critical phenomena. 
In this paper, we show that in a wide sense these theories are the same.

On one hand, the dimer model on the square grid was among the first models for which conformal invariance was rigorously demonstrated~\cite{kenyon2000confDomino,kenyon2000dominoGFF}. Moreover, exact `discrete holomorphic' structures studying dimers on `critical' graphs~\cite{kenyon2002critical} have been essential in establishing conformal invariance in more general settings, such as in the Ising model~\cite{smirnov2007Ising,chelkak2011discrete}.

On the other hand, the 2D Dirac fermion is also an important example of a conformal field theory. Moreover, it is a `building block' of the large class of \textit{rational conformal field theories} (e.g. the Ising model, Potts model,...) in the sense that its current algebra generates the operator algebras of these theories (c.f. \cite{bigYellowCFT1997}).

These two theories are closely related. First, we have that the dimer model's partition function is the determinant of a so-called \textit{Kasteleyn matrix} $\overline{\partial}$. This suggestive notation is because the Kasteleyn matrix can be thought of as a discretization of the holomorphic derivative (see Sec.~\ref{sec:dimers_isoradial}). One way to phrase the dimer partition function is thus to use the language of Grassmann calculus (see e.g.~\cite{bigYellowCFT1997} for a review on Grassmann calculus)
\begin{equation*}
    Z_{\mathrm{dimer}} = \int d\overline{\psi} d\psi e^{ - \int \overline{\psi} \, \overline{\partial} \, \psi } = \det( \overline{\partial} ).
\end{equation*}
Immediately, this shares an aesthetic similarity
\footnote{Similar observations were made in~\cite{dijkgraaf2009}.} 
with the partition function of the the two dimensional fermion,
\begin{equation*}
    Z_{\mathrm{Dirac, 2D}} = \int d\overline{\psi} d\psi e^{ - \int \overline{\psi} \, \slashed{\partial} \, \psi } = \det( \slashed{\partial} ),
\end{equation*}
especially since the two-dimensional Dirac operator $\slashed{\partial}$ is comprised of a holomorphic and antiholomorphic derivative.

This similarity is also quantitative. An early indication around 2000 was the result~\cite{kenyon2000dominoGFF} that studied fluctuations in the `height function', an integer-valued function on faces of the graph that's dual to a dimer matching (see Sec.~\ref{sec:height_functions}). These fluctuations were shown (in the sense of a weak limit) to identify the height function to the Gaussian Free Field, which is the free boson theory in 2D with action $S_{\mathrm{boson}}(\phi) = -{\pi} \int (\partial \phi)^2 $. This duality is reminiscent of the continuous version of the boson-fermion duality. 

On higher genus surfaces, there are multiple $\overline{\partial}$ operators identified with spin structures~\cite{cimasoniReshetikhin2007spin} and higher-genus bosonization formulas~\cite{alvarezGaume1987bosonHigherGenus} apply to compute partition functions, although conformal invariance on the lattice is harder to formulate and establish on higher-genus surfaces.

Later in 2011, Dubédat~\cite{Dubedat2011} considered the stronger forms of height fluctuations and shows that they agree with the Gaussian Free Field / free boson expectations. He coupled the dimer model to a background $U(1)$-connection, often taken to be singular, and computed the partition function with respect to the connection. A key tool he uses is a special discrete holomorphic function with monodromy.
In~\cite{Dubedat2018doubleDimersConformalLoopEnsemblesIsomonodromicDeformations}, he extends this analysis to coupling to certain (singular) $SL(2,\C)$ connections.
In both of these analyses, he finds that the dimer partition function with respect to this connection matches with a $\tau$ function, which is a mathematical construction of a 2D Dirac determinant in the presence of a singular connection (see the Appendix~\ref{app:diracAndTau}). This can be rephrased as follows.
Given a background lattice connection $U$ corresponding to a continuum gauge field $A$, these analyses demonstrate asymptotic equalities between the (normalized) partition functions
\begin{equation*}
    Z_{\mathrm{dimer}}(U) = \det( \overline{\partial}(U) )
    \quad \text{ and } \quad
    Z_{\mathrm{Dirac, 2D}}(A) = \det( \slashed{\partial} - \slashed{A} )
\end{equation*}
where $\overline{\partial}(U)$ is a twisted Kasteleyn operator (see Sec.~\ref{sec:graph_connection}) and $\slashed{\partial} - \slashed{A}$ is the twisted Dirac operator (see Appendix~\ref{app:diracAndTau}). 

More generally, graph connections have applications to studying the scaling limit of the dimer model. In~\cite{kenyon2014doubledimer}, Kenyon used $SL(2,\C)$ connections to prove that loops in the double-dimer model (formed by superimposing two independent dimer matchings) are distributed in a conformally invariant way, which was conjectured~\cite{sheffield2006CLE,kenyonWilson2011} to be related to $CLE(4)$, a certain `conformal loop ensemble' of loop-soups in the plane with a conformally invariant distribution.
Dubédat's work~\cite{Dubedat2018doubleDimersConformalLoopEnsemblesIsomonodromicDeformations} about $SL(2,\C)$-connections established a quantitative link between the twisted dimer partition function and certain $CLE(4)$ observables.
More recently in~\cite{dksWebs2022}, $SL(N,\C)$-twisted determinants were related to weighted sums over so-called `webs', weighted by traces of the web.

In this paper, we establish another quantitative link between the partition function of `critical' dimer models in the presence of background singular $U(N)$ connections that are flat away from a finite number of `punctures'. In particular, given a $U(N)$ lattice gauge field $U$ and corresponding continuum field $\slashed{A}$ that are flat away from a set of punctures, and a `scale' $R$ corresponding to an inverse `mesh size' of the lattice, we will have that the normalized partition functions are related as:

\begin{equation*}
    \boxed{
        \frac{\det( \overline{\partial}(U) )}{\det( \overline{\partial} )} 
        \xrightarrow{R \to \infty, \text{moments}}
        \frac{\det( \slashed{\partial} - \slashed{A} )}{\det( \slashed{\partial} )}
        \times
        e^{i \Phi_{\text{n.u.}}} .
    }
\end{equation*}

Above, `$\xrightarrow{R \to \infty}$' means up to $O(1/R)$ factors. 
We explain the meaning of $\xrightarrow{\text{moments}}$ below.
And, $e^{i \Phi_{\text{n.u.}}}$ is a \textit{non-universal} phase factor that depends on various microscopic data (see Sec.~\ref{sec:non_universality_im_part}). And, the ratio $\frac{\det( \slashed{\partial} - \slashed{A} )}{\det( \slashed{\partial} )}$ is a `regulated' fermion determinant that diverges logarithmically as $\log(R)$, but that we can relate to the standard $\tau$ function (see Sec.~\ref{app:diracAndTau}). 

This extends and complements the works of Dubédat. His strategy to equate the dimer partition functions to tau functions is variational, i.e. by displacing the punctures and finding matching differential equations for the tau function and for (logarithms of) the partition functions. 
In particular, this leads to various undetermined constants in the equations, related to the initial data of the differential equation.

In our strategy, we consider the \textit{cluster expansion} of the dimer partition function. 
We are able to explicitly relate the series in the cluster expansion to a series of holomorphic integrals; the heart of these calculations is various exact lattice-level identities of discrete exponentials on isoradial graphs (see Sec.~\ref{sec:exactExpr_zipperTraces}).
In certain cases, this expansion can be resummed and lead to determinations of the previously undetermined constants (see Sec.~\ref{sec:asympt_I2n_singleZip}). 
A quantum field theory-like calculation (see Appendix~\ref{app:tau_zipper}) confirms that the dimer cluster expansion matches the free fermion expansion after a certain resummation of the Dirac partition function. On the Dirac side, certain integrals in the cluster expansion are divergent. This is not an issue for the dimer model, since the lattice scale naturally regulates the integrals to produce finite numbers that diverge with the lattice scale (see Sec.~\ref{sec:nonSeparatedZippers}).

One disadvantage of our method is that while we can rigorously compute the terms in the cluster expansion, this method alone cannot yet rigorously show that their sum converges, whereas in the variational analysis of Dubédat this convergence is manifest. Combining his results with ours do however lead to rigorous results about the correlations by showing moments with respect to the connection converge. As an example, Dubédat shows that (let $-\pi<\alpha<\pi$) height correlations of the form 
$\braket{e^{ i \alpha ( h(F) - h(F') ) }}$ converge as
\begin{equation*}
    \log |\braket{e^{ i \alpha ( h(F) - h(F') ) }}|
    \xrightarrow{|F-F'| \to \infty} 
    -\frac{1}{2\pi^2}\alpha^2 \log(|F-F'|)
    +
    C(\alpha)
\end{equation*}
for some constant $C(\alpha)$, which agrees with classical Coulomb gas heuristics~\cite{nienhuis1984coulomb}. For us, we can rigorously compute all moments with respect to $\alpha$ allowing us to determine $C(\alpha)$
(see Sec.~\ref{sec:singleZipper}). Our methods alone can't demonstrate that the sum of moments converges, due to the presence of subleading corrections, but our moment calculations rigorously fix $C(\alpha)$. Similarly, we can rigorously demonstrate moment generating functions for twisted Kasteleyn determinants corresponding to other singular gauge fields (see Sec.~\ref{sec:multipleZipper_calculations}). 

The organization of the paper is as follows. 
In Sec.~\ref{sec:review}, we review the preliminaries of isoradial graphs, the dimer model, and graph connections. 
In Sec.~\ref{sec:singleZipper}, we compute the partition function with respect to connections flat away from two punctures.
In Sec.~\ref{sec:multipleZipper_calculations}, we compute the partition function with respect to multiple punctures.
The details of the specific calculations, the expectations, and the results are highlighted in the beginning of each of those sections. 
In Sec.~\ref{sec:discussion}, we conclude with discussion and further questions.

Various technical calculations are left in the appendix.
However, of independent interest in the Appendix~\ref{app:diracAndTau}, we give a (non-rigorous) overview of computing the partition function of Dirac fermions with respect to background connections in two dimensions. There, we employ the series expansions of the Dirac fermion in the background connections, demonstrating (after a `counterterm correction') the gauge invariance of the series. 
Representing the gauge field in two different ways leads to two different expansions of the determinant, one related to the Fuschian representation of the connection and one related to its monodromy representation. The first can be shown to satisfy the defining equations of the tau function and the second is closely related to the dimer model partition function. We haven't seen either of these expansions in the literature.

\section{Preliminaries} \label{sec:review}

In Sec.~\ref{sec:dimers_isoradial}, we discuss the basic definitions of isoradial graphs, $\overline{\partial}$ operator, and its inverse following~\cite{kenyon2002critical} (see also ~\cite{mercatDiscreteIsing}). 
Then, we discuss its relation to the dimer model. See~\cite{kenyonIntroDimer} for an introduction to the dimer model.
In Sec.~\ref{sec:graph_connection}, we discuss graph connections and their implementation in twisting the dimer model. In Sec.~\ref{sec:height_functions}, we discuss the definition of height functions associated to dimer matchings and the application of graph connections to the study of height functions. 

\subsection{Dimer Model on Isoradial Graphs, $\overline{\partial}$, and $\frac{1}{\overline{\partial}}$} \label{sec:dimers_isoradial}

An \textit{isoradial embedding} of a graph is a planar embedding of the graph for which every edge is drawn as a straight line and every face is circumscribed in a circle of radius 1; a graph is \textit{isoradial} if it has an isoradial embedding.
Equivalently, one can think of an isoradial embedding as related to an underlying \textit{rhombus graph} as follows. The rhombus graph is a planar graph whose edges are drawn as straight lines and whose faces are all rhombi of side-lengths 1. 
Then the isoradial graph can be formed by drawing edges between opposite faces of the rhombi.
Note that a rhombus graph is bipartite since all faces have an even number of sides, so this procedure will give rise to two distinct graphs that are dual to each other and that are both isoradial. 

\begin{figure}[p!]
\begin{subfigure}
  \centering
  \includegraphics[width=\linewidth]{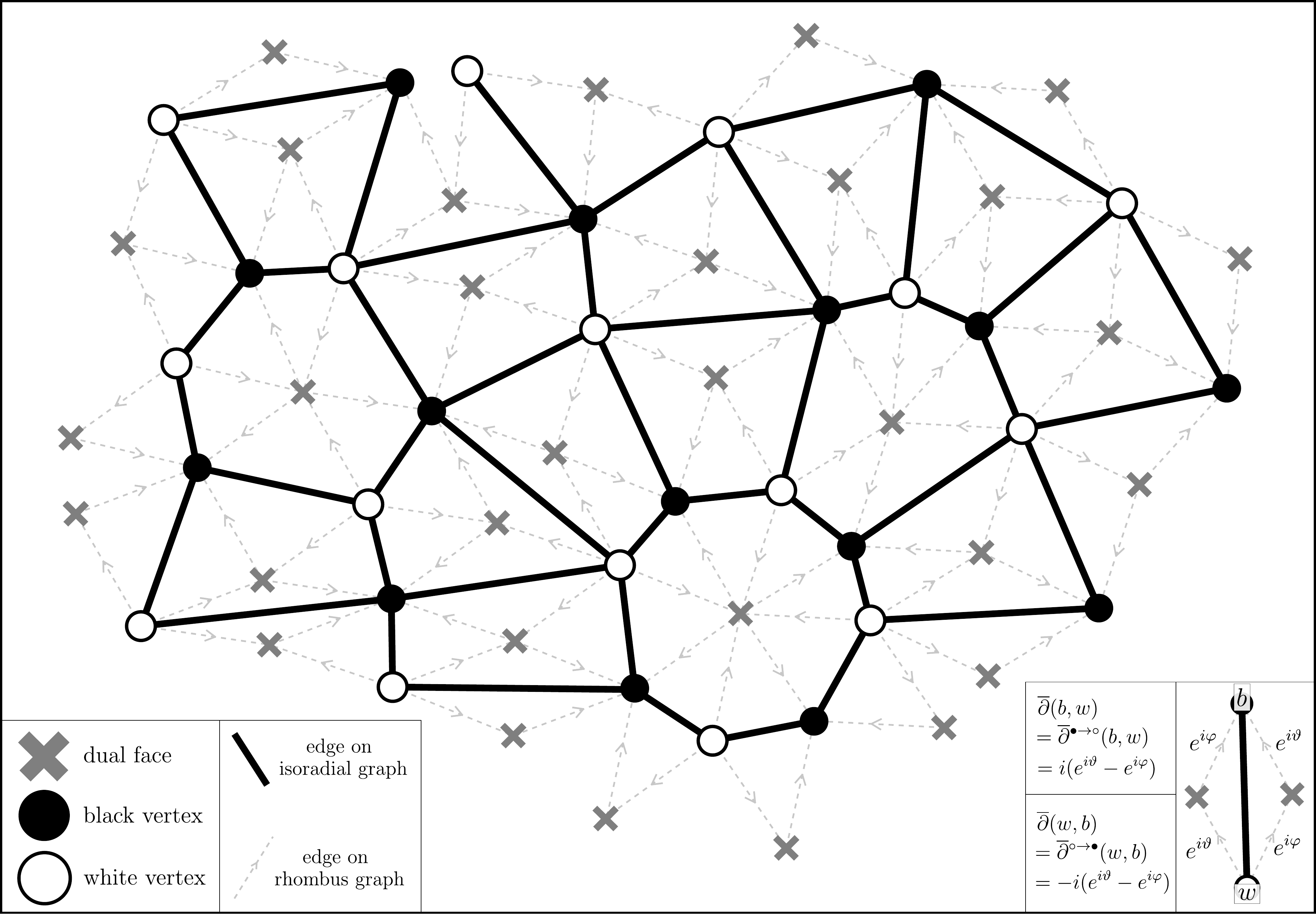}
  \caption{
      Example portion of a bipartite isoradial graph.
      (Bottom-Left) Legend labeling vertices corresponding to dual faces and black/white vertices, as well as the edges for the isoradial graph and oriented edges of the rhombus graph.
      (Bottom-Right) Depiction of the Kasteleyn operator $\overline{\partial}$ for the different directions $\overline{\partial}^{\bullet \to \circ}(b,w)$ and $\overline{\partial}^{\circ \to \bullet}(w,b)$.
  }
  \label{fig:isoradial_graph}
\end{subfigure}
\begin{subfigure}
  \centering
  \includegraphics[width=\linewidth]{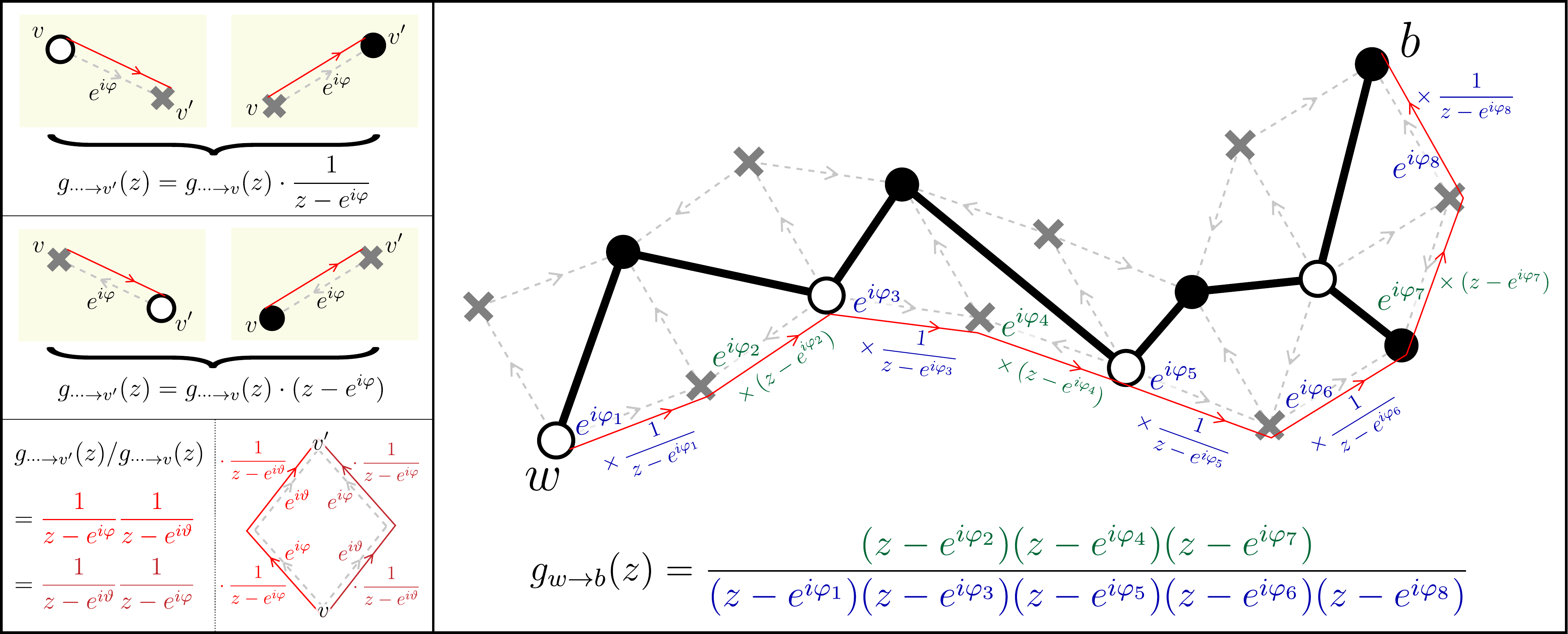}
  \caption{
    The discrete exponential function.
    (Left) Rules for relating $g_{\cdots \to v'}$ and $g_{\cdots \to v}$ for $v,v'$ adjacent on the rhombus graph. Consistency of the products around a rhombus implies each $g_{x \to y}(z)$ is well-defined and independent of the path $x \to y$.
    (Right) Example of discrete exponential between a white and black vertex.
  }
  \label{fig:discrete_exponential}
\end{subfigure}
\end{figure}

In this paper, we will consider exclusively cases where the isoradial graph itself is bipartite with a decomposition into black and white vertices. Given a bipartite isoradial graph, one can give an orientation to each edge of the rhombus graph by having them point \textit{away} from white vertices and \textit{towards} black ones. As such, every edge of the rhombus graph depicts a vector which is associated to a unit complex number. See Fig.~\ref{fig:isoradial_graph} for a depiction.

Now, we can suggestively define the $\overline{\partial}$ operator on the graph as follows. Given a black vertex $b$ connected by an edge to a white vertex $w$, the above-stated orientations of the rhombus graph edges define rhombus angles $e^{i \vartheta}$,$e^{i \varphi}$ (say $\vartheta - \pi < \varphi < \vartheta$) coming out from $w$ on the right,left side of $w$ respectively. Then, we define
\begin{equation} \label{eq:def_kasteleyn_mat}
    \overline{\partial}^{\bullet \to \circ}(b,w) 
    = 
    -\overline{\partial}^{\circ \to \bullet}(w,b)
    = 
    i (e^{i\vartheta} - e^{i\varphi}) 
    =
    - e^{i(\vartheta+\varphi)/2}  \nu(b,w)
    \quad,\quad
    \text{ where } \nu(b,w) =  2 \sin\left(\frac{\vartheta - \varphi}{2}\right) > 0.
\end{equation}
$\nu(b,w)$ is a positive \textit{weight function} on edges.
Note $\overline{\partial}^{\bullet \to \circ}(b,w)$ acts from black-to-white vertices and $\overline{\partial}^{\circ \to \bullet}(w,b) = -\overline{\partial}^{\bullet \to \circ}(b,w)$ acts from white-to-black. 
Then, we define 
\begin{equation*}
\begin{split}
    \overline{\partial}^{\bullet \to \circ}(v,v') 
    = \overline{\partial}^{\circ \to \bullet}(v,v') 
    = 0 
    \quad &, \quad
    \text{ for all vertices } v,v' \text{ not neighboring}
\end{split}
\end{equation*}
and define the full matrix
\begin{equation}
    \overline{\partial} = \overline{\partial}^{\bullet \to \circ} + \overline{\partial}^{\circ \to \bullet}.
\end{equation}
The matrix $\overline{\partial}$ can be thought of as a discretization of the antiholomorphic derivative $\partial_{\overline{z}} = \frac{1}{2}(\partial_x + i \partial_y)$. For example, the square-lattice $(\sqrt{2} \Z)^2$ is an isoradial graph, and given a function $f(v)$ on vertices of $(\sqrt{2} \Z)^2$, we have
\begin{equation*}
    (\overline{\partial} f)(v) = \sqrt{2}(f(v^{\rightarrow}) - f(v^{\leftarrow})) + i \sqrt{2}(f(v^{\uparrow}) - f(v^{\downarrow}))
\end{equation*}
where $v^{\rightarrow},v^{\leftarrow},v^{\uparrow},v^{\downarrow}$ are the vertices to the right,left,above,below $v$ on the grid. As such, a function $f$ on the vertices is said to be \textit{discrete holomorphic} or \textit{discrete analytic} if $\overline{\partial} f = 0$. For most of this paper, we will only be concerned with the half $\overline{\partial}^{\bullet \to \circ}$ of $\overline{\partial}$.

For any \textit{finite} isoradial graph, it turns out that $\overline{\partial}^{\bullet \to \circ}$ is a \textit{Kasteleyn matrix}, whose determinant gives a positive-weighted sum over \textit{dimer configurations} (i.e. perfect matchings) of the graph. In particular, we'd have 
\begin{equation} \label{eq:det_of_kasteleyn}
    \det(\overline{\partial}^{\bullet \to \circ})
    =
    e^{i \Phi_0}
    \,\,
    \sum_{\substack{\text{perfect matchings} \\ \mathcal{M} \text{ of graph}}}
    \,\,\,
    \prod_{\substack{\text{edges } {\bf e}=(b,w) \\ \text{part of } \mathcal{M}}}
    \,\,
    \nu(b,w)
\end{equation}
where $e^{i \Phi_0}$ is a unit complex number that doesn't change even if we modify the definition of the weight function $\nu(b,w)$ in Eq.~\eqref{eq:def_kasteleyn_mat}. For example, instead choosing $\nu(b,w) = 1$ would instead equate $\det(\overline{\partial}^{\bullet \to \circ})$ to the number of dimer matchings, an example of the celebrated \textit{Kasteleyn's Theorem}~\cite{kasteleyn1961}. For our purposes, we instead use the so-called \textit{critical weights} as in Eq.~\eqref{eq:def_kasteleyn_mat}.

For most of the remainder of this paper, we consider \textit{infinte} isoradial graphs. In this case, the determinant $\det(\overline{\partial}^{\bullet \to \circ})$ is infinite, although \textit{ratios} of related determinants used to characterize certain dimer correlation functions will be well-defined. On an infinte graph, there is a special class of discrete holomorphic functions that we'll refer to as \textit{discrete exponentials} or \textit{lattice exponentials}. Given two vertices $v,w$ on the rhombus graph and a path $v \to w$ on the rhombus graph, the discrete exponential is a rational function $g_{v \to w}(z)$ defined recursively. First, we'd have $g_{v \to v}(z) = 1$. Then, if $v,v'$ share an edge on the rhombus graph with the edge oriented $v \to v'$ with $v'-v = e^{i \varphi}$, we'll have $g_{\cdots \to v'}(z) = g_{\cdots \to v}(z) \cdot \frac{1}{z-e^{i \varphi}}$. The rhombus graph structure means these functions don't depend on the path $v \to w$. See Fig.~\ref{fig:discrete_exponential} for depictions and an example. 
For $v$ any vertex, $b$ any black vertex and $z$ any complex number avoiding the poles of $g_{v \to \cdots}$, it is simple to check that $\sum_{w \text{ white}}\overline{\partial}(b,w) g_{v \to w}(z) = 0$ so that $g_{v \to \boldsymbol{\cdot} }(z)$ are all discrete holomorphic. Note that the definitions also make sense for $v,w$ on the rhombus graph corresponding to faces - such functions are important in this paper's analysis.
The functions $g_{v \to w}(z)$ are referred to as exponentials because of their asymptotic behavior for $|z| \ll 1$ and $|z| \gg 1$, which are reviewed in Appendix~\ref{app:asympt_gFunc}. 

The \textit{inverse Kasteleyn matrix} $\frac{1}{\overline{\partial}}$ comprised of $\frac{1}{\overline{\partial}^{\bullet \to \circ}}$ and $\frac{1}{\overline{\partial}^{\circ \to \bullet}}$ would satisfy
\begin{equation} 
    \sum_{w} \overline{\partial}^{\bullet \to \circ}(b,w) \frac{1}{\overline{\partial}^{\bullet \to \circ}}(w,b') = \delta_{b,b'}
    \quad,\quad
    \sum_{b} \frac{1}{\overline{\partial}^{\bullet \to \circ}}(w,b) \overline{\partial}^{\bullet \to \circ}(b,w') = \delta_{w,w'}
    \quad,\quad
    \frac{1}{\overline{\partial}^{\bullet \to \circ}}(w,b) = -\frac{1}{\overline{\partial}^{\circ \to \bullet}}(b,w) \xrightarrow{b \to \infty} 0.
\end{equation}
In fact, an exact formula~\cite{kenyon2002critical} exists
\begin{equation} \label{eq:invKast_contour_formula}
    \frac{1}{\overline{\partial}^{\bullet \to \circ}}(w,b) = -\frac{1}{\overline{\partial}^{\circ \to \bullet}}(b,w) = \frac{1}{2 \pi} \int_{0}^{-e^{i \phi} \infty} dz \, g_{w \to b}(z) 
    \quad,\quad
    \text{ where } |\phi - \arg(b-w)| < \pi.
\end{equation}
In the above, the contour $0 \to -e^{i \phi} \infty$ is taken to be a straight-line ray starting going $0 \to \infty$ going out in the $-e^{i \phi}$ direction. The exact choice of contour doesn't depend on $\phi$, because $g_{w \to b}(z) \xrightarrow{b \to \infty} \frac{1}{z^2}$ and the general fact \cite{kenyon2002critical,BMS05} that the poles of $g_{w \to b}(z)$ are in the set of unit complex numbers $\{ +e^{i\phi} \,\,|\,\, \phi \in \R \,\text{ and }\, |\phi - \arg(b-w)| < \pi \}$.

\subsection{Edge-Probabilities and Correlation Kernels}
Given a finite isoradial graph and a collection of edges $\textbf{coll} = \{{\bf e}_1, \cdots, {\bf e}_n\} = \{(b_1,w_1), \cdots, (b_n,w_n)\}$, we may want to ask what the probability that all edges in $\textbf{coll}$ are contained in a given dimer matching. This probability is given by (see e.g.~\cite{kenyon1997localstat})
\begin{equation} \label{eq:edge_correlations}
    \text{Prob}(\text{all edges in }\textbf{coll}\text{ are in a matching})
    =
    \left( \prod_{i=1}^{n} \overline{\partial}(b_i,w_i) \right)
    \cdot 
    \det\left(\restr{\overline{\partial}^{-1}}{\{w_1,\cdots,w_n\},\{b_1,\cdots,b_n\}}\right),
\end{equation}
i.e. is the product of Kasteleyn edge weights times a corresponding minor of the inverse Kasteleyn matrix. 

Above, the only way that the above probability could be nonzero is if all $\{b_i\}$ and $\{w_i\}$ are distinct because all vertices can only appear in one edge at a time. In this paper, it's important to consider collections $\textbf{coll}$ of edges with repeated $\{b_i\}$ or $\{w_i\}$.
The \textit{restriction} of $\overline{\partial}^{\bullet \to \circ}_{\textbf{coll}}$ to the edges in $\textbf{coll}$ is imporant, defined as
\begin{equation}
    \overline{\partial}^{\bullet \to \circ}_{\textbf{coll}}(b,w) 
    = 
    \begin{cases}
        \overline{\partial}^{\bullet \to \circ}(b,w)\quad 
        &\text{ if } (b,w) \in \textbf{coll},
        \\
        0
        &\text{ otherwise.}
    \end{cases}
\end{equation}

Also, matrices of the form $\frac{\overline{\partial}^{\bullet \to \circ}_{\textbf{coll}}}{\overline{\partial}^{\bullet \to \circ}} := \overline{\partial}^{\bullet \to \circ}_{\textbf{coll}} \frac{1}{\overline{\partial}^{\bullet \to \circ}}$ are important in our analysis. Indeed, we can see it as a certain \textit{correlation kernel} as follows. 
Consider the principal minors of such matrix. 
For a collection $\tilde{B} = \{\tilde{b}_{1} \cdots \tilde{b}_{k}\}$ of black vertices of the graph, we can express the corresponding principal minor as:
\begin{equation}
\begin{split}
    \det\left(
        \restr{\frac{\overline{\partial}^{\bullet \to \circ}_{\textbf{coll}}}{\overline{\partial}^{\bullet \to \circ}}
        }{\tilde{B},\tilde{B}}
    \right)
    &=
    \sum_{\substack{
        \tilde{W}:=\{\tilde{w}_1,\cdots,\tilde{w}_k\} 
        \\ 
        \subset \text{ white vertices}
    }}
    \det\left(
        \restr{\overline{\partial}^{\bullet \to \circ}_{\textbf{coll}}}{\tilde{B},\tilde{W}}
    \right)
    \cdot
    \det\left(
        \restr{\frac{1}{\overline{\partial}^{\bullet \to \circ}}}{\tilde{W},\tilde{B}}
    \right)
    \\
    &=    
    \sum_{\substack{
        \tilde{W}:=\{\tilde{w}_1,\cdots,\tilde{w}_k\} 
        \\ 
        \subset \text{ white vertices}
    }}
    \sum_{\sigma \in \mathrm{Sym}(n)} (-1)^{\sigma}
    \left( \prod_{i=1}^{n} 
        \overline{\partial}^{\bullet \to \circ}_{\textbf{coll}}(\tilde{b}_i,\tilde{w}_{\sigma(i)})
    \right)
    \cdot 
    \det\left(\restr{\overline{\partial}^{-1}}{\{\tilde{w}_1,\cdots,\tilde{w}_n\},\{\tilde{b}_1,\cdots,\tilde{b}_n\}}\right)
    \\
    &=    
    \sum_{\substack{
        \tilde{W}:=\{\tilde{w}_1,\cdots,\tilde{w}_k\} 
        \\ 
        \subset \text{ white vertices}
    }}
    \sum_{\sigma \in \mathrm{Sym}(n)}
    \left( \prod_{i=1}^{n} 
        \overline{\partial}^{\bullet \to \circ}_{\textbf{coll}}(\tilde{b}_i,\tilde{w}_{\sigma(i)})
    \right)
    \cdot 
    \det\left(\restr{\overline{\partial}^{-1}}{\{\tilde{w}_{\sigma(1)},\cdots,\tilde{w}_{\sigma(n)}\},\{\tilde{b}_1,\cdots,\tilde{b}_n\}}\right)
    \\
    &= 
    \sum_{\substack{
        \tilde{W}:=\{\tilde{w}_1,\cdots,\tilde{w}_k\} 
        \\ 
        \text{ordered collection}
        \\
        \text{of white vertices}
    }}
    \text{Prob}(
        \text{all edges }  
        (\tilde{b}_1,\tilde{w}_1) \cdots (\tilde{b}_n,\tilde{w}_n)
        \text{ are in } \textbf{coll} \text{ and in a matching}    
    )
    \\
    &=
    \text{Prob}(
        \tilde{b}_1,\cdots,\tilde{b}_n
        \text{ are all matched in edges in } \textbf{coll}
    ).
\end{split}
\end{equation}
The first line is the Cauchy-Binet identity. The second line is expanding the first determinant. The third line absorbs $(-1)^\sigma$ into rearranging the rows. The fourth line uses Eq.~\eqref{eq:edge_correlations} and absorbing the two summations into a single sum.
So, all principal minors lie between $0$ and $1$ since they have a probabilistic interpretation. 
Consider the modified characteristic polynomial,
\begin{equation}
\begin{split}
    P(t) = 
    \det\left(
        \mathbbm{1}
        + 
        (t-1)
        \frac{\overline{\partial}^{\bullet \to \circ}_{\textbf{coll}}}{\overline{\partial}^{\bullet \to \circ}}
    \right)
    =&
    \frac{\det\left(\overline{\partial}^{\bullet \to \circ} + (t-1) \overline{\partial}^{\bullet \to \circ}_{\textbf{coll}}\right)}
    {\det\left(\overline{\partial}^{\bullet \to \circ}\right)}
    =
    \sum_{k=1}^{|\textbf{coll}|} t^k P_k,
    \\
    \text{where } P_k = \text{Prob}\Big(\textit{exactly } k \text{ distinct  vertices  of }& \{b_1,\cdots,b_{|\textbf{coll}|}\} \text{ are matched in edges of } \textbf{coll}\Big).
\end{split}
\end{equation}
This is because $\overline{\partial}^{\bullet \to \circ} + (t-1) \overline{\partial}^{\bullet \to \circ}_{\textbf{coll}}$ is exactly the matrix $\overline{\partial}^{\bullet \to \circ}$ except the edge weights get modified as $\nu(b,w) \to t \cdot \nu(b,w)$ for $(b,w) \in \textbf{coll}$ and remain unchanged for other edges, so the coefficient of $t^k$ in 
${\det\left(\overline{\partial}^{\bullet \to \circ} + (t-1) \overline{\partial}^{\bullet \to \circ}_{\textbf{coll}}\right)}$ 
is proportional to the number of matching containing exactly $k$ edges of $\textbf{coll}$.

Since the $P_k$ are all non-negative, we have the following lemma.
\begin{lem} \label{lem:kernel_evals}
    All real eigenvalues of $\frac{\overline{\partial}^{\bullet \to \circ}_{\textbf{coll}}}{\overline{\partial}^{\bullet \to \circ}}$ lie between $[0,1]$. And all complex eigenvalues come in conjugate pairs.
\end{lem}
\begin{proof}
    Nonzero eigenvalues $\lambda_0$ correspond to roots $t_0$ with $P(t_0) = 0$ as $\lambda_0 = \frac{1}{1-t_0}$. Since $P(t)$ has all non-negative and some positive coefficients, real $t_0$ must have $t_0 \le 0$ and complex $t_0$ must have a conjugate pair. These statements imply the lemma.
\end{proof}

We showed above lemma was for finite graphs and finite $\textbf{coll}$, the same statement holds for infinite graphs and finite $\textbf{coll}$. In particular, since the analogous $P(t)$ for infinite graphs will be a limit of polynomials of finite graphs, $P(t)$ would have non-negative coefficients that add up to 1 again.

\subsection{Graph connections and twisting $\overline{\partial}$} 
\label{sec:graph_connection}
For a group $G$, a \textit{$G$-connection} (or \textit{$G$- gauge field} or \textit{$G$- local system}) on a graph is a function 
\begin{equation}
    U : \{ \text{directed edges of graph} \} \to G 
    \quad,\quad 
    \text{so that } 
    U(b \to w) = U(w \to b)^{-1}
    \text{ for all edges } {\bf e} = (b \leftrightarrow w).
\end{equation}
Two connections $U$, $\tilde{U}$ are said to be \textit{gauge equivalent} if there exists a function 
\begin{equation}
    {\bf g} : \{ \text{vertices of graph} \} \to G
    \quad,\quad 
    \text{so that } 
    \tilde{U}(b \to w) = {\bf g}(b) U(b \to w) {\bf g}(w)^{-1}
    \text{ for all edges } {\bf e} = (b \leftrightarrow w).
\end{equation}
Such a function ${\bf g}$ is called a \textit{gauge transformation}. Any gauge transformation (say, on a finite graph) can be written as a sequence of \textit{elementary gauge transformations} for which ${\bf g}(v) = g$ for a single vertex $v$ while $g(v')=1$ for all other vertices $v'$. 

In this paper, we will be concerning ourselves with the matrix-groups $G \subset GL(N,\C)$ that act on $\C^N$. 
In this case, a coordinate-free description of a $G$-connection can be given by assigning vector spaces isomorphic to $\C^N$ to each vertex and viewing the function $U$ as a set of isomorphisms between adjacent vertices' vector spaces with respect to some basis. Then, a gauge transformation corresponds to a change-of-basis at each vertex. 

Note that for any loop of vertices $\ell = v_1 \to v_2 \to \cdots v_k \to v_1$, the product $M(\ell) = U(v_1 \to v_2) \cdot U(v_2 \to v_3) \cdot \cdots U(v_k \to v_1)$ known as the \textit{monodromy matrix} of the loop changes by conjugation $M(\ell) \mapsto \tilde{M}(\ell) = {\bf g}(v_1) M(\ell) {\bf g}(v_1)^{-1}$. In general, two connections $U, U'$ are gauge equivalent iff for a base-point $v$ and any loop $\ell$ starting at $v$, the monodromy matrices $M(\ell),M'(\ell)$ coming from $U,U'$ are related as $M'(\ell) = g M(\ell) g^{-1}$ for some $g \in G$ independent of $\ell$. The field $U$ is \textit{flat} around a face if the monodromy around it is the identity matrix. For the infinite graphs of interest to us, we will consider the cases where the monodromy is flat except only nontrivial around a finite number of faces, referred to as \textit{punctures}.

Now, given a $GL(N,\C)$- gauge field $U$ and our Kasteleyn Matrix $\overline{\partial}$, we can form a \textit{twisted} Kasteleyn matrix $\overline{\partial}(U)$ that is $\C$-valued matrix with indices in $\{1,\cdots,N\} \times \{\text{vertices}\}$ or equivalently as a $GL(N,\C)$-valued matrix with indices in $\{\text{vertices}\}$ follows. The matrix $\overline{\partial}(U)$ is comprised of the parts $\overline{\partial}^{\bullet \to \circ}(U)$ and $\overline{\partial}^{\circ \to \bullet}(U)$ and can be written
\begin{equation}
\begin{split}
    \overline{\partial}^{\bullet \to \circ}(U) (b^{(i)},w^{(j)}) = \overline{\partial}^{\bullet \to \circ}(b,w) \cdot U(b \to w)_{i,j}
    \quad&\leftrightarrow\quad
    \overline{\partial}^{\bullet \to \circ}(U) (b,w) = \overline{\partial}^{\bullet \to \circ}(b,w) \cdot U(b \to w) 
    ,\quad \text{and}
    \\
    \overline{\partial}^{\circ \to \bullet}(U) (w^{(i)},b^{(j)}) = \overline{\partial}^{\circ \to \bullet}(w,b) \cdot U(w \to b)_{i,j}
    \quad&\leftrightarrow\quad
    \overline{\partial}^{\circ \to \bullet}(U) (w,b) = \overline{\partial}^{\circ \to \bullet}(w,b) \cdot U(w \to b).
\end{split}
\end{equation}
In the above, the $n$ vertices $v^{(1)}, \cdots, v^{(N)}$ correspond to the $N$ `copies' of the vertex $v$. 

\begin{figure}[h!]
  \centering
  \includegraphics[width=\linewidth]{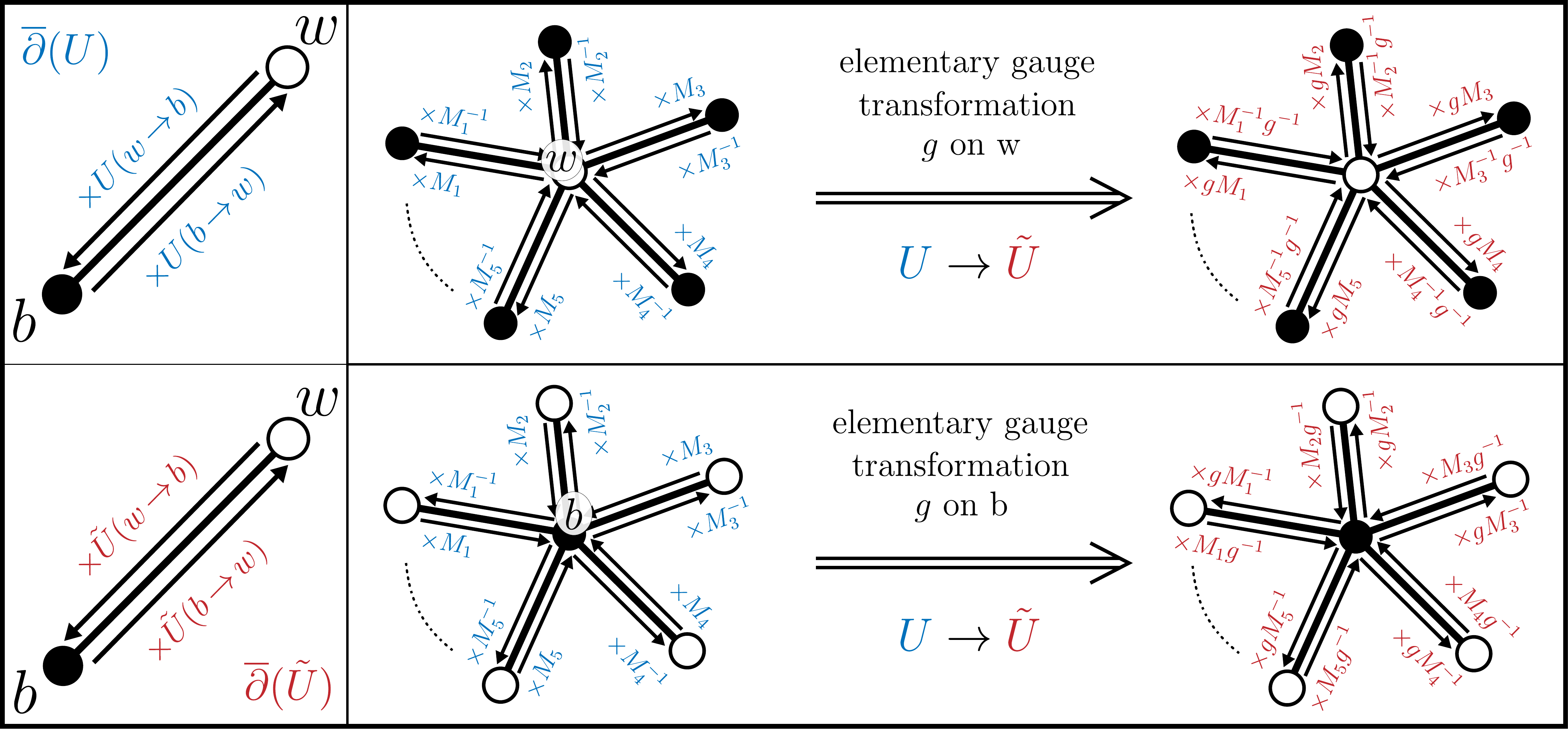}
  \caption{
    Effects of elementary gauge transformations on the twisted Kasteleyn matrix.
  }
  \label{fig:gauge_transformation}
\end{figure}

From the latter perspective, one can phrase a gauge transformation as a conjugation
\begin{equation}
    \overline{\partial}(\tilde{U})
    =
    \mathrm{diag}({\bf g}) \,
    \overline{\partial}(U) \,
    \mathrm{diag}({\bf g})^{-1}
    \quad,\quad
    \text{where }
    \mathrm{diag}({\bf g})(v,v') = {\bf g}(v) \delta_{v,v'}.
\end{equation}
As depicted in Fig.~\ref{fig:gauge_transformation}, an elementary gauge transformation $g$ on a vertex $v$ acts on the pieces $\overline{\partial}^{\bullet \to \circ}(U)$ and $\overline{\partial}^{\circ \to \bullet}(U)$ as
\begin{equation*}
\begin{split}
    &\overline{\partial}^{\bullet \to \circ}(\tilde{U})(b,w) 
    = 
    \overline{\partial}^{\bullet \to \circ}(U)(b,w) \cdot g^{-1}
    \quad,\quad
    \overline{\partial}^{\circ \to \bullet}(\tilde{U})(w,b) 
    = 
    g \cdot \overline{\partial}^{\circ \to \bullet}(U)(w,b)
    \quad,\quad
    \text{if } v = w
    \\
    &\overline{\partial}^{\bullet \to \circ}(\tilde{U})(b,w) 
    = 
    g \cdot \overline{\partial}^{\bullet \to \circ}(U)(b,w)
    \quad,\quad
    \overline{\partial}^{\circ \to \bullet}(\tilde{U})(w,b) 
    = 
    \overline{\partial}^{\circ \to \bullet}(U)(w,b) \cdot g^{-1}
    \quad,\quad
    \text{if } v = b
\end{split}
\end{equation*}
Note that this implies $\det(\overline{\partial})$ is \textit{gauge-invariant}, i.e. doesn't change upon a gauge transformation. However, the individual components do indeed change as
\begin{equation} \label{eq:gauge_covariance_kast_det}
    \frac
    {\det(\overline{\partial}^{\bullet \to \circ}(\tilde{U}))}
    {\det(\overline{\partial}^{\bullet \to \circ}(U))}
    =
    \begin{cases}
        \det(g)      &\text{ if } v \text{ black} \\
        \det(g)^{-1} &\text{ if } v \text{ white} 
    \end{cases}
    \quad,\quad
    \frac
    {\det(\overline{\partial}^{\circ \to \bullet}(\tilde{U}))}
    {\det(\overline{\partial}^{\circ \to \bullet}(U))}
    \cdot
    \begin{cases}
        \det(g)^{-1} &\text{ if } v \text{ black} \\
        \det(g)      &\text{ if } v \text{ white} 
    \end{cases}.
\end{equation}
So although $\det(\overline{\partial}^{\bullet \to \circ}(U))$ isn't gauge-invariant, if we restrict to groups $G \subset SL(N,\C)$, then we'll have $\det(g) = 1$ so that $\det(\overline{\partial}^{\bullet \to \circ}(U))$ would indeed be invariant. 
Also, note that if we restrict all edge matrices $U$ and gauge transformations to live in $U(N)$, the \textit{absolute-value} $\left|\det(\overline{\partial}^{\bullet \to \circ}(U))\right|$ will be gauge invariant.

In our setup with infinite graphs, we will only consider connections that are nontrivial in a finite region so all gauge transformations would indeed be connected by elementary gauge transformations. In particular, this means that the monodromy `around infinity' will be trivial.

\subsection{Height Functions of Dimer Matchings} \label{sec:height_functions}

Now we discuss a prototypical application of graph connections which is the study of height functions on bipartite planar graphs, for which we only need $G$ to be an abelian group. Given a \textit{reference} dimer matching, one can assign to dimer matching a \textit{height function} on the faces formed as follows. Orient the reference matching from white-to-black and orient the dimer matching from black-to-white. 
Then the combination of the matchings form a set of closed, directed loops on the graph. The loops define the contour lines of a height field $h$, where the value of $h$ at a specific face is fixed, say, to be zero. In particular, dualizing the directed edges clockwise (and ignoring doubled edges) defines a divergence-free flow on the dual graph that defines a function after fixing the value at a face. See Fig.~\ref{fig:height_function} for a depiction.  

\begin{figure}[h!]
  \centering
  \includegraphics[width=\linewidth]{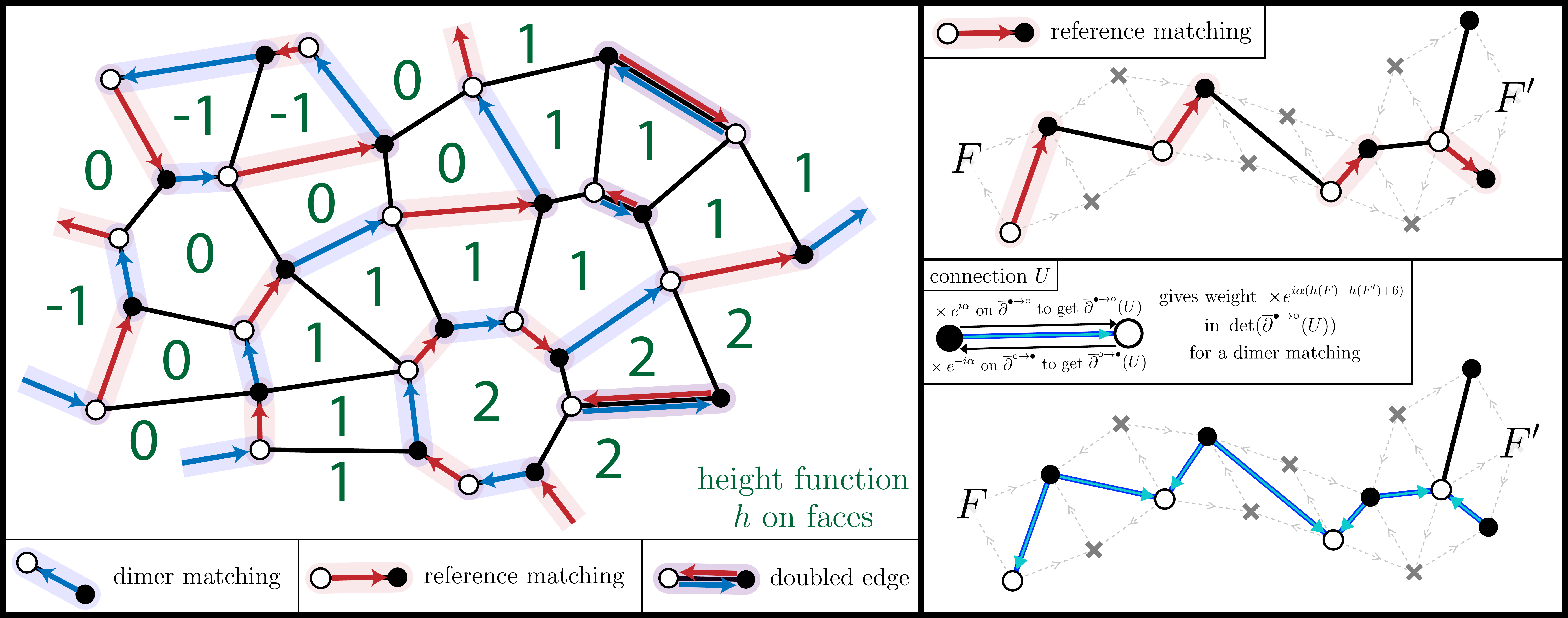}
  \caption{
    (Left) Given a reference matching and a dimer matching, we can construct a height function $h$ on the dual graph. 
    (Right) Adding a zipper connection of strength $e^{i\alpha}$ between two faces $F,F'$ modifies Eq.~\eqref{eq:det_of_kasteleyn} to count matchings with an additional weight of $e^{i\alpha(h(F)-h(F')-H_0)}$. In this case $H_0=6$.
  }
  \label{fig:height_function}
\end{figure}

As such, one may ask how to study correlation functions of differences in the height field $\{h(F) - h(F') | F,F' \in \{\text{faces}\}\}$.
One way to do this is to add a $\C^{\times}$-valued connection on the graph. For example, consider adding a connection $U$ of whose only monodromies are $e^{-i \alpha}$ counterclockwise around $F$ and $e^{i \alpha}$ counterclockwise around $F'$ and localize them on judiciously chosen edges so that either $U(b \to w) = e^{\pm i \alpha}$ or $U(b \to w) = 1$ for each edge and so that the dual of these edges form a path between $F \to F'$. Such a connection is referred to as a \textit{zipper}.
See Fig.~\ref{fig:height_function} for a depiction.
Then we can compute $\det(\overline{\partial}^{\bullet \to \circ}(U))$ using the formula Eq.~\eqref{eq:det_of_kasteleyn}, since one can think of the choice of connection modifying the weight function $\nu$ as adding a weight $\times \exp{(i\alpha(h(F)-h(F')-H_0))}$ to each dimer matching. Here, $H_0$ is some constant that depends on the specific choice of zipper. As such, if one can compute $\det(\overline{\partial}^{\bullet \to \circ}(U))$, taking moments with respect to $\alpha$ can give expectation values $\braket{(h(F)-h(F')-H_0)^k}$. Similarly, other height correlations involving faces $F_1,F_2,\cdots$ can be studied by adding multiple zippers of strengths $e^{i\alpha_1},e^{i\alpha_2},\cdots$ between various faces.

Note that the height function depends sensitively on the choice of the reference matching
\footnote{Note that there are more invariant definitions of height functions (see e.g.~\cite{KPWtreesAndMatchings1999}).}
and that different gauge-equivalent connections will change the value of $H_0$. This isn't surprising, since $e^{i\alpha}$ corresponds to a gauge field that's \textit{not} in $SL(1,\C) = \{1\}$, so gauge transformations of $U$ should be expected to change $\det( \overline{\partial}^{\bullet \to \circ} )(U)$. As such, studying height function correlations directly comes with the extra baggage of these microscopic choices.
However, restricting to to $U(1)$ connections and $U(1)$ weights leaves the absolute value $\left|\det(\overline{\partial}^{\bullet \to \circ}(U))\right|$ invariant, and the baggage of microscopic choices leads only to a non-universal \textit{phase} of the determinant. 

Similarly, we will choose to focus on computing determinants of $U(N)$-connections, which are invariant up to a non-universal phase.

\section{Two Punctures} \label{sec:singleZipper}

The first computation we do involves punctures at two far-away faces $F,F'$ on the isoradial graph and a connection $U$ flat away from these punctures. We compute the ratio of determinants
$\frac{\det \overline{\partial}^{\bullet \to \circ}(U)}{\det \overline{\partial}^{\bullet \to \circ}}$. 
Formally, both of these quantities are infinite and must be defined by taking the limit  of larger and larger finite graphs that encompass the gauge field configuration (which we stipulated were only nontrivial in a a finite region). Then after that, we'd need to take the limit of $|F - F'|$ becoming large. From our perspective, we won't need to deal with the technicality of taking larger and larger graphs, and we only worry about the latter limit. Furthermore, we will need to restrict our attention to isoradial graphs satisfying certain \textit{regularity conditions} as outlined in Appendix~\ref{app:asympt_gFunc}.

Since there are only two punctures and the monodromy at infinity is trivial, we have that the monodromies around the two punctures will be inverses of each other $M,M^{-1}$. By diagonalization, we can equate this determinant to a product of the determinants of $\C^\times$-valued connections, where the eigenvalues are in $\C^{\times}$. Even though these aren't gauge invariant individually, their products will be.
\footnote{
    Although, for eigenvalues that are unit complex numbers, the changes in partition function Eq.~\eqref{eq:gauge_covariance_kast_det} are only unit complex numbers if we restrict to unitary gauge transformations. With these restrictions in allowed gauge fields $U$, the absolute values 
    $\left|\frac{\det \overline{\partial}^{\bullet \to \circ}(U)}{\det \overline{\partial}^{\bullet \to \circ}}\right|$
    are well-defined.
    \label{foot:unitaryGaugeTransf}
} 
Say that we put a weight $e^{i \alpha}$ along a zipper connecting $F \to F'$ to form a twisted matrix $\overline{\partial}^{\bullet \to \circ}(e^{i\alpha})$. 
The ratio 
$\frac{\det \overline{\partial}^{\bullet \to \circ}(e^{i\alpha})}{\det \overline{\partial}^{\bullet \to \circ}}$ 
is related to expectation values $\braket{e^{i\alpha(h(F) - h(F'))}}$ for a suitable choice of height function (see Sec.~\ref{sec:height_functions}).
The expected convergence of height fluctuations to a Gaussian process matches
the expectation of convergences to the free fermion determinant of an abelian connection that's singular and flat away from the punctures. 
The discussion surrounding Eq.~\eqref{eq:abelian_singular_dirac_det} (see also Sec.~\ref{app:dirac_det_abelian}) gives the heuristics for such a determinant.
In particular, since there's a monodromy of $e^{i \alpha}$ around one puncture and $e^{-i\alpha}$ around the other, we would expect that the continuum limit satisfies
\begin{equation*}
    \log \frac{\det \overline{\partial}^{\bullet \to \circ}(e^{i\alpha})}{\det \overline{\partial}^{\bullet \to \circ}}
    \xrightarrow{|F' - F| \to \infty}
    -\frac{1}{2\pi^2} \alpha^2 \log(|F'-F|)
\end{equation*}
However, a subtlety is that the twisted Kasteleyn determinant must necessarily be periodic in $\alpha$ whereas $\alpha^2$ is \textit{not} periodic in $\alpha$. In~\cite{Dubedat2011}, Dubédat establishes that for $-\pi < \alpha < \pi$, this ratio does indeed converge as 
\begin{equation} \label{eq:singleZipper_result_intro}
    \log \left| \frac{\det \overline{\partial}^{\bullet \to \circ}(e^{i\alpha})}{\det \overline{\partial}^{\bullet \to \circ}} \right|
    \xrightarrow{|F' - F| \to \infty}
    -\frac{1}{2\pi^2} \alpha^2 \log(|F'-F|) + C(\alpha) + o(1)
\end{equation}
for some undetermined constant $C(\alpha)$ and that this convergence is uniform on compact subintervals of $(-\pi,\pi)$. This means that for generic \textit{unitary} $M \in U(N)$ whose eigenvalues are unit complex numbers $e^{i \alpha_1},\cdots,e^{i\alpha_N}$ with $-\pi < \alpha < \pi$, we'll get
\begin{equation} \label{eq:singleZipper_result_intro_2}
\begin{split}
    \log 
    \frac{\det \overline{\partial}^{\bullet \to \circ}(U)}{\det \overline{\partial}^{\bullet \to \circ}}
    \xrightarrow{|F' - F| \to \infty}
    &-\frac{1}{2\pi^2} (\alpha_1^2 + \cdots + \alpha_N^2) \log(|F'-F|) + C(\alpha_1) + \cdots + C(\alpha_N) + o(1)
    \\
    &+ \text{`imaginary part'}.
\end{split}
\end{equation}
In the above, there will generically be an \text{`imaginary part'} that is \textit{not universal} in the sense that it will depend on microscopic details of the graph and zipper (see Secs.~\ref{sec:real_im_parts_expansion},\ref{sec:non_universality_im_part}). 

In our analysis, we'll be able to explicitly identify the constant $C(\alpha)$ (see Sec.~\ref{sec:asympt_I2n_singleZip}) as 
\begin{equation*}
    C(\alpha) 
    = 
    -\frac{\alpha^2}{2\pi^2} ( 1  + \gamma_{\mathrm{Euler}} )
    -
    4 \sum_{n=2}^{\infty} \frac{1}{2n} \left(\frac{\alpha}{2\pi} \right)^{2n} \zeta(2n-1)
\end{equation*}
where $\zeta$ is the Riemann-zeta function. In particular, we are rigorously able to establish a generating function for the coefficients of $C(\alpha)$, where the coefficients have subleading corrections as $|F-F'| = D \to \infty$. 
Our analysis by itself does not establish that the series for $C(\alpha)$ including the subleading corrections actually converges. 
But, Dubédat's convergence result coupled with our calculation of moments means that $C(\alpha)$ is rigorously determined. 

Additionally, even though Dubédat's convergence results and our analysis focus on the eigenvalues $e^{i\alpha}$ of $M$ being unit complex numbers, we can analytically continue the moment generating function to complex $\alpha$ to get general $\C^{\times}$ eigenvalues, although we don't know how to analyze whether that sum converges. We expect that in general, this will lead to non-universal answers. In particular, the $\text{`imaginary part'}$ of the above answer will end up corresponding to odd powers of $\alpha$, which we expect not to be universal (see Sec.~\ref{sec:non_universality_im_part}). Analytically continuing $U(1)$ to $\C^\times$ will mix up the real and imaginary parts of the above expansion and thus mix up the universal and non-universal quantities.

In Sec.~\ref{sec:straightLine_zippers}, we establish a convenient microscopic choice of the zipper $F \to F'$. 
In Sec.~\ref{sec:clusterExp_singleZip}, we discuss a particular cluster expansion of the twisted determinant that is useful for this problem.
In Sec.~\ref{sec:exactExpr_zipperTraces}, we discuss exact formulas of the terms in the cluster expansion and intermediate exact formulas used throughout the paper. 
In Sec.~\ref{sec:asympt_I2n_singleZip}, we give asymptotic expressions of the terms in expansion.

\subsection{Straight-line Zippers} \label{sec:straightLine_zippers}

We'll be given a bipartite isoradial graph. Let $F,F'$ be vertices corresponding to faces of the graph, and overload $F,F'$ to also refer to the complex coordinate of the corresponding vertex on the rhombus graph. Define an angle $\theta_0$ so that $F' - F = D e^{i \theta_0}$ for some $D > 0$. Then, we have in general~\cite{BMS05} that there exists a path
\begin{equation}
    F =: v_0 \to v_1 \to v_2 \to v_3 \to \cdots \to v_{2 P - 3} \to v_{2 P - 2} \to v_{2 P - 1} \to v_{2 P} := F'
\end{equation}
of vertices $v_\cdots$ in the rhombus graph such that each
\begin{equation} \label{eq:boundedRhombusAngles}
    v_{j+1} - v_j = e^{i \phi_{j \to j+1}},
    \quad \text{where } 
    \theta_0 - \frac{\pi}{2} < \phi_{j \to j+1} < \theta_0 + \frac{\pi}{2}.
\end{equation}
Since the rhombus graph is bipartite (with black or white vertices of the isoradial graph being neighbors with the faces), each $v_{2 j} =: F_j$ defines a face $F_j$, so that $F_0 = F$ and $F_P=F'$. And, each $v_{2 j - 1}$ will be a black or white vertex, which we'll call either $w_{j}$ or $b_j$. See Fig.~\ref{fig:pathOnRhombusGraph}. Furthermore, since the path is finite we can constrain all the angles to lie in $(\theta_\text{min}(F_0 \to F_P),\theta_\text{max}(F_0 \to F_P))$ where $\theta_0 - \frac{\pi}{2} < \theta_\text{min} < \theta_\text{max} < \theta_0 + \frac{\pi}{2}$. We refer to these kinds of paths as \textit{straight-line paths} on the rhombus graph.

\begin{figure}[p!]
\begin{subfigure}
  \centering
  \includegraphics[width=\linewidth]{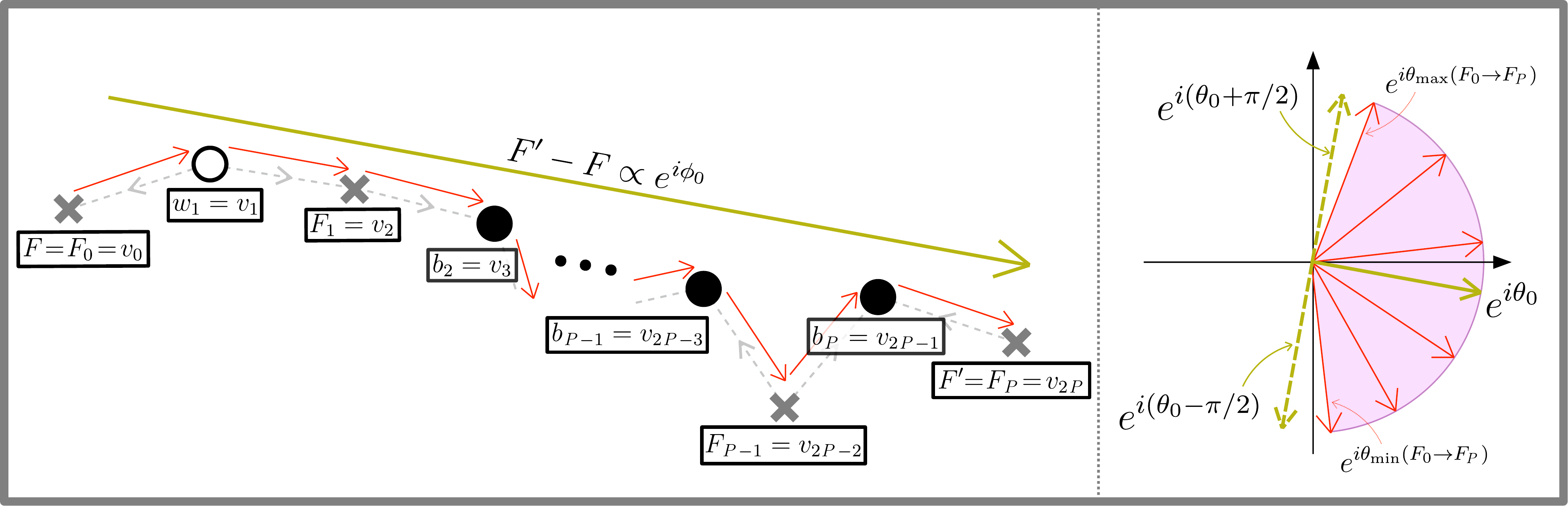}
  \caption{There exists a path between $F = F_0 = v_0 \to F' = F_P = v_{2P}$ where the angles of the path (red arrows) are constrained to be in the range $(\theta_0 - \frac{\pi}{2}, \theta_0 + \frac{\pi}{2})$ where $F' - F \propto e^{i \theta_0}$, as in the right side of the figure. By finiteness, one can constrain the angles to lie in $(\theta_\text{min}(F_0 \to F_P),\theta_\text{max}(F_0 \to F_P))$ where $\theta_0 - \frac{\pi}{2} < \theta_\text{min}(F_0 \to F_P) < \theta_\text{max}(F_0 \to F_P) < \theta_0 + \frac{\pi}{2}$.}
  \label{fig:pathOnRhombusGraph}
\end{subfigure}
\begin{subfigure}
  \centering
  \includegraphics[width=\linewidth]{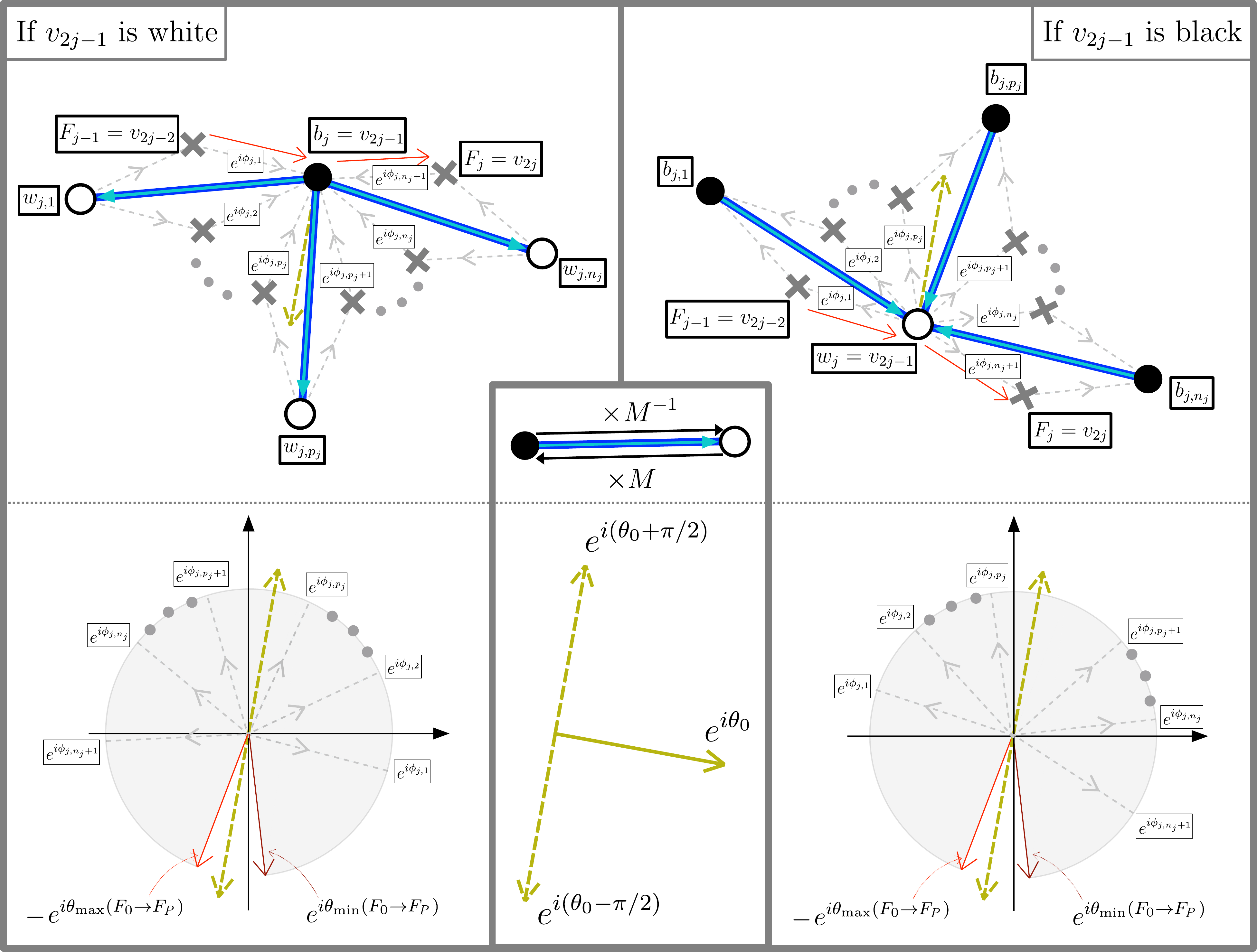}
  \caption{A connection on the isoradial graph associated to the path chosen in Fig.~\ref{fig:pathOnRhombusGraph}, whose edges weights to modify near $v_{2 j - 1}$ depend on whether $v_{2 j - 1}$ is black or white, as described in the main text. The angles $\phi_{j,k}$ that determine the orientation direction of the rhombus edges can be bounded between $\theta_\text{min}(F_0 \to F_P) < \phi_{j,k} < \theta_{max}(F_0 \to F_P) + \pi$.}
  \label{fig:pathOnRhombusGraph_connection}
\end{subfigure}
\end{figure}

Now we want to construct a zipper between $F$ and $F'$ by adding a connection on some edges of the graph along a straight-line path. This will be referred to as a \textit{straight-line zipper}. We will choose a connection with a transition matrix of $M \in U(N)$ across the zipper. In particular, we'll choose a connection so that the weights of some edges for $\overline{\partial}$ near each $v_{2 j - 1} = w_j \text{ or } b_j$ get their weights multiplied by $M^{\pm 1}$. The particular choice of which edge weights get multiplied will depend on whether $v_{2j - 1}$ is white or black. First if $v_{2j-1}=w_j$ is white, then $w_j$ has $n_j$ black neighbors $b_{j,1}, \cdots, b_{j,n_j}$ that lie on the side of the path $v_0 \to \cdots \to v_{2P}$ that points towards $e^{i(\theta_0 + \frac{\pi}{2})}$. And if $v_{2j-1}=b_j$ is black, then $b_j$ has $n_j$ white neighbors $w_{j,1}, \cdots, w_{j,n_j}$ that lie on the side of the path pointing towards $e^{i(\theta_0 - \frac{\pi}{2})}$. In both cases along the edges going around we multiply the  $M^{-1},M$ in the black$\to$white, white$\to$black direction. See Fig.~\ref{fig:pathOnRhombusGraph_connection}.

The reason we chose these specific edge weights is to constrain the directions of the oriented edges in the rhombus graph. Let's define $e^{i \phi_{j,1}}, \cdots, e^{i \phi_{j,n_j+1}}$ be the rhombus edge directions pointing away from $w_j$ going clockwise, and $e^{i \phi_{j,1}}, \cdots, e^{i \phi_{j,n_j+1}}$ be the rhombus edge directions pointing towards $b_j$ going counterclockwise. In particular, we'll have $e^{i \phi_{j,1}} = - e^{i \phi_{2j-2 \to 2j-1}}$ for $v_{2j-1}$ white and $e^{i \phi_{j,n_j+1}} = - e^{i \phi_{2j-2 \to 2j-1}}$ for $v_{2j-1}$ black, with $e^{i \phi_{j \to j+1}}$ defined in Eq.~\eqref{eq:boundedRhombusAngles}. 
Then we'll have that (see Fig.~\ref{fig:pathOnRhombusGraph_connection}) the angles can all bounded between $\theta_\text{min}(F_0 \to F_P) < \phi_{j,k} < \theta_{max}(F_0 \to F_P) + \pi$. In particular, this means that all of these angles will be bounded away from $\theta_0 - \frac{\pi}{2}$. 

\subsection{Cluster Expansion of the twisted determinant: single straight-line zipper} \label{sec:clusterExp_singleZip}

Given a straight-line zipper between $F \to F'$ as described in Sec.~\ref{sec:straightLine_zippers}, we want to expand the determinant in terms of the matrix $M \in U(N)$. Since we have only one matrix involved, we can study this problem in terms of the eigenvalues of $M$, or equivalently if we replace $M$ with a complex number $e^{i\alpha}$ in the construction. We will restrict our attention to real $\alpha \in (0,\pi)$, the other cases can be handled by complex conjugation and periodicity. Throughout the rest of this section, we will refer to $\overline{\partial}(e^{i\alpha})$ as the twisted Kasteleyn matrix.
We have
\begin{equation}
    \overline{\partial}(e^{i\alpha}) 
    = 
    \underbrace{
        \overline{\partial}^{\bullet \to \circ} + (e^{-i\alpha} - 1) \overline{\partial}_{\mathrm{zip}}^{\bullet \to \circ}
    }_{=: \overline{\partial}^{\bullet \to \circ}(e^{i\alpha}) }
    + 
    \underbrace{
        \overline{\partial}^{\circ \to \bullet} + (e^{i\alpha} - 1) \overline{\partial}_{\mathrm{zip}}^{\circ \to \bullet}
    }_{=:  \overline{\partial}^{\circ \to \bullet}(e^{i\alpha}) }
\end{equation}
where $\overline{\partial}_{\mathrm{zip}}^{\bullet \to \circ}$ and $\overline{\partial}_{\mathrm{zip}}^{\circ \to \bullet}$ are the restrictions of the $\overline{\partial}$ matrix to the zipper, for the black-to-white and white-to-black portions of the matrix respectively.

We restrict our attention to evaluating one half of the twisted matrix, namely
$\frac{\det \overline{\partial}^{\bullet \to \circ}(e^{i\alpha})}{\det \overline{\partial}^{\bullet \to \circ}}$. 
The logarithm of this ratio is more tractable to study using the identity $\tr \log = \log \det$. 
As such, we can expand
\begin{equation} \label{eq:logDetExpansion_singZip_1}
\begin{split}
    \log \frac{\det \overline{\partial}^{\bullet \to \circ}(e^{i\alpha})}{\det \overline{\partial}^{\bullet \to \circ}}
    &= 
    \tr \log 
    \left(
        \mathbbm{1} - (1 - e^{-i\alpha}) \frac{\overline{\partial}_{\mathrm{zip}}^{\bullet \to \circ}}{\overline{\partial}^{\bullet \to \circ}}
    \right) 
    = 
    -\sum_{n=1}^{\infty} \frac{1}{n} (1-e^{-i\alpha})^n 
    \tr \bigg[ \bigg( 
        \frac{\overline{\partial}_{\mathrm{zip}}^{\bullet \to \circ}}
        {\overline{\partial}^{\bullet \to \circ}}
    \bigg)^n \bigg]
\end{split}
\end{equation}
where 
$
\frac{\overline{\partial}_{\mathrm{zip}}^{\bullet \to \circ}}
    {\overline{\partial}^{\bullet \to \circ}} 
:= \overline{\partial}_{\mathrm{zip}}^{\bullet \to \circ} 
    \frac{1}{\overline{\partial}^{\bullet \to \circ}}
$. 

\subsubsection{Real and Imaginary Parts of Expansions} \label{sec:real_im_parts_expansion}
At this point, we will note that only the real parts of the series expansion are important in height correlations. By Lemma~\ref{lem:kernel_evals}, all traces of powers of 
$\frac{\overline{\partial}_{\mathrm{zip}}^{\bullet \to \circ}}
{\overline{\partial}^{\bullet \to \circ}} $
will be real. As such, the terms in the expansion corresponding to even powers of $\alpha$ will be real and those with odd powers of $\alpha$ will be imaginary.
The imaginary terms do not affect the absolute value of the determinant. 

To get the real part of the determinant, consider the equality,
\begin{equation} \label{eq:log_expansion_A}
\begin{split}
    & \frac{1}{2}\left(
        \log \left( 1 - (1-e^{-i \alpha}) y \right) + 
        \log \left( 1 - (1-e^{ i \alpha}) y \right) 
    \right)
    = 
    \frac{1}{2}\log \left(
        1 - (2 \sin(\alpha/2))^2 y(1-y)
    \right)
\end{split}
\end{equation}

This leads to the expansion
\begin{equation} \label{eq:singleZip_SeriesExpansion}
\begin{split}
    \tr \log \left( 1 - (1-e^{-i\alpha}) \frac{\overline{\partial}_{\mathrm{zip}}^{\bullet \to \circ}}
             {\overline{\partial}^{\bullet \to \circ}} \right) 
    =
    -\sum_{n=1}^{\infty} \frac{1}{2n}
    (2 \sin(\alpha/2))^{2n}
    \tr
    \left[\left(
        \frac{\overline{\partial}_{\mathrm{zip}}^{\bullet \to \circ}}
             {\overline{\partial}^{\bullet \to \circ}}
        \left(\mathbbm{1} -
        \frac{\overline{\partial}_{\mathrm{zip}}^{\bullet \to \circ}}
             {\overline{\partial}^{\bullet \to \circ}} 
        \right)
    \right)^n\right]
    + \text{`imaginary part'}.
\end{split}
\end{equation}
Note that since each $\frac{\overline{\partial}_{\mathrm{zip}}^{\bullet \to \circ}}{\overline{\partial}^{\bullet \to \circ}}$ is a finite-rank matrix, we'll have that each of 
$\tr \log \left( 1 - (1-e^{\pm i\alpha}) \frac{\overline{\partial}_{\mathrm{zip}}^{\bullet \to \circ}}
             {\overline{\partial}^{\bullet \to \circ}} \right) $
both converge absolutely for small enough $\alpha$, so both sides of the above will converge for small $\alpha$, thus their moments with respect to $\alpha$ will be given by moments with respect to the above.

It will turn out that evaluating these traces 
$
\tr
\left[\left(
    \frac{\overline{\partial}_{\mathrm{zip}}^{\bullet \to \circ}}
         {\overline{\partial}^{\bullet \to \circ}}
    \left(\mathbbm{1} -
    \frac{\overline{\partial}_{\mathrm{zip}}^{\bullet \to \circ}}
         {\overline{\partial}^{\bullet \to \circ}} 
    \right)
\right)^n\right]
$
are tractable, and the much of this section will be dedicated to evaluating these.

\subsection{Exact expressions for zipper traces} \label{sec:exactExpr_zipperTraces}

This section finds an exact expression for the traces
$
\tr
\left[\left(
    \frac{\overline{\partial}_{\mathrm{zip}}^{\bullet \to \circ}}
         {\overline{\partial}^{\bullet \to \circ}}
    \left(\mathbbm{1} -
    \frac{\overline{\partial}_{\mathrm{zip}}^{\bullet \to \circ}}
         {\overline{\partial}^{\bullet \to \circ}} 
    \right)
\right)^n\right]
$
that are part of Eq.~\eqref{eq:singleZip_SeriesExpansion}, culminating in Proposition~\ref{prop:traceId_second}. However, the intermediate steps to the proposition and the reasoning behind them are also at the heart of the multiple zipper calculations in Sec.~\ref{sec:multipleZipper_calculations}. 

The following equality is central.
\begin{lem} \label{lem:nEdges_zipIdLemma}
Let $n > 1$. Let ${\bf e}_1 = (b_1,w_1) , \cdots , {\bf e}_{n}= (b_{n},w_{n})$ be $n$ edges in the zipper. Define $F^{L}_{{\bf e}_i}, F^{R}_{{\bf e}_i}$ to be the rhombus graph vertices corresponding to the two dual faces on the dual edge of each ${\bf e}_i$ on the left,right side of $w_i \to b_i$, as in Fig.~\ref{fig:nEdges_zipIdLemma}. Then for \textit{any} vertex or dual vertex $\Tilde{v}_1, \cdots, \Tilde{v}_n$
\begin{equation} \label{eq:nEdges_zipIdLemma}
\begin{split}
    &\overline{\partial}_{\mathrm{zip}}^{\bullet \to \circ}(b_1,w_1) \cdot g_{w_1 \to b_2}(z_2) \cdot \overline{\partial}_{\mathrm{zip}}^{\bullet \to \circ}(b_2,w_2) \cdot g_{w_2 \to b_3}(z_3) \cdot \cdots \cdot \overline{\partial}_{\mathrm{zip}}^{\bullet \to \circ}(b_{n},w_{n}) \cdot g_{w_{n} \to b_1}(z_{1}) 
    \\
    =&
    \left( \frac{g_{\Tilde{v}_1 \to F^L_{1}}(z_1)}{g_{\Tilde{v}_2\to F^L_{1}}(z_2)} - \frac{g_{\Tilde{v}_1 \to F^R_{1}}(z_1)}{g_{\Tilde{v}_2 \to F^R_{1}}(z_2)} \right)
    \left( \frac{g_{\Tilde{v}_2 \to F^L_{2}}(z_2)}{g_{\Tilde{v}_3 \to F^L_{2}}(z_3)} - \frac{g_{\Tilde{v}_2 \to F^R_{2}}(z_2)}{g_{\Tilde{v}_3 \to F^R_{1}}(z_3)} \right)
    \cdots
    \left( \frac{g_{\Tilde{v}_n \to F^L_{n}}(z_{n})}{g_{\Tilde{v}_1 \to F^L_{2}}(z_{1})} - \frac{g_{\Tilde{v}_n \to F^R_{n}}(z_{n})}{g_{\Tilde{v}_1 \to F^R_{1}}(z_{1})} \right) 
    \\
    &\times 
    \frac{(-i)^n}{(z_1 - z_2) (z_2 - z_3) \cdots (z_{n} - z_{1})}.
\end{split}
\end{equation}
If $n=1$, then the expression above should be taken to modify $\frac{1}{z_1 - z_1}$ with an appropriate limit which defines a derivative
\begin{equation} \label{eq:nEdges_zipIdLemma_nEq1}
\begin{split}
    &\overline{\partial}_{\mathrm{zip}}^{\bullet \to \circ}(b_1,w_1) \cdot g_{w_1 \to b_1}(z_1)
    = `` \frac{-i}{z_1 - z_1}
    \left( \frac{g_{\Tilde{v}_1 \to F^L_{1}}(z_1)}{g_{\Tilde{v}_1 \to F^L_{1}}(z_1)} - \frac{g_{\Tilde{v}_1 \to F^R_{1}}(z_1)}{g_{\Tilde{v}_1 \to F^R_{1}}(z_1)} \right) " 
    = -i \left( \frac{\frac{d}{d z_1} g_{\Tilde{v}_1 \to F^L_{1}}(z_1)}{g_{\Tilde{v}_1 \to F^L_{1}}(z_1)} - \frac{\frac{d}{d z_1} g_{\Tilde{v}_1 \to F^R_{1}}(z_1)}{g_{\Tilde{v}_1 \to F^R_{1}}(z_1)} \right)
\end{split}
\end{equation}
\end{lem}
\begin{figure}[h!]
  \centering
  \includegraphics[width=0.8\linewidth]{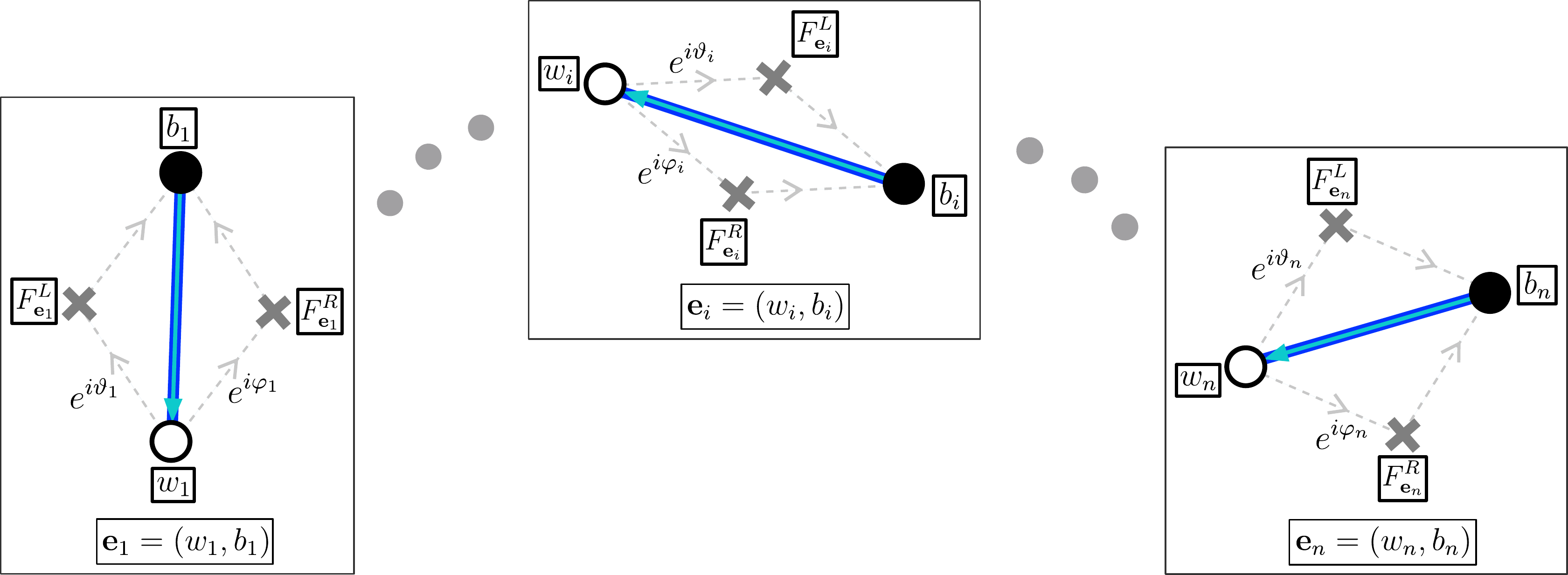}
  \caption{Given edges ${\bf e}_1,\cdots,{\bf e}_n$ along the zipper, the dual faces $F^{L}_{{\bf e}_i}, F^{R}_{{\bf e}_i}$ are on the left, right sides of the edge $w_i \to b_i$. Also, define $e^{i \varphi_i} = F_{{\bf e}_i}^R - w_i$ and $e^{i \vartheta_i} = F_{{\bf e}_i}^L - w_i$ to be the rhombus vectors along each edge.}
  \label{fig:nEdges_zipIdLemma}
\end{figure}

\begin{proof}
    We have that
    \begin{equation*}
        \overline{\partial}_{\mathrm{zip}}^{\bullet \to \circ}(b_i,w_i) = i(e^{i \vartheta_i}-e^{i \varphi_i})
    \end{equation*}
    where $e^{i \varphi_1} = F_{{\bf e}_i}^L - w_i , e^{i \vartheta_1} = F_{{\bf e}_i}^R - w_i$ as in Fig.~\ref{fig:nEdges_zipIdLemma}. Also, we have that each
    \begin{equation*}
        g_{w_i \to b_{i+1}}(z_{i+1}) = g_{b_i \to b_{i+1}}(z_{i+1}) \frac{1}{(z_{i+1}-e^{i \vartheta_i})(z_{i+1}-e^{i \varphi_i})}
    \end{equation*}
    where $b_{n+1}:=b_1$. 
    
    Let's start with the $n=1$ case. We'll have
    \begin{equation*}
    \begin{split}
        \overline{\partial}_{\mathrm{zip}}^{\bullet \to \circ}(b_1,w_1) \cdot g_{w_1 \to b_1}(z_1) 
        &= i (e^{i \vartheta_i}-e^{i \varphi_i}) \frac{1}{(z_{1}-e^{i \varphi_1})(z_{1}-e^{i \vartheta_1})} 
         = i \bigg( \frac{1}{z_{1}-e^{i \vartheta_1}} - \frac{1}{z_{1}-e^{i \varphi_1}} \bigg)
    \end{split}
    \end{equation*}
    Now, note that for any $\Tilde{v}_1$, we can write for some $k$, $\{\zeta_j\}$, $\{\pm_j\}$
    \begin{equation*}
        g_{\Tilde{v}_1 \to b_1}(z_1) = \prod_{j=1}^k (z_1 - e^{i \zeta_j})^{\pm_j 1}.
    \end{equation*}
    Also, we'd have 
    \begin{equation*}
        g_{\Tilde{v}_1 \to F_{{\bf e}_1}^L}(z_1) = g_{\Tilde{v}_1 \to b_1}(z_1) (z_1 - e^{i \varphi_1})
        \quad,\quad
        g_{\Tilde{v}_1 \to F_{{\bf e}_1}^R}(z_1) = g_{\Tilde{v}_1 \to b_1}(z_1) (z_1 - e^{i \vartheta_1}).
    \end{equation*}
    These facts together 
    imply 
    \begin{equation*}
    \begin{split}
        &\overline{\partial}_{\mathrm{zip}}^{\bullet \to \circ}(b_1,w_1) 
        \cdot
        g_{w_1 \to b_1}(z) = i \bigg( \frac{1}{z_{1}-e^{i \vartheta_1}} - \frac{1}{z_{1}-e^{i \varphi_1}} \bigg)
        \\
        &= -i \left(
                \bigg( \frac{1}{z_1 - e^{i \varphi_1}}   + \sum_{j=1}^k \pm_j \frac{1}{z_1 - e^{i \zeta_j}} \bigg) 
                - 
                \bigg( \frac{1}{z_1 - e^{i \vartheta_1}} + \sum_{j=1}^k \pm_j \frac{1}{z_1 - e^{i \zeta_j}} \bigg)
            \right) 
        = -i \left( 
                \frac{\frac{d}{d z_1} g_{\Tilde{v}_1 \to F^L_{1}}(z_1)}{g_{\Tilde{v}_1 \to F^R_{1}}(z_1)} 
                -
                \frac{\frac{d}{d z_1} g_{\Tilde{v}_1 \to F^R_{1}}(z_1)}{g_{\Tilde{v}_1 \to F^L_{1}}(z_1)}
             \right)
    \end{split}
    \end{equation*}
    the result for $n=1$.
    
    For $n>1$, we have that
    \begin{equation*}
    \begin{split}
        &\overline{\partial}_{\mathrm{zip}}^{\bullet \to \circ}(b_1,w_1) \cdot g_{w_1 \to b_2}(z_2) \cdot \overline{\partial}_{\mathrm{zip}}^{\bullet \to \circ}(b_2,w_2) \cdot g_{w_2 \to b_3}(z_3) \cdot \cdots \cdot \overline{\partial}_{\mathrm{zip}}^{\bullet \to \circ}(b_{n},w_{n}) \cdot g_{w_{n} \to b_1}(z_{1}) 
        \\
        &= 
        \bigg[ i(e^{i \vartheta_1}-e^{i \varphi_1}) \frac{1}{(z_{2}-e^{i \vartheta_1})(z_{2}-e^{i \varphi_1})} g_{b_1 \to b_{2}}(z_{2}) \bigg] 
        \cdots 
        \bigg[ i(e^{i \vartheta_n}-e^{i \varphi_n}) \frac{1}{(z_{1}-e^{i \vartheta_n})(z_{1}-e^{i \varphi_n})} g_{b_n \to b_{1}}(z_{1}) \bigg] 
        \\
        &= 
        (-i)^n g_{b_1 \to b_{2}}(z_{2}) \cdots g_{b_n \to b_{1}}(z_{1}) 
        \bigg[ \frac{1}{z_1 - e^{i \varphi_n}} - \frac{1}{z_1 - e^{i \vartheta_n}} \bigg] \cdots 
        \bigg[ \frac{1}{z_n - e^{i \varphi_{n-1}}} - \frac{1}{z_n - e^{i \vartheta_{n-1}}} \bigg] 
        \\
        &=
        (-i)^n 
        \frac{g_{\Tilde{v}_2 \to b_{2}}(z_{2})}{g_{\Tilde{v}_2 \to b_{1}}(z_{2})} 
        \cdots 
        \frac{g_{\Tilde{v}_1 \to b_{1}}(z_{1})}{g_{\Tilde{v}_1 \to b_{n}}(z_{1})} 
        \bigg[
            \frac{1}{z_n-z_1}
            \bigg( 
                \frac{z_n - e^{i \varphi_n}}{z_1 - e^{i \varphi_n}} 
                - 
                \frac{z_n - e^{i \vartheta_n}}{z_1 - e^{i \vartheta_n}} 
            \bigg)
        \bigg] 
        \cdots 
        \bigg[
            \frac{1}{z_{n-1}-z_n}
            \bigg( 
                \frac{z_{n-1} - e^{i \varphi_{n-1}}}{z_n - e^{i \varphi_{n-1}}} 
                - 
                \frac{z_{n-1} - e^{i \vartheta_{n-1}}}{z_n - e^{i \vartheta_{n-1}}} 
            \bigg)
        \bigg] 
        \\
        &= 
        \frac{(-i)^n}{(z_1-z_2) \cdots (z_n-z_1)} 
        \bigg(
            \frac{g_{\Tilde{v}_1 \to F^L_{{\bf e}_1}}(z_1)}{g_{\Tilde{v}_2 \to F^L_{{\bf e}_1}}(z_2)} 
            - 
            \frac{g_{\Tilde{v}_1 \to F^R_{{\bf e}_1}}(z_1)}{g_{\Tilde{v}_2 \to F^R_{{\bf e}_1}}(z_2)}
        \bigg) 
        \cdots 
        \bigg(
            \frac{g_{\Tilde{v}_n \to F^L_{{\bf e}_n}}(z_n)}{g_{\Tilde{v}_1 \to F^L_{{\bf e}_n}}(z_1)}
            - 
            \frac{g_{\Tilde{v}_n \to F^R_{{\bf e}_n}}(z_n)}{g_{\Tilde{v}_1 \to F^R_{{\bf e}_n}}(z_1)} 
        \bigg)
    \end{split}
    \end{equation*}
\end{proof}

A corollary of the above lemma is an explicit integral expression for $\tr \bigg[ \bigg( \frac{\overline{\partial}_{\mathrm{zip}}^{\bullet \to \circ}}{\overline{\partial}^{\bullet \to \circ}} \bigg)^n \bigg]$.

\begin{cor} \label{cor:traceId_first}
    \begin{equation} \label{eq:traceId_first}
        \tr \bigg[ \bigg( 
            \frac{\overline{\partial}_{\mathrm{zip}}^{\bullet \to \circ}}{\overline{\partial}^{\bullet \to \circ}} 
        \bigg)^n \bigg] = 
        \frac{1}{(2 \pi i)^n} 
        \int_0^{-i e^{i \theta_0}\infty} dz_1 
        \cdots 
        \int_0^{-i e^{i \theta_0} \infty} dz_n 
        \frac{1}{(z_1-z_2) \cdots (z_n-z_1)} 
        \bigg(1 - \frac{g_{F \to F'}(z_1)}{g_{F \to F'}(z_2)}\bigg) 
        \cdots
        \bigg(1 - \frac{g_{F \to F'}(z_n)}{g_{F \to F'}(z_1)}\bigg)
    \end{equation}
\end{cor}
\begin{proof}
    By the inverse Kasteleyn formula Eq.~\eqref{eq:invKast_contour_formula} and the discussion of rhombus angles in the zipper in Section~\ref{sec:straightLine_zippers} and Fig.~\ref{fig:pathOnRhombusGraph_connection}, we have that $0 \to -i e^{i \theta_0} \infty$ is a valid integration contour of $\frac{1}{\overline{\partial}^{\bullet \to \circ}}(w,b)$ for all $w,b$ part of the zipper. So,
    \begin{equation*}
        \frac{1}{\overline{\partial}^{\bullet \to \circ}}(w,b) = \frac{1}{2 \pi} \int_{0}^{-i e^{i \theta_0} \infty} dz g_{w \to b}(z).
    \end{equation*}
    
    Also, let us label the edges in the zipper from $F \to F'$ as ${\bf E}_1 \cdots {\bf E}_Q$ in order, so that the dual edges $\{{\bf E}_j^{\vee}\}$ in the order ${\bf E}_1^{\vee} \to \cdots \to {\bf E}_Q^{\vee}$ give the path $F \to F'$ on the dual graph. Note that $F^L_{{\bf E}_1} = F$, $F^R_{{\bf E}_Q} = F'$, and $F^R_{{\bf E}_j} = F^L_{{\bf E}_{j+1}}$ for all $j$. As such, summing over all edges of the right-hand-side of Eq.~\eqref{eq:nEdges_zipIdLemma} is a telescoping sum, that if we choose all $\Tilde{v}_i = F$
    \begin{equation*}
    \begin{split}
        \sum_{ \substack{ {\bf e}_1 \cdots {\bf e}_n \\ \in \\ \{ {\bf E}_1 \cdots {\bf E}_Q \} } }
        \left( 
            \frac{g_{F \to F^L_{{\bf e}_1}}(z_1)}{g_{F \to F^L_{{\bf e}_1}}(z_2)} 
            - 
            \frac{g_{F \to F^R_{{\bf e}_1}}(z_1)}{g_{F\to F^R_{{\bf e}_1}}(z_2)} 
        \right)
        \cdots
        \left( 
            \frac{g_{F \to F^L_{{\bf e}_n}}(z_{n})}{g_{F \to F^L_{{\bf e}_n}}(z_{1})} 
            - 
            \frac{g_{F \to F^R_{{\bf e}_n}}(z_{n})}{g_{F \to F^R_{{\bf e}_n}}(z_{1})} 
        \right) = 
        \left(
            1 - \frac{g_{F \to F'}(z_1)}{g_{F \to F'}(z_2)}
        \right) 
        \cdots 
        \left(
            1 - \frac{g_{F \to F'}(z_n)}{g_{F \to F'}(z_1)}
        \right) 
    \end{split}
    \end{equation*}
    because all terms cancel except for the $F^L_{{\bf E}_1}=F$ and $F^R_{{\bf E}_Q} = F'$ terms. This implies the Eq.~\eqref{eq:traceId_first}.
\end{proof}

Note that in Corollary~\ref{cor:traceId_first}, all the integration contours are the same, $0 \to -i e^{i \theta_0}$. Even though the $\frac{1}{z_{i} - z_{i+1}}$ terms becomes singular, the full integrand is never singular because  $\left( 1-\frac{g_{F \to F'}(z_i)}{g_{F \to F'}(z_{i+1})} \right)$ is a rational function with a zero at $z_i = z_{i+1}$.

There's a similar statement we can show.
\begin{prop} \label{prop:traceId_second}
    \begin{equation} \label{eq:traceId_second}
    \begin{split}
    \tr &\bigg[ \Bigg( 
        \bigg( \frac{\overline{\partial}_{\mathrm{zip}}^{\bullet \to \circ}}{\overline{\partial}^{\bullet \to \circ}} \bigg)
        \bigg( \mathbbm{1} - \frac{\overline{\partial}_{\mathrm{zip}}^{\bullet \to \circ}}{\overline{\partial}^{\bullet \to \circ}} \bigg)
    \Bigg)^n \bigg] =: \mathcal{I}_{2n} 
    \\
    &= 
    (-1)^n \tr \bigg[ \bigg(
        \frac{\overline{\partial}_{\mathrm{zip}}^{\bullet \to \circ}}{\overline{\partial}^{\bullet \to \circ}}
    \bigg)^{2n} \bigg] 
    - 
    (-1)^n {n \choose 1} \tr \bigg[ \bigg(
        \frac{\overline{\partial}_{\mathrm{zip}}^{\bullet \to \circ}}{\overline{\partial}^{\bullet \to \circ}}
    \bigg)^{2n-1} \bigg] 
    + 
    (-1)^n {n \choose 2} \tr \bigg[ \bigg(
        \frac{\overline{\partial}_{\mathrm{zip}}^{\bullet \to \circ}}{\overline{\partial}^{\bullet \to \circ}}
    \bigg)^{2n-2} \bigg] 
    \cdots
    + 
    {n \choose n} \tr \bigg[ \bigg(
        \frac{\overline{\partial}_{\mathrm{zip}}^{\bullet \to \circ}}{\overline{\partial}^{\bullet \to \circ}}
    \bigg)^{2n-n} \bigg] 
    \\
    &= 
    \frac{1}{(2 \pi)^{2n}} 
    \int_0^{-i e^{i \theta_0} \infty} dz_1 
    \cdots 
    \int_0^{i e^{i \theta_0} \infty} dz_{2n} 
    \underbrace{
        \frac{1}{(z_1-z_2) \cdots (z_{2n}-z_1)} 
        \bigg(1 - \frac{g_{F \to F'}(z_1)}{g_{F \to F'}(z_2)}\bigg) 
        \cdots 
        \bigg(1 - \frac{g_{F \to F'}(z_{2n})}{g_{F \to F'}(z_1)}\bigg) 
    }_{=: I_{2n}(z_1,\cdots,z_{2n})}
    \end{split}
    \end{equation}
    where the integration contours for $z_1, \cdots, z_{2n}$ alternate between the contours $z_j: 0 \to -i e^{i \theta_0} \infty$ for $j$ odd and $z_j: 0 \to i e^{i \theta_0} \infty$ for $j$ even. 
\end{prop}
In the above, we defined the full integral to equal $\mathcal{I}_{2n}$ and define the integrand to be $I_{2n}(z_1,\cdots,z_{2n})$.

\begin{proof}
This formula relies on the following formula, valid for any $k > 1$:
\begin{equation} \label{eq:traceId_second_pf_1}
\begin{split}
    &\frac{1}{(2 \pi i)^k} 
    \int_0^{i e^{i \theta_0} \infty} dz_1 
    \int_0^{\pm_2 i e^{i \theta_0}\infty} dz_2  
    \cdots 
    \int_0^{\pm_k i e^{i \theta_0} \infty} dz_{k} 
    \frac{1}{(z_1-z_2) \cdots (z_k-z_1)} 
    \bigg(1 - \frac{g_{F \to F'}(z_1)}{g_{F \to F'}(z_2)}\bigg) 
    \cdots 
    \bigg(1 - \frac{g_{F \to F'}(z_k)}{g_{F \to F'}(z_1)}\bigg) 
    \\
    =& 
    \frac{1}{(2 \pi i)^k} 
    \int_0^{-i e^{i \theta_0} \infty} dz_1 
    \int_0^{\pm_2 i e^{i \theta_0}\infty} dz_2  
    \cdots 
    \int_0^{\pm_k i e^{i \theta_0} \infty} dz_{k} 
    \frac{1}{(z_1-z_2) \cdots (z_k-z_1)} 
    \bigg(1 - \frac{g_{F \to F'}(z_1)}{g_{F \to F'}(z_2)}\bigg) 
    \cdots 
    \bigg(1 - \frac{g_{F \to F'}(z_k)}{g_{F \to F'}(z_1)}\bigg) 
    \\
    &- 
    \frac{1}{(2 \pi i)^{k-1}} 
    \int_0^{\pm_2 i e^{i \theta_0}\infty} dz_2  
    \cdots 
    \int_0^{\pm_k i e^{i \theta_0} \infty} dz_{k} 
    \frac{1}{(z_2-z_3) \cdots (z_k-z_2)} 
    \bigg(1 - \frac{g_{F \to F'}(z_2)}{g_{F \to F'}(z_3)}\bigg) 
    \cdots 
    \bigg(1 - \frac{g_{F \to F'}(z_k)}{g_{F \to F'}(z_2)}\bigg).
\end{split}
\end{equation}

In the above, the contours of $z_i$ for $j = 2,\cdots,k$ are either one of $0 \to \pm_j i \infty$, fixed on both sides of the equation. 
Because the integrand is cyclically symmetric, the above Eq.~\eqref{eq:traceId_second_pf_1} implies the corollary by changing the $z_{2 \ell}: 0 \to i e^{i \theta_0} \infty$ to $z_{2 \ell}: 0 \to -i e^{i \theta_0} \infty$ one at a time. 
In particular, the ${n \choose j}$ collections of changing $j$ of these contours correspond to the coefficient ${n \choose j}$ of $(-1)^j {n \choose j} \tr \bigg[ \bigg( \frac{\overline{\partial}_{\mathrm{zip}}^{\bullet \to \circ}}{\overline{\partial}^{\bullet \to \circ}} \bigg)^{2n-j} \bigg]$. 
Note that it also implies that integral is independent the \textit{order} of the contours $z_j: 0 \to \pm_j i \infty$ (i.e. only depends on the numbers $\#\{\pm_j=+\}$, $\#\{\pm_j=-\}$), since it would mean

\begin{equation*}
\begin{split}
    &\frac{1}{(2 \pi i)^k} 
    \int_0^{+i e^{i \theta_0} \infty} dz_1 
    \int_0^{-i e^{i \theta_0} \infty} dz_2  
    \int_0^{\pm_3 i e^{i \theta_0}\infty} dz_3  
    \cdots 
    \frac{1}{(z_1-z_2) \cdots (z_k-z_1)}
    \bigg(1 - \frac{g_{F \to F'}(z_1)}{g_{F \to F'}(z_2)}\bigg) 
    \cdots 
    \bigg(1 - \frac{g_{F \to F'}(z_k)}{g_{F \to F'}(z_1)}\bigg) 
    \\
    =& 
    \frac{1}{(2 \pi i)^k} 
    \int_0^{-i e^{i \theta_0} \infty} dz_1 
    \int_0^{-i e^{i \theta_0} \infty} dz_2  
    \int_0^{\pm_3 i e^{i \theta_0}\infty} dz_3 
    \cdots 
    \frac{1}{(z_1-z_2) \cdots (z_k-z_1)} 
    \bigg(1 - \frac{g_{F \to F'}(z_1)}{g_{F \to F'}(z_2)}\bigg) 
    \cdots 
    \bigg(1 - \frac{g_{F \to F'}(z_k)}{g_{F \to F'}(z_1)}\bigg) 
    \\
    &- 
    \frac{1}{(2 \pi i)^{k-1}} 
    \int_0^{-i e^{i \theta_0} \infty} dz_2  
    \int_0^{\pm_3 i e^{i \theta_0}\infty} dz_3 
    \cdots \frac{1}{(z_2-z_3) \cdots (z_k-z_2)} 
    \bigg(1 - \frac{g_{F \to F'}(z_2)}{g_{F \to F'}(z_3)}\bigg) 
    \cdots 
    \bigg(1 - \frac{g_{F \to F'}(z_k)}{g_{F \to F'}(z_2)}\bigg) 
    \\
    =& 
    \frac{1}{(2 \pi i)^k} 
    \int_0^{-i e^{i \theta_0} \infty} dz_1 
    \int_0^{+i e^{i \theta_0} \infty} dz_2 
    \int_0^{\pm_3 i e^{i \theta_0}\infty} dz_3 
    \cdots \frac{1}{(z_1-z_2) 
    \cdots (z_k-z_1)} 
    \bigg(1 - \frac{g_{F \to F'}(z_1)}{g_{F \to F'}(z_2)}\bigg) 
    \cdots 
    \bigg(1 - \frac{g_{F \to F'}(z_k)}{g_{F \to F'}(z_1)}\bigg).
\end{split}
\end{equation*}

Now we prove Eq.~\eqref{eq:traceId_second}. Writing $g_{F \to F'}(z) = \prod_{j=1}^{P} \frac{z-e^{i \alpha_j}}{z-e^{i \beta_j}}$, we have that the poles $\{e^{i \beta_j}\}$ of $g_{F \to F'}(z_1)$ are of the form $\theta_0 - \frac{\pi}{2} < \beta_j < \theta_0 + \frac{\pi}{2}\}$, whereas all of the poles $\{e^{i \alpha_j}\}$ of $\frac{1}{g_{F \to F'}(z_1)}$ are of the form $\theta_0 - 3\frac{\pi}{2} < \alpha_j < \theta_0 - \frac{\pi}{2}\}$.  As such, the strategy is to split the product involving $z_1$ into two terms 
\begin{equation*}
    \cdots \frac{\bigg(1 - \frac{g_{F \to F'}(z_1)}{g_{F \to F'}(z_2)}\bigg) \bigg(1 - \frac{g_{F \to F'}(z_k)}{g_{F \to F'}(z_1)}\bigg)}{(z_k-z_1)(z_1-z_2)} \cdots = \cdots \underbrace{\frac{\bigg(1 - \frac{g_{F \to F'}(z_1)}{g_{F \to F'}(z_2)}\bigg)}{(z_k-z_1)(z_1-z_2)}}_{\text{term 1}} \cdots - \cdots \underbrace{\frac{\bigg(\frac{g_{F \to F'}(z_k)}{g_{F \to F'}(z_1)} - \frac{g_{F \to F'}(z_k)}{g_{F \to F'}(z_2)}\bigg)}{(z_k-z_1)(z_1-z_2)}}_{\text{term 2}} \cdots.
\end{equation*}
Then, for each of each of the terms one can shift the contours in such a way to avoid the poles $\{e^{i \alpha}\}$ or $\{e^{i \beta}\}$. In particular one would shift clockwise for `term 1' and counterclockwise for `term 2' since `term 1' has poles at $\{z_1 = e^{i \alpha}\}$ and `term 2' has poles at $\{z_1 = e^{i \beta}\}$. However, this procedure has an issue in that each of the `term 1' and `term 2' have a pole at $z_1 = z_k$. (Note that there's no pole at $z_1 = z_2$.) So there's an issue that the integrals of these individual terms will be undefined on $z_1: 0 \to \pm_k i \infty$ on whichever part of the equation that is. So to rectify this, one first needs to regulate the integration contours so that the $z_1$ contour is separated from the others, which must be done before spliting up the term. See Fig.~\ref{fig:rotatingContours_singZip} for a visual aid for the following computation.

\begin{figure}[h!]
  \centering
  \includegraphics[width=0.8\linewidth]{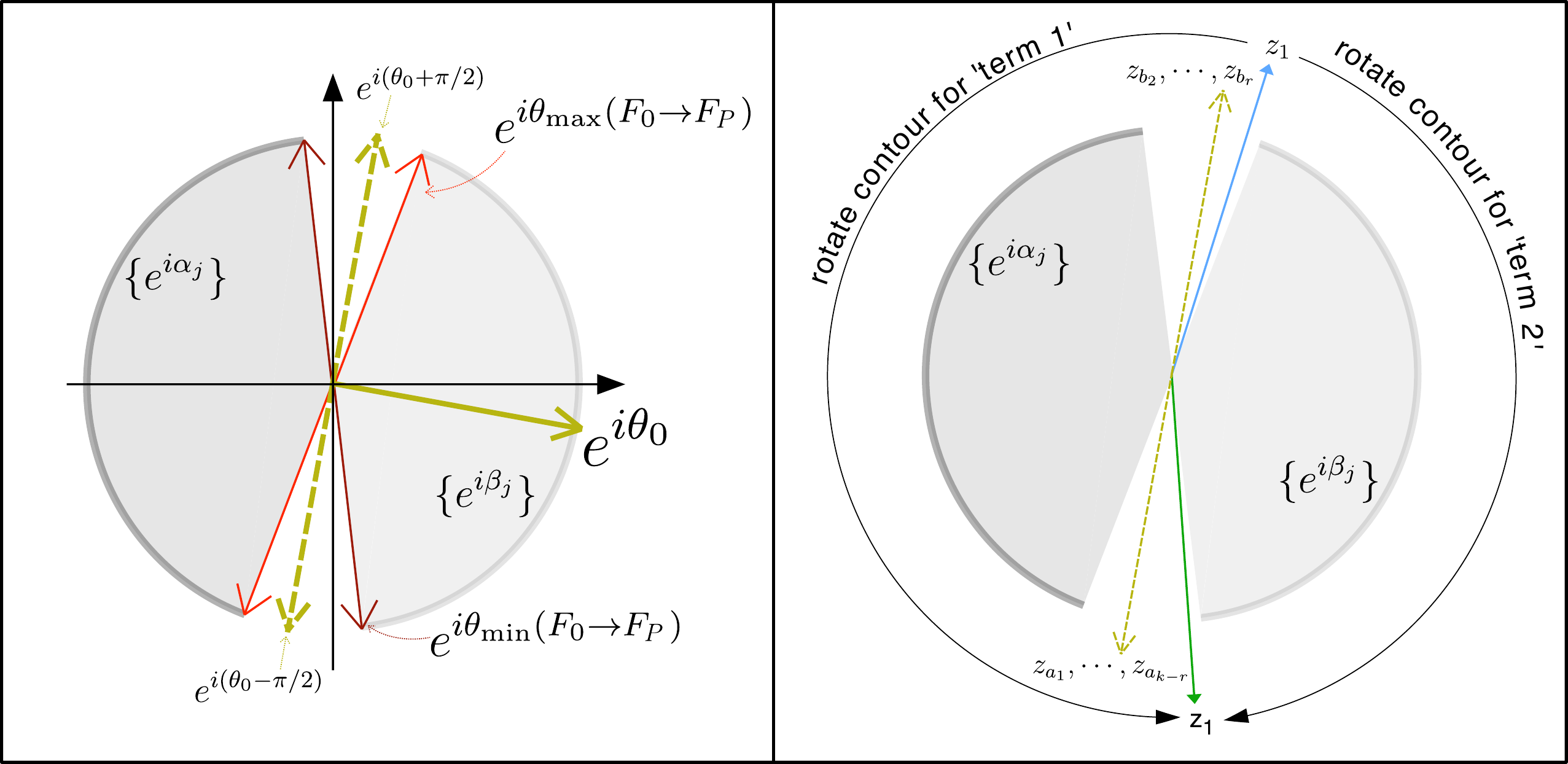}
  \caption{(Left) Depiction of the poles $\{e^{i \alpha}\}$ and $\{e^{i \beta}\}$ related to the angles $e^{i\theta_0}$, $e^{i \theta_\text{min,max}(F \to F')}$ of Fig.~\ref{fig:pathOnRhombusGraph_connection}. The poles lie on unit circle, along the thick parts of the circular wedges. (Right) We slighly deform the $z_1$ contour counterclockwise from the original $0 \to -i e^{i \theta_0} \infty$. Then, rotating the countours to slightly clockwise from $0 \to i e^{i\theta_0} \infty$ separately for each of `term 1' and `term 2', as described in the main text, will give a at $z_1 = z_k$.}
  \label{fig:rotatingContours_singZip}
\end{figure}

First, let $r = \#\{\pm_j=+\}$ and $1 \le a_1 < \cdots < a_{k-r} \le k$ be the coordinates for which $\pm_{a_j} = -$, and let $1 = b_1 < b_2 < \cdots < b_{r} \le k$ be the coordinates with $\pm_{b_j} = +$ . We set $b_1 = 1$ because on the left-hand-side of this proposition, $z_1$ takes the contour $z_1: 0 \to +i \infty e^{i \theta_0}$. 

We can start by slightly rotating the $z_1$ contour to be $0 \to i e^{i (\theta_0 - \epsilon)}$ for $0 < \epsilon \ll 1$ small enough to avoid the poles $\{e^{i \beta_j}\}$ before splitting up the term. Then, we'll rotate this shifted $z_1$ contour to the new contour $z_1: 0 \to -i e^{i (\theta_0 + \epsilon)} \infty$ clockwise for `term 1' and counterclockwise for `term 2'. For `term 1', this rotation will only cross the pole at $z_1 = z_k$ leaving $(-2 \pi i)$ times that residue, which leaves a residue of
\begin{equation*}
    - \frac{1}{(2 \pi i)^{k-1}} \int_0^{-i e^{i \theta_0} \infty} dz_2  \int_0^{\pm_3 i e^{i \theta_0}\infty} dz_3  \cdots \frac{1}{(z_2-z_3) \cdots (z_k-z_2)} \bigg(1 - \frac{g_{F \to F'}(z_2)}{g_{F \to F'}(z_3)}\bigg) \cdots \bigg(1 - \frac{g_{F \to F'}(z_k)}{g_{F \to F'}(z_2)}\bigg)
\end{equation*}
while the rotation for `term 2' will leave no residues. At the end of these rotations we can then recombine the `term 1' and `term 2' into the original integrand on their common contour and rotate it back to $0 \to -i e^{i \theta_0} \infty$. This new integral plus the residue term are what we wanted to obtain.
\end{proof}

\subsubsection{Non-Universality of Imaginary Parts of log-determinant} \label{sec:non_universality_im_part}
In general, the imaginary parts of the above expansion will generally be non-universal and depend on microscopic lattice details. To see an example, we consider the $\alpha^1$ term in the expansion, corresponds to
$\tr \bigg[
    \frac{\overline{\partial}_{\mathrm{zip}}^{\bullet \to \circ}}{\overline{\partial}^{\bullet \to \circ}}  
\bigg]$.
By Corollary~\ref{cor:traceId_first} and writing 
$g_{F \to F'}(z) = \prod_{j=1}^{P} \frac{z-e^{i \alpha_j}}{z-e^{i \beta_j}}$, we'll have
\begin{equation}
    \tr \bigg[
        \frac{\overline{\partial}_{\mathrm{zip}}^{\bullet \to \circ}}{\overline{\partial}^{\bullet \to \circ}}  
    \bigg]
    =
    -\frac{1}{2 \pi i} \int_0 ^ {i e^{i\theta_0} \infty} dz
    \frac{\frac{d}{dz}g_{F \to F'}(z)}{g_{F \to F'}(z)}
    =
    \frac{1}{2 \pi} \sum_{j=1}^P (\alpha_j - \beta_j)
    .
\end{equation}
where the angles $\alpha,\beta$ can be uniquely specified by the integration contour along $0 \to i e^{i \theta_0} \infty$. This quantity depends sensitively on the underlying lattice and can't be written as a universal function of $|F-F'|$, even in the limit of large $|F-F'|$.

In the next Section~\ref{sec:asympt_I2n_singleZip}, we'll see that each of the 
$\tr 
\bigg[ \Bigg( 
    \bigg( \frac{\overline{\partial}_{\mathrm{zip}}^{\bullet \to \circ}}{\overline{\partial}^{\bullet \to \circ}} \bigg)
    \bigg( \mathbbm{1} - \frac{\overline{\partial}_{\mathrm{zip}}^{\bullet \to \circ}}{\overline{\partial}^{\bullet \to \circ}} \bigg)
\Bigg)^n \bigg]$
will asymptotically equal a universal quantity, a uniquely specified function in $|F-F'|$.

However, we also have a series expansion $t = \sum_{n=1}^{\infty} \frac{1}{n} {2n-2 \choose n-1} (t(1-t))^n$ for $0 < t < 1/2$. Thus, we may expect the equality 
\[
    \tr 
    \bigg[
        \frac{\overline{\partial}_{\mathrm{zip}}^{\bullet \to \circ}}{\overline{\partial}^{\bullet \to \circ}} 
    \bigg]
    \stackrel{?}{=}
    \sum_{n=1}^{\infty}
    \frac{1}{n} {2n-2 \choose n-1}
    \tr 
    \bigg[ \Bigg( 
        \bigg( \frac{\overline{\partial}_{\mathrm{zip}}^{\bullet \to \circ}}{\overline{\partial}^{\bullet \to \circ}} \bigg)
        \bigg( \mathbbm{1} - \frac{\overline{\partial}_{\mathrm{zip}}^{\bullet \to \circ}}{\overline{\partial}^{\bullet \to \circ}} \bigg)
    \Bigg)^n \bigg].
\]
As such, we may be concerned that an infinite sum of `universal' terms may sum up to a non-universal one.
There are two issues with this reasoning. First, the equality $t = \sum_{k=1}^{\infty} \frac{1}{k} {2n-2 \choose n-1} (t(1-t))^n$ doesn't hold for all $t$ so there may be convergence issues for eigenvalues outside of $[0,1/2]$. Second, each 
$\tr 
    \bigg[ \Bigg( 
        \bigg( \frac{\overline{\partial}_{\mathrm{zip}}^{\bullet \to \circ}}{\overline{\partial}^{\bullet \to \circ}} \bigg)
        \bigg( \mathbbm{1} - \frac{\overline{\partial}_{\mathrm{zip}}^{\bullet \to \circ}}{\overline{\partial}^{\bullet \to \circ}} \bigg)
    \Bigg)^n \bigg]$
will have sub-leading corrections for which we don't have uniform control, so resummation of these lower-order terms may lead to non-universal sums.

In general, odd powers of $\alpha$ in the expansion of $\log(1 - (1-e^{-i \alpha}) y)$ will be accompanied by such infinite sums in $y(1-y)$, whereas even powers of $\alpha$ will be accompanied by finite sums, by Eq.~\eqref{eq:log_expansion_A}. These finite sums of universal quantities will also be universal. So we'll see that even powers are universal while odd powers are not. In particular, the real part will be universal while the imaginary part isn't.

\subsection{Asymptotics of Eq.~\eqref{eq:traceId_second} for $\mathcal{I}_{2n}$} \label{sec:asympt_I2n_singleZip}
Now, we give asymptotic estimates of the integrals Eq.~\eqref{eq:traceId_second}. In this section, we'll assume WLOG that $e^{i \theta_0} = 1$. And for notational compactness, we'll write $g(z) = g_{F \to F'}(z)$ throughout this section. The computation relies on the asymptotics of the function $g_{F \to F'}(z) = \prod_{j=1}^{P} \frac{z-e^{i \alpha_j}}{z-e^{i \beta_j}}$ from~\cite{kenyon2002critical} and also Appendix~\ref{app:asympt_gFunc}. We'll also define in this section
\begin{equation}
    D := |F' - F|
\end{equation}
to be the distance between the dual faces on the rhombus graph. 

The remainder of this section will be devoted towards the following proposition.
\begin{prop} \label{prop:asymptIntegral_singleZip}
\begin{equation} \label{eq:asymptIntegral_singleZip}
\begin{split}
    &\mathcal{I}_{2n} \cdot (1 + O(1/D)) + O(1/D) 
    =
    \mathcal{J}_{2n} 
    \\
    &:=
    \frac{2}{(2\pi)^{2n}} \int_{0 \xrightarrow{-} 1} dw_1 \int_{0 \xrightarrow{+} 1} dw_2 \cdots \int_{0 \xrightarrow{-} 1} dw_{2n-1} \int_{\mathcal{C}_\text{out}} dw_{2n} \frac{1}{(w_1 - w_2) \cdots (w_{2n} - w_1)} \bigg(1 - e^{-D|w_{2n} - w_{1}|}\bigg).
\end{split}
\end{equation}
In the above (for $\pm \in \{+,-\}$), the contours ${0 \xrightarrow{\pm} 1}$ refer to contours $0 \to 1$ in the complex plane, except deformed slightly in the imaginary direction $\pm i$ away from the real axis. And, we define $\mathcal{C}_\text{out}$ as the contour consisting of the union of the lines $\{0 \to -\infty\} \cup \{\infty \to 1\}$ along the real axis. 
\end{prop}

Originally, we wanted to evaluate $\log \frac{\det \overline{\partial}^{\bullet \to \circ}(e^{i \alpha})}{\det \overline{\partial}^{\bullet \to \circ}}$. The Eq.~(\ref{eq:singleZip_SeriesExpansion},\ref{eq:traceId_second},\ref{eq:asymptIntegral_singleZip}) imply that
\begin{equation*}
    \log \frac{\det \overline{\partial}^{\bullet \to \circ}(e^{i\alpha})}{\det \overline{\partial}^{\bullet \to \circ}}
    =
    -\sum_{n=1}^{\infty} \frac{1}{2n} (2 \sin(\alpha/2))^{2n} \mathcal{I}_{2n}
     + \text{`imaginary part'}.
\end{equation*}
In Appendix~\ref{app:sum_leading_order_I2n}, we instead consider the sum over the leading-order contributions $-\sum_{n=1}^{\infty} \frac{1}{2n} (2 \sin(\alpha/2))^{2n} \mathcal{J}_{2n}$, which is expected to give 
$ \log \frac{\det \overline{\partial}^{\bullet \to \circ}(e^{i\alpha})}{\det \overline{\partial}^{\bullet \to \circ}} $
up to negligible corrections. This gives the final result Eq.~\eqref{eq:final_sum_J2n}
\begin{equation} \label{eq:final_sum_J2_reproduced}
    -\sum_{n=1}^{\infty} \frac{1}{2n} (2 \sin(\alpha/2))^{2n} \mathcal{J}_{2n} + O(D^{-1/3})
    =
    -\frac{\alpha^2}{2\pi^2} \left( 1 + \log(D) + \gamma_{\mathrm{Euler}} \right) 
    -
    4 \sum_{n=2}^{\infty} \frac{1}{2n} \left(\frac{\alpha}{2\pi} \right)^{2n} \zeta(2n-1)
\end{equation}
where $\zeta$ is the Riemann-zeta function.

By using the series expansion
\begin{equation}
    -\frac{(2 \arcsin(\beta/2))^2}{2\pi^2}
    =
    -\sum_{n=1}^{\infty} \frac{1}{n} \frac{2}{\pi^2} \frac{(2n-2)!!}{(2n-1)!!} \beta^{2n},
\end{equation}
we can substitute in $\alpha = 2 \arcsin(\beta/2)$ to deduce
\begin{equation}
    \mathcal{J}_{2n} + O(D^{-1/3}) = \frac{4}{\pi^2} \frac{(2n-2)!!}{(2n-1)!!} \log(D) + c_{2n}
\end{equation}
where $c_{2n}$ are constants that can be determined by Eq.~\eqref{eq:final_sum_J2_reproduced} but for which we don't know a closed form.

\subsubsection{Setting up the calculation}
Using the regularity conditions in Appendix~\ref{app:asympt_gFunc}, the asymptotic form of $g(z)$ and $\frac{1}{g(z)}$ functions for $D \gg 1$ will be
\begin{equation} \label{eq:g_asympt}
    g_{F \to F'}(-|z|) \cdot (1 + O(1/D)) =
    \begin{cases}
        \gamma e^{- D |z|} &\text{ if } |z| < D^{-2/3} \\
        \approx 0 &\text{ if }  D^{-2/3} < |z| < D^{2/3} \\
        e^{-D/|z|} &\text{ if } {\sqrt{D}} < |z|
    \end{cases}
\end{equation}
and
\begin{equation} \label{eq:gInv_asympt}
    \frac{1}{g_{F \to F'}(|z|)} \cdot (1 + O(1/D)) =
    \begin{cases}
        \frac{1}{\gamma}\gamma e^{- D |z|} &\text{ if } |z| < D^{-2/3} \\
        \approx 0 &\text{ if } D^{-2/3} < |z| < D^{2/3} \\
       e^{-D/|z|} &\text{ if } D^{2/3} < |z|
    \end{cases} 
\end{equation}
where $\approx 0$ means exponentially small and can be treated as negligible in the integrals (see Prop.~\ref{prop:g_expDecay_cone}). And above, we define $\gamma := g_{F \to F'}(0) = \prod_{j=1}^{P} \frac{e^{i \alpha_j}}{e^{i \beta_j}}$ which is a unit complex number whose value depends on the microscopic details of the specific path and won't make a difference in the asymptotic evaluation of the integrals.

Note that these asymptotics only give useful information about $z$ on the negative real axis for $g(z)$ and on the positive real axis for $\frac{1}{g(z)}$ and don't directly give information about the integrand on the contours $z_j: 0 \to \pm i \infty$. 
As such, to actually use these asymptotic expansions, we need to find a way to shift the contours from the original $0 \to \pm i \infty$ to $0 \to \pm \infty$. 
We can't directly do this, for two reasons. 
First, the integrand has poles at both $\{z_j = e^{i\alpha}\}$ and $\{z_j = e^{i\beta}\}$ since they contain pieces of both $g(z_j)$ and $\frac{1}{g(z_j)}$, so these contour deformations will cross many poles. 
Second, the asymptotics for $g(z)$ and $\frac{1}{g(z)}$ aren't even useful on the same halves of the real axis.
To deal with this, we need to split up the terms in such a way that each term only contains one of $g(z_j)$ or $\frac{1}{g(z_j)}$. 

To do this, first write the integrand of Eq.~\eqref{eq:traceId_second} as
\begin{equation}
\begin{split}
I_{2n} = \frac{1}{(z_1 - z_2) \cdots (z_{2n-1}-z_{2n})(z_{2n}-z_1)}
  &\left[
    \left(1 - \frac{g(z_1)}{g(z_2)}\right)
    \left(1 - \frac{g(z_3)}{g(z_4)}\right)
    \cdots
    \left(1 - \frac{g(z_{2n-1})}{g(z_{2n})}\right)
  \right] \\
    \times
  &\left[
    \left(1 - \frac{g(z_2)}{g(z_3)}\right)
    \cdots
    \left(1 - \frac{g(z_{2n-2})}{g(z_{2n-1})}\right)
    \left(1 - \frac{g(z_{2n})}{g(z_{1})}\right)
  \right].
\end{split}
\end{equation}
If we expand out the second line of the above into its $2^{n}$ factors and consider the terms gotten by multiplying the first line by those factors, then we can see that every term will only have $g(z_j)$ either in the numerator or the denominator. For example, the case $n=2$ has an integrand that can be written as
\begin{equation}
\begin{split}
    I_4 = \frac{1}{(z_1 - z_2) \cdots (z_{4}-z_1)} \times
    \Bigg[ 
        &
        \underbrace{
            \bigg( 1 - \frac{g(z_1)}{g(z_2)} \bigg) 
            \bigg( 1 - \frac{g(z_3)}{g(z_4)} \bigg)
        }_{}
        - 
        \underbrace{
            \bigg( 1 - \frac{g(z_1)}{g(z_2)} \bigg) 
            \frac{g(z_2)}{g(z_3)}
            \bigg(1 - \frac{g(z_3)}{g(z_4)}\bigg)
        }_{} 
        \\ 
        -&
        \underbrace{
            \bigg(1 - \frac{g(z_1)}{g(z_2)}\bigg)
            \bigg(1 - \frac{g(z_3)}{g(z_4)}\bigg) 
            \frac{g(z_4)}{g(z_1)}
        }_{} 
        + 
        \underbrace{
            \bigg(1 - \frac{g(z_1)}{g(z_2)}\bigg) 
            \frac{g(z_2)}{g(z_3)}
            \bigg(1 - \frac{g(z_3)}{g(z_4)}\bigg) 
            \frac{g(z_4)}{g(z_1)} 
        }_{} 
    \Bigg].
\end{split}
\end{equation}
For each $z_i$, each of the terms in with an underbrace in the above can be simplified to only contain a factor of $g(z_i)$ or $\frac{1}{g(z_i)}$. Also, note that the configuration of the contours of integration alternating between $0 \to +i \infty$ and  $0 \to -i \infty$ mean that we can deform the pieces of the full integral for $n=2$ to equal
without crossing singularities and without residues 

\begin{equation} \label{eq:I4_splitUp_terms}
\begin{split}
    \mathcal{I}_4 = \frac{1}{(2 \pi)^4} 
    \Bigg[ 
        &
        \underbrace{
            \int_{0}^{-\infty} d z_1 
            \int_{0}^{+\infty} d z_2 
            \int_{0}^{-\infty} d z_3 
            \int_{0}^{+\infty} d z_4 
            \frac{1}{(z_1 - z_2) \cdots (z_4 - z_1)} 
            \bigg(1 - \frac{g(z_1)}{g(z_2)}\bigg)
            \bigg(1 - \frac{g(z_3)}{g(z_4)}\bigg)
        }_{\text{`Part (00)'}}
        \\
        -&
        \underbrace{
            \int_{0}^{-\infty} d z_1 
            \int_{0}^{-\infty} d z_2 
            \int_{0}^{+\infty} d z_3 
            \int_{0}^{+\infty} d z_4 
            \frac{1}{(z_1 - z_2) \cdots (z_4 - z_1)} 
            \bigg(1 - \frac{g(z_1)}{g(z_2)}\bigg) 
            \frac{g(z_2)}{g(z_3)}
            \bigg(1 - \frac{g(z_3)}{g(z_4)}\bigg)
        }_{\text{`Part (10)'}}
        \\
        -&
        \underbrace{
            \int_{0}^{+\infty} d z_1 
            \int_{0}^{+\infty} d z_2 
            \int_{0}^{-\infty} d z_3 
            \int_{0}^{-\infty} d z_4 
            \frac{1}{(z_1 - z_2) \cdots (z_4 - z_1)} 
            \bigg(1 - \frac{g(z_1)}{g(z_2)}\bigg)
            \bigg(1 - \frac{g(z_3)}{g(z_4)}\bigg) 
            \frac{g(z_4)}{g(z_1)}
        }_{\text{`Part (01)'}}
        \\
        +&
        \underbrace{
            \int_{0}^{+\infty} d z_1 
            \int_{0}^{-\infty} d z_2 
            \int_{0}^{+\infty} d z_3 
            \int_{0}^{-\infty} d z_4 
            \frac{1}{(z_1 - z_2) \cdots (z_4 - z_1)} 
            \bigg(1 - \frac{g(z_1)}{g(z_2)}\bigg) 
            \frac{g(z_2)}{g(z_3)}
            \bigg(1 - \frac{g(z_3)}{g(z_4)}\bigg) 
            \frac{g(z_4)}{g(z_1)} 
    \Bigg]
    }_{\text{`Part (11)'}}.
\end{split}
\end{equation}
In particular, we label the terms `Part($x_1 x_2$)' for $x_1,x_2 \in \{0,1\}$. If $x_j=1$, then the integrand of the term will have a factor $-\frac{g(z_{2j})}{g(z_{2j+1})}$ being multiplied in and the integration contours will be $z_{2j}: 0 \to -\infty$ and $z_{2j+1}: 0 \to +\infty$, whereas if $x_j=\text{0}$, the integrand of the term will have a factor $1$ being multiplied in and the integration contours will be $z_{2j}: 0 \to +\infty$ and $z_{2j+1}: 0 \to -\infty$.

For general $n$, in the same exact manner as above, we can split up the integral into $2^{n}$ parts, labeled `Part($x_1 \cdots x_n$)' for $x_1,\cdots,x_n \in \{0,1\}$. Note that since each of the pieces has a pole at $z_{2j} = z_{2j+1}$ (where again, $z_{2n+1} \equiv z_1$), one may be concerned that moving the contours may produce residues. However, note that in this procedure that we are shifting the contours $z_{2j}: 0 \to i\infty$ and $z_{2j+1}: 0 \to -i\infty$ to $z_{2j}: 0 \to \pm_{2j}\infty$ and $z_{2j+1}: 0 \to -\pm_{2j}\infty$ for all $2^n$ choices of $\{\pm_{2j}\}$. So the fact that these contours always lay opposite to each other means that this procedure will never have the contours crossing any poles and this procedure doesn't produce any residues.

It will be instructive to analyze the asymptotics of the $n=1$ case, $\mathcal{I}_2$, before moving on to the general case. 

\subsubsection{The $n=1$ case, $\mathcal{I}_2$}
Somewhat simpler, the $\mathcal{I}_2$ integral can be written
\begin{equation}
\begin{split}
    \mathcal{I}_2 &= \frac{1}{(2 \pi)^2} \Bigg[ 
      \underbrace{\int_{0}^{-\infty} d z_1 \int_{0}^{+\infty} d z_2 \frac{1}{(z_1 - z_2)(z_2 - z_1)} \bigg(1 - \frac{g(z_1)}{g(z_2)}\bigg)}_{\text{`Part (0)'}} 
     - \underbrace{\int_{0}^{+\infty} d z_1 \int_{0}^{-\infty} d z_2  \frac{1}{(z_1 - z_2)(z_2 - z_1)} \bigg(\frac{g(z_2)}{g(z_1)} - 1\bigg)}_{\text{`Part (1)'}} \Bigg]
\end{split}
\end{equation}
which can be analyzed using the asymptotics in Eqs.~(\ref{eq:g_asympt},\ref{eq:gInv_asympt}). Let's first analyze `Part (0)', as the story for `Part (1)' will be very similar. 

Note that we can use the asymptotic form Eq.~\eqref{eq:g_asympt} of $g(z)$ to write 

\begin{equation} \label{eq:asympts_oneZipper_order2_part0}
\begin{split}
    \text{`Part (0)'} \cdot (1 + O(1/D)) &= \frac{1}{(2 \pi)^2}\int_{0}^{-1} d z_1 \int_{0}^{\infty} d z_2 \frac{1}{(z_1 - z_2)(z_2 - z_1)} \bigg(1 - e^{D(z_1-z_2)}\bigg) + O(1/D)  \\
    &\quad + \frac{1}{(2 \pi)^2}\int_{0}^{-1} d z_1 \int_{0}^{\infty} d z_2 \frac{1}{(z_1 - z_2)(z_2 - z_1)} \bigg(1 - e^{D(\frac{1}{z_1}-\frac{1}{z_2})}\bigg) + O(1/D) \\
    &= \frac{2}{(2 \pi)^2}\int_{0}^{-1} d z_1 \int_{0}^{\infty} d z_2 \frac{1}{(z_1 - z_2)(z_2 - z_1)} \bigg(1 - e^{D(z_1-z_2)}\bigg) + O(1/D) \\
    &= \frac{1}{2 \pi^2} (\log(D) + \gamma_\text{Euler} + 1) + O(1/D)
\end{split}
\end{equation}
In the above, the second equality follows from changing variables $z_1 \mapsto \frac{1}{z_1}$ and $z_2 \mapsto \frac{1}{z_2}$ in the second term. In the third equality, $\gamma_{\text{Euler}}$ the Euler gamma constant, and the result can be found from explicitly evaluating the integral. However, the first line requires some justification, as follows.

First, we note that in the `corner region' of $z_1 \in (-D^{2/3},-\infty)$ and $z_2 \in (0,D^{-2/3})$, the term will be $O(1/D)$ since it will be of the form 
\begin{equation}
\begin{split}
    &-\frac{1}{(2 \pi)^2} \int_{-D^{2/3}}^{-\infty} dz_1 \int_{0}^{\frac{1}{D^{2/3}}} dz_1 \frac{1}{(z_1-z_2)^2}(1 - \gamma e^{D(\frac{1}{z_1} - z_2)}) + O(1/D) \\
    =& -\frac{1}{(2 \pi)^2} \int_{0}^{-D^{-2/3}} dz_1 \int_{0}^{D^{-2/3}} dz_1 \frac{1}{(1-z_1 z_2)^2}(1 - \gamma e^{D(z_1 - z_2)}) + O(1/D) = O(1/D)
\end{split}
\end{equation}
The first line uses the asymptotic form of $g(z)$ to write the integral up to a lower order $O(1/D)$ term. The second line changes variables $z_1 \to \frac{1}{z_1}$, and the third line follows because $|\frac{1 - \gamma e^{D(z_1 - z_2)}}{(1 -z_1 z_2)^2}| < 2$ on the region, and the region itself has area $1/D^{4/3}$. As such, we can asymptotically ignore that part.

Second, even though asymptotically $g(z_1) \approx 0$ and $\frac{1}{g(z_2)} \approx 0$ on the middle regions $z_1 \in (-D^{2/3},-D^{-2/3})$ and $z_2 \in (D^{-2/3},D^{2/3})$, we can substitute $g(z_1) \to e^{D z_1}$ and $\frac{1}{g(z_2)} \to e^{-D z_2}$ into the integrals on the full interval at the cost of a lower-order term, since the exponentially small terms will contribute at most an $O(1/D)$ contribution.

The computation for `Part (1)' goes through exactly the same and gives the same asymptotic result. As such, we have 
\begin{equation}
    \mathcal{I}_2 = \text{`Part (0)'} + \text{`Part (1)'} = \frac{1}{\pi^2}(\log(D) + \gamma_\text{Euler} + 1) \cdot (1 + O(1/D)),
\end{equation}
which is consistent with the Proposition~\ref{prop:asymptIntegral_singleZip}.

Now, we show how show the equality with the integral over the $\{w_j\}$ variables in Proposition~\ref{prop:asymptIntegral_singleZip}. Let's analyze `Part (0)' first. We have

\begin{equation}
\begin{split}
    \text{`Part (0)'} &\cdot (1 + O(1/D)) + O(1/D) 
    \\
    &= 
    \frac{2}{(2 \pi)^2}
    \int_{0}^{-1} d z_1 \int_{0}^{\infty} d z_2 
    \frac{1}{(z_1 - z_2)(z_2 - z_1)} \big(1 - e^{D(z_1-z_2)}\big) 
    \\
    &= 
    \frac{2}{(2 \pi)^2} 
    \int_{0}^{-1} d z_1 \int_{0}^{\infty} d z_2 
    \int_{0}^{D} d w_1 \int_{0}^{-\infty} d w_2 \,\, 
    e^{z_1(w_1-w_2)} e^{z_2(w_2-w_1)} 
    \\
    &= 
    \frac{2}{(2 \pi)^2} 
    \int_{0}^{D} d w_1 \int_{0}^{-\infty} d w_2 
    \frac{1}{(w_1 - w_2)(w_2 - w_1)} \big(1 - e^{(w_2-w_1)}\big).
\end{split}
\end{equation}
The first line was derived previously. The second line uses the `Schwinger parameterization trick' substituting 
$\frac{1-e^{D(z_1-z_2)}}{z_1-z_2} = -\int_{0}^{D} dw_1 \, e^{w_1(z_1-z_2)}$
and 
$\frac{1}{z_2-z_1} = -\int_{0}^{-\infty} dw_2 \, e^{w_2(z_2-z_1)}$. Note that the $w_2$ integral converges since $z_2-z_1 > 0$.
The last line follows from evaluating the integrals over $z_1,z_2$, noting that the $z_2$ integral converges if we choose $w_1: 0 \to D$ and $w_2: 0 \to -\infty$ to be straight lines along the real axis, which would imply $w_2-w_1<0$. 

Similarly, we can write

\begin{equation}
\begin{split}
    \text{`Part (1)'} &\cdot (1 + O(1/D)) + O(1/D) 
    \\
    &= 
    \frac{2}{(2 \pi)^2}
    \int_{0}^{+1} d z_1 \int_{0}^{-\infty} d z_2 
    \frac{1}{(z_1 - z_2)(z_2 - z_1)} \big( 1 - e^{D(z_1-z_2)} \big) e^{D(z_2-z_1)} 
    \\
    &= 
    \frac{2}{(2 \pi)^2} 
    \int_0^{+1} d z_1 \int_{0}^{-\infty} d z_2 
    \int_{0}^{D} d w_1 \int_{\infty}^{D} d w_2 \,\, 
    e^{w_1(z_1-z_2)} e^{w_2(z_2-z_1)} e^{D(z_2-z_1)} \big(e^{D(z_1-z_2)}-1\big)
    \\
    &= 
    \frac{2}{(2 \pi)^2} 
    \int_{0}^{D} d w_1 \int_{\infty}^{D} d w_2
    \frac{1}{(w_1 - w_2)(w_2 - w_1)} \big(1 - e^{-(w_2-w_1)}\big),
\end{split}
\end{equation}
noting similarly that all of the integrals going out to infinity converge.
In total, we find that up to additive and multiplicative factors of $O(1/D)$, we have `Part (0)' + `Part (1)' equal the integral specified in Proposition~\ref{prop:asymptIntegral_singleZip}, after changing variables $w_{1,2} \mapsto w_{1,2} / D$.

\subsubsection{The general case}

Now, let us consider the general case inputting the asymptotic form of Eqs.~(\ref{eq:g_asympt},\ref{eq:gInv_asympt}) into the decomposition discussed following Eq.~\eqref{eq:I4_splitUp_terms}. 

First, let's again consider the `corner regions' of all of the `Part($\cdots$)', where for some $j \neq k$ we have $|z_j| \in (0, D^{-2/3})$ and $|z_k| \in (D^{2/3}, \infty)$ and for all other $\ell \neq j,k$, either $|z_{\ell}| \in (0, D^{-2/3})$ or $|z_{\ell}| \in (D^{2/3}, \infty)$. It will be a general fact that the integrals on these corner regions will be lower-order terms as $D \to \infty$, as follows.

\begin{lem} \label{lem:cornerRegion_oneZip}
Define the integrands $I_{2n}(z_1, \cdots, z_{2n} ; x_1 \cdots x_n)$ of each of the `Part($x_1 \cdots x_n$)'. Then, for all $j \neq k$, we have 
\begin{equation}
    \cdots \int_{0}^{\pm D^{-2/3}} dz_j \cdots \int_{\pm D^{2/3}}^{\pm \infty} dz_k \cdots I_{2n}(z_1, \cdots, z_{2n} ; x_1 \cdots x_n) = O(1/D)
\end{equation}
where the integration regions for all $z_{\ell}, \ell \neq j,k$ can be either one of $(0,\pm D^{-2/3})$ or $\pm (D^{2/3},\infty)$.
\end{lem}
In the above lemma, the choice of signs $\pm$ for the intervals is specified by the integration contours corresponding to `Part($x_1 \cdots x_n$)'.
\begin{proof}
    First, we change variables $z_p \to |z_p| = \pm z_p$ for all $p$. At this point, the integrand can be written as a product of various factors of one of the following forms
    
    \begin{equation*}
    \begin{split}
        \frac{1-\frac{g(z_{2m-1})}{g(z_{2m})}}{|z_{2m-1}| + |z_{2m}|}\frac{1}{|z_{2m}| + |z_{2m+1}|} & \text{ for } z_{2m-1} \in (0,-\infty) \text{ and } z_{2m} \in (0,+\infty), \text{ i.e. } x_m = x_{m-1} = 0, 
        \\
        \frac{1-\frac{g(z_{2m})}{g(z_{2m-1})}}{|z_{2m-1}| + |z_{2m}|}\frac{1}{|z_{2m}| + |z_{2m+1}|} & \text{ for } z_{2m-1} \in (0,+\infty) \text{ and } z_{2m} \in (0,-\infty), \text{ i.e. } x_m = x_{m-1} = 1, 
        \\
        \frac{g(z_{2m-1})-g(z_{2m})}{|z_{2m-1}| - |z_{2m}|}\frac{1}{|z_{2m}| + |z_{2m+1}|} & \text{ for } z_{2m-1} \in (0,-\infty) \text{ and } z_{2m} \in (0,-\infty), \text{ i.e. } x_m = 1,  x_{m-1} = 0, 
        \\
        \frac{\frac{1}{g(z_{2m-1})}-\frac{1}{g(z_{2m})}}{|z_{2m-1}| - |z_{2m}|}\frac{1}{|z_{2m}| + |z_{2m+1}|} & \text{ for } z_{2m-1} \in (0,+\infty) \text{ and } z_{2m} \in (0,+\infty), \text{ i.e. } x_m = 0,  x_{m-1} = 1
    \end{split}
    \end{equation*}
    depending on the specific `Part($x_1 \cdots x_n$)'. 
    
    Next, we change variables to all lie in $\text{int}_0 := (0,D^{-2/3})$ by sending $|z_p| \to \Tilde{z}_p := |z_p|^{\pm 1}$, depending on the interval chosen. At this point, we will use the bounds in Propositions~\ref{prop:1minusg_estimate},\ref{prop:gminusg_estimate} and Eq.~\eqref{eq:gMinusg_lt_1Minusg}. In addition, we use the fact that for these ranges of $\{z_p\}$ that
    \begin{equation*}
        \left| 1 - \frac{g(z_{2m-1})}{g(z_{2m})} \right| \le 2 \quad,\quad \left| 1 - \frac{g(z_{2m})}{g(z_{2m-1})} \right| \le 2 \quad,\quad \left| \frac{1}{g(z_{2m-1})} - \frac{1}{g(z_{2m})} \right| \le 2 \quad,\quad \left| g(z_{2m-1}) - g(z_{2m}) \right| \le 2.
    \end{equation*}
    Then, after having distributed the Jacobians $-\frac{1}{z^2}$ throughout the denominators appropriately, we can bound the full integral as a product of factors of the form
    \begin{equation*}
    \begin{split}
        \frac{2}{\Tilde{z}_{2m-1} + \Tilde{z}_{2m}} \quad,\quad \frac{2}{1+ \Tilde{z}_{2m-1} \Tilde{z}_{2m}} \quad,\quad \frac{1}{\Tilde{z}_{2m} + \Tilde{z}_{2m+1}} \quad,\quad \frac{1}{1+ \Tilde{z}_{2m} \Tilde{z}_{2m+1}},
    \end{split}
    \end{equation*}
    where all $\Tilde{z}_{2m} \in (0,D^{-2/3})$. The fact that the original integration regions were $|z_j|: 0 \to D^{-2/3}$ and $|z_k|: D^{2/3} \to \infty$ means that there's at least one factor of $\frac{1}{1+ \Tilde{z}_{r} \Tilde{z}_{r+1}}$ in the bounding integral. Note that $\frac{1}{1+ \Tilde{z}_{r} \Tilde{z}_{r+1}} \le 1$. As such, we can bound the full integral as a constant power of two times the product of at least two integrals of the form
    \begin{equation}
    \begin{split}
        &\int_{0}^{D^{-2/3}} d\Tilde{z}_r \int_{0}^{D^{-2/3}} d\Tilde{z}_{r+1} \cdots \int_{0}^{D^{-2/3}} d\Tilde{z}_{r+s} \frac{1}{\Tilde{z}_{r} + \Tilde{z}_{r+1}} \frac{1}{\Tilde{z}_{r+1} + \Tilde{z}_{r+2}} \cdots \frac{1}{\Tilde{z}_{r+s-1} + \Tilde{z}_{r+s}} \\
        =& \frac{1}{D^{2/3}} \int_{0}^{1} d\Tilde{z}_r \int_{0}^{1} d\Tilde{z}_{r+1} \cdots \int_{0}^{1} d\Tilde{z}_{r+s} \frac{1}{\Tilde{z}_{r} + \Tilde{z}_{r+1}} \frac{1}{\Tilde{z}_{r+1} + \Tilde{z}_{r+2}} \cdots \frac{1}{\Tilde{z}_{r+s-1} + \Tilde{z}_{r+s}},
    \end{split}
    \end{equation}
    where we can allow possibly $s=0$, in which case the relevant integral would just be $\int_{0}^{D^{-2/3}} d\Tilde{z}_r = \frac{1}{D^{2/3}}$.
    
    The fact that the integrals $\int_{0}^{1} d\Tilde{z}_r \int_{0}^{1} d\Tilde{z}_{r+1} \cdots \int_{0}^{1} d\Tilde{z}_{r+s} \frac{1}{\Tilde{z}_{r} + \Tilde{z}_{r+1}} \cdots \frac{1}{\Tilde{z}_{r+s-1} + \Tilde{z}_{r+s}}$ are convergent, finite numbers independent of $D$ and that there are at least two of them means that the full integral can be bounded by a constant times $\frac{1}{D^{4/3}}$, which implies what we wanted to show.
\end{proof}

Now, we can proceed to show the proposition. In particular, we will show the following expression for `Part($x_1 \cdots x_n$)'.

\begin{lem} \label{lem:part_x1__xn_int}
Fix some $x_1, \cdots, x_n \in \{0,1\}$. Then,
\begin{equation}
    \text{`Part}(x_1,\cdots,x_n)\text{'}\cdot(1+O(1/D)) + O({1}/{D}) = \int_{0}^{D} dw_1 \int_{\mathcal{C}_{x_1}} dw_2 \cdots \int_{0}^{D} dw_{2n-1} \int_{\mathcal{C}_{x_n}} dw_{2n} \frac{1 - e^{-|w_{2n} - w_1|}}{(w_1 - w_2) \cdots (w_{2n}-w_1)}.
\end{equation}
In the above, the contours for $w_{2j}$ are $\mathcal{C}_{x_j}$, defined as follows. We define $\mathcal{C}_0 = 0 \to -\infty$ and $\mathcal{C}_{1} = +\infty \to D$. 
\end{lem}

Note that this lemma implies the Proposition~\ref{prop:asymptIntegral_singleZip}. In particular, note that all of the integrands above have the same integrand. As such, adding up all of the integration contours over all assignments of $\{x_j\}$ will give (after changing all variables of all $w_j \mapsto w_j/D$) 
\begin{equation*}
\begin{split}
    \mathcal{I}_{2n} &= \sum_{x_1,\cdots,x_n \in \{0,1\}} \text{`Part}(x_1,\cdots,x_n)\text{'} \cdot (1 + O(1/D))
    \\
    &= \frac{2}{(2\pi)^{2n}} \int_{0}^{1} dw_1 \int_{\mathcal{C}_\text{out}} dw_2 \cdots \int_{0}^{1} dw_{2n-1} \int_{\mathcal{C}_\text{out}} dw_{2n} \frac{1 - e^{-D|w_{2n} - w_1|}}{(w_1 - w_2) \cdots (w_{2n}-w_1)} + O({1}/{D}).
\end{split}
\end{equation*}
Note that for fixed $w_{2n}$, one can set $e^{-|w_{2n}-w_{1}|}$ with one of $e^{w_{2n}-w_{1}}$ or $e^{w_{1}-w_{2n}}$ depending on if $w_{2n}<0$ or $w{2n}>0$. So, the integrands are all holomorphic as functions for all other $w_{\cdots}$ and vanish as $\frac{1}{w^2_{\cdots}}$. As such, one can deform the contours without changing the value of the integrands, as long as a $w_{j}$ contour doesn't cross any $w_{j \pm 1}$ contour in the process. In particular, one can change the contours to be as in the integral Proposition~\ref{prop:asymptIntegral_singleZip}. 

The proof of the lemma will be as follows, similar to the $n=1$ analysis of $\mathcal{I}_2$ from earlier.
\begin{proof}[Proof of Lemma~\ref{lem:part_x1__xn_int}]

First, note that we can use the asymptotics of $g(z)$ and the definition of each of the `Part($x_1,\cdots,x_n$)' as follows. Justification for each equality will be provided afterward.
\begin{equation}
\begin{split}
    &\text{`Part}(x_1,\cdots,x_n)\text{'} \cdot (1 + O(1/D)) + O(1/D) \\
    &= 
    \frac{1}{(2 \pi)^{2n}}
    \int_{0}^{-(-1)^{x_{n}}} d z_1 
    \int_{0}^{(-1)^{x_{1}}\infty} d z_2 
    \int_{0}^{-(-1)^{x_{1}}\infty} d z_3 
    \int_{0}^{(-1)^{x_{2}}\infty} d z_4 
    \cdots 
    \int_{0}^{-(-1)^{x_{n-1}}\infty} d z_{2n-1} 
    \int_{0}^{(-1)^{x_{n}}\infty} d z_{2n} 
    \\
    &\quad\quad 
    \left\{
        \frac{
            (1 - e^{D(z_1-z_2)}) \bigg(-e^{D(z_2 - z_3)} \bigg)^{x_1}
            (1 - e^{D(z_3-z_4)}) \bigg(-e^{D(z_4 - z_5)} \bigg)^{x_2}
            \cdots
            (1 - e^{D(z_{2n-1}-z_{2n})}) \bigg(-e^{D(z_{2n} - z_1)} \bigg)^{x_n}
        }{(z_1-z_2)\cdots(z_{2n}-z_1)} 
    \right\} 
    \\
    &\,\,
    +
    \frac{1}{(2 \pi)^{2n}}
    \int_{-(-1)^{x_{n}}}^{-(-1)^{x_{n}}\infty} d z_1 
    \int_{0}^{(-1)^{x_{1}}\infty} d z_2 
    \int_{0}^{-(-1)^{x_{1}}\infty} d z_3 
    \int_{0}^{(-1)^{x_{2}}\infty} d z_4 
    \cdots 
    \int_{0}^{-(-1)^{x_{n-1}}\infty} d z_{2n-1} 
    \int_{0}^{(-1)^{x_{n}}\infty} d z_{2n} 
    \\
    &\quad\quad 
    \left\{
        \frac{
            (1 - e^{D(\frac{1}{z_1}-\frac{1}{z_2})}) \bigg(-e^{D(\frac{1}{z_2} - \frac{1}{z_3})} \bigg)^{x_1}
            (1 - e^{D(\frac{1}{z_3}-\frac{1}{z_4})}) \bigg(-e^{D(\frac{1}{z_4} - \frac{1}{z_5})} \bigg)^{x_2}
            \cdots
            (1 - e^{D(\frac{1}{z_{2n-1}}-\frac{1}{z_{2n}})}) \bigg(-e^{D(\frac{1}{z_{2n}} - \frac{1}{z_1})} \bigg)^{x_n}
        }{(z_1-z_2)\cdots(z_{2n}-z_1)} 
    \right\}
    \\
    &= \frac{2}{(2 \pi)^{2n}}
    \int_{0}^{-(-1)^{x_{n}}} d z_1 
    \int_{0}^{(-1)^{x_{1}}\infty} d z_2 
    \int_{0}^{-(-1)^{x_{1}}\infty} d z_3 
    \int_{0}^{(-1)^{x_{2}}\infty} d z_4 
    \cdots 
    \int_{0}^{-(-1)^{x_{n-1}}\infty} d z_{2n-1} 
    \int_{0}^{(-1)^{x_{n}}\infty} d z_{2n} 
    \\
    &\quad\quad 
        \left\{
        \frac{
            (1 - e^{D(z_1-z_2)}) \bigg(-e^{D(z_2 - z_3)} \bigg)^{x_1}
            (1 - e^{D(z_3-z_4)}) \bigg(-e^{D(z_4 - z_5)} \bigg)^{x_2}
            \cdots
            (1 - e^{D(z_{2n-1}-z_{2n})}) \bigg(-e^{D(z_{2n} - z_1)} \bigg)^{x_n}
        }{(z_1-z_2)\cdots(z_{2n}-z_1)} 
    \right\} 
    \\
    &= \frac{2}{(2 \pi)^{2n}}
    \int_{0}^{-(-1)^{x_{n}} } d z_1 
    \int_{0}^{(-1)^{x_{1}}\infty} d z_2 
    \int_{0}^{-(-1)^{x_{1}}\infty} d z_3 
    \int_{0}^{(-1)^{x_{2}}\infty} d z_4 
    \cdots 
    \int_{0}^{-(-1)^{x_{n-1}}\infty} d z_{2n-1} 
    \int_{0}^{(-1)^{x_{n}}\infty} d z_{2n} 
    \\
    &\quad\quad
    \int_{0}^{D} d w_1 
    \int_{\mathcal{C}_{x_1}} d w_2
    \int_{0}^{D} d w_3 
    \int_{\mathcal{C}_{x_2}} d w_4 
    \cdots
    \int_{0}^{D} d w_{2n-1} 
    \int_{\mathcal{C}_{x_{2n}}} d w_{2n} \,\, 
    \Big\{ e^{w_1(z_1 - z_{2})} e^{w_2(z_2 - z_{3})} \cdots e^{w_{2n}(z_{2n} - z_{1})} \Big\}
    \\
    &= 
    \frac{2}{(2 \pi)^{2n}} 
    \int_{0}^{D} dw_1 
    \int_{\mathcal{C}_{x_1}} dw_2 
    \cdots 
    \int_{0}^{D} dw_{2n-1} 
    \int_{\mathcal{C}_{x_n}} dw_{2n} 
    \frac{1 - e^{-|w_{2n} - w_1|}}{(w_1 - w_2) \cdots (w_{2n}-w_1)}
\end{split}
\end{equation}
The first equality of the above equation follows from substituting the asymptotics from Eqs.~(\ref{eq:g_asympt},\ref{eq:gInv_asympt}) into the integrand of `Part($x_1,\cdots,x_n$)'. There, we need to add in the $O(1/D)$ correction to extend the approximate $g(z) \approx e^{\pm D z_j}$ or $g(z) \approx e^{\pm D \frac{1}{z_j}}$ from the regions $|z| < D^{-2/3}$ or $|z| > D^{2/3}$. 
The second equality follows from substituting $z_j \to \frac{1}{z_j}$ for each $j$ in the second term. 
The third equality can be justified slightly differently for the $w_j$ integrals depending on if $j$ even versus $j$ odd. The integrals over $w_{2j}$ come from the `Schwinger parameterization' of substituting $-\frac{\big(-e^{D(z_{2j}-z_{2j+1})}\big)^{x_1}}{z_{2j}-z_{2j+1}} = \int_{\mathcal{C}_{x_j}} dw_{2j} \, e^{w_{2j} (z_{2j}-z_{2j+1})}$. 
There, the contour $w_{2j}: 0 \to -\infty$ for $x_j=0$ or $w_{2j}: \infty \to D$ for $x_j=1$ work for each pair $z_{2j} , z_{2j+1}$. 
In particular, since the $z_{2j}$ and $z_{2j+1}$ contours are always opposite each other, these contours always lead to convergent integrals. 
The integrals $w_{2j-1}$ follow from $-\frac{1-e^{D(z_{2j-1}-z_{2j})}}{z_{2j-1}-z_{2j}} = \int_{0}^{D} dw_{2j-1} \, e^{w_{2j-1}(z_{2j-1}-z_{2j})}$. The fourth equality follows from integrating over the $\{z_j\}$, noting that all of the integrals converge because of the orientations of the $\{w_j\}$.
\end{proof}

\section{Multiple Punctures} \label{sec:multipleZipper_calculations}


Now, we consider the analogous calculation for multiple punctures and a gauge field that's flat away from a finite number of well-separated punctures. A key example considered by Dubédat in~\cite{Dubedat2011} is a $U(1)$-valued gauge field when for $s$ pairs of faces $F_1,F_1' , \cdots , F_s,F_s'$ are all well-separated and the field's monodromy around $F_j$ is the unit complex number $e^{i \alpha_j}$ and the monodromy around $F'_j$ is $e^{-i\alpha_j}$. 
Such connections are related to the height correlators $\braket{\exp\left( i \sum_{j=1}^{s} \alpha_j (h(F_j) - h(F'_j)) \right)}$, and the limit $\alpha_j \to \pi$ can be applied to the strength of monomer insertions around the $F_j,F'_j$.
Note that by the Footnote~\ref{foot:unitaryGaugeTransf}, the absolute value $\left| \frac{\det \overline{\partial}^{\bullet \to \circ}(e^{i\alpha_1},\cdots,e^{i\alpha_s})}{\det \overline{\partial}^{\bullet \to \circ}} \right|$ is well-defined if we restrict to unitary gauge transformations. 

In accordance with the free fermion determinant and Gaussian field expectations Eq.~\eqref{eq:abelian_singular_dirac_det} and the results for two punctures Eq.~\eqref{eq:singleZipper_result_intro_2}, one would expect 
\begin{equation*}
    \log \left|\frac{\det \overline{\partial}^{\bullet \to \circ}(e^{i\alpha_1},\cdots,e^{i\alpha_s})}{\det \overline{\partial}^{\bullet \to \circ}}\right|
    \xrightarrow{|F' - F| \to \infty}
    -\frac{1}{2\pi^2} \sum_{j=1}^{s} \left( \alpha_j^2 \log(|F'_j-F_j|) + C(\alpha_j) \right)
    +\frac{1}{\pi^2} \sum_{1 \le i<j \le s} \alpha_i \alpha_j \log \left( \frac{|F_i-F_j||F_i'-F'_j|}{|F_i-F'_j||F'_i-F_j|} \right)
    + o(1).
\end{equation*}
In fact, Dubédat verfies that this holds and converges for $-\pi < \alpha_j < \pi$. Furthermore, in~\cite{Dubedat2018doubleDimersConformalLoopEnsemblesIsomonodromicDeformations}, Dubédat extends this analysis to certain classes of $SL(2,\C)$ connections that are well-separated as above and shows that for non-abelian connections, well-separated zippers the analogous determinant 
$\log \frac{\det \overline{\partial}^{\bullet \to \circ}(U)}{\det \overline{\partial}^{\bullet \to \circ}}$
converges to a certain `$\tau$' function, which can be interpreted as a generalization of the free fermion determinants except for non-abelian connections.

In our analysis, we are able to demonstrate that (in a sense of moments, up to leading-order as the lattice-spacing goes to zero, to be described shortly) $\log \frac{\det \overline{\partial}^{\bullet \to \circ}(U)}{\det \overline{\partial}^{\bullet \to \circ}}$ matches with the free-fermion determinant. In particular, we'll have that for well-separated zippers approximating paths ${\bf P}_j: \tilde{F}_j \to \tilde{F}'_j$ in the complex-plane and endowed with a graph connection representing a jump matrix $U_j \in U(N)$ with eigenvalues $e^{i \alpha^{(j)}_1},\cdots,e^{i \alpha^{(j)}_N}$ going right-to-left across ${\bf P}_j$, the determinant is asymptotically
\begin{equation} \label{eq:zipper_expansion_lattice_intro}
\begin{split}
    &\log \frac{\det \overline{\partial}^{\bullet \to \circ}(U)}{\det \overline{\partial}^{\bullet \to \circ}}
    \\
    &\xrightarrow[R \to \infty]{\textbf{moments}}
    \\
    &\sum_{j=1}^{s}
    \left(
        -\frac{1}{2 \pi^2} ({\alpha^{(j)}_1}^2 + \cdots + {\alpha^{(j)}_N}^2) \log\left(R|\tilde{F}'_j - \tilde{F}_j|\right)
        +
        C(\alpha^{(j)}_1) + \cdots + C(\alpha^{(j)}_N)
    \right)
    + \text{`imaginary part'}
    \\
    &-\sum_{n=2}^{\infty} \frac{1}{n}
    \sum_{\substack{j_1 \cdots j_{n} \in \{1 \cdots s\} \\ !(j_1 = \cdots = j_{n})}} 
    \frac{\tr \left[ (1 - U_{j_1}^{-1}) \cdots (1 - U_{j_n}^{-1}) \right]}{(2 \pi i)^{n}}
    \left\{
    \int_{{\bf P}_{j_1}^{(c(1))}} dw_1 \cdots \int_{{\bf P}_{j_n}^{(c(n))}} dw_{\ell_{1}+\cdots+\ell_{n}}
    \frac{1}{(w_1-w_2)\cdots(w_{n}-w_1)}
    +
    \mathrm{c.c.}
    \right\}
\end{split}
\end{equation}
where above, $c(1) \cdots c(n)$ satisfy $c(k) = 1, c(k+1) = 2, \cdots, c(k+\ell) = \ell+1$ for $j_{k-1} \neq j_{k} = \cdots = j_{k+\ell} \neq j_{k+\ell+1}$ and $j_{k+m} = j_{k}$ for all $k$, the curves ${\bf P}_{j}^{(1)},{\bf P}_{j}^{(2)},{\bf P}_{j}^{(3)},\cdots$ are disjoint curves $\tilde{F}_j \to \tilde{F}'_j$ near ${\bf P}_j$ going right-to-left (see Fig.~\ref{fig:splitContours_appendix}), and $R$ is a parameter for which $1/R$ can be thought of as the lattice-spacing (see Sec.~\ref{sec:multipleZippers_Setup}).
The notation $\xrightarrow[R \to \infty]{\textbf{moments}}$ refers to the fact all moments with respect to the matrix elements of $U_1^{-1},\cdots,U_s^{-1}$ 
\footnote{More precisely, writing $U_j = \exp[ \sum \kappa_a^{(j)} T^a ]$ for $T^a$ a basis of the Lie algebra of whatever group we choose to work with, the moments with respect to the $\kappa_a^{(j)}$ are determined by the formula. In particular, the matrix elements themselves generally won't be independent since we work with $U(N)$.}
have a well-defined limit as $R \to \infty$ and can be rigorously computed using the above formula. 
The cluster expansion Eq.~\eqref{eq:seriesExpansion_multZippers_1}, the results from the single-zipper case Eq.~\eqref{eq:singleZipper_result_intro_2}, and the Proposition~\ref{prop:regulatedContours_integral} taken together give the equation above.
One can compare this the free-fermion calculation Eq.~\eqref{eq:zipper_expansion_final} and see that it agrees.

The first line comes from terms in the cluster expansion that are of the same form as those in Sec.~\ref{sec:straightLine_zippers}. In particular, the `imaginary part' above will be a sum of the `imaginary parts' of the corresponding calculations from Sec.~\ref{sec:straightLine_zippers}, which we argued in Sec.~\ref{sec:non_universality_im_part} be non-universal. The other terms in the cluster expansion correspond to the second line of the expansion. They will all be universal in the sense that their asymptotics only depend on the continuum paths that the gauge-field configurations approximate. These terms will generically have real and imaginary parts which are both universal.

Above, we've been focusing on well-separated zippers, for which all of the above integrals are well-defined and convergent. When the zippers are allowed to share endpoints, then some of the integrals will diverge logarithmically. On the lattice, everything will be a finite number for a fixed $R$, and it turns out that for well-chosen zippers, a simple modification to the integrals (see Sec.~\ref{sec:nonSeparatedZippers}) will give a finite result with a logarithmic divergence in $R$. 

While we can rigorously compute everything with respect to the moments as above, we do not know how to analyze the convergence of the cluster expansion, in particular of the behavior of the subleading terms. In~\cite{Dubedat2018doubleDimersConformalLoopEnsemblesIsomonodromicDeformations}, this convergence was established for certain subclasses of well-separated zipper configurations with a reflectional symmetry around $\R \subset \C$ and for \textit{unipotent} $U_j \in SL(2,\R)$. 
In this case, the first line in the RHS of Eq.~\eqref{eq:zipper_expansion_lattice_intro} diverging logarithmically with $R$ vanishes.
In particular, choosing unipotent $\{U_j\}$ will automatically give zero for the two-puncture calculations of Sec.~\ref{sec:singleZipper}, including for the `non-universal' imaginary parts. As such in this case, our series expansion is universal.
Then, our result can be seen as rigorously giving a series expansion for Dubédat's version of the tau function. 

In Sec.~\ref{sec:multipleZippers_Setup}, we go through the geometric preliminaries needed to compare the lattice and continuum zipper setups for well-separated zippers. In Sec.~\ref{sec:multZip_clusterExpansion}, we derive the cluster expansion used in our analysis. In Sec.~\ref{sec:multZip_evaluating_clusterExp_terms}, we evaluate the terms in the cluster expansion for well-separated zippers. In Sec.~\ref{sec:nonSeparatedZippers}, we explain how to modify the above steps for non-separated zippers.

\subsection{Setup for well-separated Zippers} \label{sec:multipleZippers_Setup}

The general strategy of constructing well-separated zippers will be similar for this computation as it was for the previous case of a single pair of faces $(F,F')$. In particular, we will consider a deformation $\overline{\partial}(U)$ of the Kasteleyn matrix $\overline{\partial}$ by multiplying the edge weights on the graph by weights that correspond to a connection on the graph with monodromies that jump as $U_j$ left-to-right across $F_j \to F'_j$.

We will be able to arrange these monodromies by considering a collection of $s$ zippers, which are each a path on the dual graph connecting $F_j \to F'_j$. Then we will multiply the edge weights by $U_j^{\pm 1}$ on the edges corresponding to the dual edges of the path, where $U_j$ will be multiplied in the direction right-to-left relative to the path and $U_j^{-1}$ will be multiplied in the direction right-to-left.
\footnote{
    Note that these definitions seem `opposite' to the fact that the monodromy in the continuum is multiplied by $U_j$ going left-to-right. A way to see that they should be inverses is that on the graph (see Sec.~\ref{sec:graph_connection}), the gauge-invariant monodromy for a path $v_1 \to v_2 \to \cdots \to v_k \to v_1$ is computed as $U(v_1 \to v_2) U(v_2 \to v_3) \cdot \cdots \cdot U(v_k \to v_1)$. However in the continuum, the monodromy matrices would be multiplied in the opposite order $\tilde{U}(v_k \to v_1) \cdot \cdots \cdot \tilde{U}(v_2 \to v_3) \tilde{U}(v_1 \to v_2)$. 
    As such, identifying the graph monodromies of $U$ as the inverses of the continuum monodromies $\tilde{U}$ is natural and consistent.
}

However, there will be a difference from the previous section in the specific zippers we use to do our calculation. Previously the weights $U_j^{-1}$ were always pointing in the direction of black-to-white vertices, and we chose zippers along a dual that were approximated by a straight line going from $F \to F'$. In the setup for this calculation, we will still keep the $U_j^{-1}$ factors pointing in the black-to-white direction. However, we will abandon specifically choosing the paths to be `straight lines'. Specifically, we want to keep the zippers themselves well-separated from each other so that the zippers don't have any crossings. The reason for this is with some foresight from the previous discussion is that in the continuum limit, the series expansions for these will require asymptotics of the functions $g_{v \to w}(z)$ for $|v-w|$ large. With even more foresight, we'll have that the integrals will involve an integrand $\frac{1}{(w_1 - w_2) \cdots (w_n - w_1)}$ where the $w_{j}$ are on contours that the zipper approximate. If the zippers cross, then these integrals will diverge, or at least be ambiguous.

\begin{figure}[p!]
\begin{subfigure}
  \centering
  \includegraphics[width=\linewidth]{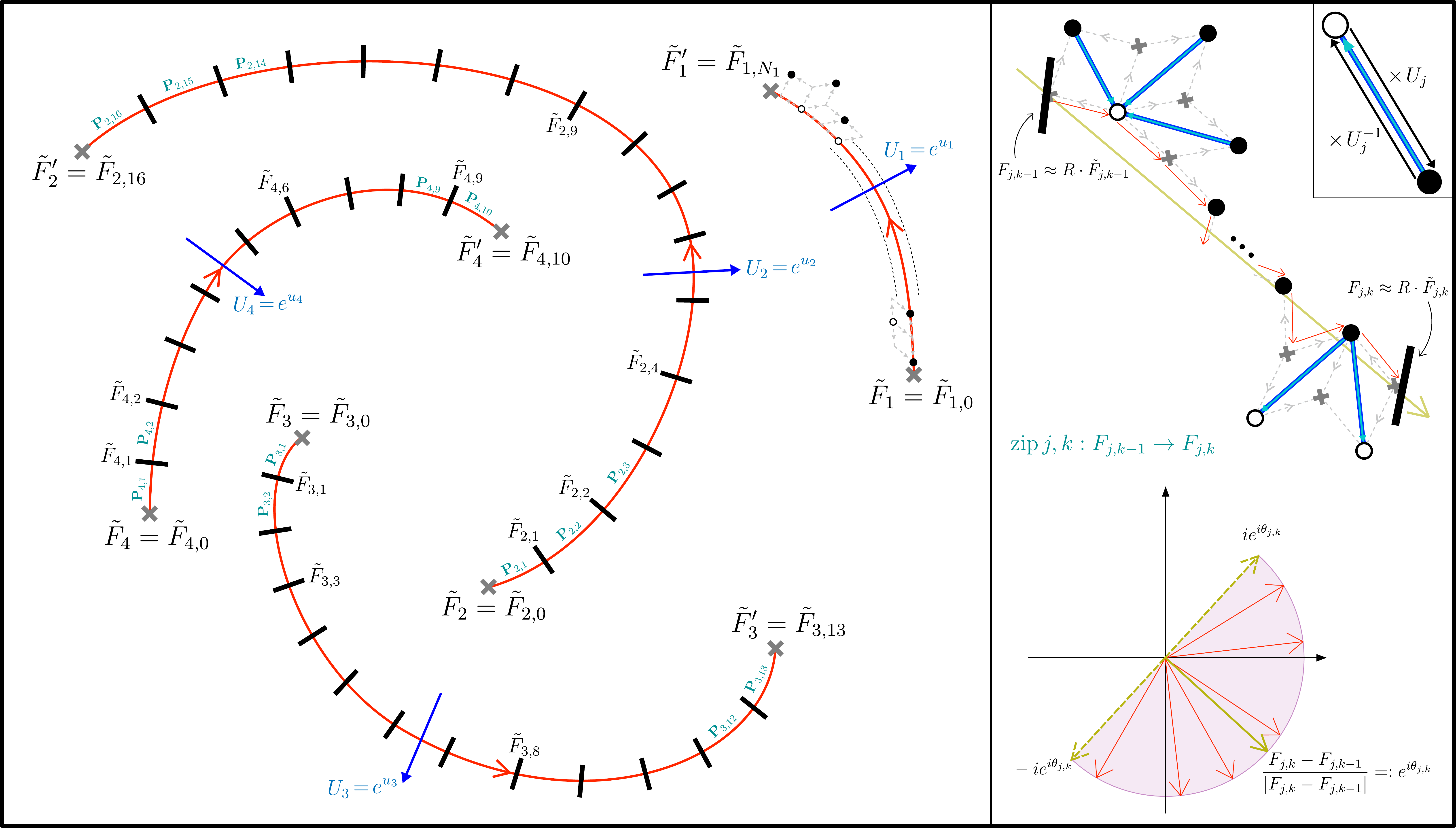}
  \caption{ Example of multiple zippers that are well-separated. (Left) For $j=1\cdots4$, there are four continuous, nonintersecting curves $\bP_j$ in $\C$ that go between $\tilde{F}_j \to \tilde{F}'_j$. They are split into $N_j$ segments, where $N_2 = 16, N_3 = 13, N_4 = 10$ and $N_1$ is unspecified. On the curve $\bP_1$, a depiction (not to scale) of the corresponding rhombus graph is given. 
  We associate a matrix $U_j$ for going right-to-left across the path ${\bf P}_j$, which gives a flat connection on the punctured plane.
  (Right) For some $R \gg 1$, we choose faces $F_{j,k} \approx R \tilde{F}_{j,k}$. Then for each of the segments $\bP_{j,k} : \tilde{F}_{j,k-1} \to \tilde{F}_{j,k}$, we associate a monotone path on the rhombus graph $\mathrm{zip} \, j,k : F_{j,k-1} \to F_{j,k}$. On the top, we depict the straight-line zipper path together with the connection that we put in the graph, associated with the connection in the continuum. On the bottom, we show the allowed range of rhombus angles for these straight-line paths.}
  \label{fig:multipleZippers_example}
\end{subfigure}
\begin{subfigure}
  \centering
  \includegraphics[width=\linewidth]{multipleZippers_splitUpAngle.pdf}
  \caption{ Depicition of necessary conditions for splitting up the curve into sufficiently small segments, namely the existence of angles $\theta^{\mathrm{mid} \,,\, j}_{\{j,j'\},\{k,k'\}} , \theta^{\mathrm{mid} \,,\, j'}_{\{j,j'\},\{k,k'\}}$ and a midpoint $\tilde{F}^{\mathrm{mid}}_{\{j,j'\},\{k,k'\}}$ for which $\tilde{F}_{j,k-1}, \tilde{F}_{j,k}$ and $\tilde{F}_{j',k'-1}, \tilde{F}_{j',k'}$ are in arcs that subtend an angle smaller than $\epsilon/8$ from the midpoint, and for which the rays from the midpoints with angles $\theta^{\mathrm{mid} \,,\, j}_{\{j,j'\},\{k,k'\}} , \theta^{\mathrm{mid} \,,\, j'}_{\{j,j'\},\{k,k'\}}$ going in both directions intersect both line segments.}
  \label{fig:multipleZippers_splitUpAngle}
\end{subfigure}
\end{figure}

As such, we will need to split up our paths into smaller pieces, and if there are crossing between the zippers, then possibly the asymptotics for $g_{v \to w}(z)$ for $v,w$ near the crossings will no longer be valid and the continuum expressions for the integrals wouldn't make sense. Throughout this section, we will assume that all of the zippers in our analysis are contained in a convex
\footnote{By `convex', we mean that for any two vertices in the region, the `straight-line' paths analogous to the ones defined in Section~\ref{sec:straightLine_zippers} and Fig.~\ref{fig:pathOnRhombusGraph} are contained within the region.}
region of the rhombus graph that satisfy the regularity conditions of Appendix~\ref{app:asympt_gFunc}. 

We starting by fixing some notation and talk about how exactly we will phrase the `continuum limit' of our problem. First, note that in the continuum limit, the relative distances between the $\{F_j\}$ and $\{F'_j\}$ will be much larger than the scales relevant to the asymptotics of $g(z)$ and the length 1 of the rhombus edges. And, we are considering the limit as the relative distances goes to $\infty$. As such, we will instead find it convenient to define \textit{fixed} complex numbers $\{\tilde{F}_j\}$ and $\{\tilde{F}'_j\}$ and a additional parameter $R$ for which the coordinate of the faces will be $F_j \approx R \tilde{F}_j$ and $F'_j \approx R \tilde{F}'_j$, where `$\approx$' means that the $F_j,F'_j$ are the closest dual face to the given complex number. We will be looking at terms at leading order in $R$ as $R \to \infty$.

Next, we need to consider how to choose our zippers. First, we want to choose \textit{continuum} paths $\bP_j : \tilde{F}_j \to \tilde{F}'_j$ in the complex plane that are smooth and do not mutually intersect. We will want our zippers to `approximate' these continuum paths in some sense. Recall that in the case of a single zipper, we essentially chose the continuum path $\bP_j$ to be a straight line between $F \to F'$, and we were able to approximate the zipper by choosing a monotone rhombus path. For our case, we may not be able to choose straight paths without them intersecting. So instead, we will split the continuous curve into sufficiently small but still finite-sized pieces on which the zippers can be chosen to be the monotone paths, where `sufficiently small' will be defined soon. In particular, we'll choose to split up the paths into small enough regions that the asymptotics of the functions $g_{F \to F'}(z)$ can be applied for $F,F'$ being endpoints of two different segments. Every continuous path $\bP_j : \tilde{F}_j \to \tilde{F}'_j$ will be split up and turned into some number $N_j$ of segments. There will be $N_j+1$ total points on each $\bP_j$, labeled by complex numbers $\tilde{F}_{j} =: \tilde{F}_{j,0}, \tilde{F}_{j,1}, \tilde{F}_{j,2}, \cdots, \tilde{F}_{j,N_j-1}, \tilde{F}_{j,N_j} := \tilde{F}'_j$. We will give a label $\bP_{j,k} : \tilde{F}_{j,k-1} \to \tilde{F}_{j,k}$ to each of the segments, which each correspond to a monotone, straight-line, segments. Also, we will use $\mathrm{zip} \, j,k : F_{j,k-1} \to F_{j,k}$ to refer to these straight-line segments on the rhombus graph. 
Also, we will find it useful to define the angle $e^{i \theta_{j,k}} = \frac{F_{j,k}-F_{j,k-1}}{|F_{j,k}-F_{j,k-1}|}$. 
See Fig.~\ref{fig:multipleZippers_example} for a depiction. 

Now, we discuss what we mean by `sufficiently small' for the decomposition into segments. This will require the definition of some new auxiliary quantities. The estimate Prop.~\ref{prop:g_expDecay_cone} from Appendix~\ref{app:asympt_gFunc} uses the regularity conditions to show that there exists an angle $\epsilon / 8$ (related to the $\epsilon$ used in the appendix) for which $g_{F \to F'}(z)$ is exponentially suppressed for $z = -|z|e^{i \theta}$ in an intermediate range of $D^{-2/3} < |z| < D^{2/3}$ angles in a range of $\theta \in (\arg(F'-F) - \epsilon/8, \arg(F'-F) + \epsilon/8)$. For the segments to be `sufficiently small' means that for each pair of segments $(\mathrm{zip} \, j,k \, , \, \mathrm{zip} \, j',k')$, there exists a complex number $\tilde{F}^{\mathrm{mid}}_{\{j,j'\},\{k,k'\}}$ in the region subtended by the pair, and two angles $\theta^{\mathrm{mid} \,,\, j}_{\{j,j'\},\{k,k'\}} , \theta^{\mathrm{mid} \,,\, j'}_{\{j,j'\},\{k,k'\}}$ such that both $\tilde{F}_{j,k}, \tilde{F}_{j,k+1}$ are contained in the cone $\tilde{F}^{\mathrm{mid}}_{\{j,j'\},\{k,k'\}} \bigsqcup_{\theta - \theta^{\mathrm{mid} \,,\, j}_{\{j,j'\},\{k,k'\}} \in (-\epsilon/8,\epsilon/8)} e^{i\theta}\R_{>0} $ emanating from $\tilde{F}^{\mathrm{mid}}_{\{j,j'\},\{k,k'\}}$ and so that both $\tilde{F}_{j',k'}, \tilde{F}_{j',k'+1}$ are in the cone $\tilde{F}^{\mathrm{mid}}_{\{j,j'\},\{k,k'\}} \bigsqcup_{\theta - \theta^{\mathrm{mid} \,,\, j'}_{\{j,j'\},\{k,k'\}} \in (-\epsilon/8,\epsilon/8)} e^{i\theta}\R_{>0} $. Furthermore, we require that the rays from $\tilde{F}^{\mathrm{mid}}_{\{j,j'\},\{k,k'\}}$ going out in both directions along the angles $\theta^{\mathrm{mid} \,,\, j}_{\{j,j'\},\{k,k'\}} , \theta^{\mathrm{mid} \,,\, j'}_{\{j,j'\},\{k,k'\}}$ intersect both segments. See Fig.~\ref{fig:multipleZippers_splitUpAngle} for a depiction.

Note that this decomposition is always possible. For $j \neq j'$, it is clear sufficiently subdivided curves will satisfy this since the distance between different segments is bounded from below. For $j = j'$, segments that are far enough apart and small enough can satisfy this. However for segments that are close, the fact that the curves are smooth means that within some neighborhood, the curves well approximate straight lines, and any choice of midpoint $\tilde{F}^{\mathrm{mid}}_{\{j,j'\},\{k,k'\}}$ along the $\bP_j$ between the two segments will be a suitable choice. For \textit{non-adjacent} segments within such a neighborhood, we want to choose $\tilde{F}^{\mathrm{mid}}_{\{j,j'\},\{k,k'\}}$ to be along $\bP_j$ \textit{strictly} in between the segments. For \textit{adjacent} segments $\bP_{j,k},\bP_{j,k+1}$ that share $\tilde{F}_{j,k}$, we will choose $\tilde{F}^{\mathrm{mid}}_{\{j,j\},\{k,k+1\}} = \tilde{F}_{j,k}$. 

\subsection{Cluster Expansion of the twisted determinant: multiple zippers} \label{sec:multZip_clusterExpansion}
Now we will descibe the cluster expansion used in this analysis. 
Consider the Kasteleyn matrix $\overline{\partial}(U)$ as described before, with a connection along the previously mentioned zippers connections contributing $U_1,\cdots,U_s$ to the monodromy going right-to-left across $F_1 \to F'_1, \cdots, F_s \to F'_s$. This is defined by
\begin{equation}
    \overline{\partial}(U) = 
    \underbrace{
        \overline{\partial}^{\bullet \to \circ} 
        + 
        (U^{-1}_1 - 1) \overline{\partial}^{\bullet \to \circ}_{\mathrm{zip} \, 1}
        + \cdots +
        (U^{-1}_s - 1) \overline{\partial}^{\bullet \to \circ}_{\mathrm{zip} \, s}
    }_{=: \overline{\partial}^{\bullet \to \circ}(U) }\
    \,\,+\,\,
    \underbrace{
        \overline{\partial}^{\bullet \to \circ}
        + 
        (U_1 - 1) \overline{\partial}^{\circ \to \bullet}_{\mathrm{zip} \, 1} 
        + \cdots + 
        (U_s - 1) \overline{\partial}^{\circ \to \bullet}_{\mathrm{zip} \, s} 
    }_{=: \overline{\partial}^{\circ \to \bullet}(U) }
\end{equation}
where each $\overline{\partial}^{\bullet \to \circ}_{\mathrm{zip} \, k}$ or $\overline{\partial}^{\circ \to \bullet}_{\mathrm{zip} \, k}$ is the restriction of the $\overline{\partial}$ operator to the $k$th zipper in the black-to-white or white-to-black directions of the lattice. 
As such, we can expand 
\begin{equation*}
\begin{split}
    & \log 
    \frac{ \det \overline{\partial}^{\bullet \to \circ}(U)}{ \det \overline{\partial}^{\bullet \to \circ}}
    = \tr \log \left( \mathbbm{1} 
        - (1-U^{-1}_1) \frac{\overline{\partial}^{\bullet \to \circ}_{\mathrm{zip} \, 1}}{\overline{\partial}^{\bullet \to \circ}}
        - \cdots 
        - (1-U^{-1}_s) \frac{\overline{\partial}^{\bullet \to \circ}_{\mathrm{zip} \, s}}{\overline{\partial}^{\bullet \to \circ}}
      \right)
    \\
    & = 
    -\sum_{n=1}^{\infty} \frac{1}{n} \sum_{i_1 \cdots i_n \in \{1 \cdots s\}} 
        \tr \left[ (1-U^{-1}_1) \cdots (1-U^{-1}_n) \right] \cdot
        \tr \left[
            \frac{\overline{\partial}^{\bullet \to \circ}_{\mathrm{zip} \, i_1}}{\overline{\partial}^{\bullet \to \circ}}
            \cdots
            \frac{\overline{\partial}^{\bullet \to \circ}_{\mathrm{zip} \, i_n}}{\overline{\partial}^{\bullet \to \circ}}
        \right].
\end{split}
\end{equation*}

Note that the terms where $i_1 = i_2 = \cdots = i_n = j$ for all $n$ with some fixed $j$ correspond to the series expansion of 
$\log
\frac
{ \det \left( \overline{\partial}^{\bullet \to \circ} + (1-U_j^{-1})\overline{\partial}^{\bullet \to \circ}_{\mathrm{zip} \, j} \right) }
{ \det \left( \overline{\partial}^{\bullet \to \circ} \right) }$, which was already analyzed in Sec.~\ref{sec:singleZipper}.
We will consider subtracting out these terms and analyzing the remainder. So, we can similarly expand out
\begin{equation} \label{eq:seriesExpansion_multZippers_1}
\begin{split}
    & 
    \frac{ \det \overline{\partial}^{\bullet \to \circ}(U)}{ \det \overline{\partial}^{\bullet \to \circ}}
    -
    \sum_{j=1}^{s}
    \log \frac{ \det \left( \overline{\partial}^{\bullet \to \circ} + (1-U_j^{-1})\overline{\partial}^{\bullet \to \circ}_{\mathrm{zip} \, j} \right) }{ \det \overline{\partial}^{\bullet \to \circ} }
    \\
    & = 
    -\sum_{n=1}^{\infty} \frac{1}{n} \sum_{\substack{i_1 \cdots i_n \in \{1 \cdots s\} \\ i_1 \cdots i_n \text{ not all same}}} 
        \tr \left[ (1-U^{-1}_1) \cdots (1-U^{-1}_n) \right] \cdot
        \tr \left[
            \frac{\overline{\partial}^{\bullet \to \circ}_{\mathrm{zip} \, i_1}}{\overline{\partial}^{\bullet \to \circ}}
            \cdots
            \frac{\overline{\partial}^{\bullet \to \circ}_{\mathrm{zip} \, i_n}}{\overline{\partial}^{\bullet \to \circ}}
        \right].
\end{split}
\end{equation}

Similarly to the series expansion of the single-zipper case, we will show that the leading order contributions to the traces $\tr \left[ \frac{\overline{\partial}^{\bullet \to \circ}_{\mathrm{zip} \, i_1}}{\overline{\partial}^{\bullet \to \circ}} \cdots \frac{\overline{\partial}^{\bullet \to \circ}_{\mathrm{zip} \, i_n}}{\overline{\partial}^{\bullet \to \circ}} \right]$ can be expressed by some continuous integrals over $n$ variables $w_1 \cdots w_n$.

\subsection{Evaluating $\tr \left[ \frac{\overline{\partial}^{\bullet \to \circ}_{\mathrm{zip} \, i_1}}{\overline{\partial}^{\bullet \to \circ}} \cdots \frac{\overline{\partial}^{\bullet \to \circ}_{\mathrm{zip} \, i_n}}{\overline{\partial}^{\bullet \to \circ}} \right]$ for well-separated zippers} \label{sec:multZip_evaluating_clusterExp_terms}

We want to asymptotically evaluate the terms on the left-hand side of Eq.~\eqref{eq:seriesExpansion_multZippers_1} when the $i_1 \cdots i_n$ are not all the same, in the limit of the parameter $R \to \infty$ for fixed $\{ \tilde{F}_j, \tilde{F}'_j \}$. In particular the leading order term will schematically be of the form 
\begin{equation*}
    \tr \left[ \frac{\overline{\partial}^{\bullet \to \circ}_{\mathrm{zip} \, i_1}}{\overline{\partial}^{\bullet \to \circ}} \cdots \frac{\overline{\partial}^{\bullet \to \circ}_{\mathrm{zip} \, i_n}}{\overline{\partial}^{\bullet \to \circ}} \right] \xrightarrow{R \to \infty} 
    \text{``} 
    \frac{1}{(2\pi i)^n}
    \int_{\bP_{i_1}} dw_1 \cdots \int_{\bP_{i_n}} dw_n \frac{1}{(w_1 - w_2) \cdots (w_n - w_1)} + \mathrm{c.c.}
    \text{''}
\end{equation*}
A priori, the integral in quotes on the right-hand-side above is not well-defined because it will be singular whenever $i_{k} = i_{k+1}$. However, we will see instead that the integral above will be a certain \textit{regularization} of the singular expression above. The regularization will involve slightly deforming the contours with with the same $i_{k} = \cdots = i_{k+k'} =: j$ into $k'+1$ contours that are close to $\bP_{j}$ that don't mutually intersect except at the endpoints $\tilde{F}_{j}, \tilde{F}'_j$.

We now describe the particular regularization that pops out of the series expansion. First, note that since $\bP_j$ is a directed curve in the complex plane, there is a canonical `left' and `right' side of the curve. Now suppose $i_{k} = \cdots = i_{k+r-1} =: j \in \{1 \cdots s\}$ are a maximal set of cyclically adjacent indices that are all equal among the indices, i.e. for which $i_{k-1} \neq j$, $i_{k+r} \neq j$. Then we can define $\bP^{(1)}_{j} \cdots \bP^{(r)}_{j}$ to be $r$ non-intersecting curves near $\bP_{j}$ with endpoints at $\tilde{F}_j \to \tilde{F}'_j$ with `left-to-right' order  $\bP^{(r)}_{j} \to \bP^{(r-1)}_{j} \to \cdots \to \bP^{(2)}_{j} \to \bP^{(1)}_{j}$. Then, we'll have the regulated integrals that show up in the series expansions will have $w_{k}, \cdots, w_{k+r-1}$ vary over the curves $\bP^{(1)}_{j}, \cdots, \bP^{(r)}_{j}$. See the below proposition for the formal statement and Figs.~\ref{fig:regulatedContours_example},\ref{fig:splitContours_appendix} for various examples.

\begin{figure}[h!]
  \centering
  \includegraphics[width=\linewidth]{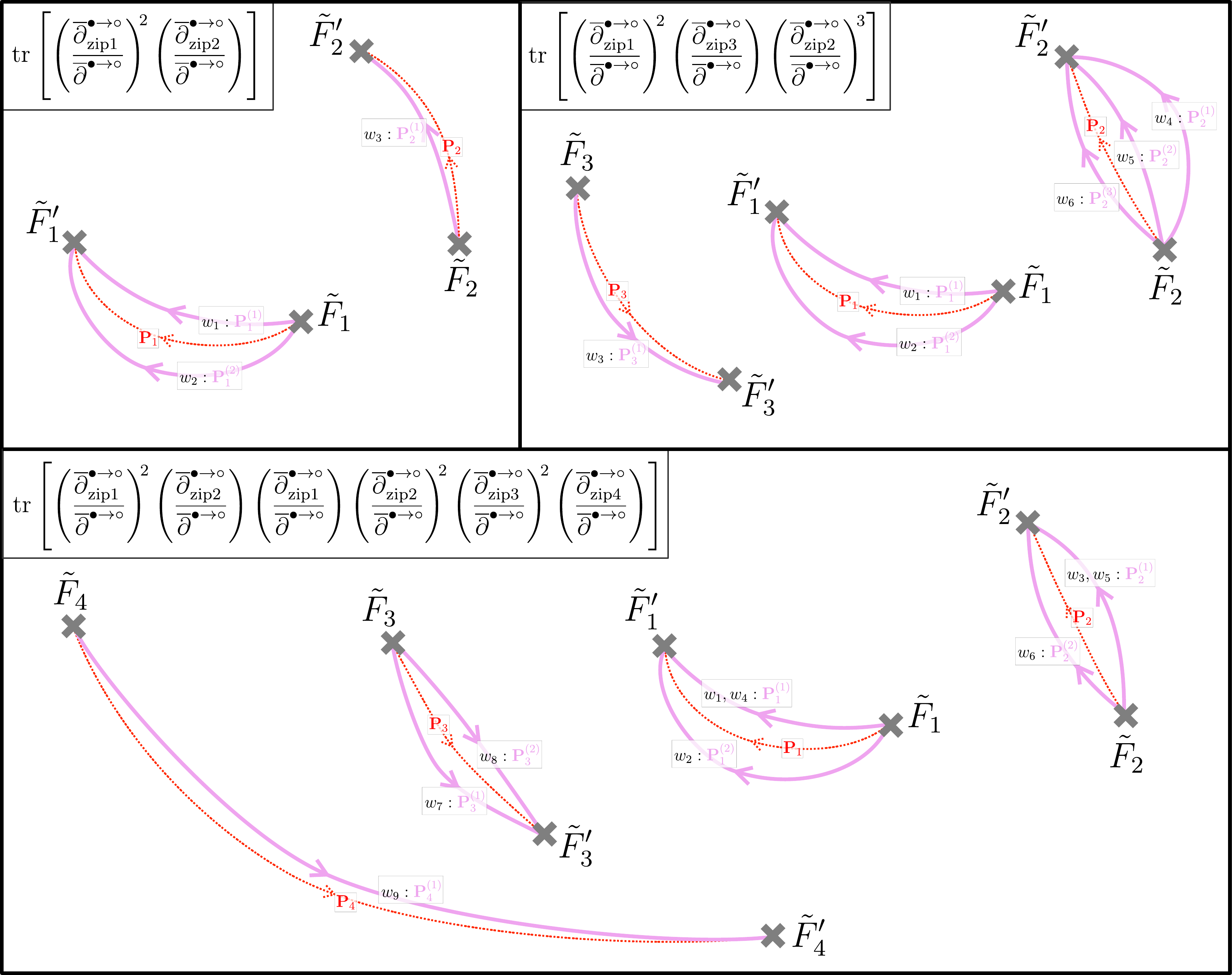}
  \caption{ Various examples of the regulated integration contours of the integrand $\frac{1}{(2 \pi i)^n} \frac{1}{(w_1-w_2) \cdots (w_n-w_1)}$ that appear in the evaluation of $\tr \left[ \frac{\overline{\partial}^{\bullet \to \circ}_{\mathrm{zip} \, i_1}}{\overline{\partial}^{\bullet \to \circ}} \cdots \frac{\overline{\partial}^{\bullet \to \circ}_{\mathrm{zip} \, i_n}}{\overline{\partial}^{\bullet \to \circ}} \right]$. The dotted red lines are the original smooth curves $\{\bP_{\cdots}\}$ used to define the zippers, and are labeled. The solid pink lines are the deformed contours $\{\bP^{(\cdots)}_{\cdots}\}$, and these paths together with the assignments of variables $w_1 \cdots w_n$ to the path (corresponding to the integral in the Prop.~\ref{prop:regulatedContours_integral}) are labeled. (Top-Left) An $n=3$ example with two zippers. (Top-Right) An $n=6$ example with three zippers. (Bottom) An $n=9$ example with four zippers. }
  \label{fig:regulatedContours_example}
\end{figure}

\begin{prop} \label{prop:regulatedContours_integral}
    Let $i_1 \cdots i_n \in \{1 \cdots s\}$ not all be the same index. Let us give adjacent labels, where $i_1 = \cdots = i_{r_1} =: j_1$, $i_{r_1+1} = \cdots = i_{r_1+r_2} =: j_2$, $\cdots$, $i_{r_1+\cdots+r_{t-1}+1} = \cdots = i_{r_1+\cdots+r_{t-1}+r_t} =: j_t$,  with $r_1+\cdots+r_t=m$. Without loss of generality, suppose that $i_1=j_1 \neq j_t = i_n$. Then, we have
    \begin{equation} \label{eq:regulatedContours_integral}
    \begin{split}
        &\tr \left[ \frac{\overline{\partial}^{\bullet \to \circ}_{\mathrm{zip} \, i_1}}{\overline{\partial}^{\bullet \to \circ}} \cdots \frac{\overline{\partial}^{\bullet \to \circ}_{\mathrm{zip} \, i_n}}{\overline{\partial}^{\bullet \to \circ}} \right] \cdot (1 + O(1/R)) + O(1/R) 
        = \tr \left[ 
        \left(\frac{\overline{\partial}^{\bullet \to \circ}_{\mathrm{zip} \, j_1}}{\overline{\partial}^{\bullet \to \circ}}\right)^{r_1}
        \cdots 
        \left(\frac{\overline{\partial}^{\bullet \to \circ}_{\mathrm{zip} \, j_t}}{\overline{\partial}^{\bullet \to \circ}}\right)^{r_t} \right] \cdot (1 + O(1/R)) + O(1/R)
        \\
        &=
        \frac{1}{(2 \pi i)^n} 
        \underbrace{
        \int_{\bP^{(1)}_{j_1}} dw_{1}
        \cdots
        \int_{\bP^{(r_1)}_{j_1}} dw_{r_1}}_{i_1 = \cdots = i_{r_1} = j_1} 
        \cdots
        \underbrace{\int_{\bP^{(1)}_{j_t}} dw_{r_1+\cdots+r_{t-1}+1} 
        \cdots \int_{\bP^{(r_t)}_{j_t}} dw_n
        }_{i_{r_1+\cdots+r_{t-1}+1} = \cdots = i_n = j_t}
        \frac{1}{(w_1 - w_2) \cdots (w_n-w_1)} + \mathrm{c.c.}
    \end{split}
    \end{equation}
    where the contours $\bP^{(\cdots)}_{j_{\cdots}}$ are a choice of regulated contours as defined above. The constants in the different $O(1/R)$ can be chosen to depend only on the continuous paths $\{\bP_j\}$ and their decompositions into segments.
\end{prop}
The choice of $i_1 \neq i_n$ is done for notational convenience. Note that since the integrands and traces are cyclically symmetric in the $\{ i_1 \cdots i_n \}$, all cases of them not being all the same index can be reduced to the form above.

\begin{proof}
    We will use the decomposition of the zippers into the segments $\{\mathrm{zip} \, j,k = F_{j,k} \to F_{j,k+1}\}$. So for each of the full zippers `$\mathrm{zip} \, j$', we can write
    \begin{equation}
        \overline{\partial}^{\bullet \to \circ}_{\mathrm{zip} \, j} = \overline{\partial}^{\bullet \to \circ}_{\mathrm{zip} \, j,1} + \cdots + \overline{\partial}^{\bullet \to \circ}_{\mathrm{zip} \, j,N_j}
    \end{equation}
    where each of the $\overline{\partial}^{\bullet \to \circ}_{\mathrm{zip} \, j,k}$ is the restriction of the $\overline{\partial}$ operator to the black-to-white part of to the graph along the zipper. From here, our strategy will be to analyze the traces $\tr \left[ \frac{\overline{\partial}^{\bullet \to \circ}_{\mathrm{zip} \, i_1,i'_1}}{\overline{\partial}^{\bullet \to \circ}} \cdots \frac{\overline{\partial}^{\bullet \to \circ}_{\mathrm{zip} \, i_n,i'_n}}{\overline{\partial}^{\bullet \to \circ}} \right]$, where each $i'_k \in \{1 \cdots N_{i_k}\}$. 
    
    It turns out that these traces will have precisely the same form the main proposition. We first establish some notation. Let $\{\bP_{i_{k},i'_{k}}\}$ be the continuous curves along the segments. And for some fixed $j,j'$, we analogously let $\bP^{(1)}_{j,j'}, \bP^{(2)}_{j,j'}, \bP^{(3)}_{j,j'}, \cdots$ be some sequence of curves $\tilde{F}_{j,j'-1} \to \tilde{F}_{j,j'}$ nearby the original $\bP_{j,j'}$ that don't intersect except at the endpoints and so that $\bP^{(1)}_{j,j'} \to \bP^{(2)}_{j,j'} \to \bP^{(3)}_{j,j'} \to \cdots$ is the right-to-left order of the curves with respect to the orientation of $\bP_{j,j'}$. Then, we'll have the following lemma
    \begin{lem} \label{lem:regulatedContours_integral_splitUp}
    Let $q_1 \cdots q_{\tau} \in \{1 \cdots s\}$ not all be the same, and let each $q'_{k} \in \{1 \cdots N_{q_k}\}$, and let positive integers $\{\rho_{k}\}$ satisfy $\rho_1 + \cdots \rho_{\tau} = n$. Then,
    \begin{equation} \label{eq:regulatedContours_integral_splitUp}
    \begin{split}
        &\tr \left[ 
        \left( \frac{\overline{\partial}^{\bullet \to \circ}_{\mathrm{zip} \, q_1,q'_1}}{\overline{\partial}^{\bullet \to \circ}} \right)^{\rho_1}
        \cdots
        \left( \frac{\overline{\partial}^{\bullet \to \circ}_{\mathrm{zip} \, q_{\tau},q'_{\tau}}}{\overline{\partial}^{\bullet \to \circ}} \right)^{\rho_{\tau}}
        \right] \cdot (1 + O(1/R)) + O(1/R)
        \\
        &=
        \frac{1}{(2 \pi i)^n} 
        \underbrace{\int_{\bP^{(1)}_{q_1, q'_1}} dw_{1} \cdots \int_{\bP^{(\rho_1)}_{q_1, q'_1}} dw_{\rho_1}}
        \cdots
        \underbrace{\int_{\bP^{(1)}_{q_{\tau}, q'_{\tau}}} dw_{\rho_1+\cdots+\rho_{\tau-1} + 1} \cdots \int_{\bP^{(\rho_{\tau})}_{q_{\tau}, q'_{\tau}}} dw_n}
        \frac{1}{(w_1 - w_2) \cdots (w_n-w_1)} + \mathrm{c.c.}
    \end{split}
    \end{equation}
    \end{lem}
    Note that this proposition implies the full result. In particular, we can first specify the $\{q_k\} , \{\rho_k\}$ so that 
   \begin{equation*}
       ( \underbrace{q_1, \cdots, q_1}_{\rho_1 \text{ times}} ,
         \underbrace{q_2, \cdots, q_2}_{\rho_2 \text{ times}} ,
         \cdots ,
         \underbrace{q_{\tau}, \cdots, q_{\tau}}_{\rho_{\tau} \text{ times}}
       ) = ( i_1, \cdots, i_n ).
   \end{equation*}
    Then we can consider summing over all possible $\{q'_k\}$. Summing over the left-hand-sides of Eq.~\eqref{eq:regulatedContours_integral_splitUp} clearly yields the left-hand-side of Eq.~\eqref{eq:regulatedContours_integral}. Similarly, one can sum over the right-hand-sides of Eq.~\eqref{eq:regulatedContours_integral_splitUp} to yield the right-hand-side of Eq.~\eqref{eq:regulatedContours_integral}, possibly after some contour deformations. As such, proving the above lemma implies the main proposition.
    
    In the remainder of this proof, instead proving the formula for the general case, we will more explicitly show the formula for various simpler cases to avoid tedious notation and explain at the end how the steps generalize in the most general case. 
    
    \subsubsection{Setup and Notation} \label{sec:setupNotation_first_proof}
    
    We introduce some notational compression as follows. Denote $F_{q_j,q'_j-1} =: {\bf F}_{j}$ and $F_{q_j,q'_j} =: {\bf F}'_{j}$ for $j = 1, \cdots, \tau$, so that all of the segments ${\bf F}_{j} \to {\bf F}'_{j}$ will be distinct. 
    Furthermore, denote the paths on the rhombus graph as $F^{\mathrm{mid}}_{\{q_{j-1},q_j\},\{q'_{j-1},q'_j\}} =: {\bf F}^{m}_{j-1,j}$, $\theta^{q_{j},q'_{j}} =: \Theta_{j}$ , $\theta^{\mathrm{mid} , q_{j-1}}_{\{q_{j-1},q_j\},\{q'_{j-1},q'_j\}} =: \Theta^{m ,  j-1}_{j-1,j}$ and $\theta^{\mathrm{mid} \,,\, q_{j}}_{\{q_{j-1},q_j\},\{q'_{j-1},q'_j\}} =: \Theta^{m,j}_{j-1,j}$. 
    For the continuum coordinates, similarly denote $\tilde{F}_{q_j,q'_j-1} =: \tilde{{\bf F}}_{j}$ and $\tilde{F}_{q_j,q'_j} =: \tilde{{\bf F}}'_{1}$, and also $\tilde{{\bf F}}^{\mathrm{mid}}_{\{q_{j-1},q_j\},\{q'_{j-1},q'_j\}} =: \tilde{{\bf F}}^{m}_{j-1,j}$.
    And denote each $\mathrm{zip} \, {q_j,q'_j} =: \widetilde{\mathrm{zip}}_{j}$. We will also denote the continuum paths $\bP^{(j)}_{q_k,q'_k}$ as $\tilde{{\bf F}}_{k} \xrightarrow{(j)} \tilde{{\bf F}}'_{k}$ and also $\bP_{q_k,q'_k}$ as $\tilde{{\bf F}}_{k} \to \tilde{{\bf F}}'_{k}$. In all the above, the indices are cyclic so that $j-1 = 0$ corresponds to the index $\tau$.
    
    To show the lemma, we first claim the following exact equality. 
    \begin{equation} \label{eq:multZip_telescopedSum}
    \begin{split}
        &\tr \left[ 
        \left( \frac{\overline{\partial}^{\bullet \to \circ}_{\widetilde{\mathrm{zip}}_{1}}}{\overline{\partial}^{\bullet \to \circ}} \right)^{\rho_1}
        \cdots
        \left( \frac{\overline{\partial}^{\bullet \to \circ}_{\widetilde{\mathrm{zip}}_{\tau}}}{\overline{\partial}^{\bullet \to \circ}} \right)^{\rho_{\tau}}
        \right] 
        \\ 
        &=\quad \frac{1}{(2 \pi i)^n} 
          {\int}_{0 \to -\exp\left[i \Theta^{m,1}_{\tau,1}\right]\infty}
          dz_1 \,\,
          {\int}_{0 \to -i e^{i \Theta_{1}} \infty}
          dz_2 \,\,
          \,\cdots\, 
          {\int}_{0 \to -i e^{i \Theta_{1}} \infty}
          dz_{\rho_1} \,\,
        \\
        &\quad\quad
          {\int}_{0 \to -\exp\left[i \Theta^{m,2}_{1,2}\right]\infty}
          dz_{\rho_1+1} \,\,
          {\int}_{0 \to -i e^{i \Theta_{2}} \infty}
          dz_{\rho_1+2} \,\,
          \,\cdots\, 
          {\int}_{0 \to -i e^{i \Theta_{2}} \infty}
          dz_{\rho_1+\rho_2} \,\,
        \\
        &\quad\quad \cdots
        \\
        &\quad\quad
         {\int}_{0 \to -\exp\left[i \Theta^{m,\tau}_{\tau-1,\tau}\right]\infty}
         dz_{\rho_1+\cdots+\rho_{\tau-1}+1} \,\,
         {\int}_{0 \to -i e^{i \Theta_{\tau}} \infty}
         dz_{\rho_1+\cdots+\rho_{\tau-1}+2} \,\,
         \,\cdots\, 
         {\int}_{0 \to -i e^{i \Theta_{\tau}} \infty}
         dz_n \,\,
        \frac{1}{(z_1-z_2) \cdots (z_n-z_1)}
        \\
        &\quad\quad\quad
        \cdot\big\{ \,
        \mathcal{G}_1\left(z_{1},\cdots,z_{\rho_1+1}\right)
        \cdot
        \mathcal{G}_2\left(z_{\rho_1+1},z_{\rho_1+2},\cdots,z_{\rho_1+\rho_2+1}\right)
        \,\cdots\,
        \mathcal{G}_{\tau}\left(z_{\rho_1+\cdots+\rho_{\tau-1}+1},z_{\rho_1+\cdots+\rho_{\tau-1}+2},\cdots,z_n,z_1\right)
        \big\}
        \\&\\
        &\text{where for each $j = 1 \cdots \tau$, for $z_{\rho_{1} + \cdots + \rho_{\tau} + 1} = z_{n+1} := z_1$, we define the functions}
        \\
        &\mathcal{G}_{j}\left(z_{\rho_1+\cdots+\rho_{j-1}+1},\cdots,z_{\rho_1+\cdots+\rho_{j}+1}\right)
        =
        \begin{cases}
            &\left( 
                \frac{g_{{\bf F}^{m}_{j-1,j} \to {\bf F}_{j}}(z_{\rho_1+\cdots+\rho_{j-1}+1})}
                     {g_{{\bf F}^{m}_{j,j+1} \to {\bf F}_{j}}(z_{\rho_1+\cdots+\rho_{j-1}+2})}
                -
                \frac{g_{{\bf F}^{m}_{j-1,j} \to {\bf F}'_{j}}(z_{\rho_1+\cdots+\rho_{j-1}+1})}
                     {g_{{\bf F}^{m}_{j,j+1} \to {\bf F}'_{j}}(z_{\rho_1+\cdots+\rho_{j-1}+2})}
            \right) \quad\text{if $\rho_j = 1$}
            \\ & \dotfill \\
            &\left( 
                \frac{g_{{\bf F}^{m}_{j-1,j} \to {\bf F}_{j}}(z_{\rho_1+\cdots+\rho_{j-1}+1})}
                     {1}
                - 
                \frac{g_{{\bf F}^{m}_{j-1,j} \to {\bf F}'_{j}}(z_{\rho_1+\cdots+\rho_{j-1}+1})}
                     {g_{{\bf F}_{j} \to {\bf F}'_{j}}(z_{\rho_1+\cdots+\rho_{j-1}+2})}
            \right)
            \\
            &\times
            \left( 
                1 - \frac{g_{{\bf F}_{j} \to {\bf F}'_{j}}(z_{\rho_1+\cdots+\rho_{j-1}+2})}{g_{{\bf F}_{j} \to {\bf F}'_{j}}(z_{\rho_1+\cdots+\rho_{j-1}+3})} 
            \right)
            \times \cdots \times
            \left( 
                1 - \frac{g_{{\bf F}_{j} \to {\bf F}'_{j}}(z_{\rho_1+\cdots+\rho_{j-1}+\rho_{j}-1})}{g_{{\bf F}_{j} \to {\bf F}'_{j}}(z_{\rho_1+\cdots+\rho_{j-1}+\rho_{j}})}
            \right)
            \\
            &\times \left( 
                \frac{1}
                     {g_{{\bf F}^{m}_{j,j+1} \to {\bf F}_{j}}(z_{\rho_1+\cdots+\rho_{j-1}+\rho_{j}+1})}
                - 
                \frac{g_{{\bf F}_{j}\to {\bf F}'_{j}}(z_{\rho_1+\cdots+\rho_{j-1}+\rho_{j}})}
                     {g_{{\bf F}^{m}_{j,j+1} \to {\bf F}'_{j}}(z_{\rho_1+\cdots+\rho_{j-1}+\rho_{j}+1})}
            \right) \quad\text{if $\rho_j > 1$}.
        \end{cases}
    \end{split}
    \end{equation}
    
    This equality follows from a few steps. First, we recall the fact from Eq.~\eqref{eq:invKast_contour_formula}
    \begin{equation*}
        \frac{1}{\overline{\partial}^{\bullet \to \circ}}(w,b) = \frac{1}{2\pi} \int_{0 \to -e^{i \theta} \infty} dz g_{w \to b}(z) 
        \quad \text{ for any angle } \theta \in (\arg(b-w)-\frac{\pi}{2},\arg(b-w)+\frac{\pi}{2})
    \end{equation*}
    from Lemma~\ref{lem:nEdges_zipIdLemma}. Next, we note that $\theta = \Theta_j + \frac{\pi}{2}$ is an angle that works uniformly for all white vertices $w \in \widetilde{\mathrm{zip}}_j$ and black vertices $b \in \widetilde{\mathrm{zip}}_{j}$, as in Fig.~\ref{fig:pathOnRhombusGraph_connection}. And similarly, $\theta = \Theta^{m,j}_{j-1,j}$ works uniformly for for all white vertices $w \in \widetilde{\mathrm{zip}}_{j-1}$ and black vertices $b \in \widetilde{\mathrm{zip}}_{j}$, which can be seen from Fig.~\ref{fig:multipleZippers_splitUpAngle}. Then, the equality above follows in a similar way to the Corollary~\ref{cor:traceId_first}. In particular, summing the equality of Lemma~\ref{lem:nEdges_zipIdLemma} over all edges in the zipper gives the integrand above via a telescoping sum similar to Corollary~\ref{cor:traceId_first} by choosing each $\tilde{v}_{\rho_1 + \cdots + \rho_{j-1} + 1} =  {\bf F}^{m}_{j-1,j}$ and each $\tilde{v}_{\rho_1 + \cdots + \rho_{j-1} + 2} = \cdots = \tilde{v}_{\rho_1 + \cdots + \rho_{j-1} + \rho_{j}} = {\bf F}_{j}$ if $\rho_{j} > 1$, and the equality follows from the shared integration contours of all expressions. 
    
    \subsubsection{A first example}
    
    The first case we consider will be with $n = 3$ and $\tau = 3$, so that necessarily $\rho_1 = \rho_2 = \rho_3 = 1$. One can look at Fig.~\ref{fig:example_small_zippers_proof} for visual aid, but the specifics of the zipper configurations will not yet be important. The equality Eq.~\eqref{eq:multZip_telescopedSum} will read
    \begin{equation*}
    \begin{split}
        \tr \left[ 
        \frac{\overline{\partial}^{\bullet \to \circ}_{\widetilde{\mathrm{zip}} \, 1}}{\overline{\partial}^{\bullet \to \circ}} 
        \frac{\overline{\partial}^{\bullet \to \circ}_{\widetilde{\mathrm{zip}} \, 2}}{\overline{\partial}^{\bullet \to \circ}} 
        \frac{\overline{\partial}^{\bullet \to \circ}_{\widetilde{\mathrm{zip}} \, 3}}{\overline{\partial}^{\bullet \to \circ}} 
        \right] =
        &
        \frac{1}{(2 \pi i)^3} 
        \int_{0}^{-\exp\left[i \Theta^{m , 1}_{3,1}\right] \infty} dz_{1}
        \int_{0}^{-\exp\left[i \Theta^{m , 2}_{1,2}\right] \infty} dz_{2}
        \int_{0}^{-\exp\left[i \Theta^{m , 3}_{2,3}\right] \infty} dz_{3}
        \frac{1}{(z_1 - z_2)(z_2 - z_3)(z_3 - z_1)} 
        \\
        & \times 
        \left( 
            \frac{g_{{\bf F}^{m}_{3,1} \to {\bf F}_{1}} (z_1)}{g_{{\bf F}^{m}_{1,2} \to {\bf F}_{1}} (z_2)}
          - \frac{g_{{\bf F}^{m}_{3,1} \to {\bf F}'_{1}}(z_1)}{g_{{\bf F}^{m}_{1,2} \to {\bf F}'_{1}}(z_2)}
        \right)
        \left( 
            \frac{g_{{\bf F}^{m}_{1,2} \to {\bf F}_{2}} (z_2)}{g_{{\bf F}^{m}_{2,3} \to {\bf F}_{2}} (z_3)}
          - \frac{g_{{\bf F}^{m}_{1,2} \to {\bf F}'_{2}}(z_2)}{g_{{\bf F}^{m}_{2,3} \to {\bf F}'_{2}}(z_3)}
        \right)
        \left( 
            \frac{g_{{\bf F}^{m}_{2,3} \to {\bf F}_{3}} (z_3)}{g_{{\bf F}^{m}_{3,1} \to {\bf F}_{3}} (z_1)}
          - \frac{g_{{\bf F}^{m}_{2,3} \to {\bf F}'_{3}}(z_3)}{g_{{\bf F}^{m}_{3,1} \to {\bf F}'_{3}}(z_1)}
        \right)
    \end{split}
    \end{equation*}
    Now from here, we can apply the asymptotics of the functions $g_{\cdots \to \cdots}(z)$ from Appendix~\ref{app:asympt_gFunc}. In particular, we have the exponential suppression of all factors $g_{{\bf F}^{m}_{j-1,j} \to {\bf F}_{j}} (z_j)$ , $g_{{\bf F}^{m}_{j-1,j} \to {\bf F}'_{j}} (z_j)$ , $\frac{1}{g_{{\bf F}^{m}_{j-1,j} \to {\bf F}_{j-1}} (z_{j})}$ , $\frac{1}{g_{{\bf F}^{m}_{j-1,j} \to {\bf F}'_{j-1}} (z_{j})}$ in intermediate values of $|z_{j}|$ because of how we  chose the segments ${\bf F}_{\cdots} \to {\bf F}_{\cdots}$ and angles $\Theta^{m , \cdots}_{\cdots,\cdots}$ as in Fig.~\ref{fig:multipleZippers_splitUpAngle}.
    The asymptotics and some more manipulations will give us 
    \begin{equation} \label{eq:firstExample_asymptotics_equation}
    \begin{split}
        &\tr \left[ 
        \frac{\overline{\partial}^{\bullet \to \circ}_{\widetilde{\mathrm{zip}} \, 1}}{\overline{\partial}^{\bullet \to \circ}} 
        \frac{\overline{\partial}^{\bullet \to \circ}_{\widetilde{\mathrm{zip}} \, 2}}{\overline{\partial}^{\bullet \to \circ}} 
        \frac{\overline{\partial}^{\bullet \to \circ}_{\widetilde{\mathrm{zip}} \, 3}}{\overline{\partial}^{\bullet \to \circ}} 
        \right] \cdot (1 + O(1/R)) + O(1/R) 
        \\
        &=
        \frac{1}{(2 \pi i)^3} 
        \int_{0}^{-\exp\left[i \Theta^{m , 1}_{3,1}\right]} dz_{1}
        \int_{0}^{-\exp\left[i \Theta^{m , 2}_{1,2}\right] \infty} dz_{2}
        \int_{0}^{-\exp\left[i \Theta^{m , 3}_{2,3}\right] \infty} dz_{3}
        \frac{1}{(z_1 - z_2)(z_2 - z_3)(z_3 - z_1)} 
        \\
        & \quad\quad\quad\quad\quad\quad\times
        \left( 
            \exp\left[z_1 \overline{({\bf F}_1 - {\bf F}^{m}_{3,1})} - z_2 \overline{({\bf F}_1 - {\bf F}^{m}_{1,2})}\right]
          - \exp\left[z_1 \overline{({\bf F}'_1 - {\bf F}^{m}_{3,1})} - z_2 \overline{({\bf F}'_1 - {\bf F}^{m}_{1,2})}\right]
        \right)
        \\
        & \quad\quad\quad\quad\quad\quad\times
        \left( 
            \exp\left[z_2 \overline{({\bf F}_2 - {\bf F}^{m}_{1,2})} - z_3 \overline{({\bf F}_2 - {\bf F}^{m}_{2,3})}\right]
          - \exp\left[z_2 \overline{({\bf F}'_2 - {\bf F}^{m}_{1,2})} - z_3 \overline{({\bf F}'_2 - {\bf F}^{m}_{2,3})}\right]
        \right)
        \\
        & \quad\quad\quad\quad\quad\quad\times
        \left( 
            \exp\left[z_3 \overline{({\bf F}_3 - {\bf F}^{m}_{2,3})} - z_1 \overline{({\bf F}_3 - {\bf F}^{m}_{3,1})}\right]
          - \exp\left[z_3 \overline{({\bf F}'_3 - {\bf F}^{m}_{2,3})} - z_1 \overline{({\bf F}'_3 - {\bf F}^{m}_{3,1})}\right]
        \right)
        \\
        & \quad +
        \frac{1}{(2 \pi i)^3} 
        \int_{-\exp\left[i \Theta^{m , 1}_{3,1}\right]}^{-\exp\left[i \Theta^{m , 1}_{3,1}\right] \infty} dz_{1}
        \int_{0}^{-\exp\left[i \Theta^{m , 2}_{1,2}\right] \infty} dz_{2}
        \int_{0}^{-\exp\left[i \Theta^{m , 3}_{2,3}\right] \infty} dz_{3}
        \frac{1}{(z_1 - z_2)(z_2 - z_3)(z_3 - z_1)} 
        \\
        & \quad\quad\quad\quad\quad\quad\quad\times
        \left( 
            \exp\left[\frac{1}{z_1} {({\bf F}_1 - {\bf F}^{m}_{3,1})} - \frac{1}{z_2} {({\bf F}_1 - {\bf F}^{m}_{1,2}})\right]
            - 
            \exp\left[\frac{1}{z_1} {({\bf F}'_1 - {\bf F}^{m}_{3,1})} - \frac{1}{z_2} {({\bf F}'_1 - {\bf F}^{m}_{1,2}})\right]
        \right)
        \\
        & \quad\quad\quad\quad\quad\quad\quad\times
        \left( 
            \exp\left[\frac{1}{z_2} {({\bf F}_2 - {\bf F}^{m}_{1,2})} - \frac{1}{z_3} {({\bf F}_2 - {\bf F}^{m}_{2,3})}\right]
            - 
            \exp\left[\frac{1}{z_2} {({\bf F}'_2 - {\bf F}^{m}_{1,2})} - \frac{1}{z_3} {({\bf F}'_2 - {\bf F}^{m}_{2,3})}\right]
        \right)
        \\
        & \quad\quad\quad\quad\quad\quad\quad\times
        \left( 
            \exp\left[\frac{1}{z_3} {({\bf F}_3 - {\bf F}^{m}_{2,3})} - \frac{1}{z_1} {({\bf F}_3 - {\bf F}^{m}_{3,1})}\right]
            - 
            \exp\left[\frac{1}{z_3} {({\bf F}'_3 - {\bf F}^{m}_{2,3})} - \frac{1}{z_1} {({\bf F}'_3 - {\bf F}^{m}_{3,1})}\right]
        \right)
        \\
        &=
        \frac{1}{(2 \pi i)^3} 
        \int_{0}^{-\exp\left[-i \Theta^{m , 1}_{3,1}\right] \infty} dz_{1}
        \int_{0}^{-\exp\left[-i \Theta^{m , 2}_{1,2}\right] \infty} dz_{2}
        \int_{0}^{-\exp\left[-i \Theta^{m , 3}_{2,3}\right] \infty} dz_{3}
        \frac{1}{(z_1 - z_2)(z_2 - z_3)(z_3 - z_1)} 
        \\
        & \quad\quad\quad\quad\quad\quad\times
        \left( e^{{\bf F}'_1  (z_1-z_2)} - e^{{\bf F}_1 (z_1-z_2)} \right)
        \left( e^{{\bf F}'_2  (z_2-z_3)} - e^{{\bf F}_2 (z_2-z_3)} \right)
        \left( e^{{\bf F}'_3  (z_3-z_1)} - e^{{\bf F}_3 (z_3-z_1)} \right)
        \\
        & \quad + \mathrm{c.c.} + O(1/R)
        \\
        &=
        \frac{1}{(+2 \pi i)^3} 
        \int_{0}^{-\exp\left[-i \Theta^{m , 1}_{3,1}\right] \infty} dz_{1}
        \int_{0}^{-\exp\left[-i \Theta^{m , 2}_{1,2}\right] \infty} dz_{2}
        \int_{0}^{-\exp\left[-i \Theta^{m , 3}_{2,3}\right] \infty} dz_{3}
        \int_{{\bf F}_1}^{{\bf F}'_1} dw_{1}
        \int_{{\bf F}_2}^{{\bf F}'_2} dw_{2}
        \int_{{\bf F}_3}^{{\bf F}'_3} dw_{3}
        \\
        & \quad\quad\quad\quad\quad\quad
        \left\{ e^{w_1 (z_1-z_2)} e^{w_2 (z_2-z_3)} e^{w_3 (z_3-z_1)} \right\}
        \\
        & \quad + \mathrm{c.c.} + O(1/R)
        \\
        &=
        \frac{1}{(+2 \pi i)^3}
        \int_{{\bf F}_1}^{{\bf F}'_1} dw_{1}
        \int_{{\bf F}_2}^{{\bf F}'_2} dw_{2}
        \int_{{\bf F}_3}^{{\bf F}'_3} dw_{3}
        \frac{1}{(w_1 - w_2)(w_2 - w_3)(w_3 - w_1)} 
        + \mathrm{c.c.} + O(1/R)
        \\
        &=
        \frac{1}{(+2 \pi i)^3}
        \int_{\tilde{{\bf F}}_1}^{\tilde{{\bf F}}'_1} dw_{1}
        \int_{\tilde{{\bf F}}_2}^{\tilde{{\bf F}}'_2} dw_{2}
        \int_{\tilde{{\bf F}}_3}^{\tilde{{\bf F}}'_3} dw_{3}
        \frac{1}{(w_1 - w_2)(w_2 - w_3)(w_3 - w_1)} 
        + \mathrm{c.c.} + O(1/R)
    \end{split}
    \end{equation}
    where the integrals $w_j : {\bf F}_j \to {\bf F}'_j$ must be taken to be homotopic to the straight line contours ${\bf F}_j \to {\bf F}'_j$ (a fact that will be explained soon), which is consistent with what we wanted to show. Now we explain the above steps and some subtleties in their derivation.
    
    The first equality above is entirely analogous to the first equality of Eq.~\eqref{eq:asympts_oneZipper_order2_part0}, which required a few steps of justification. First, it required the asymptotics of $g_{\cdots \to \cdots}(z_j)$ for the values of $|z_j|$ small (giving the integrand of the first term) and of $|z_j|$ large (giving the second integrand).
    \footnote{A subtle point is that with our choices, sometimes ${\bf F}^{m}_{j,j+1}$ equals one of the ${\bf F}_{j} , {\bf F}'_{j} , {\bf F}_{j+1} , {\bf F}'_{j+1}$ precisely when the two segments correspond to adjacent segments along the zipper. In this case each $|{\bf F}^{m}_{j,j+1} - {\bf F}_{j}|$ , $|{\bf F}^{m}_{j,j+1} - {\bf F}'_{j}|$ , $|{\bf F}^{m}_{j,j+1} - {\bf F}_{j+1}|$ , $|{\bf F}^{m}_{j,j+1} - {\bf F}'_{j+1}|$ is either exactly zero or large when $R \gg 1$, so the asymptotic expansions still are valid, in particular because $g_{F \to F}(z) = 1$ for any face $F$ matches up with $\exp[z(F-F)]=1$.}
    Then, a similar computation as in the Lemma~\ref{lem:cornerRegion_oneZip} shows that the `corner regions' give an $O(1/R)$ contribution to the integrals. Then, extending the asymptotics from the regions of $|z_j|$ small to the whole integration regions give another small $O(1/R)$ contribution.

    The second equality again has a few parts. First, the integrand simplifies because the factors involving the auxiliary midpoints ${\bf F}^{m}_{j-1,j}$ cancel out.
    \footnote{This should not be a surprise, since the original integrand in terms of the $g_{\cdots \to \cdots}$ didn't depend on the exact choice of of the midpoints. Rather, those midpoints were needed as auxiliaries to justify the asymptotic expressions above.}
    Next, we extend the contours of integration of the $z_1$ contour from $|z_1|: 0 \to 1$ to $|z_1|: 0 \to \infty$ and respectively $|z_1|: 1 \to \infty$ to $|z_1|: 0 \to \infty$ for each of the terms. This only gives an extra $O(1/R)$ factor.
    \footnote{
        Note a difference in this computation versus Eq.~\eqref{eq:asympts_oneZipper_order2_part0} in the single zipper case is that we were not able to extend the $z_1$ contour to infinity, because the integrands were more singular and extending the contour would not give a small term (in fact it would diverge) if we did that. Since the integrands here decay more rapidly, we are allowed to extend the integrals at the cost of a small term. Crucially, this depended on our condition that $q_1 \neq q_{\tau}$ (here, $\tau = 3$) which implied that the segments lie on different zippers and the faces $\{{\bf F}_1,{\bf F}'_1\}$ are distinct from $\{{\bf F}_{3},{\bf F}'_{3}\}$. If instead we had, say ${\bf F}_1 = {\bf F}_3$, then 
        $\frac{\left( e^{{\bf F}_3  (z_3-z_1)} - e^{{\bf F}'_3 (z_3-z_1)} \right) \left( e^{{\bf F}_1  (z_1-z_2)} - e^{{\bf F}'_1 (z_1-z_2)} \right)}{(z_3-z_1)(z_1-z_2)} = 
        e^{{\bf F}'_1 (z_3-z_2)} \frac{\left(1 - e^{({\bf F}'_1 - {\bf F}_1)(z_1-z_2)}\right)\left(1 - e^{({\bf F}'_3 - {\bf F}_1)(z_3-z_1)}\right)}{(z_3-z_1)(z_1-z_2)}$ 
        would \textit{not} be exponentially suppressed in $z_1$, so it would not be justified to extend the integration contour.
        \label{foot:extension_warning}
    }
    
    Then, we note that upon changing all variables $z_j \mapsto \frac{1}{z_j}$, the second term and first term are complex conjugates; we choose express the result in terms of the second term.
    
    The third equality follows from the Schwinger parameterizations $ \int_{{\bf F}_j}^{{\bf F}'_j} dw_j e^{w_j (z_{j}-z_{j+1})} = \frac{e^{{\bf F}'_j  (z_{j}-z_{j+1})} - e^{{\bf F}_j (z_{j}-z_{j+1})}}{z_{j}-z_{j+1}}$. At this stage, the exact choice of contours $w_j : {\bf F}_j \to {\bf F}'_j$ may not appear to matter a priori. However, to perform the integrals over all the $z_j$ in the fourth equality, this choice of contour being (homotopic to) a straight line is crucial. In particular, suppose instead that we chose a contour $w_j : {\bf F}_j \to {\bf F}'_j$ that `wrapped around' one of the contours ${\bf F}_{j \pm 1} \to {\bf F}'_{j \pm 1}$. Then, there would be values of $(w_j - w_{j-1})$ {\color{RoyalPurple} [resp. $(w_{j+1} - w_{j})$]} for which the integral over $z_j$ {\color{RoyalPurple} [resp. $z_{j+1}$]} does not converge, because $\int_{0}^{-e^{-i \Theta}\infty} dz \, e^{\kappa z}$ will diverge if $e^{i\Theta}$ and $\kappa$ form an obtuse angle. By Fig.~\ref{fig:multipleZippers_splitUpAngle},  $\mathrm{Re} ( -e^{-i\Theta^{m,j}_{j-1,j}} (w_j-w_{j-1}) ) < 0$ for $w_{j-1}$,$w_j$ along the straight-line segments ${\bf F}_{j-1} \to {\bf F}'_{j-1}$,${\bf F}_{j} \to {\bf F}'_{j}$, so the integral over $z_j$ converges. 

    Indeed, the value of the integral will change upon changing the contour to wrap around in such a way due to the poles at $w_{j} = w_{j \pm 1}$. Such subtleties of contour integration will play a big role in the resulting integrals of more general terms.
    
    The last equality follows from changing variables $w_{j} \mapsto R \cdot w_{j}$, which again changes the integral by a small $O(1/R)$ correction because each ${\bf F}_{j} \approx R \cdot \tilde{{\bf F}}_j$.
    
    \subsubsection{More Setup and More Notation} 
    Before moving on to more representative examples, it will be helpful to setup some more notation to help compactify the computations. 
    Generally, we would like to apply the asymptotics of the functions $g_{\cdots \to \cdots}(z)$ as we did in the previous example. 
    However, if we had some $\rho_j > 1$, then the contours of integration in Eq.~\eqref{eq:multZip_telescopedSum} for $z_{\rho_1 + \cdots + \rho_{j-1} + 2},\cdots,z_{\rho_1 + \cdots + \rho_{j-1} + \rho_{j}}$ would be $0 \to -i e^{i \Theta_j} \infty$ and the integrand would have factors of both $g_{{\bf F}_j \to {\bf F}'_j}$ and $\frac{1}{g_{{\bf F}_j \to {\bf F}'_j}}$. Both of these factors have asymptotic expansions that decay on opposite contours $0 \to -e^{i \Theta_{1}}\infty$ and $0 \to e^{i \Theta_{1}}\infty$, which are mutually incompatible and also incompatible with the original contour.
    
    The strategy to deal with this will be the similar to the single-zipper case 
    in that we need to split up the integrand into various parts so that only one of $g_{{\bf F}_j \to {\bf F}'_j}$ or $\frac{1}{g_{{\bf F}_j \to {\bf F}'_j}}$ 
    appears per term for these variables and then rotate their contours to their relevant rays for asymptotic analysis, often at the cost of residues.
    However, the details of this procedure will be slightly different. In particular, we'll split up the term containing variables
    $z_{\rho_1+\cdots+\rho_{j-1}+2},z_{\rho_1+\cdots+\rho_{j-1}+3}$ and rotate those contours, which would give a residue term.
    Then for the rest of the terms, we do the same thing recursively.
    
    As such, for any even integer $2\lambda_j$ with $0 \le 2\lambda_j \le \rho_j$, we are motivated to write
    
    \begin{equation}
    \begin{split}
        \mathcal{G}_j
        =
        \sum_{\zeta_j^{(2\lambda_j),2},\zeta_j^{(2\lambda_j),4},\cdots,\zeta_j^{(2\lambda_j),2\lambda_j} \in \{0,1\}}
        &g_{{\bf F}^{m}_{j-1,j} \to {\bf F}_{j}}(z_{\rho_1+\cdots+\rho_{j-1}+1})
        \\
        &\times\left( 1 - \frac{g_{{\bf F}_{j} \to {\bf F}'_{j}}(z_{\rho_1+\cdots+\rho_{j-1}+1})}
                        {g_{{\bf F}_{j} \to {\bf F}'_{j}}(z_{\rho_1+\cdots+\rho_{j-1}+2})} \right)
        \times
        \left(
            -\frac{g_{{\bf F}_{j} \to {\bf F}'_{j}}(z_{\rho_1+\cdots+\rho_{j-1}+2})}{g_{{\bf F}_{j} \to {\bf F}'_{j}}(z_{\rho_1+\cdots+\rho_{j-1}+3})} 
        \right)^{\zeta_j^{(2\lambda_j),2}}
        \\
        &\times
        \left( 1 - \frac{g_{{\bf F}_{j} \to {\bf F}'_{j}}(z_{\rho_1+\cdots+\rho_{j-1}+3})}{g_{{\bf F}_{j} \to {\bf F}'_{j}}(z_{\rho_1+\cdots+\rho_{j-1}+4})} \right)
        \times
        \left(
            -\frac{g_{{\bf F}_{j} \to {\bf F}'_{j}}(z_{\rho_1+\cdots+\rho_{j-1}+4})}{g_{{\bf F}_{j} \to {\bf F}'_{j}}(z_{\rho_1+\cdots+\rho_{j-1}+5})} 
        \right)^{\zeta_j^{(2\lambda_j),4}}
        \\
        &\times
        \cdots
        \\
        &\times
        \left( 1 - \frac{g_{{\bf F}_{j} \to {\bf F}'_{j}}(z_{\rho_1+\cdots+\rho_{j-1}+2\lambda_j-1})}{g_{{\bf F}_{j} \to {\bf F}'_{j}}(z_{\rho_1+\cdots+\rho_{j-1}+2\lambda_j})} \right)
        \times
        \left(
          - \frac{g_{{\bf F}_{j} \to {\bf F}'_{j}}(z_{\rho_1+\cdots+\rho_{j-1}+2\lambda_j})}{g_{{\bf F}_{j} \to {\bf F}'_{j}}(z_{\rho_1+\cdots+\rho_{j-1}+2\lambda_j+1})}
        \right)^{\zeta_j^{(2\lambda_j),2\lambda_j}}
        \\
        &\times
        \left( 1 - \frac{g_{{\bf F}_{j} \to {\bf F}'_{j}}(z_{\rho_1+\cdots+\rho_{j-1}+2\lambda_j+1})}{g_{{\bf F}_{j} \to {\bf F}'_{j}}(z_{\rho_1+\cdots+\rho_{j-1}+2\lambda_j+2})} \right)
        \times
        \cdots
        \times
        \left( 1 - \frac{g_{{\bf F}_{j}\to {\bf F}'_{j}}(z_{\rho_1+\cdots+\rho_{j-1}+\rho_{j}})}
        {g_{{\bf F}_{j}\to {\bf F}'_{j}}(z_{\rho_1+\cdots+\rho_{j-1}+\rho_{j}+1})} \right)
        \\
        &\times
        \frac{1}{g_{{\bf F}^{m}_{j,j+1} \to {\bf F}_{j}}(z_{\rho_1+\cdots+\rho_{j-1}+\rho_{j}+1})}
        \\
        =:
        \sum_{\zeta_j^{(2\lambda_j),2},\zeta_j^{(2\lambda_j),4},\cdots,\zeta_j^{(2\lambda_j),2\lambda_j} \in \{0,1\}}
        &
        \mathcal{G}_{j}^{(2\lambda_j)}
        \left\llbracket\zeta_j^{(2\lambda_j),2},\zeta_j^{(2\lambda_j),4},\cdots,\zeta_j^{(2\lambda_j),2\lambda_j}\right\rrbracket
        (z_{\rho_1+\cdots+\rho_{j-1}+1},\cdots,z_{\rho_1+\cdots+\rho_{j-1}+\rho_j+1})
        .
    \end{split}
    \end{equation}
    Above, we split up each factor $\mathcal{G}_j$ into $2^{\lambda_j}$ terms by expanding out the sub-factors
    \begin{equation*}
        \left( 
            1 - \frac{g_{{\bf F}_{j}\to {\bf F}'_{j}}(z_{\rho_1+\cdots+\rho_{j-1}+2})}
                     {g_{{\bf F}_{j}\to {\bf F}'_{j}}(z_{\rho_1+\cdots+\rho_{j-1}+3})} 
        \right),
        \left( 
            1 - \frac{g_{{\bf F}_{j}\to {\bf F}'_{j}}(z_{\rho_1+\cdots+\rho_{j-1}+4})}
                     {g_{{\bf F}_{j}\to {\bf F}'_{j}}(z_{\rho_1+\cdots+\rho_{j-1}+5})} 
        \right),
        \cdots,
        \left( 
            1 - \frac{g_{{\bf F}_{j}\to {\bf F}'_{j}}(z_{\rho_1+\cdots+\rho_{j-1}+2\lambda_j})}
                     {g_{{\bf F}_{j}\to {\bf F}'_{j}}(z_{\rho_1+\cdots+\rho_{j-1}+2\lambda_j+1})}
        \right)
    \end{equation*}
    into its two components, and named them based on the term $\zeta_j^{(2\lambda_j),2},\zeta_j^{(2\lambda_j),4},\cdots,\zeta_j^{(2\lambda_j),2\lambda_j}$
    corresponding to them. Note that each component 
    $\mathcal{G}_{j}^{(2\lambda_j)}
        \left\llbracket\zeta_j^{(2\lambda_j),2},\zeta_j^{(2\lambda_j),4},\cdots,\zeta_j^{(2\lambda_j),2\lambda_j}\right\rrbracket$
    has asymptotics for $g_{{\bf F}_j \to {\bf F}'_j}(z_{\rho_1 + \cdots + \rho_{j-1} + 2\ell})$ that we desire along the axis $(-1)^{\zeta_j^{(2\lambda_j),2\ell}} \exp[i\Theta_1] \infty$ and for $g_{{\bf F}_j \to {\bf F}'_j}(z_{\rho_1 + \cdots + \rho_{j-1} + 2\ell + 1})$ along $(-1)^{1-\zeta_j^{(2\lambda_j),2\ell}} \exp[i\Theta_1] \infty$ if $2\ell+1 \le \rho_j$. This motivates the following definitions.
    
    Now for each $0 \le 2\lambda_j \le \rho_j$ and $\zeta_j^{(2\lambda_j),2\ell} \in \{0,1\}$, define
    \begin{equation}
        \overline{[2\lambda_j+1]} := \mathrm{min}(2\lambda_j+1,\rho_j)
        \quad\quad\quad\quad\quad
        \zeta_j^{(2\lambda_j),2\ell+1} = 1 - \zeta_j^{(2\lambda_j),2\ell} 
        \quad \text{ if } 
        2\lambda_j +1 \le \rho_j
    \end{equation}
    and 
    the strings of characters
    \begin{equation}
        \eta_j^{(2\lambda_j),2\ell}
        := 
        \begin{cases}
            +- \quad \text{ if } \zeta_j^{(2\lambda_j),2\ell} = 0 \\
            -+ \quad \text{ if } \zeta_j^{(2\lambda_j),2\ell} = 1
        \end{cases}
        \text{for } 2\ell < \rho_j
        \quad\quad\quad\quad\quad
        \eta_j^{(2\lambda_j),2\ell}
        := 
        \begin{cases}
            + \quad \text{ if } \zeta_j^{(2\lambda_j),2\ell} = 0 \\
            - \quad \text{ if } \zeta_j^{(2\lambda_j),2\ell} = 1
        \end{cases}
        \text{for } 2\ell = \rho_j.
    \end{equation}
    From here, we will define quantities
    \begin{equation} \label{eq:splitUp_termTilde_recursive}
    \begin{split}
        &\widetilde{\mathrm{Term}}
        \left[ 
            \eta_1^{(2\lambda_1),2} \cdots \eta_1^{(2\lambda_1),2\lambda_1} 
                \,\,{}_{ \rho_1 - \overline{[2\lambda_1+1]}-1 }
            \,\,\Big|\,\, \cdots \,\,\Big|\,\,
            \eta_{\tau}^{(2\lambda_{\tau}),2} \cdots \eta_{\tau}^{(2\lambda_{\tau}),2\lambda_{\tau}} 
                \,\,{}_{ \rho_{\tau} - \overline{[2\lambda_{\tau}+1]}-1 }
        \right]
        \\
        &:=
        \frac{1}{(z_1-z_2) \cdots (z_n-z_1)}
        \\
        &\quad\quad\quad\quad
        \cdot\Bigg\{ \,
        \mathcal{G}_{1}^{(2\lambda_1)}
        \left\llbracket
            \zeta_1^{(2\lambda_1),2},\zeta_1^{(2\lambda_1),4},\cdots,\zeta_1^{(2\lambda_1),2\lambda_1}
        \right\rrbracket
        (z_{1},\cdots,z_{\rho_1+1})
        \\
        &\quad\quad\quad\quad\quad
        \times\,\cdots\,\times
        \\
        &\quad\quad\quad\quad\quad
        \times
        \mathcal{G}_{\tau}^{(2\lambda_{\tau})}
        \left\llbracket
            \zeta_{\tau}^{(2\lambda_{\tau}),2},\zeta_{\tau}^{(2\lambda_{\tau}),4},\cdots,\zeta_{\tau}^{(2\lambda_{\tau}),2\lambda_{\tau}}
        \right\rrbracket
        (z_{\rho_1+\cdots+\rho_{\tau-1}+1},\cdots,z_{\rho_1+\cdots+\rho_{\tau-1}+\rho_{\tau}},z_{1})
        \Bigg\}
    \end{split}
    \end{equation}
    which will be the integrand of the integral 
    \begin{equation} \label{eq:termDefinition}
    \begin{split}
        &\mathrm{Term}
        \left[ 
            \eta_1^{(2\lambda_1),2} \cdots \eta_1^{(2\lambda_1),2\lambda_1} 
                \,\,{}_{ \rho_1 - \overline{[2\lambda_1+1]}-1 }
            \,\,\Big|\,\, \cdots \,\,\Big|\,\,
            \eta_{\tau}^{(2\lambda_{\tau}),2} \cdots \eta_{\tau}^{(2\lambda_{\tau}),2\lambda_{\tau}} 
                \,\,{}_{ \rho_{\tau} - \overline{[2\lambda_{\tau}+1]}-1 }
        \right]
        \\
        &= \frac{1}{(2 \pi i)^n} 
        {\int}_{0 \to -\exp\left[i \Theta^{m,1}_{\tau,1}\right]\infty}
        dz_1 
        \\
        &\quad\quad\quad\quad\,\,
        {\int}_{0 \to (-1)^{\zeta_{1}^{(2 \lambda_1),2}}  \exp[i \Theta_{1}] \infty}
        dz_2
        \,\,\cdots\,\,
        {\int}_{0 \to (-1)^{\zeta_{1}^{(2 \lambda_1),\overline{[2\lambda_1+1]}}} \exp[i \Theta_{1}] \infty}
        dz_{\overline{[2\lambda_1+1]}} \,\,
        \\
        &\quad\quad\quad\quad\,\,
        {\int}_{0 \to -i \exp[i \Theta_{1}] \infty}
        dz_{\overline{[2\lambda_1+1]}+1} 
        \,\,\cdots\,\,
        {\int}_{0 \to -i \exp[i \Theta_{1}] \infty}
        dz_{\rho_1} \,\,
        \\
        &\quad\quad\quad\quad \cdots\cdots
        \\
        &\quad\quad\quad\quad\,\,
        {\int}_{0 \to -\exp\left[i \Theta^{m,1}_{\tau,1}\right]\infty} dz_{\rho_1+\cdots+\rho_{\tau-1}+1}
        \\
        &\quad\quad\quad\quad\,\,
        {\int}_{0 \to (-1)^{\zeta_{\tau}^{(2 \lambda_{\tau}),2}}  \exp[i \Theta_{\tau}] \infty}
        dz_{\rho_1+\cdots+\rho_{\tau-1}+2}
        \,\,\cdots\,\,
        {\int}_{0 \to (-1)^{\zeta_{\tau}^{(2 \lambda_{\tau}),\overline{[2\lambda_{\tau}+1]}}} \exp[i \Theta_{\tau}] \infty}
        dz_{\rho_1+\cdots+\rho_{\tau-1}+\overline{[2\lambda_{\tau}+1]}} \,\,
        \\
        &\quad\quad\quad\quad\,\,
        {\int}_{0 \to -i \exp[i \Theta_{\tau}] \infty}
        dz_{\rho_1+\cdots+\rho_{\tau-1}+\overline{[2\lambda_{\tau}+1]}+1} 
        \,\,\cdots\,\,
        {\int}_{0 \to -i \exp[i \Theta_{\tau}] \infty}
        dz_{\rho_1+\cdots+\rho_{\tau-1}+\rho_{\tau}} \,\,
        \\
        &\quad\quad\quad\quad\quad
        \left\{
            \widetilde{\mathrm{Term}}
            \bigg[ 
                \eta_1^{(2\lambda_1),2} \cdots \eta_1^{(2\lambda_1),2\lambda_1} 
                    \,\,{}_{ \rho_1 - \overline{[2\lambda_1+1]}-1 }
                \,\,\Big|\,\, \cdots \,\,\Big|\,\,
                \eta_{\tau}^{(2\lambda_{\tau}),2} \cdots \eta_{\tau}^{(2\lambda_{\tau}),2\lambda_{\tau}} 
                    \,\,{}_{ \rho_{\tau} - \overline{[2\lambda_{\tau}+1]}-1 }
            \bigg]
        \right\}
    \end{split}
    \end{equation}
    formed by pushing the contours for each $z_{\rho_1+\cdots+\rho_{j-1}+2},\cdots,z_{\rho_1+\cdots+\rho_{j-1}+\overline{[2\lambda_j+1]}}$
    to their locations relevant for asymptotic analysis and leaving the contours for $z_{\rho_1+\cdots+\rho_{j-1}+\overline{[2\lambda_j+1]}+1},\cdots,z_{\rho_1+\cdots+\rho_{j-1}+\rho_j}$
    unrotated, along $0 \to -i \exp[i \Theta_j]\infty$.
    Note that by definition, these integrands satisfy
    \begin{equation}
    \begin{split}
    &\widetilde{\mathrm{Term}}
    \bigg[ 
        \cdots \,\,\Big|\,\,
        \eta_j^{(2\lambda_j),2} \cdots \eta_j^{(2\lambda_j),2\lambda_j} 
        \,\,{}_{ \rho_j - \overline{[2\lambda_j+1]}-1 }
        \,\,\Big|\,\, \cdots 
    \bigg]
    \\
    &\quad=
    \widetilde{\mathrm{Term}}
    \bigg[ 
        \cdots \,\,\Big|\,\,
        \eta_j^{(2\lambda_j),2} \cdots \eta_j^{(2\lambda_j),2\lambda_j} + -
        \,\,{}_{ \rho_j - \overline{[2\lambda_j+1]}- 3 }
        \,\,\Big|\,\, \cdots 
    \bigg]+ 
    \widetilde{\mathrm{Term}}
    \bigg[ 
        \cdots \,\,\Big|\,\,
        \eta_j^{(2\lambda_j),2} \cdots \eta_j^{(2\lambda_j),2\lambda_j} - +
        \,\,{}_{ \rho_j - \overline{[2\lambda_j+1]}- 3 }
        \,\,\Big|\,\, \cdots 
    \bigg]
    \\
    &\text{ if } \rho_j - \overline{[2\lambda_j+1]}-1 > 1
    \\\\ 
    &\widetilde{\mathrm{Term}}
    \bigg[ 
        \cdots \,\,\Big|\,\,
        \eta_j^{(2\lambda_j),2} \cdots \eta_j^{(2\lambda_j),2\lambda_j} 
        \,\,{}_{1}
        \,\,\Big|\,\, \cdots 
    \bigg]
    \\
    &\quad=
    \widetilde{\mathrm{Term}}
    \bigg[ 
        \cdots \,\,\Big|\,\,
        \eta_j^{(2\lambda_j),2} \cdots \eta_j^{(2\lambda_j),2\lambda_j} +
        \,\,{}_{0}
        \,\,\Big|\,\, \cdots 
    \bigg]
    +
    \widetilde{\mathrm{Term}}
    \bigg[ 
        \cdots \,\,\Big|\,\,
        \eta_j^{(2\lambda_j),2} \cdots \eta_j^{(2\lambda_j),2\lambda_j} -
        \,\,{}_{0}
        \,\,\Big|\,\, \cdots 
    \bigg]
    \end{split}
    \end{equation}
    which tells us that the integrands can be recursively defined.

    \subsubsection{Another representative example} \label{sec:mainProof_another_rep_example}
    
    \begin{figure}[p!]
    \begin{subfigure}
        \centering
        \includegraphics[width=0.7\linewidth]{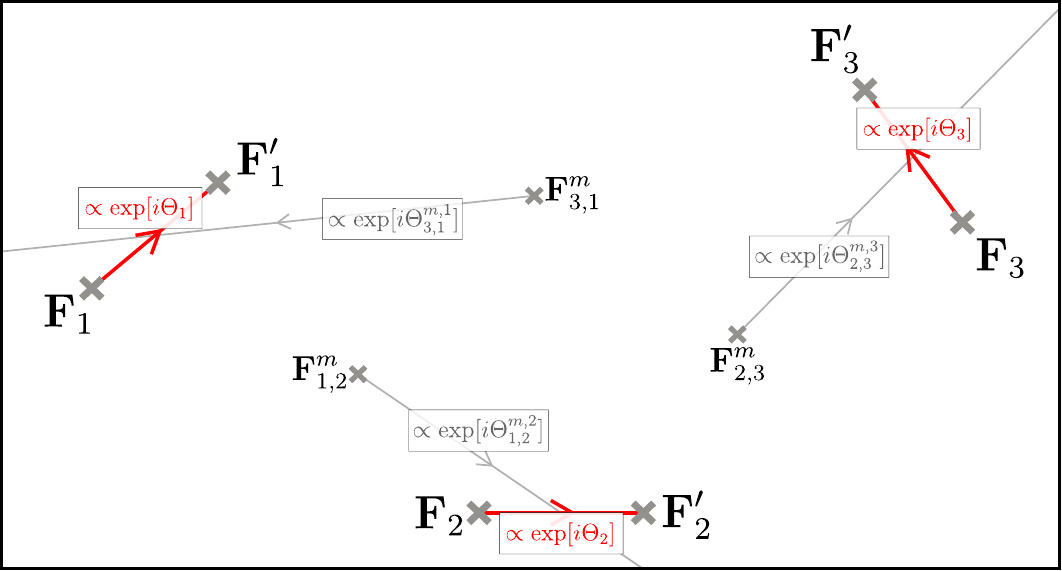}
        \caption{ Example configurations of zipper segments $\widetilde{\mathrm{zip}} \, 1$,  $\widetilde{\mathrm{zip}} \, 2$,  $\widetilde{\mathrm{zip}} \, 3$ and related auxiliary constructions. }
        \label{fig:example_small_zippers_proof}
    \end{subfigure}
    \vspace{2cm}
    \begin{subfigure}
        \centering
        \includegraphics[width=\linewidth]{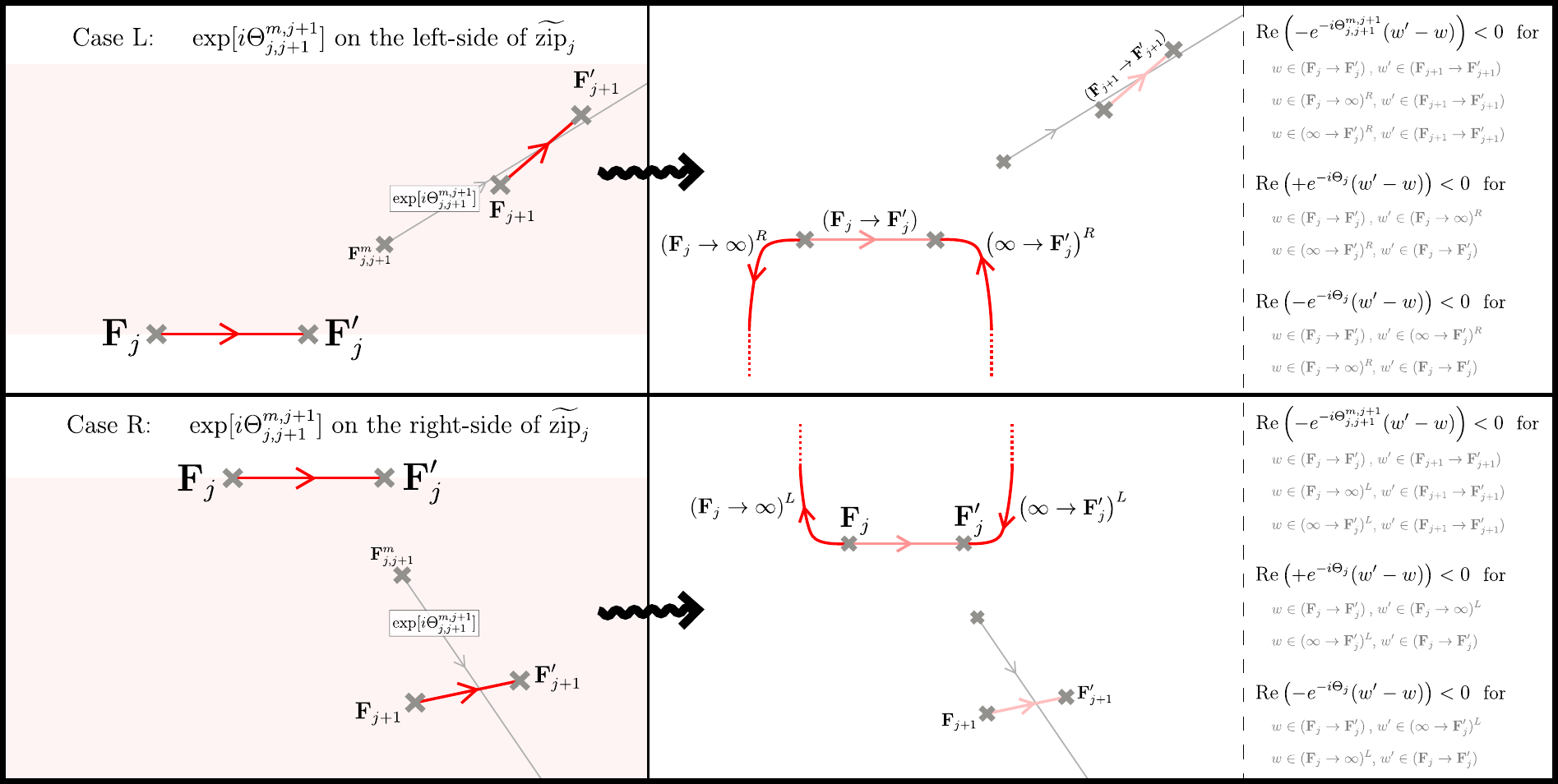}
        \caption{ Given the straight-line segment $({\bf F}_j \to {\bf F}'_j)$, we define several contours as follows. 
        The $({\bf F}_j \to \infty)^L$ and $({\bf F}_j \to \infty)^R$ go out a small amount in the $-\exp[i \Theta_{j}]$ direction then along an infinite ray in the $+i \exp[i \Theta_{j}]$ and $-i \exp[i \Theta_{j}]$ directions respectively. 
        The $(\infty \to {\bf F}'_j)^L$ and $(\infty \to {\bf F}'_j)^R$ start at infinity going along rays in the  $-i \exp[i \Theta_{j}]$ and $+i \exp[i \Theta_{j}]$ directions respectively before turning to the $-\exp[i \Theta_{j}]$ direction to ending at ${\bf F}'_j$. 
        Depending on the orientation of $({\bf F}_{j+1} \to {\bf F}'_{j+1})$, particularly whether $\exp[i \Theta^{m,j+1}_{j,j+1}]$ lies on the left-side (i.e. Case L) or right-side (i.e. Case R) of $({\bf F}_j \to {\bf F}'_j)$, the contours $({\bf F}_j \to \infty)^L$, $({\bf F}_j \to \infty)^R$ or $(\infty \to {\bf F}'_j)^L$, $(\infty \to {\bf F}'_j)^R$ will have the required orientation to apply the Schwinger parameterization trick to evaluate the integrals relevant to our paper. }
        \label{fig:zipper_contours_to_infty}
    \end{subfigure}
    \end{figure}

    Now, we are in a position to consider another representative example, in which we utilize and give examples of the notation we built above. 
    We take $n=6$, $\tau = 3$, and $\rho_1 = 3$, $\rho_2 = 2$, $\rho_3 = 1$ with the zipper segments as in Fig.~\ref{fig:example_small_zippers_proof}. 
    In this example, we'll see several general features of the contour integration procedures that will generalize.
    
    The equality Eq.~\eqref{eq:multZip_telescopedSum} reads
    \begin{equation}
    \begin{split}
        &\tr \left[ 
        \left( \frac{\overline{\partial}^{\bullet \to \circ}_{\widetilde{\mathrm{zip}} \, 1}}{\overline{\partial}^{\bullet \to \circ}} \right)^3
        \left( \frac{\overline{\partial}^{\bullet \to \circ}_{\widetilde{\mathrm{zip}} \, 2}}{\overline{\partial}^{\bullet \to \circ}} \right)^2
        \left( \frac{\overline{\partial}^{\bullet \to \circ}_{\widetilde{\mathrm{zip}} \, 3}}{\overline{\partial}^{\bullet \to \circ}} \right)
        \right]
        =
        {\mathrm{Term}}( {}_2 \,\,|\,\, {}_1 \,\,|\,\, {}_0 )
        \\
        &=
        \frac{1}{(2 \pi i)^6} 
        \int_{0}^{-\exp\left[i \Theta^{m , 1}_{3,1}\right] \infty} dz_{1}
        \int_{0}^{-i\exp\left[i \Theta^{(1)}_{1}\right] \infty}    dz_{2}
        \int_{0}^{-i\exp\left[i \Theta^{(2)}_{1}\right] \infty}    dz_{3}
        \int_{0}^{-\exp\left[i \Theta^{m , 2}_{1,2}\right] \infty} dz_{4}
        \int_{0}^{-i\exp\left[i \Theta_{2}\right] \infty}          dz_{5}
        \\
        &\quad\quad\quad
        \int_{0}^{-\exp\left[i \Theta^{m , 3}_{2,3}\right] \infty} dz_{6}
        \frac{1}{(z_1 - z_2)(z_2 - z_3)(z_3 - z_4)(z_4 - z_5)(z_5 - z_6)(z_6 - z_1)} 
        \left( 
            \frac{g_{{\bf F}^{m}_{3,1} \to {\bf F}_{1}} (z_1)}{1}
          - \frac{g_{{\bf F}^{m}_{3,1} \to {\bf F}'_{1}}(z_1)}{g_{{\bf F}_{1} \to {\bf F}'_{1}}(z_2)}
        \right)
        \\
        & \quad\quad\quad\times 
        \Bigg\{
            \underbrace{
            \quad\quad \begin{matrix}&\, \vspace{-4pt} \\ &1 \\ &\, \vspace{-4pt} \end{matrix} \quad\quad
            }_{\widetilde{\mathrm{Term}}( +- {}_0 \,\,|\,\, {}_1 \,\,|\,\, {}_0 )}
            -
            \underbrace{
            \frac{g_{{\bf F}_{1} \to {\bf F}'_{1}}(z_2)}{g_{{\bf F}_{1} \to {\bf F}'_{1}}(z_3)}
            }_{\widetilde{\mathrm{Term}}( -+ {}_0 \,\,|\,\, {}_1 \,\,|\,\, {}_0 )}
        \Bigg\}
        \left( 
            \frac{1}{g_{{\bf F}^{m}_{1,2} \to {\bf F}_{1}} (z_4)}
          - \frac{g_{{\bf F}_{1} \to {\bf F}'_{1}}(z_3)}{g_{{\bf F}^{m}_{1,2} \to {\bf F}'_{1}}(z_4)}
        \right)
        \left( 
            \frac{g_{{\bf F}^{m}_{1,2} \to {\bf F}_{2}} (z_4)}{1}
          - \frac{g_{{\bf F}^{m}_{1,2} \to {\bf F}'_{2}}(z_4)}{g_{{\bf F}_{2} \to {\bf F}'_{2}}(z_5)}
        \right)
        \\
        & \quad\quad\quad\times 
        \left( 
            \frac{1}{g_{{\bf F}^{m}_{2,3} \to {\bf F}_{2}} (z_6)}
          - \frac{g_{{\bf F}_{2} \to {\bf F}'_{2}}(z_5)}{g_{{\bf F}^{m}_{2,3} \to {\bf F}'_{2}}(z_6)}
        \right)
        \left( 
            \frac{g_{{\bf F}^{m}_{2,3} \to {\bf F}_{3}} (z_6)}{g_{{\bf F}^{m}_{3,1} \to {\bf F}_{3}} (z_1)}
          - \frac{g_{{\bf F}^{m}_{2,3} \to {\bf F}'_{3}}(z_6)}{g_{{\bf F}^{m}_{3,1} \to {\bf F}'_{3}}(z_1)}
        \right),
    \end{split}
    \end{equation}
    where we split up the various underbraced terms in the integrand and define the `$\widetilde{\mathrm{Term}}(\cdots|\cdots|\cdots$)' as the term in the integrand corresponding to them in the same way as defined earlier. 
    
    \begin{figure}[p!]
    \begin{subfigure}
        \centering
        \includegraphics[width=0.8\linewidth]{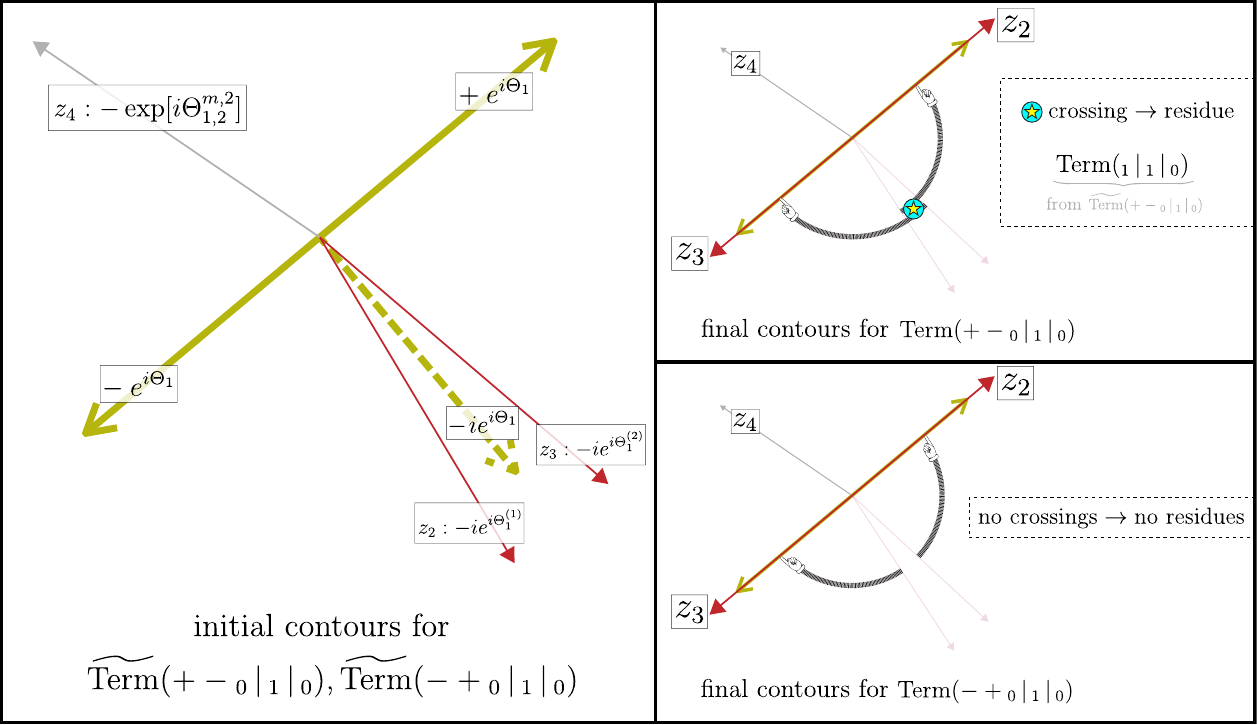}
        \caption{Rotating the $z_2,z_3$ contours in this example from around the $-i e^{i \Theta_{1}}$ axis to their relevant final $+ e^{i \Theta_{1}}$ or $- e^{i \Theta_{1}}$ axes for each of the integrands $\widetilde{\mathrm{Term}} ( + - {}_0 \,\,|\,\, {}_1 \,\,|\,\, {}_0 )$, $\widetilde{\mathrm{Term}} ( - + {}_0 \,\,|\,\, {}_1 \,\,|\,\, {}_0 )$, matching the constructions in Fig.~\ref{fig:example_small_zippers_proof}.  Because of the zipper segments' orientations, these deformations will never cross the $-\exp[i \Theta^{m,2}_{1,2}]$ axis.
        (Left) Relevant integration axes for $z_2,z_3,z_4$.
        (Right-Top) The contour deformation to turn $\widetilde{\mathrm{Term}} ( + - {}_0 \,\,|\,\, {}_1 \,\,|\,\, {}_0 )$ with the initial contours into ${\mathrm{Term}} ( + - {}_0 \,\,|\,\, {}_1 \,\,|\,\, {}_0 )$ has crossings of the $z_2$ and $z_3$ axes, which lead to a residue ${\mathrm{Term}} ( {}_1 \,\,|\,\, {}_1 \,\,|\,\, {}_0 )$ like in Eq.~\eqref{eq:contourDeformation_residue_example}. 
        (Right-Bottom) The deformation to turn the initial integrals over $\widetilde{\mathrm{Term}} ( - + {}_0 \,\,|\,\, {}_1 \,\,|\,\, {}_0 )$ into ${\mathrm{Term}} ( - + {}_0 \,\,|\,\, {}_1 \,\,|\,\, {}_0 )$ have no crossings, thus no residues.}
        \label{fig:example_z2z3_rotations}
    \end{subfigure}
    \vspace{0.5cm}
    \begin{subfigure}
        \centering
        \includegraphics[width=0.95\linewidth]{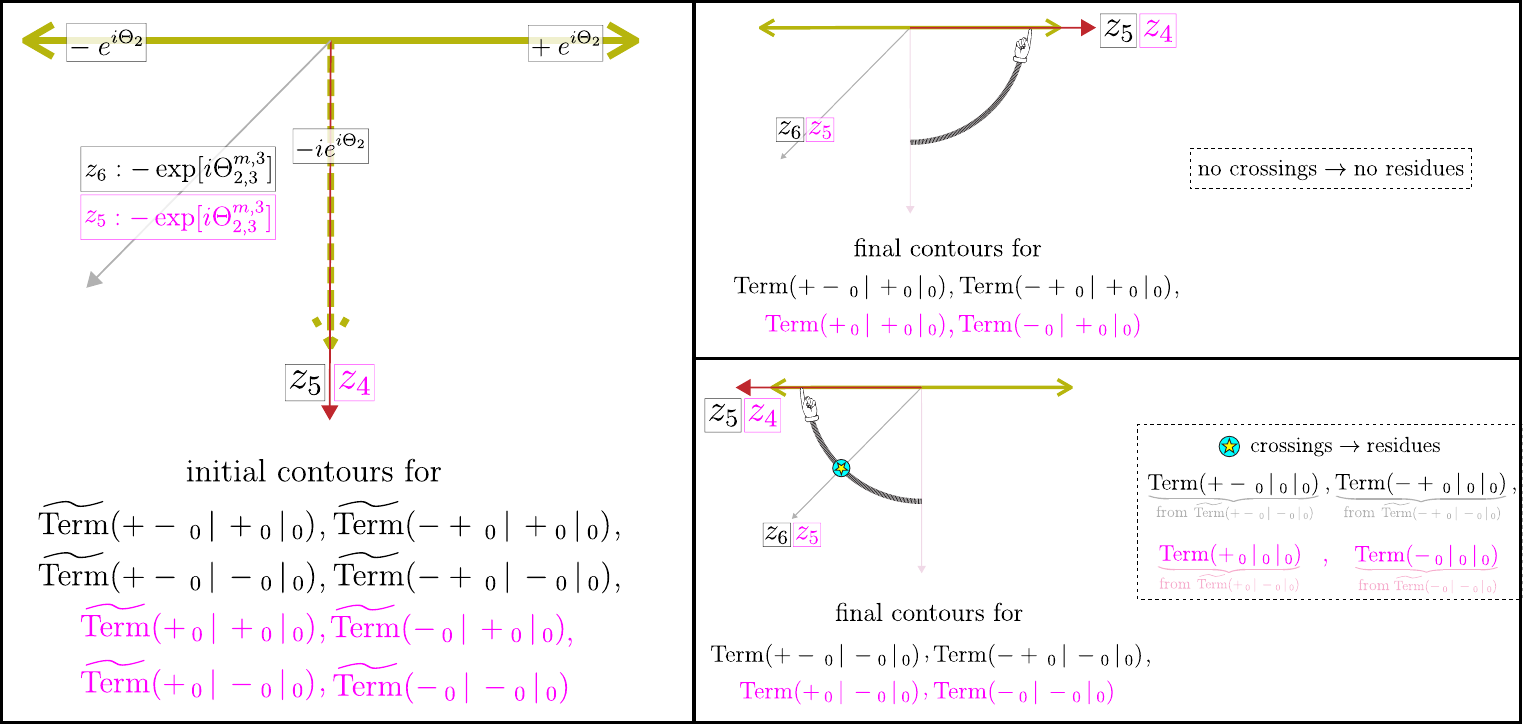}
        \caption{
        Rotating the contours in this example from around the $-i e^{i \Theta_{2}}$ axis to their relevant final $+ e^{i \Theta_{2}}$ or $- e^{i \Theta_{2}}$ axes for the zippers in Fig.~\ref{fig:example_z2z3_rotations}.
        Some of the contour rotations give residues because of the relative configurations of zippers.
        In black are the relevant contours and constructions for
        ${\widetilde{\mathrm{Term}}}( + - {}_0 \,\,|\,\, + {}_0 \,\,|\,\, {}_0 )$,
        ${\widetilde{\mathrm{Term}}}( - + {}_0 \,\,|\,\, + {}_0 \,\,|\,\, {}_0 )$,
        ${\widetilde{\mathrm{Term}}}( + - {}_0 \,\,|\,\, - {}_0 \,\,|\,\, {}_0 )$,
        ${\widetilde{\mathrm{Term}}}( - + {}_0 \,\,|\,\, - {}_0 \,\,|\,\, {}_0 )$
        and in pink are the relevant ones for
        ${\widetilde{\mathrm{Term}}}( + \,{}_0 \,\,|\,\, + {}_0 \,\,|\,\, {}_0 )$,
        ${\widetilde{\mathrm{Term}}}( - \,{}_0 \,\,|\,\, + {}_0 \,\,|\,\, {}_0 )$,
        ${\widetilde{\mathrm{Term}}}( + \,{}_0 \,\,|\,\, - {}_0 \,\,|\,\, {}_0 )$,
        ${\widetilde{\mathrm{Term}}}( - \,{}_0 \,\,|\,\, - {}_0 \,\,|\,\, {}_0 )$.}
        \label{fig:example_z5_rotations}
    \end{subfigure}
    \end{figure}
    
    Above, we defined two split-up contours along the two angles $-i e^{i \Theta_1^{(1)}}, -i e^{i \Theta_1^{(2)}}$ are `near enough' $-i e^{i \Theta_1}$. We choose $-i e^{i \Theta_1^{(1)}}$ lying slightly clockwise to $-i e^{i \Theta_1^{(2)}}$ as in Fig.~\ref{fig:example_z2z3_rotations}.
    Here, `near enough' means the contours avoid the poles of $g_{{\bf F}_1 \to {\bf F}'_1}$ or $\frac{1}{g_{{\bf F}_1 \to {\bf F}'_1}}$. We need to split up the contours because the integrands ${\widetilde{\mathrm{Term}}( +- {}_0 \,\,|\,\, {}_1 \,\,|\,\, {}_0 )}$ and ${\widetilde{\mathrm{Term}}( -+ {}_0 \,\,|\,\, {}_1 \,\,|\,\, {}_0 )}$ have poles at $z_2=z_3$, so we need to split up the contours to get a well-defined integral.
    
    Now, we carry out the contour deformations the above integrals into `$\mathrm{Term}(+- {}_{0} | {}_{1} | {}_0$)',`$\mathrm{Term}(-+ {}_{0} | {}_{1} | {}_0$)' respectively so that we can apply the asymptotics of the $g_{\cdots}$ to the $z_2,z_3$ contours.
    As we alluded to before, some of contour deformations will often be at the cost of a residue caused by the contours crossing at one of poles at $z_2 = z_3$.
    See Fig.~\ref{fig:example_z2z3_rotations} for a depiction of this contour rotation procedure.
    
    Because of the way the zippers are oriented as in Fig.~\ref{fig:example_small_zippers_proof}, we'll have that there are only crossings at $z_2 = z_3$ to obtain $\mathrm{Term} ( +- {}_0 \,\,|\,\, {}_1 \,\,|\,\, {}_0 )$ and there is no such crossing to obtain $\mathrm{Term} ( -+ {}_0 \,\,|\,\, {}_1 \,\,|\,\, {}_0 )$.
    And, we can arrange the crossings of $z_2 = z_3$ to occur along the contour $0 \to -i e^{i \Theta_1} \infty$. 
    These crossings give rise to residues. 
    For example, the integral over $\widetilde{\mathrm{Term}}( +- {}_{0} \,\,|\,\, {}_{1} \,\,|\,\, {}_{0} )$ will become 
    \begin{equation} \label{eq:contourDeformation_residue_example}
    \begin{split}
        & \frac{1}{(2 \pi i)^6} 
        \int_{0}^{-\exp\left[i \Theta^{m , 1}_{3,1}\right] \infty} dz_{1}
        \int_{0}^{-i\exp\left[i \Theta^{(1)}_{1}\right] \infty}    dz_{2}
        \int_{0}^{-i\exp\left[i \Theta^{(2)}_{1}\right] \infty}    dz_{3}
        \int_{0}^{-\exp\left[i \Theta^{m , 2}_{1,2}\right] \infty} dz_{4}
        \int_{0}^{-i\exp\left[i \Theta_{2}\right] \infty}          dz_{5}
        \\
        & \quad\quad\quad
        \int_{0}^{-\exp\left[i \Theta^{m , 3}_{2,3}\right] \infty} dz_{6}
        \Big\{ \widetilde{\mathrm{Term}}( +- {}_{0} \,\,|\,\, {}_{1} \,\,|\,\, {}_{0} ) \Big\}
        \\
        &= {\mathrm{Term}}( +- {}_{0} \,\,|\,\, {}_{1} \,\,|\,\, {}_{0} )
        \\
        &\quad+
        \frac{1}{(2 \pi i)^5} 
        \int_{0}^{-\exp\left[i \Theta^{m , 1}_{3,1}\right] \infty} dz_{1}
        \int_{0}^{-i\exp\left[i \Theta_{1}\right] \infty}          dz_{2}
        \int_{0}^{-\exp\left[i \Theta^{m , 2}_{1,2}\right] \infty} dz_{4}
        \int_{0}^{-i\exp\left[i \Theta_{2}\right] \infty}          dz_{5}
        \\
        & \quad\quad\quad
        \int_{0}^{-\exp\left[i \Theta^{m , 3}_{2,3}\right] \infty} dz_{6}
        \frac{1}{(z_1 - z_2)(z_2 - z_4)(z_4 - z_5)(z_5 - z_6)(z_6 - z_1)}
        \\
        & \quad\quad\quad\times
        \left( 
            \frac{g_{{\bf F}^{m}_{3,1} \to {\bf F}_{1}} (z_1)}{1}
          - \frac{g_{{\bf F}^{m}_{3,1} \to {\bf F}'_{1}}(z_1)}{g_{{\bf F}_{1} \to {\bf F}'_{1}}(z_2)}
        \right)
        \Bigg( 
            \underbrace{
                \frac{1}{g_{{\bf F}^{m}_{1,2} \to {\bf F}_{1}} (z_4)}
            }_{\widetilde{\mathrm{Term}} ( + {}_0  \,\,|\,\, {}_1 \,\,|\,\, {}_0 )}
            - 
            \underbrace{
                \frac{g_{{\bf F}_{1} \to {\bf F}'_{1}}(z_2)}{g_{{\bf F}^{m}_{1,2} \to {\bf F}'_{1}}(z_4)}
            }_{\widetilde{\mathrm{Term}} ( - {}_1  \,\,|\,\, {}_1 \,\,|\,\, {}_0 )}
        \Bigg)
        \left( 
            \frac{g_{{\bf F}^{m}_{1,2} \to {\bf F}_{2}} (z_4)}{1}
            -
            \frac{g_{{\bf F}^{m}_{1,2} \to {\bf F}'_{2}}(z_4)}{g_{{\bf F}_{2} \to {\bf F}'_{2}}(z_5)}
        \right)    
        \\
        & \quad\quad\quad\times    
        \left( 
            \frac{1}{g_{{\bf F}^{m}_{2,3} \to {\bf F}_{2}} (z_6)}
          - \frac{g_{{\bf F}_{2} \to {\bf F}'_{2}}(z_5)}{g_{{\bf F}^{m}_{2,3} \to {\bf F}'_{2}}(z_6)}
        \right)
        \left( 
            \frac{g_{{\bf F}^{m}_{2,3} \to {\bf F}_{3}} (z_6)}{g_{{\bf F}^{m}_{3,1} \to {\bf F}_{3}} (z_1)}
          - \frac{g_{{\bf F}^{m}_{2,3} \to {\bf F}'_{3}}(z_6)}{g_{{\bf F}^{m}_{3,1} \to {\bf F}'_{3}}(z_1)}
        \right)
        \\
        &= {\mathrm{Term}}( +- {}_{0} \,\,|\,\, {}_{1} \,\,|\,\, {}_{0} )
         + {\mathrm{Term}}( {}_{1} \,\,|\,\, {}_{1} \,\,|\,\, {}_{0} ).
    \end{split}
    \end{equation}
    In the right-side of the first equality, we get the rotated $\mathrm{Term} ( + - {}_0 \,\,|\,\, {}_1 \,\,|\,\, {}_0 )$ and the plus a residue term. For the residue term, after changing notation $z_4,z_5,z_6 \mapsto z_3,z_4,z_5$, we can split the integrand up into two parts that correspond to $\widetilde{\mathrm{Term}} ( + \, {}_0 \,\,|\,\, {}_1 \,\,|\,\, {}_0 ) + \widetilde{\mathrm{Term}} ( - \, {}_0 \,\,|\,\, {}_1 \,\,|\,\, {}_0 )$ which add up to ${\mathrm{Term}} ( {}_1 \,\,|\,\, {}_1 \,\,|\,\, {}_0 ) $ when integrated. 
    In total, we summarize Fig.~\ref{fig:example_z2z3_rotations} and the above calculations by the following equality:
    \begin{equation} \label{eq:example_equality_pt1}
    \begin{split}
        {\mathrm{Term}}( {}_2 \,\,|\,\, {}_1 \,\,|\,\, {}_0 )
        =
        \underbrace{
            {\mathrm{Term}}( + - {}_0 \,\,|\,\, {}_1 \,\,|\,\, {}_0 ) 
        }_{\text{from integration of } {\widetilde{\mathrm{Term}}}( + - {}_0 \,\,|\,\, {}_1 \,\,|\,\, {}_0 ) }
        \quad\quad
        +
        \underbrace{
            {\mathrm{Term}}( - + {}_0 \,\,|\,\, {}_1 \,\,|\,\, {}_0 ) 
        }_{\text{from integration of } {\widetilde{\mathrm{Term}}}( - + {}_0 \,\,|\,\, {}_1 \,\,|\,\, {}_0 ) } 
        +
        \underbrace{
            {\mathrm{Term}}( {}_1 \,\,|\,\, {}_1 \,\,|\,\, {}_0 )
        }_{\text{residue}}
    \end{split}
    \end{equation}
    
    At this point, we still haven't rotated all of the contours away from the ray $-i e^{i \Theta_1}$ since the integrals over
    $\widetilde{\mathrm{Term}} ( + {}_0 \,\,|\,\, {}_1 \,\,|\,\, {}_0 )$,$\widetilde{\mathrm{Term}} ( - {}_0 \,\,|\,\, {}_1 \,\,|\,\, {}_0 )$ 
    remain unrotated. We follow the same contour rotation procedure as before. This time, the integrand has a pole at $z_2 = z_3$. However, similarly to Fig.~\ref{fig:example_z2z3_rotations} this is not crossed. Additionally, we will not encounter any more residues since we only rotate one contour. This gives
    \begin{equation} \label{eq:example_equality_pt2}
    \begin{split}
        &{\mathrm{Term}}( {}_1 \,\,|\,\, {}_1 \,\,|\,\, {}_0 )
        \\
        &=
        \frac{1}{(2 \pi i)^5} 
        \int_{0}^{-\exp\left[i \Theta^{m , 1}_{3,1}\right] \infty} dz_{1}
        \int_{0}^{-i\exp\left[i \Theta_{1}\right] \infty}          dz_{2}
        \int_{0}^{-\exp\left[i \Theta^{m , 2}_{1,2}\right] \infty} dz_{3}
        \int_{0}^{-i\exp\left[i \Theta_{2}\right] \infty}          dz_{4}
        \int_{0}^{-\exp\left[i \Theta^{m , 3}_{2,3}\right] \infty} dz_{5}
        \\
        &\quad\quad\quad
        \Big\{ 
          {\widetilde{\mathrm{Term}} ( + \,{}_0  \,\,|\,\, {}_1 \,\,|\,\, {}_0 )} + {\widetilde{\mathrm{Term}} ( - \,{}_0  \,\,|\,\, {}_1 \,\,|\,\, {}_0 )} 
        \Big\}
        \\
        &= {{\mathrm{Term}} ( + \,{}_0 \,\,|\,\, {}_1 \,\,|\,\, {}_0 )} 
         + {{\mathrm{Term}} ( - \,{}_0  \,\,|\,\, {}_1 \,\,|\,\, {}_0 )}.
    \end{split}
    \end{equation}
    
    Now, we do the same procedure as above to rotate all of the remaining contours away from $-i e^{i \Theta_2}$.
    Recall from Eq.~\eqref{eq:splitUp_termTilde_recursive} that each
    ${\widetilde{\mathrm{Term}}}(\, \cdots \,\,|\,\,   {}_1 \,\,|\,\, \cdots \,)
    ={\widetilde{\mathrm{Term}}}(\, \cdots \,\,|\,\, + {}_0 \,\,|\,\, \cdots \,)
    +{\widetilde{\mathrm{Term}}}(\, \cdots \,\,|\,\, - {}_0 \,\,|\,\, \cdots \,)$.
    We need to rotate the the $z_5$ contour for
    ${\widetilde{\mathrm{Term}}}( + - {}_0 \,\,|\,\, + {}_0 \,\,|\,\, {}_0 )$,
    ${\widetilde{\mathrm{Term}}}( + - {}_0 \,\,|\,\, - {}_0 \,\,|\,\, {}_0 )$,
    ${\widetilde{\mathrm{Term}}}( - + {}_0 \,\,|\,\, + {}_0 \,\,|\,\, {}_0 )$,
    ${\widetilde{\mathrm{Term}}}( - + {}_0 \,\,|\,\, - {}_0 \,\,|\,\, {}_0 )$
    and respectively the $z_4$ contour for
    ${\widetilde{\mathrm{Term}}}( + {}_0 \,\,|\,\, + {}_0 \,\,|\,\, {}_0 )$,
    ${\widetilde{\mathrm{Term}}}( + {}_0 \,\,|\,\, - {}_0 \,\,|\,\, {}_0 )$,
    ${\widetilde{\mathrm{Term}}}( - {}_0 \,\,|\,\, + {}_0 \,\,|\,\, {}_0 )$,
    ${\widetilde{\mathrm{Term}}}( - {}_0 \,\,|\,\, - {}_0 \,\,|\,\, {}_0 )$.
    In this case however, because of the ways that the zippers are configured, some of the contour rotations will cross the the pole along the $z_6$ and respectively $z_5$ contour which give some residues. See Fig.~\ref{fig:example_z5_rotations}.

    For example, we get
    \begin{equation} \label{eq:example_equality_pt3}
    \begin{split}
        &{\mathrm{Term}}( + - {}_0 \,\,|\,\,   {}_1 \,\,|\,\, {}_0 )
        \\
        &=
        \frac{1}{(2 \pi i)^6} 
        \int_{0}^{-\exp\left[i \Theta^{m , 1}_{3,1}\right] \infty} dz_{1}
        \int_{0}^{-i\exp\left[i \Theta^{(1)}_{1}\right] \infty}    dz_{2}
        \int_{0}^{-i\exp\left[i \Theta^{(2)}_{1}\right] \infty}    dz_{3}
        \int_{0}^{-\exp\left[i \Theta^{m , 2}_{1,2}\right] \infty} dz_{4}
        \int_{0}^{-i\exp\left[i \Theta_{2}\right] \infty}          dz_{5}
        \\
        & 
        \quad\quad\quad\quad
        \int_{0}^{-\exp\left[i \Theta^{m , 3}_{2,3}\right] \infty} dz_{6}
        \left\{
        \widetilde{\mathrm{Term}}( + - {}_0 \,\,|\,\, + \,{}_0 \,\,|\,\, {}_0 ) 
        +
        \widetilde{\mathrm{Term}}( + - {}_0 \,\,|\,\, - \,{}_0 \,\,|\,\, {}_0 )
        \right\}
        \\
        &=
        {\mathrm{Term}}( + - {}_0 \,\,|\,\, + \,{}_0 \,\,|\,\, {}_0 ) 
        +
        {\mathrm{Term}}( + - {}_0 \,\,|\,\, - \,{}_0 \,\,|\,\, {}_0 )
        +
        {\mathrm{Term}}( + - {}_0 \,\,|\,\,   {}_0 \,\,|\,\, {}_0 ).
    \end{split}
    \end{equation}
    In general, we get the equalities
    \begin{equation} \label{eq:example_equality_pt4}
    \begin{split}
    &
    {\mathrm{Term}}( + - {}_0 \,\,|\,\,   {}_1 \,\,|\,\, {}_0 )
    =
    {\mathrm{Term}}( + - {}_0 \,\,|\,\, + {}_0 \,\,|\,\, {}_0 ) 
    +
    {\mathrm{Term}}( + - {}_0 \,\,|\,\, - {}_0 \,\,|\,\, {}_0 )
    +
    {\mathrm{Term}}( + - {}_0 \,\,|\,\,   {}_0 \,\,|\,\, {}_0 )
    \\
    &
    {\mathrm{Term}}( - + {}_0 \,\,|\,\,   {}_1 \,\,|\,\, {}_0 )
    =
    {\mathrm{Term}}( - + {}_0 \,\,|\,\, + {}_0 \,\,|\,\, {}_0 ) 
    +
    {\mathrm{Term}}( - + {}_0 \,\,|\,\, - {}_0 \,\,|\,\, {}_0 )
    +
    {\mathrm{Term}}( - + {}_0 \,\,|\,\,   {}_0 \,\,|\,\, {}_0 )
    \\
    &
    {\mathrm{Term}}( + \,{}_0 \,\,|\,\,   {}_1 \,\,|\,\, {}_0 )
    =
    {\mathrm{Term}}( + \,{}_0 \,\,|\,\, + {}_0 \,\,|\,\, {}_0 ) 
    +
    {\mathrm{Term}}( + \,{}_0 \,\,|\,\, - {}_0 \,\,|\,\, {}_0 )
    +
    {\mathrm{Term}}( + \,{}_0 \,\,|\,\,   {}_0 \,\,|\,\, {}_0 )
    \\
    &
    {\mathrm{Term}}( - \,{}_0 \,\,|\,\,   {}_1 \,\,|\,\, {}_0 )
    =
    {\mathrm{Term}}( - \,{}_0 \,\,|\,\, + {}_0 \,\,|\,\, {}_0 ) 
    +
    {\mathrm{Term}}( - \,{}_0 \,\,|\,\, - {}_0 \,\,|\,\, {}_0 )
    +
    {\mathrm{Term}}( - \,{}_0 \,\,|\,\,   {}_0 \,\,|\,\, {}_0 ).
    \end{split}
    \end{equation}
    In total, by the equalities Eq.~(\ref{eq:example_equality_pt1},\ref{eq:example_equality_pt2},\ref{eq:example_equality_pt3},\ref{eq:example_equality_pt4}), we get
    \begin{equation} \label{eq:example_term210_rotated}
    \begin{split}
        {\mathrm{Term}}( {}_2 \,\,|\,\,   {}_1 \,\,|\,\, {}_0 )
        &=
        {\mathrm{Term}}( + - {}_0 \,\,|\,\, + {}_0 \,\,|\,\, {}_0 )
        +
        {\mathrm{Term}}( + - {}_0 \,\,|\,\, - {}_0 \,\,|\,\, {}_0 )
        +
        {\mathrm{Term}}( - + {}_0 \,\,|\,\, + {}_0 \,\,|\,\, {}_0 )
        +
        {\mathrm{Term}}( - + {}_0 \,\,|\,\, - {}_0 \,\,|\,\, {}_0 )
        \\
        &\,\,\,+
        {\mathrm{Term}}( + \,{}_0 \,\,|\,\, + {}_0 \,\,|\,\, {}_0 )
        +
        {\mathrm{Term}}( + \,{}_0 \,\,|\,\, - {}_0 \,\,|\,\, {}_0 )
        +
        {\mathrm{Term}}( - \,{}_0 \,\,|\,\, + {}_0 \,\,|\,\, {}_0 )
        +
        {\mathrm{Term}}( - \,{}_0 \,\,|\,\, - {}_0 \,\,|\,\, {}_0 )
        \\
        &\,\,\,+
        {\mathrm{Term}}( + - {}_0 \,\,|\,\, {}_0 \,\,|\,\, {}_0 )
        +
        {\mathrm{Term}}( - + {}_0 \,\,|\,\, {}_0 \,\,|\,\, {}_0 )
        \\
        &\,\,\,+
        {\mathrm{Term}}( + \,{}_0 \,\,|\,\, {}_0 \,\,|\,\, {}_0 )
        +
        {\mathrm{Term}}( - \,{}_0 \,\,|\,\, {}_0 \,\,|\,\, {}_0 ).
    \end{split}
    \end{equation}
    
    Now, we're in a position to apply the asymptotic analysis to these terms. One of these terms is
    \begin{equation}
    \begin{split}
        &{\mathrm{Term}}( + - {}_0 \,\,|\,\, - {}_0 \,\,|\,\, {}_0 )
        \\
        &=
        \frac{1}{(2 \pi i)^6} 
        \int_{0}^{-\exp\left[i \Theta^{m , 1}_{3,1}\right] \infty} dz_{1}
        \int_{0}^{+\exp\left[i \Theta_{1}\right] \infty} dz_{2}
        \int_{0}^{-\exp\left[i \Theta_{1}\right] \infty} dz_{3}
        \int_{0}^{-\exp\left[i \Theta^{m , 2}_{1,2}\right] \infty} dz_{4}
        \int_{0}^{-i\,\exp\left[i \Theta_{2}\right] \infty} dz_{5}
        \\
        & \quad\quad\quad
        \int_{0}^{-\exp\left[i \Theta^{m , 3}_{2,3}\right] \infty} dz_{6}
        \frac{1}{(z_1 - z_2)(z_2 - z_3)(z_3 - z_4)(z_4 - z_5)(z_5 - z_6)(z_6 - z_1)}
        \left( 
            \frac{g_{{\bf F}^{m}_{3,1} \to {\bf F}_{1}} (z_1)}{1}
          - \frac{g_{{\bf F}^{m}_{3,1} \to {\bf F}'_{1}}(z_1)}{g_{{\bf F}_{1} \to {\bf F}'_{1}}(z_2)}
        \right)
        \\
        & \quad\quad\quad\quad\quad\quad\quad\quad\quad\quad\quad\quad\times
        (1)
        \left( 
            \frac{1}{g_{{\bf F}^{m}_{1,2} \to {\bf F}_{1}} (z_4)}
          - \frac{g_{{\bf F}_{1} \to {\bf F}'_{1}}(z_3)}{g_{{\bf F}^{m}_{1,2} \to {\bf F}'_{1}}(z_4)}
        \right)
        \left( 
            \frac{g_{{\bf F}^{m}_{1,2} \to {\bf F}_{2}} (z_4)}{1}
          - \frac{g_{{\bf F}^{m}_{1,2} \to {\bf F}'_{2}}(z_4)}{g_{{\bf F}_{2} \to {\bf F}'_{2}}(z_5)}
        \right)
        \\
        & \quad\quad\quad\quad\quad\quad\quad\quad\quad\quad\quad\quad\times
        \left( - \frac{g_{{\bf F}_{2} \to {\bf F}'_{2}}(z_5)}{g_{{\bf F}_{2} \to {\bf F}'_{2}}(z_6)} \right)
        \frac{1}{g_{{\bf F}^{m}_{2,3} \to {\bf F}_{2}} (z_6)}
        \left( 
            \frac{g_{{\bf F}^{m}_{2,3} \to {\bf F}_{3}} (z_6)}{g_{{\bf F}^{m}_{3,1} \to {\bf F}_{3}} (z_1)}
          - \frac{g_{{\bf F}^{m}_{2,3} \to {\bf F}'_{3}}(z_6)}{g_{{\bf F}^{m}_{3,1} \to {\bf F}'_{3}}(z_1)}
        \right)
    \end{split}
    \end{equation}
    whose asymptotics give 
    \begin{equation} \label{eq:example_asymptotic_splitTerm}
    \begin{split}
        &{\mathrm{Term}}( + - {}_0 \,\,|\,\, - {}_0 \,\,|\,\, {}_0 )
        \cdot (1 + O(1/R)) + O(1/R)
        \\
        &=
        \frac{1}{(+2 \pi i)^6} 
        \int_{0}^{-\exp\left[-i \Theta^{m , 1}_{3,1}\right] \infty} dz_{1}
        \int_{0}^{+\exp\left[-i \Theta_{1}\right] \infty} dz_{2}
        \int_{0}^{-\exp\left[-i \Theta_{1}\right] \infty} dz_{3}
        \int_{0}^{-\exp\left[-i \Theta^{m , 2}_{1,2}\right] \infty} dz_{4}
        \int_{0}^{-\exp\left[-i \Theta_{2}\right] \infty} dz_{5}
        \\
        & \quad\quad\quad
        \int_{0}^{-\exp\left[-i \Theta^{m , 3}_{2,3}\right] \infty} dz_{6}
        \frac{\left( e^{{\bf F}'_1  (z_1-z_2)} - e^{{\bf F}_1 (z_1-z_2)} \right)}{z_1 - z_2}
        \frac{\left(                           - e^{{\bf F}_1 (z_2-z_3)} \right)}{z_2 - z_3}
        \frac{\left( e^{{\bf F}'_1  (z_3-z_4)} - e^{{\bf F}_1 (z_3-z_4)} \right)}{z_3 - z_4}
        \\
        & \quad\quad\quad\quad\quad\quad\quad\quad\quad\quad\quad\quad\times
        \frac{\left( e^{{\bf F}'_2  (z_4-z_5)} - e^{{\bf F}'_2 (z_4-z_5)} \right)}{z_4 - z_5}
        \frac{\left( e^{{\bf F}'_2  (z_5-z_6)                           } \right)}{z_5 - z_6}
        \frac{\left( e^{{\bf F}'_3  (z_6-z_1)} - e^{{\bf F}'_3 (z_6-z_1)} \right)}{z_6 - z_1}
        \\
        & \quad + \mathrm{c.c.} + O(1/R)
        \\
        &=
        \frac{1}{(+2 \pi i)^6} 
        \int_{0}^{-\exp\left[-i \Theta^{m , 1}_{3,1}\right] \infty} dz_{1}
        \int_{0}^{+\exp\left[-i \Theta_{1}\right] \infty} dz_{2}
        \int_{0}^{-\exp\left[-i \Theta_{1}\right] \infty} dz_{3}
        \int_{0}^{-\exp\left[-i \Theta^{m , 2}_{1,2}\right] \infty} dz_{4}
        \int_{0}^{-\exp\left[-i \Theta_{2}\right] \infty} dz_{5}
        \\
        & \quad\quad\quad
        \int_{0}^{-\exp\left[-i \Theta^{m , 3}_{2,3}\right] \infty} dz_{6}
        \int_{{\bf F}_1}^{{\bf F}'_1} dw_{1}
        \int_{({\bf F}_1 \to \infty)^L} dw_{2}
        \int_{{\bf F}_1}^{{\bf F}'_1} dw_{3}
        \int_{{\bf F}_2}^{{\bf F}'_2} dw_{4}
        \int_{(\infty \to {\bf F}'_2)^R} dw_{5}
        \int_{{\bf F}_3}^{{\bf F}'_3} dw_{6}
        \\
        & \quad\quad\quad\quad\quad\quad
        \left\{ e^{w_1 (z_1-z_2)} e^{w_2 (z_2-z_3)} e^{w_3 (z_3-z_4)} e^{w_4 (z_4-z_5)} e^{w_5 (z_5-z_6)} e^{w_6 (z_6-z_1)} \right\}
        \\
        & \quad + \mathrm{c.c.} + O(1/R)
        \\
        &=
        \frac{1}{(+2 \pi i)^6} 
        \int_{{\bf F}_1}^{{\bf F}'_1} dw_{1}
        \int_{({\bf F}_1 \to \infty)^L} dw_{2}
        \int_{{\bf F}_1}^{{\bf F}'_1} dw_{3}
        \int_{{\bf F}_2}^{{\bf F}'_2} dw_{4}
        \int_{(\infty \to {\bf F}'_2)^R} dw_{5}
        \int_{{\bf F}_3}^{{\bf F}'_3} dw_{6}
        \\
        & \quad\quad\quad\quad\quad\quad
        \left\{
        \frac{1}{(w_1-w_2)(w_2-w_3)(w_3-w_4)(w_4-w_5)(w_5-w_6)(w_6-w_1)}
        \right\}
        \\
        & \quad + \mathrm{c.c.} + O(1/R).
    \end{split}
    \end{equation}
    The first equality above follows from the same reasoning as second equality of Eq.~\eqref{eq:firstExample_asymptotics_equation}. 
    In the second equality above, the contours $w_2: ({\bf F}_1 \to \infty)^L$ and $w_5: (\infty \to {\bf F}'_2)^R$ refer to integration contours as defined and depicted in Fig.~\ref{fig:zipper_contours_to_infty}. As in that figure, we note that the relevant contours are configured in such a way that the individual integrals over each of the $z_{\cdots}$ converge.
    The second equality itself is then similar to the third equality of Eq.~\eqref{eq:firstExample_asymptotics_equation}.
    The third equality above again is similar to the last equality of Eq.~\eqref{eq:firstExample_asymptotics_equation}, relying on the convergence of the integrals from the Schwinger trick as in Fig.~\ref{fig:zipper_contours_to_infty}. Again, all of the integrals ${\bf F}_j \to {\bf F}'_j$ are taken to be the straight-line segments between the faces.
    
    \begin{figure}[p!]
    \begin{subfigure}
        \centering
        \includegraphics[width=\linewidth]{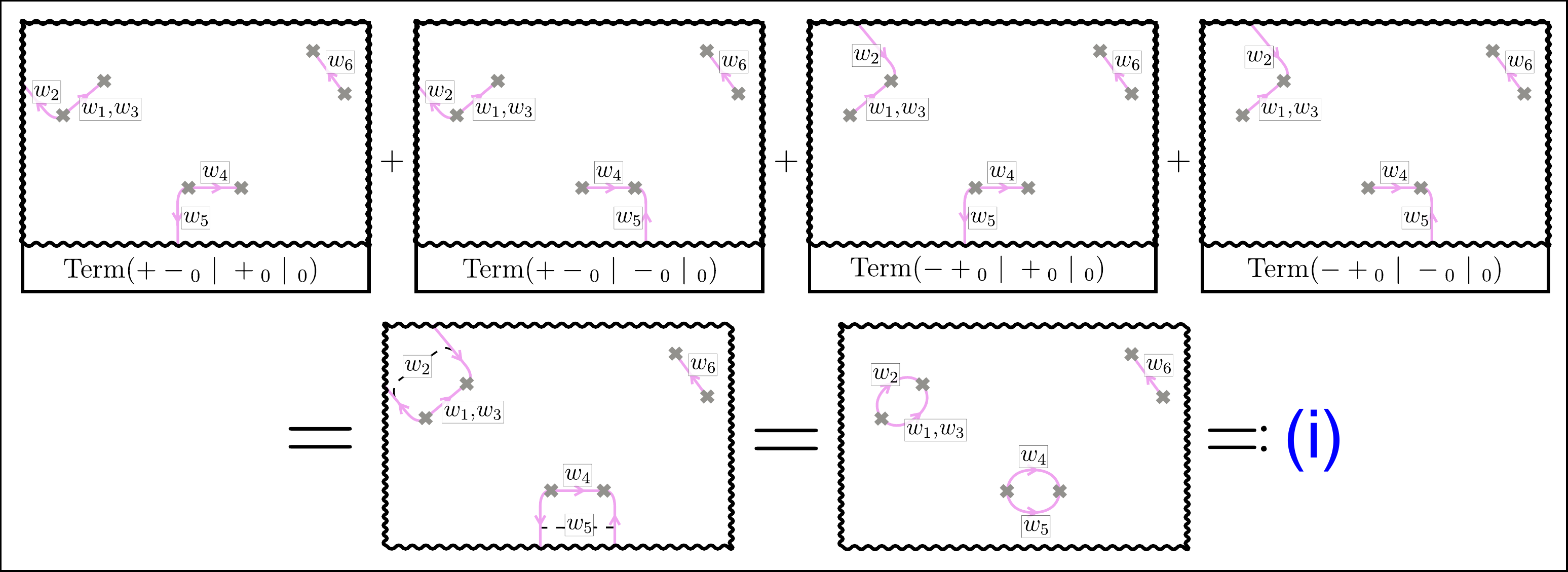}
        \caption{Contours for the terms
        ${\mathrm{Term}}( + - {}_0 \,\,|\,\, + \,{}_0 \,\,|\,\, {}_0 )$,
        ${\mathrm{Term}}( + - {}_0 \,\,|\,\, - \,{}_0 \,\,|\,\, {}_0 )$,
        ${\mathrm{Term}}( - + {}_0 \,\,|\,\, + \,{}_0 \,\,|\,\, {}_0 )$,
        ${\mathrm{Term}}( - + {}_0 \,\,|\,\, - \,{}_0 \,\,|\,\, {}_0 )$
        with integrand $\frac{1}{(2 \pi i)^6} \frac{1}{(w_1-w_2)(w_2-w_3)(w_3-w_4)(w_4-w_5)(w_5-w_6)(w_6-w_1)}$,
        up to additive and multiplicative factors of $O(1/R)$.
        When adding up these terms, note that we can deform the $w_2$ contour to the position as shown. This gives a sub-term we call $(i)$.
        }
        \label{fig:rep_example_contours_210_1st}
    \end{subfigure}
    \vspace{3.5cm}
    \begin{subfigure}
        \centering
        \includegraphics[width=\linewidth]{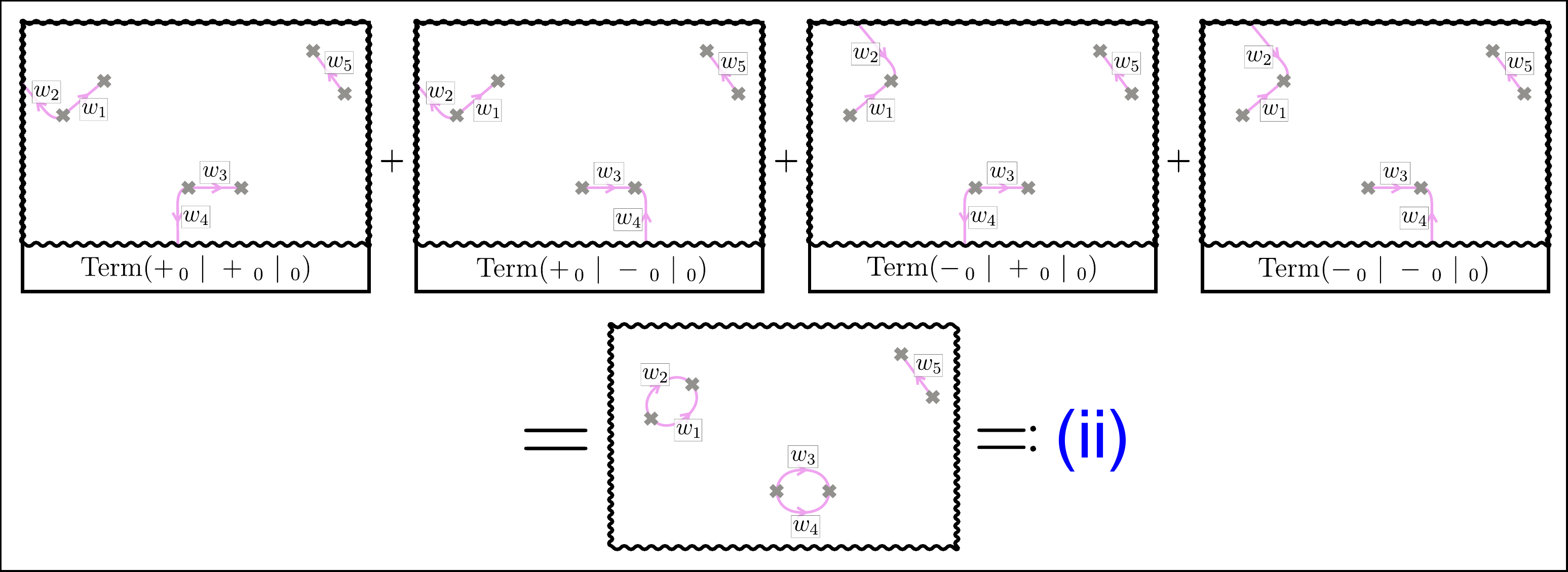}
        \caption{Contours for the terms
        ${\mathrm{Term}}( + {}_0 \,\,|\,\, + \,{}_0 \,\,|\,\, {}_0 )$,
        ${\mathrm{Term}}( + {}_0 \,\,|\,\, - \,{}_0 \,\,|\,\, {}_0 )$,
        ${\mathrm{Term}}( - {}_0 \,\,|\,\, + \,{}_0 \,\,|\,\, {}_0 )$,
        ${\mathrm{Term}}( - {}_0 \,\,|\,\, - \,{}_0 \,\,|\,\, {}_0 )$
        with integrand $\frac{1}{(2 \pi i)^5} \frac{1}{(w_1-w_2)(w_2-w_3)(w_3-w_4)(w_4-w_5)(w_5-w_1)}$,
        up to additive and multiplicative factors of $O(1/R)$,
        which are then added up. This gives a sub-term we call $(ii)$.
        }
        \label{fig:rep_example_contours_210_2nd}
    \end{subfigure}
    \end{figure}

    \begin{figure}[p!]
    \begin{subfigure}
        \centering
        \includegraphics[width=\linewidth]{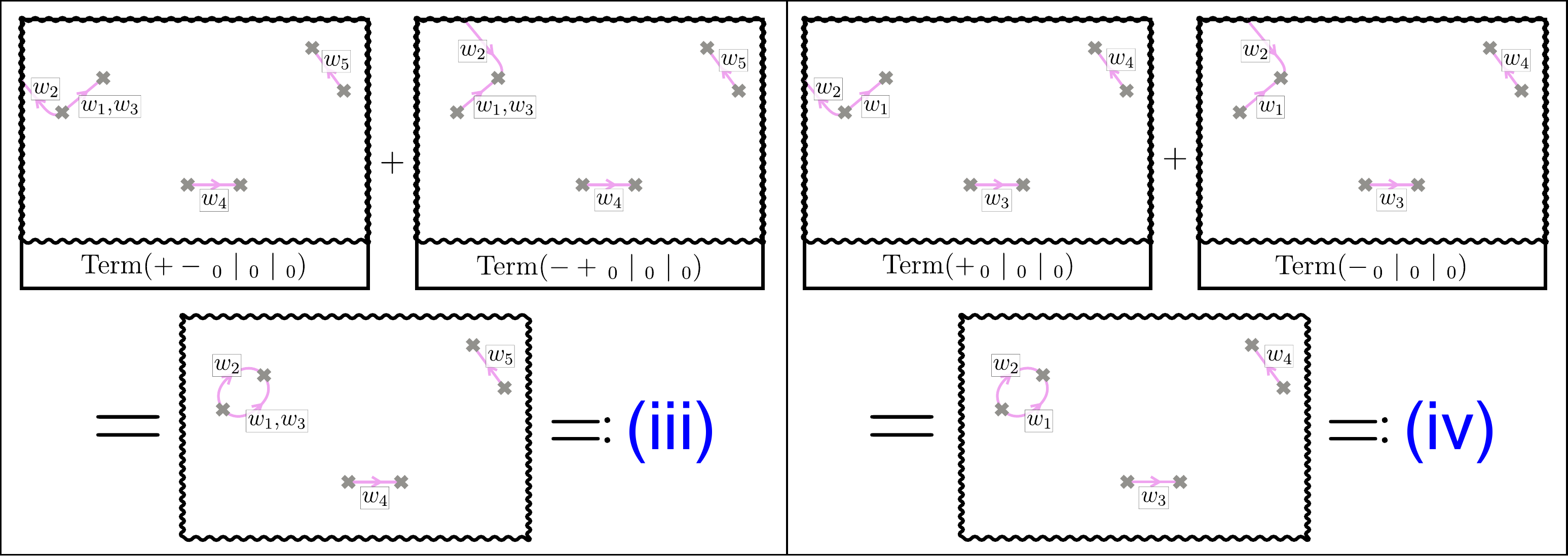}
        \caption{Adding the terms (Left)
        ${\mathrm{Term}}( + - {}_0 \,\,|\,\, {}_0 \,\,|\,\, {}_0 )$,
        ${\mathrm{Term}}( - + {}_0 \,\,|\,\, {}_0 \,\,|\,\, {}_0 )$
        with integrand $\frac{1}{(2 \pi i)^5} \frac{1}{(w_1-w_2)(w_2-w_3)(w_3-w_4)(w_4-w_5)(w_5-w_1)}$,
        and the terms (Right)
        ${\mathrm{Term}}( + {}_0 \,\,|\,\, {}_0 \,\,|\,\, {}_0 )$,
        ${\mathrm{Term}}( - {}_0 \,\,|\,\, {}_0 \,\,|\,\, {}_0 )$
        with integrand $\frac{1}{(2 \pi i)^4} \frac{1}{(w_1-w_2)(w_2-w_3)(w_3-w_4)(w_4-w_1)}$,
        up to additive and multiplicative factors of $O(1/R)$,
        which are then added up. These gives sub-terms we call $(iii)$ and $(iv)$ respectively.
        }
        \label{fig:rep_example_contours_210_3rd}
    \end{subfigure}
    \vspace{1cm}
    \begin{subfigure}
        \centering
        \includegraphics[width=\linewidth]{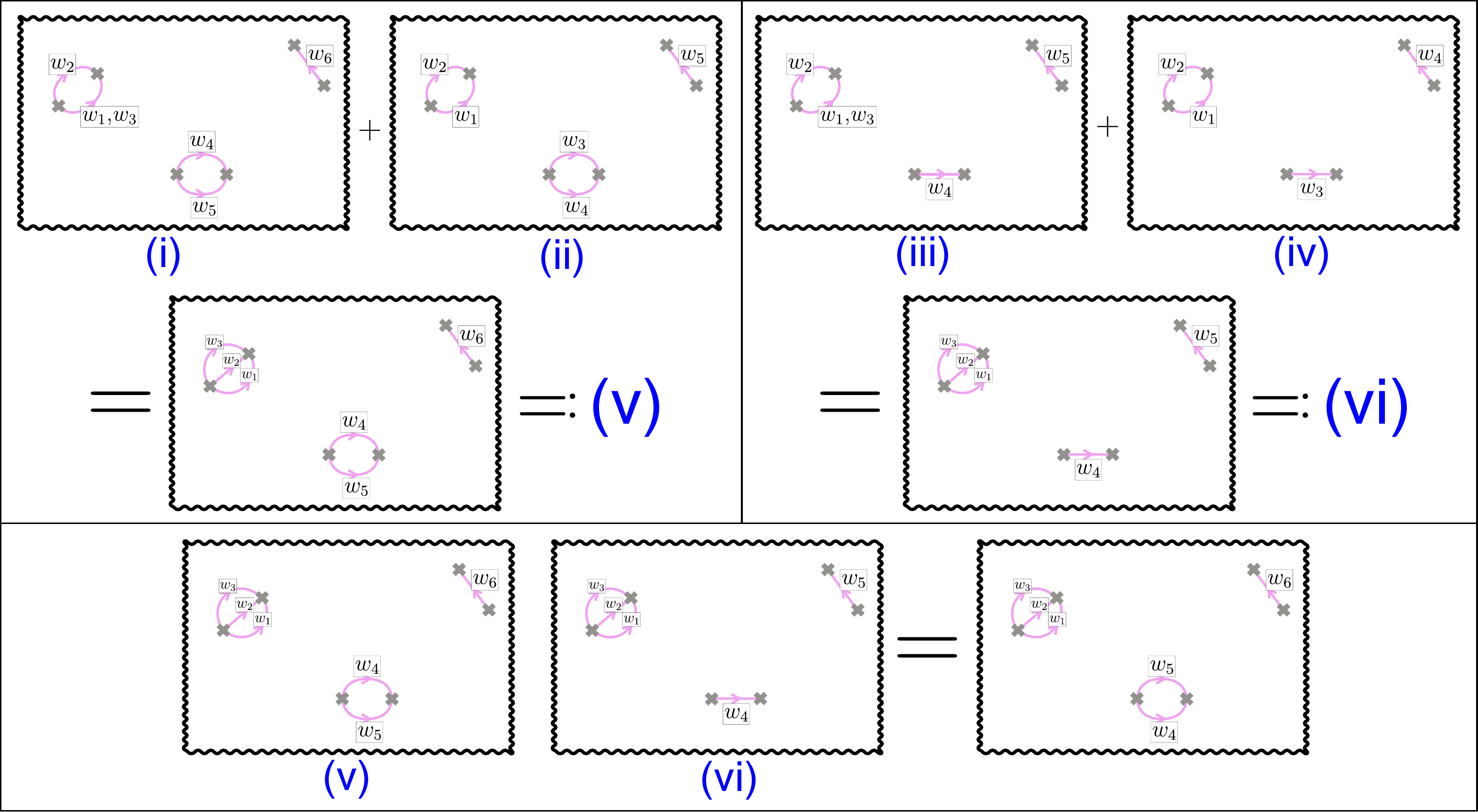}
        \caption{
            (Top-Left) Adding up the sub-terms $(i)+(ii)$ gives a sub-term we call $(v)$. This equality follows from a residue coming from a contour deformation wrapping the $w_3$ contour around the $w_2$ one in $(i)$. 
            (Top-Right) Similarly, adding up the sub-terms $(iii)+(iv)$ gives a sub-term we call $(vi)$.
            (Bottom) In $(v)$, wrapping the $w_5$ contour around the $w_6$ contour gives a residue which when added to $(vi)$ gives the contour shown, which matches the integral in Lemma~\ref{lem:regulatedContours_integral_splitUp}.
        }
        \label{fig:rep_example_contours_210_4th}
    \end{subfigure}
    \end{figure}

    Similarly, we can depict the integration contours for all of the terms on the right-hand-side of Eq.~\eqref{eq:example_term210_rotated}. See Figs.~\ref{fig:rep_example_contours_210_1st},\ref{fig:rep_example_contours_210_2nd},\ref{fig:rep_example_contours_210_3rd} for the contours and the integrals they represent. Also in those figures, we add various groups of contours so that all terms with a contour $({\bf F}_j \to \infty)^{L,R}$ on a variable get paired with another term with a contour $(\infty \to {\bf F}'_j)^{L,R}$ on the same variable. As such, all contours will start with some ${\bf F}_j$ and end with some ${\bf F}'_j$. As such, we can deform all contours to be close to the straight-line segment ${\bf F}_j \to {\bf F}'_j$. However, we have to be mindful of the order of the contours around ${\bf F}_j \to {\bf F}'_j$ because crossing a contour $w_{k}$ with a contour $w_{k+1}$ will lead to a residue because of the factor $\frac{1}{2 \pi i}\frac{1}{w_k-w_{k+1}}$ in the integrand. 
    
    As in the Figs.~\ref{fig:rep_example_contours_210_1st},\ref{fig:rep_example_contours_210_2nd},\ref{fig:rep_example_contours_210_3rd}, we choose to define the terms $(i),(ii),(iii),(iv)$ as the groups of terms in Eq.~\eqref{eq:example_term210_rotated} we described above to add up. Then in Fig.~\ref{fig:rep_example_contours_210_4th}, we note that similar deformations of contours and accounting of their residues gives the integral we wanted to show in Lemma~\ref{lem:regulatedContours_integral_splitUp}, up to additive and multiplicative factors of $O(1/R)$.

\subsubsection{Argument for general terms}

\begin{figure}[p!]
\begin{subfigure}
    \centering
    \includegraphics[width=\linewidth]{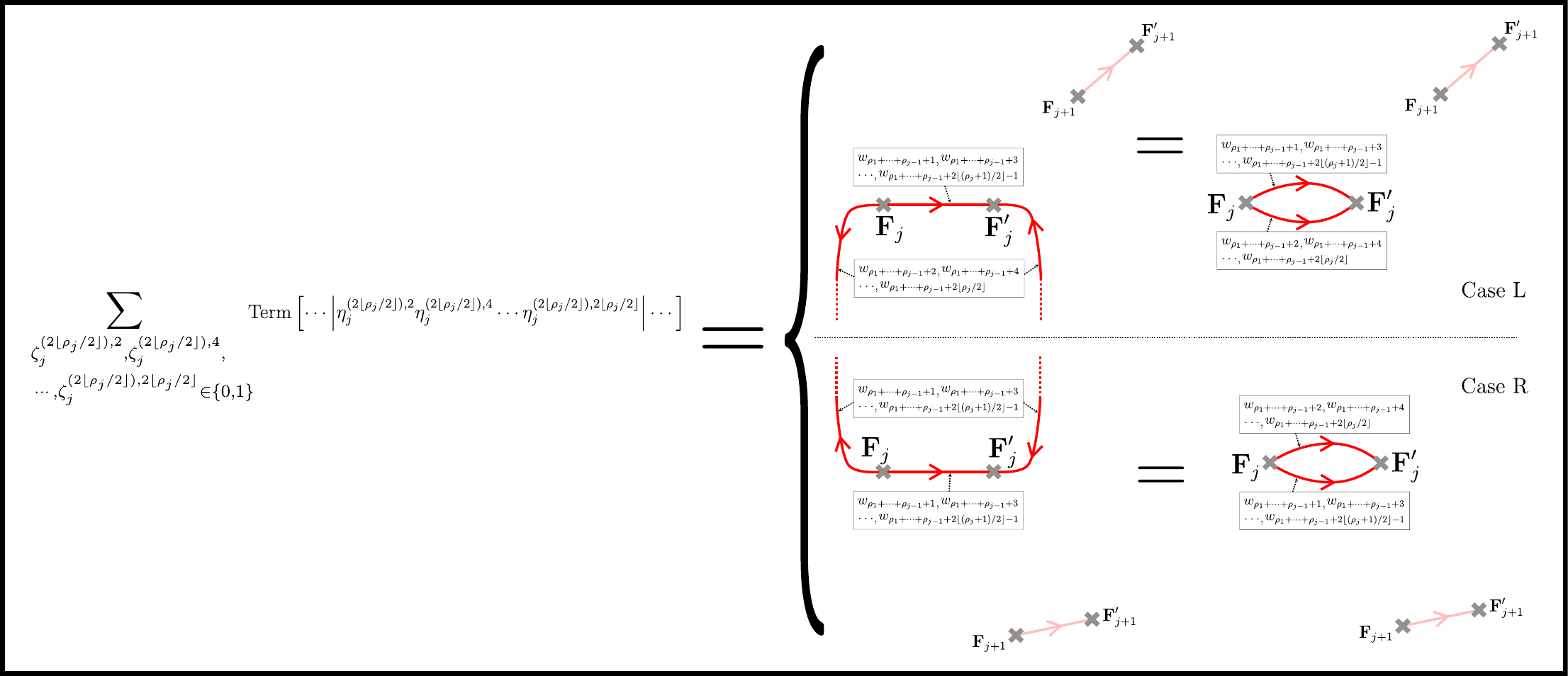}
    \caption{
    $\sum_{
        \zeta_j^{(2\floor{\rho_j/2}),2},
        \zeta_j^{(2\floor{\rho_j/2}),4},
        \cdots,
        \zeta_j^{(2\floor{\rho_j/2}),2\floor{\rho_j/2}},
        \in \{0,1\}
    }
    \mathrm{Term}
    \left[\,\cdots \,\Big|\, 
        \eta_j^{(2\floor{\rho_j/2}),2},
        \eta_j^{(2\floor{\rho_j/2}),4},
        \cdots,
        \eta_j^{(2\floor{\rho_j/2}),2\floor{\rho_j/2}} \, {}_{0}
    \,\Big|\, \cdots \,\right]$
    asymptotes to a contour integral of $\frac{1}{(2 \pi i)^n} \frac{1}{(w_1-w_2) \cdots (w_n-w_1)}$ with $\rho_j$ contours near $( {\bf F}_j \to {\bf F}'_j )$ in the configuration shown above. The specific configurations depend on whether the zipper $( {\bf F}_{j+1} \to {\bf F}'_{j+1} )$ and $\exp[i \Theta^{m,j+1}_{j,j+1}]$ give a Case L or Case R configuration (see Fig.~\ref{fig:zipper_contours_to_infty}).
    }
    \label{fig:sum_of_terms_zipper}
\end{subfigure}
\vspace{1cm}
\begin{subfigure}
    \centering
    \includegraphics[width=\linewidth]{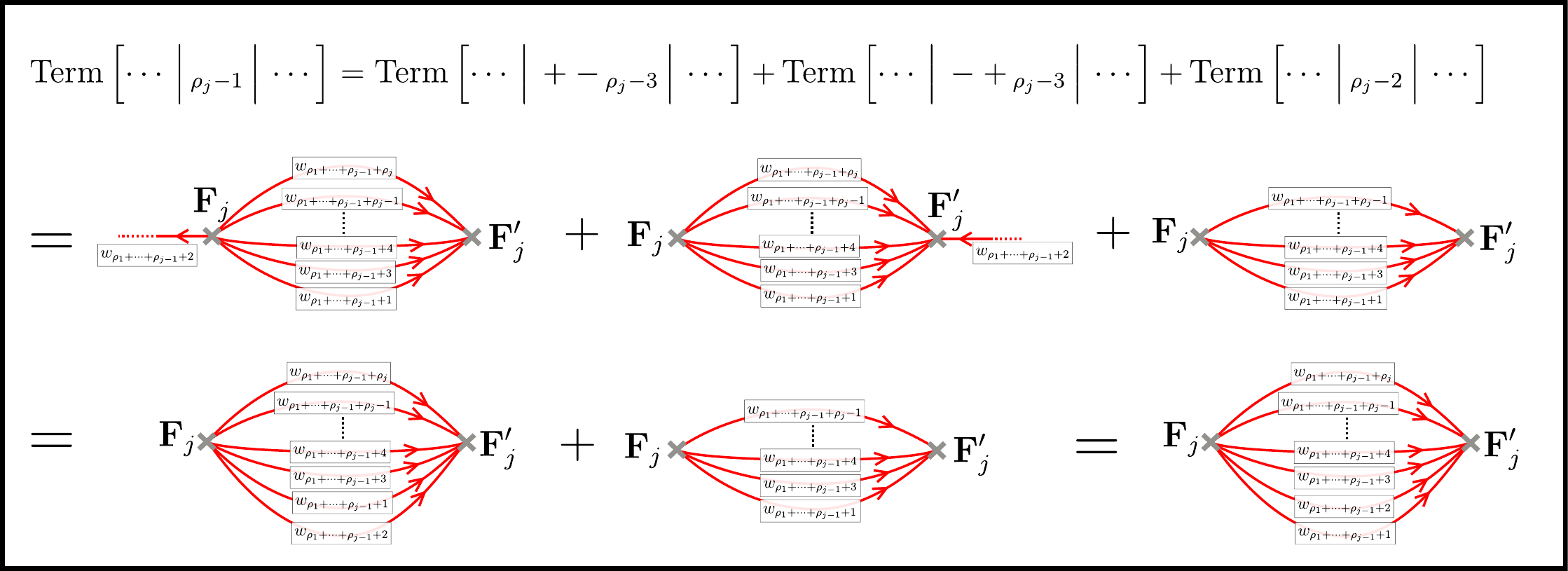}
    \caption{
        Inductive argument for the sum over contour integrals in the expansion into terms of the form Fig.~\ref{fig:sum_of_terms_zipper} giving the desired contour integral for our Lemma. The $\mathrm{Term}\left[\cdots \,\Big|\, + - {}_{\rho_j - 3} \,\Big|\, \cdots \right]$ and $\mathrm{Term}\left[\cdots \,\Big|\, - + {}_{\rho_j - 3} \,\Big|\, \cdots \right]$ can inductively be shown to have the configurations shown above, with $\rho_j$ contours $w_{\rho_1 + \cdots + \rho_{j-1} + 1},\cdots,w_{\rho_1 + \cdots + \rho_{j-1} + \rho_j}$ and an integrand $\frac{1}{(2 \pi i)^n} \frac{1}{(w_1-w_2)\cdots(w_n-w_1)}$. These have contours for $w_{\rho_1+\cdots+\rho_{j-1}+2}$ being $({\bf F}_j \to \infty)$ and $(\infty \to {\bf F}'_j)$ respectively, and add up to a closed contour $({\bf F}_j \to {\bf F}'_j)$ on the right-side of all other contours relative to the others.
        The $\mathrm{Term}\left[\cdots \,\Big|\, + - {}_{\rho_j - 3} \,\Big|\, \cdots \right]$ is the residue term, and has $\rho_j-1$ contours $w_{\rho_1 + \cdots + \rho_{j-1} + 1},\cdots,w_{\rho_1 + \cdots + \rho_{j-1} + \rho_j-1}$ and an integrand $\frac{1}{(2 \pi i)^{n-1}} \frac{1}{(w_1-w_2)\cdots(w_{n-1}-w_1)}$, with the arguments and contours for $w_{\rho_1 + \cdots + \rho_{j-1} + \rho_j},\cdots,w_{n-1}$ being the same as $w_{\rho_1 + \cdots + \rho_{j-1} + \rho_j + 1},\cdots,w_n$ respectively. This residue term acts as the residue that arises from switching the order of the $w_{\rho_1 + \cdots + \rho_{j-1} + 1},w_{\rho_1 + \cdots + \rho_{j-1} + 2}$ contours in the previous equality.
    }
    \label{fig:general_zipper_proof}
\end{subfigure}
\end{figure}

The arguments and notation we used for the example above will readily generalize. In particular, the residue calculus in the spectral variables $z_1,z_2,\cdots$ conspire to produce a series of terms which eventually we'll be able to transform into real-space variables $w_1,w_2,\cdots$ via the Schwinger parameterization trick. Then, we regroup these terms in convenient ways so that all contours run between some ${\bf F}_j \to {\bf F}'_j$ in some orders. Then, adding the contours in some special orders together with some residue calculus in the real-space variables $w_1,w_2,\cdots$ will conspire to produce the integral shown in Lemma~\ref{lem:regulatedContours_integral_splitUp}. 

In general, we have the equality
\begin{equation}
    \mathrm{Term}(\cdots \,|\, \cdots \, {}_{\rho_j} \,|\, \cdots) 
    =
    \mathrm{Term}(\cdots \,|\, \cdots +- {}_{\rho_j-2} \,|\, \cdots) 
    +
    \mathrm{Term}(\cdots \,|\, \cdots -+ {}_{\rho_j-2} \,|\, \cdots) 
    +
    \mathrm{Term}(\cdots \,|\, \cdots {}_{\rho_j-1} \,|\, \cdots) 
    \quad
    \text{ if  } \rho_j \ge 2
\end{equation}
where the $\cdots$ refer to any assignment of the variables $\eta^{\cdots}_{\cdots} \in \{+ -, - +,+,-\}$ as described in Sec.~\ref{sec:mainProof_another_rep_example}.
This equality follows from the same reasoning as in Eq.~\eqref{eq:example_equality_pt1}. Similarly, we find that
\begin{equation}
    \mathrm{Term}(\cdots \,|\, \cdots \, {}_{1} \,|\, \cdots) 
    =
    \mathrm{Term}(\cdots \,|\, \cdots + \, {}_{0} \,|\, \cdots) 
    +
    \mathrm{Term}(\cdots \,|\, \cdots - \, {}_{0} \,|\, \cdots) 
    +
    \begin{cases}
        \mathrm{Term}(\cdots \,|\, \cdots \,{}_{0} \,|\, \cdots) 
        &\quad  \text{ for Case L}
        \\
        0
        &\quad  \text{ for Case R}
    \end{cases}
\end{equation}
where Case L / Case R are described in Fig.~\ref{fig:zipper_contours_to_infty}.
A prototype calculation for Case R is given in Eq.~\eqref{eq:example_equality_pt2}, while a prototype for Case L is in Eq.~\eqref{eq:example_equality_pt3}. This equality follows because rotating the contours $z_{\cdots}$ from $-i \exp[i \Theta_j]$ to $+ \exp[i \Theta_j]$ or $- \exp[i \Theta_j]$ gives a residue from crossing the $-\exp[i\Theta^{m,j+1}_{j,j+1}]$ axis in Case L, like in Fig.~\ref{fig:example_z5_rotations} whereas this crossing is avoided in Case R, like in Fig.~\ref{fig:example_z2z3_rotations}.

In general, note that these equalities allow us to expand 
$\tr \left[ 
    \left( \frac{\overline{\partial}^{\bullet \to \circ}_{\widetilde{\mathrm{zip}}_{1}}}{\overline{\partial}^{\bullet \to \circ}} \right)^{\rho_1}
    \cdots
    \left( \frac{\overline{\partial}^{\bullet \to \circ}_{\widetilde{\mathrm{zip}}_{\tau}}}{\overline{\partial}^{\bullet \to \circ}} \right)^{\rho_{\tau}}
\right]$
in terms of ``$\mathrm{Term}$''s of the form 
$\mathrm{Term}(\cdots \, {}_{0} \,|\, \cdots \, {}_{0} \,|\, \cdots)$
for which we can apply the asymptotic analysis of the functions $g_{v \to w}(z)$. In particular, we get this sum is over all admissible assignments of the variables
assignment of the variables $\eta^{\cdots}_{\cdots} \in \{+ -, - +,+,-\}$ consistent with the index-lengths. Now, the same asymptotic analysis from Eq.~\eqref{eq:example_asymptotic_splitTerm} and the reasoning from the examples of adding up all assignments of $\eta^{\cdots}_{\cdots} \in \{+ -, - +,+,-\}$ in Figs.~\ref{fig:rep_example_contours_210_1st},\ref{fig:rep_example_contours_210_2nd},\ref{fig:rep_example_contours_210_3rd} shows that in general, the sum
\begin{equation*}
    \sum_{
        \zeta_j^{(2\floor{\rho_j/2}),2},
        \zeta_j^{(2\floor{\rho_j/2}),4},
        \cdots,
        \zeta_j^{(2\floor{\rho_j/2}),2\floor{\rho_j/2}},
        \in \{0,1\}
    }
    \mathrm{Term}
    \left[\,\cdots \,\Big|\, 
        \eta_j^{(2\floor{\rho_j/2}),2},
        \eta_j^{(2\floor{\rho_j/2}),4},
        \cdots,
        \eta_j^{(2\floor{\rho_j/2}),2\floor{\rho_j/2}}
        \, {}_{0}
    \,\Big|\, \cdots \,\right]
\end{equation*}
has the asymptotic form of a certain contour integral with integrand 
$\frac{1}{(2 \pi i)^n} \frac{1}{(w_1-w_2) \cdots (w_n-w_1)} $ and for which the contours $z_{\rho_1+\cdots+\rho_{j-1}+1},\cdots,z_{\rho_1+\cdots+\rho_{j-1}+\rho_j}$ can be taken to be in a neighborhood of $ ( {\bf F}_j \to {\bf F}'_j ) $ in a certain order. The specific order depends on whether $({\bf F}_{j} \to {\bf F}_{j})$ and $({\bf F}_{j+1} \to {\bf F}_{j+1})$ are oriented as in Case L or Case R. See Fig.~\ref{fig:sum_of_terms_zipper} for depictions.

And in general, an inductive argument summarized in Fig.~\ref{fig:general_zipper_proof} shows that this sum of the contour integrals gives the desired integral in the contour configuration described in this lemma. The base cases for the various Case L and Case R are essentially given by the example in the previous subsection, culminating in Fig.~\ref{fig:rep_example_contours_210_4th}. This concludes the proof.

\end{proof}

\subsection{Modifications for computing $\tr \left[ \frac{\overline{\partial}^{\bullet \to \circ}_{\mathrm{zip} \, i_1}}{\overline{\partial}^{\bullet \to \circ}} \cdots \frac{\overline{\partial}^{\bullet \to \circ}_{\mathrm{zip} \, i_n}}{\overline{\partial}^{\bullet \to \circ}} \right]$ with zippers sharing endpoints.} \label{sec:nonSeparatedZippers}

Now, we consider the modifications needed to generalize the previous subsections computations to connections where the zippers share endpoints. 
Note that if there are variables $w_{j},w_{j+1}$ along zippers that are well-separated and don't share endpoints, the continuum integrals in Proposition~\ref{prop:regulatedContours_integral} converge. 
However, if we replace the continuum integrals in Proposition~\ref{prop:regulatedContours_integral} with zippers so that every pair of adjacent variables $w_{j},w_{j+1}$ vary along zippers that share endpoints, the resulting integral is actually divergent. Since we work on a lattice, for any finite $R$, the final result should be be finite, although the final result should be expected to diverge as $R \to \infty$. 

\begin{figure}[h!]
    \centering
    \includegraphics[width=0.7\linewidth]{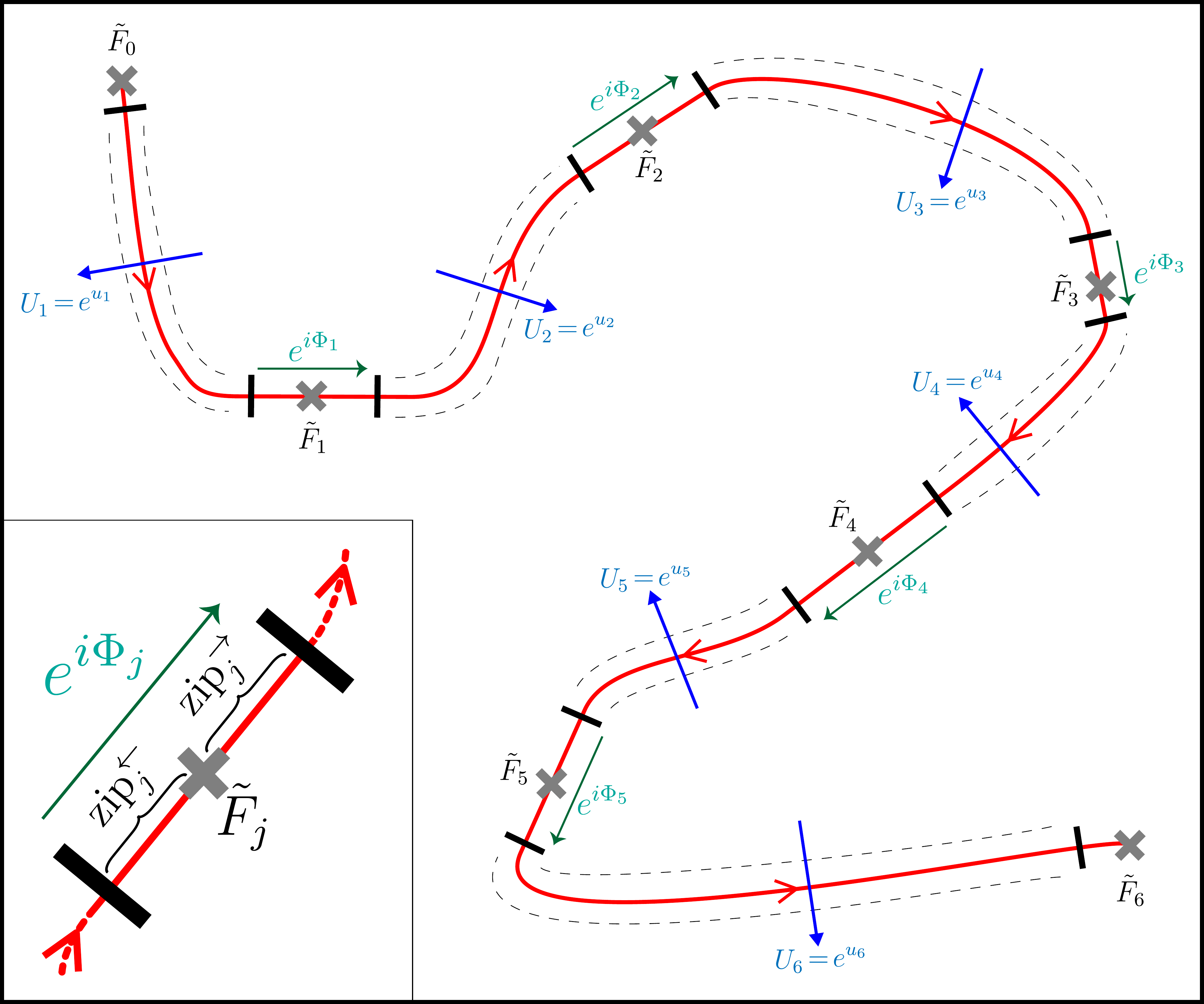}
    \caption{
        Configurations of zippers as described in Sec.~\ref{sec:nonSeparatedZippers}.
    }
    \label{fig:nonSeparated_zippers}
\end{figure}

To do the computation for finite $R$, we will choose a specific representation of the zipper gauge field, as in Fig.~\ref{fig:nonSeparated_zippers}.
Given punctures $\tilde{F}_0,\cdots,\tilde{F}_{s} \in \C$, we can choose continuum curve passing through the points in the order $\tilde{F}_0 \to \tilde{F}_1 \to \cdots \to \tilde{F}_{s-1} \to \tilde{F}_{s}$ and so that the neighborhood of the curve in the neighborhood of each of the interior points  $\tilde{F}_1, \cdots, \tilde{F}_{s-1}$ is a straight line segment. Then, for each of the parts of the curve $\tilde{F}_{j-1} \to \tilde{F}_{j}$, we assign the connection $U_j = e^{u_j}$ going across the segment. 
Then, we break up the zippers into small segments in the same way as Figs.~\ref{fig:multipleZippers_example},\ref{fig:multipleZippers_splitUpAngle}, making sure that the segments on either side of $\tilde{F}_j$ for $j=1,\cdots,s-1$ are straight lines parallel to each other. We refer to these parallel straight-line segments on the left- and right-side of $\tilde{F}_j$ as $\mathrm{zip}^{\leftarrow}_j$ and $\mathrm{zip}^{\rightarrow}_j$ respectively, and we define their direction along the curve as $e^{i\Phi_j}$.

Now, we want to describe the modification to the expression of Lemma~\ref{lem:regulatedContours_integral_splitUp} in the situation where $\mathrm{zip}\, q_1,q'_1 = \mathrm{zip}_j^{\rightarrow}$ and $\mathrm{zip}\, q_n,q'_n = \mathrm{zip}_j^{\leftarrow}$ or vice-versa for some $j$. 
\textbf{\textit{In summary, given a kind of configuration of zippers, the only modification is to replace the integrand }}
\begin{equation*}
    \frac{1}{(2 \pi i)^n} \frac{1}{(w_1-w_2)\cdots(w_n-w_1)} + \mathrm{c.c.}
    \mapsto
    \frac{1}{(2 \pi i)^n} \frac{1-e^{-R|w_n-w_1|}}{(w_1-w_2)\cdots(w_n-w_1)} + \mathrm{c.c.},
\end{equation*}
so that the term for $i_1,\cdots,i_n$ not all equal is asymptotically
\begin{equation} \label{eq:zipper_expansion_lattice_nonSep}
\begin{split}
    &\tr \left[ \frac{\overline{\partial}^{\bullet \to \circ}_{\mathrm{zip} \, i_1}}{\overline{\partial}^{\bullet \to \circ}} \cdots \frac{\overline{\partial}^{\bullet \to \circ}_{\mathrm{zip} \, i_n}}{\overline{\partial}^{\bullet \to \circ}} \right]
    \xrightarrow{R \to \infty}
    \frac{1}{(2 \pi i)^{n}}
    \int_{{\bf P}_{j_1}^{(c(1))}} dw_1 \cdots \int_{{\bf P}_{j_n}^{(c(n))}} dw_{\ell_{1}+\cdots+\ell_{n}}
    \frac{1-e^{-R|w_n-w_1|}}{(w_1-w_2)\cdots(w_n-w_1)} + \mathrm{c.c.}.
\end{split}
\end{equation}
We can compare this to the analogous factor of $(1-e^{-D|w_{2n}-w_1|})$ which appeared in the single-zipper case of Propsition.~\ref{prop:asymptIntegral_singleZip}.

In general, the only step in the previous subsection's arguments that don't carry through is the analog of the second equality in Eq.~\eqref{eq:firstExample_asymptotics_equation}. 
Using the same notational compression as Sec.~\ref{sec:setupNotation_first_proof}, this step involved changing the integration bounds of $z_1$ from $0 \to -\exp\left[i \Theta_{\tau,1}^{m,1} \right]$ to $0 \to -\exp\left[i \Theta_{\tau,1}^{m,1} \right]\infty$. If the zippers were all well-separated, we explained there how this extension did not affect the asymptotic analysis and warned in the Footnote~\ref{foot:extension_warning} how this breaks down if $\mathrm{zip}\, q_1,q'_1 = \mathrm{zip}_j^{\rightarrow}$ and $\mathrm{zip}\, q_n,q'_n = \mathrm{zip}_j^{\leftarrow}$ {\color{RoyalPurple}[resp. vice-versa]} for some $j$. In this case, we have $\exp\left[i \Theta_{\tau,1}^{m,1} \right] = e^{i\Phi_j}$ {\color{RoyalPurple}[resp. $-e^{i\Phi_j}$]}. Then, the analog of the fourth equality of Eq.~\eqref{eq:firstExample_asymptotics_equation} will have the integral over $z_1$ return $\frac{1-e^{e^{-i\Phi_j}(w_n-w_1)}}{w_n-w_1} = \frac{1-e^{-|w_n-w_1|}}{w_n-w_1}$ {\color{RoyalPurple}[resp. $\frac{1-e^{-e^{i\Phi_j}(w_n-w_1)}}{w_n-w_1} = \frac{1-e^{-|w_n-w_1|}}{w_n-w_1}$]}. The factor of $R$ in front of $|w_n-w_1|$ comes from replacing all ${\bf F}_{\cdots} \mapsto {\bf \tilde{F}}_{\cdots}$.

While the above argument gives the $\frac{1-e^{-R|w_n-w_1|}}{w_n-w_1}$ factor for adjacent segments, one may as well stick in that factor for the whole integral, since $e^{-R|w_n-w_1|}$ is negligibly small for non-adjacent segments.

\section{Discussion} \label{sec:discussion}

We've shown how the series expansion of isoradial dimer models in the presence of singular connections flat away from a fixed set of punctures matches an analogous series expansion for free fermions in the scaling limit of the lattice size going to zero. 

It would be nice to reproduce the convergence results of Dubédat~\cite{Dubedat2011,Dubedat2018doubleDimersConformalLoopEnsemblesIsomonodromicDeformations} in our framework.
Furthermore, it would be interesting to see if the series expansions for more general correlation functions could be explicitly resummed. In~\cite{Dubedat2011}, the question of evaluating height function correlations like $|\braket{e^{i \, 2\pi /3 (h(F_1) + h(F_2) + h(F_3))}}|$ was raised. In our framework, this corresponds to three punctures and two zippers of strength $e^{i \, 2\pi/3}$ between them. 
The results of Sec.~\ref{sec:nonSeparatedZippers} give a series expansion for these correlations, but a few thing are unclear, namely whether the subleading corrections make a difference to the summation or whether we should even expect a nice universal answer to these correlations in the first place.

While we spent our time analyzing this for the case of singular connections, an analogous statement should hold if we consider a \textit{smooth}, non-flat connections on $\C$ approximated by a graph connection. In particular, the expansion Eq.~\eqref{eq:diracDeterminant_full} is expected to be recovered from considering the analogous cluster expansion of the graph connection. 
In the case of a smooth connection, the difference between the twisted $\overline{\partial}(U)$ and the original $\overline{\partial}$ will be of the same order as the lattice scale; the leading order terms from summing over all edges in the graph should lead to the same expression as the continuous Dirac fermions. 
We leave details and precise formulations of this to future work.

Also, we expect these results to hold in more general settings than isoradial dimers. In particular, we expect the same result to hold for any liquid phase dimer model~\cite{kso2003amoebae}. 
For example, recent progress towards generalizing isoradial graphs to support dimer models with higher-genus spectral curves has been made. In~\cite{bcdT2022isoradialimmersion}, so-called \textit{isoradial immersions} are classified and shown~\cite{bcdT2020ellipticGenus1} to have analogous discrete exponentials and contour integral formulas for inverse Kasteleyn matrices in various phases of more general dimer models (see also ~\cite{bcdT2020ellipticGenusHigher,fock2015inverseSpectral,ggk2022inverseSpecMapDimers}). We expect that in the liquid phase of critical dimers, analogous computations to the exact identities of Sec.~\ref{sec:exactExpr_zipperTraces} would facilitate these computations. 

While the expectation for liquid phase dimers is clear, it is not clear what to expect for dimers in a gaseous phase away from criticality. It would be interesting if they could be identified with, for example, adding a fermionic mass term or more general terms to the Lagrangian.

And more generally and vaguely, it would be nice to elucidate and expand on the close relationship between the dimer model and the free fermion model. 
The free fermion is a `basis' of rational conformal field theories in the sense that products and quotients of its operator algebras generate many if not all of them~\cite{bigYellowCFT1997}.
Since many statistical models' correlations are described by rational theories, it is interesting to ask how generally various models' correlations could be expressed in terms of the dimer model, either in an exact or scaling-limit sense. 

\section*{Acknowledgments}
We acknowledge the Yale Math Department for keeping us fed. We also thank Rick Kenyon for helpful discussions and suggestions related to this problem, Andy Neitzke for discussions about tau functions, Fedor Petrov and Terry Tao for MathOverflow discussions~\cite{mathOverflowMultIntegralZeta2021,mathOverflowDiagRestrHilbTransf2022} that helped us with the computations in Appendix~\ref{app:sum_leading_order_I2n}, and Mathematica~\cite{Mathematica} for many computations.

\appendix

\section{Asymptotics of the functions $g_{v \to w}(z)$} \label{app:asympt_gFunc}

We now discuss asymptotics of the functions $g(z) = g_{v \to w}(z) = \prod_{j=1}^{P} \frac{z-e^{i \alpha_j}}{z-e^{i\beta_j}}$ in the limit of large distance between vertices corresponding to dual faces on the rhombus grid $D := |w-v| \to \infty$, following and extending the analysis of~\cite{kenyon2002critical}.

In a similar vein as Sec.~\ref{sec:straightLine_zippers}), the fact that we are dealing with a rhombus graph means~\cite{kenyon2002critical,BMS05} that there exists some $\phi_0$ for which $\{\beta_j\}$ are angles within some range $( \arg(w-v) - \frac{\pi}{2} , \arg(w-v) + \frac{\pi}{2} )$ and the $\{\alpha_j\}$ are in a range $( \arg(w-v) + \frac{\pi}{2} , \arg(w-v) + \frac{3\pi}{2} )$. 

To study the function $g(z)$, we need to impose certain \textit{regularity conditions} on the angles $\{\alpha_j\}$ and $\{\beta_j\}$. In particular, we want that these angles are always bounded away from the extremal values, that there exists angles $\phi_0$ and $\epsilon > 0$ such that 
\begin{equation*}
- \frac{\pi}{2} + \epsilon < \beta_j - \phi_0 < \frac{\pi}{2} - \epsilon \quad\text{ and }\quad \frac{\pi}{2} + \epsilon < \alpha_j - \phi_0 < \frac{3\pi}{2} - \epsilon \quad \text{ with } |\arg(w-v) - \phi_0| < \epsilon,
\end{equation*}
uniformly for all $v,w$ in the portion of the graph we care about. This condition will automatically be satisfied if the rhombus graph is periodic, or more generally if the total possible number of angles is finite (studied in~\cite{BMS05} as the \textit{quasicrystallograhic} condition). 

In the remainder section, we will assume this regularity condition. In the remainder of this section, we'll WLOG that $\arg(w-v) = 0$.

Note that this regularity condition implies there is some constant $c_{\epsilon} = \cos(\epsilon) < 1$ such that $\mathrm{Re} \, e^{i (\beta_j - \phi_0)} = \cos(\beta_j - \phi_0) > c_{\epsilon}$ and $-\mathrm{Re} \, e^{i (\alpha_j - \phi_0)} = -\cos(\alpha_j - \phi_0) > c_{\epsilon}$, so that 
\begin{equation} \label{eq:D_N_bounds}
\begin{split}
    &D = |w-v| \ge \mathrm{Re}((w-v)e^{-i \phi_0}) = \mathrm{Re} \, \sum_{j=1}^{P} (e^{i (\beta_j - \phi_0)} - e^{i (\alpha_j - \phi_0)}) > 2 P \cdot c_{\epsilon}  \quad,\quad D = \left|  \sum_{j=1}^{P} (e^{i \beta_j} - e^{i \alpha_j}) \right| < 2 P,
    \\
    &\text{so} \,\, \frac{D}{2} < P <  \frac{1}{2 c_{\epsilon}} \cdot D
\end{split}
\end{equation}
using our assumptions about the $\{\alpha_j\}, \{\beta_j\}$. Now, we can express $g(z)$ using a series expansion. 

For $|z|<1$, we have
\begin{equation} \label{eq:g_le1_seriesExp}
\begin{split}
    g(z) &= \prod_{j=1}^{P} \frac{z-e^{i \alpha_j}}{z-e^{i\beta_j}} = \left( \prod_{j=1}^{P} \frac{e^{i\alpha_j}}{e^{i\beta_j}} \right) \cdot \exp \left[ z \sum_{j=1}^P (e^{-i \beta_j} - e^{-i \alpha_j}) + \frac{z^2}{2} \sum_{j=1}^P (e^{-2 i \beta_j} - e^{-2 i \alpha_j}) + \frac{z^3}{3} \sum_{j=1}^P (e^{-3 i \beta_j} - e^{-3 i \alpha_j}) + \cdots \right] \\
    &= \gamma \cdot \exp \left[ \overline{(w-v)} z + h_1(z) \right]. 
\end{split}
\end{equation}
In the above, we defined $\gamma := \left( \prod_{j=1}^{P} \frac{e^{i\alpha_j}}{e^{i\beta_j}} \right)$ in addition to $h_1(z) = \frac{z^2}{2} \sum_{j=1}^P (e^{-2 i \beta_j} - e^{-2 i \alpha_j}) + \frac{z^3}{3} \sum_{j=1}^P (e^{-3 i \beta_j} - e^{-3 i \alpha_j}) + \cdots$. We can bound $h_1(z)$ and its derivatives using Eq.~\eqref{eq:D_N_bounds}, so that for $|z|<1$, 
\begin{equation} \label{eq:g_le1_hEstimates}
    |h_1(z)| < 2 P (z-\log{(1-z)}) < \frac{D}{c_\epsilon} (z - \log(1-z))
\end{equation}

For $|z|>1$, we have
\begin{equation} \label{eq:g_ge1_seriesExp}
\begin{split}
    g(z) &= \prod_{j=1}^{P} \frac{z-e^{i \alpha_j}}{z-e^{i\beta_j}} = \exp \left[\frac{1}{z} \sum_{j=1}^P (e^{i \beta_j} - e^{i \alpha_j}) + \frac{1}{2 z^2} \sum_{j=1}^P (e^{2 i \beta_j} - e^{2 i \alpha_j}) + \frac{1}{3 z^3} \sum_{j=1}^P (e^{3 i \beta_j} - e^{3 i \alpha_j}) + \cdots \right] \\
    &= \exp \left[(w-v) z + h_2(1/z) \right]. 
\end{split}
\end{equation}
Above, we defined $h_2(z) = \frac{z^2}{2} \sum_{j=1}^P (e^{2 i \beta_j} - e^{2 i \alpha_j}) + \frac{z^3}{3} \sum_{j=1}^P (e^{3 i \beta_j} - e^{3 i \alpha_j}) + \cdots$. Again, we can bound $h_1(z)$ and its derivatives using Eq.~\eqref{eq:D_N_bounds}, so that for $|z| > 1$, 
\begin{equation} \label{eq:g_le1_hEstimates2}
    |h_2(z)| < 2 P (z - \log{(1-z)}) < \frac{D}{c_{\epsilon}} (z - \log{(1-z)})
\end{equation}

Now, let's consider how to bound $|g(z)|$ for $z$ lying in a cone $z \in \bigsqcup_{\theta \in (-\frac{\epsilon}{8},\frac{\epsilon}{8})} e^{i \theta} \R_{< 0}$. In particular, we will see that for intermediate values of $ D^{-2/3} < |z| < D^{2/3} $ in the above cone, $g(z)$ is exponentially suppressed. 
\begin{prop} \label{prop:g_expDecay_cone}
Let $D$ be sufficiently large. Then there exists a constant $c'_{\epsilon}$ depending only on $\epsilon$ for which 
\begin{equation}
    |g(z)| < \exp\left[ - \frac{D^{1/3}}{c'_{\epsilon}} \right]
\end{equation}
for all $z = -e^{i\theta}|z|$ with $\theta \in (-\frac{\epsilon}{8},\frac{\epsilon}{8})$ and $D^{-2/3} < |z| < D^{2/3}$. 
\end{prop}
In~\cite{kenyon2002critical}, a similar inequality was demonstrated (in our notation) along the ray $e^{i \phi_0} \R_{<0}$. This statement now is a strengthening of the inequality in a cone containing $e^{i \phi_0} \R_{<0}$.
\begin{proof}
    The point of our proof will be to balance two competing effects. First, the rhombus angles that are bounded away from $\pm i$ will be contribute to the exponential decay of $g(z)$. Second, the rhombus angles close to $\pm i$ may cause $g(z)$ to grow. As such, if we can show that there aren't that many rhombus angles close to $\pm i$, then the exponential decay of the ones away from $\pm i$ will outweigh the factors of contributing to possible exponential growth.
    
    Recall that $w-v = \sum_{j=1}^{P}(e^{i \beta_j} - e^{i \alpha_j})$. Let us label $\{e^{i \beta_1}, \cdots, e^{i\beta_P}, -e^{i\alpha_1}, \cdots, -e^{i\alpha_P}\}$ as $\{e^{i\tilde{\beta}_1}, \cdots, e^{i\tilde{\beta}_{2P}}\}$ so that all $-\frac{\pi}{2} < \tilde{\beta} < \frac{\pi}{2}$.
    Note that the regularity condition implies that the total possible range of the angles is an interval of length $(\pi - 2 \epsilon)$ by our regularity conditions.
    
    Furthermore, let's say that the angles $\{\tilde{\beta}_1, \cdots, \tilde{\beta}_{k}\}$ are the `problematic angles' that lie in a range $(\frac{\pi}{2} - \frac{\epsilon}{2}, \frac{\pi}{2})$. Since the total range of possible angles has length $\pi - 2 \epsilon$, we'll have that if $k > 0$, then all other angles $\{e^{i\tilde{\beta}_{k+1}}, \cdots, e^{i\tilde{\beta}_{2P}}\}$ will necessarily lie in the range $(-\frac{\pi}{2} + \frac{3 \epsilon}{2}, \frac{\pi}{2} - \frac{\epsilon}{2})$. We will henceforth assume that $k>0$, so that the problematic angles are in the range $(\frac{\pi}{2} - \frac{\epsilon}{2}, \frac{\pi}{2})$. If there are no problematic angles in this range, we can consider instead they exist in the range $(-\frac{\pi}{2}, -\frac{\pi}{2} + \frac{\epsilon}{2})$ and that the rest of the angles exist in $(\frac{\pi}{2} + \frac{\epsilon}{2}, \frac{\pi}{2} - \frac{3\epsilon}{2})$. If there are no problematic angles in either of these ranges, then the estimates can be derived similarly to~\cite{kenyon2002critical} and will be stronger than in the rest of this proof.
    
    We want to show that the fraction of problematic angles is less than $1/2$ and also bounded away from $1/2$. Note that since $w - v = \sum_{j=1}^{2P} e^{i \tilde{\beta}_j}$ is real, we have
    \begin{equation*}
        \sum_{j=1}^{k} \sin(\tilde{\beta}_j) = -\sum_{j=k+1}^{2P} \sin(\tilde{\beta}_j).
    \end{equation*}
    Also, we have that for $j=1,\cdots,k$, 
    $\sin(\tilde{\beta}_j) > \sin(\frac{\pi}{2} - \frac{\epsilon}{2}) = \cos(\frac{\epsilon}{2}) $
    . And for $j=k+1,\cdots,2P$, $-\sin(\tilde{\beta}_j) < \sin(\frac{\pi}{2} - \frac{3 \epsilon}{2}) = \cos(\frac{3\epsilon}{2})$. These with the above equation give the bound
    \begin{equation*}
       k \cdot \cos(\epsilon/2) < \sum_{j=1}^{k} \sin(\tilde{\beta}_j) = -\sum_{j=k+1}^{2P} \sin(\tilde{\beta}_j) < (2P - k) \cos(3\epsilon/2)
    \end{equation*}
    which implies
    \begin{equation} \label{eq:gConeBound_k_vs_N}
       \frac{k}{2 P} < \frac{1}{2 + \frac{\cos(\epsilon/2) - \cos(3\epsilon/2)}{\cos(3\epsilon/2)}}
    \end{equation}
    which is bounded away from $1/2$ as we wanted.
    
    Now, we turn back to analyzing $|g(z)|$ in our range. Since $z = -|z|e^{i \theta}$, we have 
    \begin{equation} \label{eq:gConeBound_abs_g}
    \begin{split}
        |g(z)| 
        &= \left| \prod_{j=1}^{P} \frac{z-e^{i \alpha_j}}{z-e^{i\beta_j}} \right|
        = \left| \prod_{j=1}^{P} \frac{|z|+e^{i (\alpha_j-\theta)}}{|z|+e^{i(\beta_j-\theta)}} \right|
        = \prod_{j=1}^{P} \sqrt{ \frac{|z|^2 + 2|z| \cos(\alpha_j-\theta) + 1}{|z|^2 + 2|z| \cos(\beta_j-\theta) + 1} } = \prod_{j=1}^{P} \sqrt{ \frac{ 1 + \frac{2 \cos(\alpha_j - \theta)}{|z|+1/|z|} }{ 1 + \frac{2 \cos(\beta_j - \theta)}{|z|+1/|z|} } }.
    \end{split}
    \end{equation}
    Let's now bound each of the numerators and denominators in the above.
    
    First, consider $\alpha_j$ or $\beta_j$ being an `unproblematic' angle $\tilde{\beta}_{2k+1},\cdots,\tilde{\beta}_{2P}$, meaning that 
    $\alpha_j \in (\frac{\pi}{2} + \frac{\epsilon}{2} , \frac{3\pi}{2} - \frac{3\epsilon}{2})$ 
    and 
    $\beta_j \in (-\frac{\pi}{2} + \frac{\epsilon}{2} , \frac{\pi}{2} - \frac{3\epsilon}{2})$. 
    Combining with the original restriction 
    $\theta \in (-\frac{\epsilon}{8},\frac{\epsilon}{8})$,
    we get that for these unproblematic angles, $\alpha_j - \theta \in (\frac{\pi}{2} + \frac{3 \epsilon}{8} , \frac{3\pi}{2} - \frac{11\epsilon}{8})$ so that $-\sin({3\epsilon}/{8}) < \cos(\alpha_j - \theta) < 0$. 
    Similarly, $\beta_j - \theta \in (-\frac{\pi}{2} + \frac{3 \epsilon}{8} , \frac{\pi}{2} - \frac{11\epsilon}{8})$ so that $\cos(\beta_j - \theta) > \sin({3\epsilon}/{8})$.
    These give the estimates
    \begin{equation} \label{eq:gConeBound_unprob_terms}
    \begin{split}
        &
        \left|1 + \frac{2}{|z|+1/|z|} \cos(\alpha_j - \theta) \right|
        < 
        \left|1 - \frac{2}{|z|+1/|z|} \sin({3\epsilon}/{8})\right| 
        < 
        \left|1 - \frac{1}{|z|+1/|z|} \sin({3\epsilon}/{8})\right| 
        < 
        \exp\left[- \frac{1}{|z|+1/|z|} \sin({3\epsilon}/{8})\right],
        \\
        & \text{ and } \quad
        \left|\frac{1}{1 + \frac{2}{|z|+1/|z|} \cos(\beta_j - \theta)} \right| 
        <
        \left| 1 - \frac{1}{|z|+1/|z|} \sin({3\epsilon}/{8}) \right| 
        <
        \exp\left[ - \frac{1}{|z|+1/|z|} \sin({3\epsilon}/{8})\right].
    \end{split}
    \end{equation}
    In the latter set of inequalities, the first step uses $\frac{1}{1+x} < 1 - \frac{x}{2}$ for $x \in (0,1)$. Note that these terms contribute to an exponential decay of $|g(z)|$.
    
    Now, we consider how to bound the `problematic angles' which satisfy
    $\alpha_j \in (\frac{3\pi}{2} - \frac{\epsilon}{2} , \frac{3\pi}{2})$ 
    and 
    $\beta_j \in (-\frac{\pi}{2} - \frac{\epsilon}{2} , \frac{\pi}{2})$.
    For these angles, note that $\alpha_j - \theta \in (\frac{3\pi}{2} - \frac{3\epsilon}{8},\frac{3\pi}{2} + \frac{\epsilon}{8})$ and $\beta_j - \theta \in (\frac{\pi}{2} - \frac{3\epsilon}{8},\frac{\pi}{2} + \frac{\epsilon}{8})$. We get that
    $-\sin(3\epsilon/8) < \cos(\alpha_j - \theta) < \sin(\epsilon/8)$ 
    and
    $-\sin(\epsilon/8) < \cos(\beta_j - \theta) < \sin(3\epsilon/8)$. This gives
    \begin{equation} \label{eq:gConeBound_prob_terms}
    \begin{split}
        &
        \left|1 + \frac{2}{|z|+1/|z|} \cos(\alpha_j - \theta)\right| 
        < 
        \left|1 + \frac{2}{|z|+1/|z|} \sin({\epsilon}/{8}) \right| 
        < 
        \left|1 + \frac{1}{|z|+1/|z|} \sin({3\epsilon}/{8}) \right| 
        < 
        \exp\left[\frac{1}{|z|+1/|z|} \sin({3\epsilon}/{8})\right],
        \\
        & \text{ and } \quad
        \frac{1}{\left|1 + \frac{2}{|z|+1/|z|} \cos(\beta_j - \theta)\right|} 
        < 
        \frac{1}{\left|1 - \frac{2}{|z|+1/|z|} \sin(\epsilon/8)\right|} 
        < 
        1 + \frac{1}{|z|+1/|z|} \sin({3\epsilon}/{8}) \le \exp\left[ \frac{1}{|z|+1/|z|} \sin({3\epsilon}/{8}) \right].
    \end{split}
    \end{equation}
    In the first line, the second inequality uses $2 \sin(\epsilon/8) < \sin(3\epsilon/8)$ for all $\epsilon \in (0,\frac{\pi}{2})$.
    In the second line, the second inequality uses $\frac{1}{1-x \sin(\epsilon/8)} < 1 + \frac{x}{2} \sin(3\epsilon/8)$ for $x \in (0,1)$ and $\epsilon \in (0,\frac{\pi}{2})$.
    
    Combining Eqs.~(\ref{eq:gConeBound_k_vs_N},\ref{eq:gConeBound_abs_g},\ref{eq:gConeBound_unprob_terms},\ref{eq:gConeBound_prob_terms}) shows that there's some constant 
    $\tilde{c}'_{\epsilon} = \left( 1 - \frac{2}{2 + \frac{\cos(\epsilon/2) - \cos(3\epsilon/2)}{\cos(3\epsilon/2)}} \right) \sin(\frac{3\epsilon}{8}) > 0$
    for which
    \begin{equation*}
        |g(z)| < \exp[-2P \tilde{c}'_{\epsilon} \frac{1}{|z|+1/|z|}].
    \end{equation*}
    Now, we use the fact that $\frac{1}{|z|+1/|z|} < D^{1/3}$ for $D^{-2/3} < |z| < D^{2/3}$ and that $P > D/2$ to prove the in the proposition. 
\end{proof}

Now, we want to show the following statements comparing the asymptotic forms of expressions that show up throughout this paper.
\begin{prop} \label{prop:1minusg_estimate} 
Let $z_1 = e^{i \theta_1}|z_1|, z_2 = e^{i\theta_2}|z_2|$. 
Suppose either 
$0 < |z_1|,|z_2| < D^{-2/3}$
or
$D^{2/3} < |z_1|,|z_2| < \infty$. 
And suppose
$\pi - \epsilon/8 < \theta_1 < \pi + \epsilon/8$ and $-\epsilon/8 < \theta_2 < \epsilon/8$, 
Then there exist positive $D_{\epsilon}$ and $C_{\epsilon}$  depending only on $\epsilon$ such that for $D > D_{\epsilon}$
\begin{equation*}
     \frac{1 - \frac{g(z_1)}{g(z_2)}}{z_1 - z_2} = \left(1 + \frac{C(z_1,z_2)}{D}\right) \frac{1 - e^{D (z_1 - z_2)}}{z_1 - z_2}, 
     \quad \text{ where } |C(z_1,z_2)| < C_{\epsilon}
\end{equation*}
\end{prop}

\begin{prop} \label{prop:gminusg_estimate} 
Let $z_1 = e^{i \theta_1}|z_1|, z_2 = e^{i\theta_2}|z_2|$. 
Suppose either 
$0 < |z_1|,|z_2| < D^{-2/3}$
or
$D^{2/3} < |z_1|,|z_2| < \infty$. 
And suppose
$\pi - \epsilon/8 < \theta_1 < \pi + \epsilon/8$ and $\pi - \epsilon/8 < \theta_2 < \pi + \epsilon/8$,
\begin{equation*}
     \frac{g(z_1) - g(z_2)}{z_1 - z_2} = \left(1 + \frac{C(z_1,z_2)}{D}\right) \frac{e^{D z_1} - e^{D z_2}}{z_1 - z_2}, \quad \text{ where } |C(z_1,z_2)| < C_{\epsilon}
\end{equation*}
or that both $z_1,z_2$ both satisfy either $0 < z_1,-z_2 < D^{-2/3}$ or $D^{2/3} < z_1,-z_2 < \infty$.
\end{prop}

Note that for Proposition~\ref{prop:1minusg_estimate} the $z_1,z_2$ have approximately opposite arguments close to the negative,positive real axis, and $\frac{g(z_1)}{g(z_2)}$ is decays exponentially in that region, whereas for Proposition~\ref{prop:1minusg_estimate} $z_1,z_2$ have approximately the same arguments close to the negative real axis. Also, for all $|z_1|,|z_2| > 0$, note that 
\begin{equation} \label{eq:gMinusg_lt_1Minusg}
    0 < \frac{ \exp[-D|z_2|] - \exp[-D|z_1|] }{ |z_1| - |z_2| } < \frac{ 1 - \exp[-D(|z_1|+|z_2|)] }{ |z_1| + |z_2| },
\end{equation}
which can be combined with the above Propositions to estimate certain quantities that show up in our analyses.

\begin{proof}[Proof of Proposition~\ref{prop:1minusg_estimate}]
We will show the statement for $-D^{-2/3} < z_1,z_2 < 0$. The other cases of $-\infty < z_1,z_2 < -D^{2/3}$ or involving $\frac{1}{g(z)}$ follows from the exact same reasoning using the above estimates of $g(z)$ after changing variables $z \to \frac{1}{z}$ or using the similar estimates for $\frac{1}{g(z)}$.

First, let's write 
\begin{equation} \label{eq:1minusg_estimate_1}
    \frac{1 - \frac{g(z_1)}{g(z_2)}}{z_1 - z_2} 
    = 
    \frac{1 - \exp\left[D(z_1 - z_2) + h_1(z_1) + h_2(z_2)\right]}{z_1 - z_2}
    = 
    \exp\left[h_1(z_1)+h_1(z_2)\right] \frac{\exp\left[-h_1(z_1)-h_2(z_2)\right] - \exp\left[D(z_1 - z_2)\right]}{z_1 - z_2}.
\end{equation}
We want to estimate the real and imaginary parts of the factor $\exp\left[\pm(h_1(z_1)+h_2(z_2))\right]$. First, we need to estimate $|h_1(z_1)+h_2(z_2)|$. By Eq.~(\ref{eq:g_le1_hEstimates},\ref{eq:g_le1_hEstimates2}),
\begin{equation} \label{eq:abs_hPlush_estimate}
\begin{split}
    |h_1(z_1)+h_2(z_2)|
    &< 
    \frac{D}{c_{\epsilon}} 
    (|z_1| + |z_2| - \log(1-|z_1|) - \log(1-|z_2|)) 
    <
    \frac{D}{c_{\epsilon}} \big( (|z_1|+|z_2|) - \log( 1 - (|z_1|+|z_2|) ) \big) 
    \\
    &< 
    \frac{D}{c_{\epsilon}} (|z_1|+|z_2|)^2, \quad \text{ on } |z_1|+|z_2| \in (0,2 D^{-2/3}) \text{ for } D > 10
\end{split}
\end{equation}
This gives (for $D > 10$)
\begin{equation}
\begin{split}
    \exp\left[- \frac{D}{c_{\epsilon}} (|z_1|+|z_2|)^2 \right] 
    < 
    \mathrm{Re} \exp[\pm(h_1(z_1)+h_2(z_2))] 
    &<
    \exp\left[ \frac{D}{c_{\epsilon}} (|z_1|+|z_2|)^2 \right],
    \\
    \text{ and } \quad
    | \mathrm{Im} \exp[\pm(h_1(z_1)+h_2(z_2))] | 
    &< 
    \exp\left[ \frac{D}{c_{\epsilon}} (|z_1|+|z_2|)^2 \right] \cdot \frac{D}{c_{\epsilon}} (|z_1|+|z_2|)^2
\end{split}
\end{equation}
At this point, we want to estimate $\exp[\frac{D}{c_{\epsilon}} (|z_1|+|z_2|)^2]$. By the Taylor approximation theorem, we can write
\begin{equation}
\begin{split}
    \exp[\frac{D}{c_{\epsilon}} (|z_1|+|z_2|)^2] 
    =
    1 + R(|z_1|+|z_2|), \text{ where } |R(y)| 
    &< y
    \cdot \mathrm{max}_{|z_1|+|z_2| \in (0, 2 D^{-2/3})} \left\{  \frac{D}{c_{\epsilon}} 2 y \exp\left[ \frac{D}{c_{\epsilon}} (|z_1|+|z_2|)^2 \right] \right\}
    \\
    &= \frac{(2 D)^{-1/3}}{c_{\epsilon}} \exp \left[\frac{(2 D)^{-1/3}}{2 c_{\epsilon}} \right].
\end{split}
\end{equation}
Combining the above equations gives that there exists positive $D_{1,\epsilon}, C_{1,\epsilon}$ such that if $D > D_{1,\epsilon}$,
\begin{equation} \label{eq:1minusg_estimate_2}
    \exp[\pm(h_1(z_1) + h_2(z_2))] 
    = 
    1 + \frac{C_{1,\pm}(z_1,z_2)}{D}, \quad \text{ where } |C_{1,\pm}(z_1,z_2)| < C_{1,\epsilon}.
\end{equation}
for $|z_1|,|z_2|  < D^{-2/3}$.
Similarly, we can say that 
\begin{equation*}
    \frac{\exp[-h_1(z_1)-h_2(z_2)] - e^{D (z_1 - z_2)}}{z_1 - z_2} 
    = 
    \frac{1 - e^{D (z_1 - z_2)}}{z_1 - z_2}
    \left( 1 - \frac{1 - \exp[-h_1(z_1)-h_2(z_2)]}{1 - e^{D(z_1 - z_2)}} \right).
\end{equation*}
So, using the estimate Eq.~\eqref{eq:abs_hPlush_estimate} and considering the range of arguments of $z_1,z_2$, we can say that there exist positive constants $D_{2,\epsilon}, C_{2,\epsilon}$ such that if $D > D_{2,\epsilon}$ then
\begin{equation} \label{eq:1minusg_estimate_3}
    \frac{\exp[-h_1(z_1)-h_2(z_2)] - e^{D (z_1 - z_2)}}{z_1 - z_2} 
    =
    \frac{1 - e^{D (z_1 - z_2)}}{z_1 - z_2}
    \left( 1 + \frac{C_2(z_1,z_2)}{D} \right), 
    \quad \text{ where } |C_2(z_1,z_2)| < C_{2,\epsilon} .
\end{equation}
From here, the proposition follows by combining Eqs.~(\ref{eq:1minusg_estimate_1},\ref{eq:1minusg_estimate_2},\ref{eq:1minusg_estimate_3}) choosing $C_{\epsilon} = C_{1,\epsilon} \cdot C_{2,\epsilon}$ and $D_{\epsilon} = \mathrm{max}(10,D_{1,\epsilon},D_{2,\epsilon})$.
\end{proof}

\begin{proof}[Proof of Proposition~\ref{prop:gminusg_estimate}]
This estimate is proved quite similarly to the Proposition~\ref{prop:1minusg_estimate}. Again, we'll just show this in the range $0 < |z_1| , |z_2| < D^{-2/3}$, and the rest can be proved similarly. First write 
\begin{equation*}
    \frac{g(z_1) - g(z_2)}{z_1 - z_2} = e^{h_1(z_1)} \frac{\exp[-D z_1] - \exp[-D z_2 - h_1(z_1) + h_1(z_2)]}{z_1 - z_2}.
\end{equation*}
Similar estimates for $e^{h_1(z_1)}$ and $\exp[- h_1(z_1) + h_1(z_2)]$ to the previous section show that there's constants $C_{3,\epsilon},D_{4,\epsilon}$ so that for $D  > D_{3,\epsilon}$ large enough
\begin{equation*}
    e^{h_1(z_1)} = \left(1 + \frac{C_{3}(z_1)}{D}\right), \quad \text{ where } |C_{3}(z_1)| < C_{3,\epsilon} 
\end{equation*}
and for $D  > D_{4,\epsilon}$,
\begin{equation*}
    \frac{\exp[-D z_1] - \exp[-D z_2 - h_1(z_1) + h_1(z_2)]}{z_1 - z_2} = \frac{\exp[-D z_1] - \exp[-D z_2]}{z_1 - z_2} \left( 1 + \frac{C_{3}(z_1,z_2)}{D} \right), \quad \text{ where } |C_{4}(z_1,z_2)| < C_{4,\epsilon}.
\end{equation*}
Then, we can choose $\tilde{C}_{\epsilon} =  C_{3,\epsilon} \cdot C_{4,\epsilon}$ and $\tilde{D}_{\epsilon} = \max(D_{3,\epsilon},D_{4,\epsilon})$.
\end{proof}

\section{Spectral Theory of Finite Hilbert Transformations} \label{app:finiteHilbTransf}

Consider a finite interval $(a,b) \subset \R$ and the operator 
$T_{(a,b)}: L^2(a,b) \to L^2(a,b)$ defined as 
\begin{equation}
   (T_{(a,b)} f) (w) := \frac{1}{i \pi} \mathrm{p.v.} 
   \int_{a}^{b} \frac{f(w')}{w'-w} dw'
\end{equation}
where $\mathrm{p.v.}$ refers to the Cauchy principal value of the integral. Such a $T_{(a,b)}$ is called a `Finite Hilbert Transformation', since it is a restriction of the Hilbert transformation on $L^2(\R) \to L^2(\R)$,
\begin{equation}
    (T f) (w) := \frac{1}{i \pi} \mathrm{p.v.} 
   \int_{-\infty}^{\infty} \frac{f(w')}{w'-w} dw'.
\end{equation}
Note $T$ has a spectrum of ${\pm 1}$, where the $+1$ eigenspace is generated by plane waves $\{ e^{i k w} | k>0 \}$ and the $-1$ eigenspace is generated by plane waves $\{ e^{i k w} | k<0 \}$. As such, the spectrum of $T_{(a,b)}$ is contained entirely in the range $[-1,1]$.

In~\cite{Koppelman1959}, it is found that the functions 
\begin{equation} \label{eq:finiteHilb_eigenfucns}
    f_{\xi}(w) := \frac{\exp\left[ \frac{1}{2 \pi i} \left(\log \frac{1+\xi}{1-\xi}\right) \left(\log \frac{b-w}{w - a}\right) \right]}{ \sqrt{ (b - w) (w - a) } }
\end{equation}
are eigenfunctions that satisfying $(T f_{\xi})(z) = \xi f_{\xi}(z)$. 

This was solved for in~\cite{Koppelman1959} by supposing $(T f_{\xi})(z) = \xi f_{\xi}(z)$ and using the ansatz that one can continue $f_{\xi}$ analytically, to consider the functions $f^{\star}_{\xi}(w) := \frac{1}{2 \pi i} \int_{a}^{b} \frac{f_{\xi}(w')}{w'-w} dw'$ (defined holomorphically for $w \in \C$ \textit{without} using a principal value). In particular, one would have the relations,
\begin{equation} \label{eq:barrierEqs_0}
\begin{split}
    f^{\star}_{\xi}(w + i 0^+) - f^{\star}_{\xi}(w - i 0^+) 
    &= 
    f_{\xi}(w), \quad \text{for } w \in (a,b) 
    \\
    f^{\star}_{\xi}(w + i 0^+) - f^{\star}_{\xi}(w - i 0^+) 
    &=
    0, \quad \text{for } w \in \R - (a,b) 
    \\
    f^{\star}_{\xi}(w + i 0^+) + f^{\star}_{\xi}(w - i 0^+) 
    &= \frac{1}{\pi i} \mathrm{p.v.} \int_{a}^{b} \frac{f_{\xi}(w')}{w'-w} dw' 
    = (T_{(a,b)} f_{\xi})(w) = \xi f_{\xi}(w), \quad \text{for } w \in (a,b) 
    \\
    f^{\star}_{\xi}(\infty) &= 0.
\end{split}
\end{equation}
The first two equalities follows from the residue theorem showing a branch cut along $(a,b)$, and the third one follows from the ansatz of $f_{\xi}(w)$ being an eigenfunction. The above equalities combine to give the `barrier equations'
\begin{equation} 
\begin{split}
    (\xi-1) f^{\star}_{\xi}(w + i 0^+) - (\xi+1) f^{\star}_{\xi}(w - i 0^+) &= 0, \quad \text{for } w \in (a,b) \\
            f^{\star}_{\xi}(w + i 0^+) -         f^{\star}_{\xi}(w - i 0^+) &= 0, \quad \text{for } w \in (a,b).
\end{split}
\end{equation}
In particular, one can guess that for $\xi \in (-1,1)$
\begin{equation}
    h_{\xi}^{\pm}(w) = 
    \exp\left[ 
    \frac{1}{2 \pi i} \left( \log \frac{1+\xi}{1-\xi} \pm i\pi \right) 
    \int_{a}^{b} \frac{1}{w'-w} dw'
    \right] 
\end{equation}
are also solutions that have a branching structure that solve the barrier equation. However, these functions satisfy $h_{\xi}^{\pm}(\infty) = 1$. As such, any multiple of $f^{\star}_{\xi}(w) = h_{\xi}^{+}(w) - h_{\xi}^{-}(w)$ is a solution that vanishes at $\infty$. From here, one can use Eq.~\eqref{eq:barrierEqs_0} to find that the $f_{\xi}(w)$ defined in Eq.~\eqref{eq:finiteHilb_eigenfucns} work. 

There is also a direct derivation as follows. For now, suppose $-1 < \xi < 0$ so that $\log \frac{1+\xi}{1-\xi} > 0$.
\begin{equation} \label{eq:alternateDerivation_finiteHilbertEigenfunctions}
\begin{split}
    (T f_{\xi}) (w) 
    &= 
    \frac{1}{i \pi} \mathrm{p.v.} \int_{a}^{b} 
    \frac{\exp\left[ \frac{1}{2 \pi i} \left(\log \frac{1+\xi}{1-\xi}\right) \left(\log \frac{b-w'}{w'-a}\right) \right]}
    { \sqrt{ (b - w') (w' - a) } } 
    \frac{1}{w'-w} dw' 
    \\
    &=
    \mathrm{p.v.} \int_{-\infty}^{\infty} \frac{1}{2 \pi i} 
    \frac{\exp\left[ \frac{1}{2 \pi i} \left(\log \frac{1+\xi}{1-\xi}\right) W' \right]}{\cosh(\frac{W'}{2})} 
    \frac{1}{\left( \frac{b-a}{e^{W'}+1} + a \right) -w} 
    dW',
\end{split}
\end{equation}
where the second equality follows from the change of variables $W' = \log \frac{b-w'}{w'-a}$. At this point, note that the integrand has poles at $W' \in 2 \pi i \Z  + \log \frac{b-w}{w-a}$, and has a residue of 
$\frac{1}{2 \pi i}\frac{2 (-1)^{k+1}}{\sqrt{(b-w)(w-a)}} \left( \frac{1+\xi}{1-\xi} \right)^k \exp\left[ \frac{1}{2 \pi i} \left(\log \frac{1+\xi}{1-\xi}\right) \left(\log \frac{b-w}{w-a}\right) \right]  $
at each pole $W' = 2 \pi i  k + \log \frac{b-w}{w-a}$. We can close the contour in the upper-half plane, since $\log \frac{1+\xi}{1-\xi} > 0$. Since we are considering a principal value integral, the full integral is the sum over $2 \pi i$ times the poles with $k > 0$ and $\pi i$ times the $k=0$ pole, since the $k=0$ pole lies on the real axis. This sum over poles gives a geometric series that ends up equalling $\xi f_{\xi}(w)$. Essentially the same argument for $0 < \xi < 1$ works, except one needs to close the contour in the lower-half plane.

Upon proper normalization of the functions $f_{\xi}(w)$, these functions form an orthonormal basis and there's a Fourier-inversion formula using these eigenfunctions. In particular, the formulas
\begin{equation} \label{eq:inversionFormulas_finiteHilbert}
\begin{split}
    f(w) &= \frac{1}{\pi} \sqrt{\frac{b-a}{2}} \int_{-1}^{1} \frac{\exp\left[ \frac{1}{2 \pi i} \left(\log \frac{1+\xi}{1-\xi}\right) \left(\log \frac{b-w}{w-a}\right) \right]}{ \sqrt{ (1 - \xi^2) (b - w) (w - a) } } g(\xi) d\xi \\
    g(\xi) &= \frac{1}{\pi} \sqrt{\frac{b-a}{2}} \int_{a}^{b} \frac{\exp\left[-\frac{1}{2 \pi i} \left(\log \frac{1+\xi}{1-\xi}\right) \left(\log \frac{b-w}{w-a}\right) \right]}{ \sqrt{ (1 - \xi^2) (b - w) (w - a) } } f(w) dw
\end{split}
\end{equation}
give equations for the spectral transform of functions in the `real space' $L^2(a,b)$ to functions in the `spectral space' $L^2(-1,1)$.

A derivation of the above follows from the standard Fourier Inversion formula
\begin{equation}
    F(x) = \frac{1}{\sqrt{2\pi}} \int_{-\infty}^{\infty} e^{i k x}  G(k) dk \quad,\quad G(k) = \frac{1}{\sqrt{2\pi}} \int_{-\infty}^{\infty} e^{-i k x} F(x) dx
\end{equation}
upon the changes of variables and reassignments
\begin{equation*}
\begin{split}
    k = \frac{1}{\sqrt{2 \pi}} \log \frac{1+\xi}{1-\xi}
        \quad&,\quad 
    x = -\frac{1}{\sqrt{2 \pi}} \log \frac{b-w}{w-a}
    \\
    g(\xi) = \left( \frac{2}{\pi} \right)^{1/4} \frac{1}{\sqrt{1-\xi^2}} G(u(\xi))
        \quad&,\quad
    f(w) = \left( \frac{2}{\pi} \right)^{1/4} \sqrt{\frac{b-a}{(b-w)(w-a)}} F(w(x)).
\end{split}
\end{equation*}
Note that the above assignments of $f,g$ with respect to $F,G$ are consistent with $\int_{-\infty}^{\infty} |G(k)|^2 dk = \int_{-1}^{1} |g(\xi)|^2 d\xi$ and  $\int_{-\infty}^{\infty} |F(x)|^2 dx = \int_{a}^{b} |f(w)|^2 dw$.

In this paper, we considered $a=0,b=1$ and the operators
\begin{equation}
    (T^{\pm}f)(w) := \frac{1}{i \pi} \int_{0 \xrightarrow{\pm} 1} \frac{f(w')}{w'-w} dw'
\end{equation}
defined on functions on a neighborhood of $(0,1)$. The fact that all eigenfunctions $f_{\xi}(w)$ are holomorphic in $w$ around $(0,1)$ mean that
\begin{equation}
    (T^{\pm}f_{\xi})(w) = (\xi \pm 1) f_{\xi}(w)
\end{equation}
by the residue theorem. 

Similarly, we eventually find it useful to express $\frac{1}{w-w'}$ for $w' \in (0,1) \pm i 0^+$ that lie slightly above or below the real-axis around $(0,1)$. The inversion formulas in Eq.~\eqref{eq:inversionFormulas_finiteHilbert} give
\begin{equation} \label{eq:spectralRep_1_over_w}
\begin{split}
    \frac{1}{w-w'} &= \frac{i}{2 \pi} \int_{-1}^{1} \frac{-\xi \pm 1}{1 - \xi^2} f_{-\xi}(w') f_{\xi}(w) d\xi, \quad \text{ for } w' \in (0,1) \pm i 0^+.
\end{split}
\end{equation}
This follows because an expression for the `$g(\xi)$' of Eq.~\eqref{eq:inversionFormulas_finiteHilbert} ends up returning $(-\xi \pm 1)\frac{f_{-\xi}(w')}{\pi \sqrt{2} \sqrt{1-\xi^2}}$, which  can be derived similarly to the discussion around Eq.~\eqref{eq:alternateDerivation_finiteHilbertEigenfunctions}, where the $-\xi \pm 1$ factor follows because the above integral is not a principal value integral and integrating around the pole along the real axis modifies the sum over poles.

\section{Summing over the leading-order contributions, $-\sum_{n=1}^{\infty} \frac{1}{2n} (2 \sin(\alpha/2))^2 \mathcal{J}_{2n}$ for a single zipper} 
\label{app:sum_leading_order_I2n}

In Eq.~\eqref{eq:asymptIntegral_singleZip}, we defined the leading-order contribution to $\mathcal{I}_{2n}$ as
\begin{equation}
    \mathcal{J}_{2n} := \frac{2}{(2\pi)^{2n}} \int_{0 \xrightarrow{-} 1} dw_1 \int_{0 \xrightarrow{+} 1} dw_2 \cdots \int_{0 \xrightarrow{-} 1} dw_{2n-1} \int_{\mathcal{C}_\text{out}} dw_{2n} \frac{1}{(w_1 - w_2) \cdots (w_{2n} - w_1)} \bigg(1 - e^{-D|w_{2n} - w_{1}|}\bigg),
\end{equation}
so that Proposition~\ref{prop:asymptIntegral_singleZip} says
$\mathcal{I}_{2n} \cdot (1+O(1/D)) + O(1/D) = \mathcal{J}_{2n}$.

We will instead consider the similar sum
$-\sum_{n=1}^{\infty} \frac{1}{2n} (2 \sin(\alpha/2))^{2n} \mathcal{J}_{2n}$
which will sum over all the universal leading-order contributions. This sum can be tackled by the spectral theory of the operators $T^{\pm}$
\begin{equation*}
    (T^{\pm} f)(w) := \frac{1}{i \pi} \int_{0 \xrightarrow{\pm} 1} \frac{f(w')}{w'-w} dw' 
\end{equation*}
defined on functions in a neighborhood of $(0, 1)$. These operators' spectral theory is closely related to that of the `Finite Hilbert Transformation'. This operation is reviewed in Appendix~\ref{app:finiteHilbTransf}.

The relevant points are that these operators are diagonalized by certain functions $\{f_{\xi}(w)\}_{\xi \in (-1,1)}$ (see Eq.~\eqref{eq:finiteHilb_eigenfucns}), so that
\begin{equation*}
    (T^{\pm} f)(w) = (\xi \pm 1) f_{\xi}(w).
\end{equation*}
After an additional normalization, the $f_{\xi}(w)$ can be thought of as orthonormal, and there is an analog of the Fourier-inversion formula (see Eq.~\eqref{eq:inversionFormulas_finiteHilbert}) using these functions that is an isometry between the `real space' $L^2(0,1)$ and the `spectral space' $L^2(-1,1)$ parameterized by the $\xi$ variable. As in Eq.~\eqref{eq:spectralRep_1_over_w}, the inversion formula allows us to write
\begin{equation}
\begin{split}
    \frac{1}{w_1-w_2} &= \frac{i}{2 \pi} \int_{-1}^{1} \frac{-\xi + 1}{1 - \xi^2} f_{-\xi}(w_2) f_{\xi}(w_1) d\xi,
\end{split}
\end{equation}
since $w_2-w_1$ lies slightly above the real axis. At this point, we can write 
\begin{equation*}
\begin{split}
    \mathcal{J}_{2n} :&= 2 \frac{(i \pi)^{2n-2}}{(2\pi)^{2n}} \int_{\mathcal{C}_\text{out}} dw_{2n} \int_{0 \xrightarrow{-} 1} dw_1  \int_{-1}^{1} d\xi \, f_{\xi}(w_1) \frac{1 - e^{-D|w_{2n} - w_{1}|}}{w_{2n}-w_1} \frac{i}{2 \pi} \frac{-\xi + 1}{1 - \xi^2}
    \\ 
    &\quad\quad\quad
    \int_{0 \xrightarrow{-} 1} \frac{dw_{2n-1}}{i \pi} \frac{1}{w_{2n-1}-w_{2n}} \cdots \int_{0 \xrightarrow{-} 1} \frac{dw_3}{i \pi} \frac{1}{w_3-w_4} \int_{0 \xrightarrow{+} 1} \frac{dw_2}{i \pi} \frac{1}{w_2-w_3} f_{-\xi}(w_2)
    \\
    &= 2 \frac{(i \pi)^{2n-2}}{(2\pi)^{2n}} \int_{\mathcal{C}_\text{out}} dw_{2n} \int_{0 \xrightarrow{-} 1} dw_1  \int_{-1}^{1} d\xi \, f_{\xi}(w_1) \frac{1 - e^{-D|w_{2n} - w_{1}|}}{w_{2n}-w_1} \frac{i}{2 \pi} \frac{-\xi + 1}{1 - \xi^2}
    \\
    &\quad\quad\quad
    \int_{0 \xrightarrow{-} 1} \frac{dw_{2n-1}}{i \pi} \frac{1}{w_{2n-1}-w_{2n}} \cdots \int_{0 \xrightarrow{-} 1} \frac{dw_3}{i \pi} \frac{-\xi + 1}{w_3-w_4} f_{-\xi}(w_3)
    \\
    &= \cdots
    \\
    &= 2 \frac{(i \pi)^{2n-2}}{(2\pi)^{2n}} \int_{\mathcal{C}_\text{out}} dw_{2n} \int_{0 \xrightarrow{-} 1} dw_1  \int_{-1}^{1} d\xi \, f_{\xi}(w_1) \frac{1 - e^{-D|w_{2n} - w_{1}|}}{w_{2n}-w_1} \frac{i}{2 \pi} \frac{-\xi + 1}{1 - \xi^2} (-\xi-1)^{n-2} (-\xi+1)^{n-1}
    \\
    &\quad\quad\quad
    \int_{0 \xrightarrow{-} 1} \frac{dw_{2n-1}}{i \pi} \frac{1}{w_{2n-1}-w_{2n}} f_{-\xi}(w_{2n-1}).
\end{split}
\end{equation*}
At each step above, we used the fact that $(T^{\pm}f_{-\xi})(w) = (-\xi \pm 1) f_{-\xi}(w)$. For $j$ even, the $(n-1)$ integrals so far over $w_j$ correspond to $T^{+}$. For $j$ odd, the $(n-2)$ integrals so far over $w_j$ correspond to $T^{-}$. 

We want to do the same for the integral over $w_{2n-1}$. However, at this point we have an issue that $w_{2n}$ is not in $(0,1)$. One way to evaluate the next integral is to analytically continue the solution by evaluating the integral for $w'_{2n} \in (0,1) + i0^+$ and considering the analytic continuation of that solution to $w_{2n}$ while avoiding the branch cut at $(0,1)$. Substituting in the expression for $f_{\xi}(w_{2n})$ as in Eq.~\eqref{eq:finiteHilb_eigenfucns} and doing this integral gives 
\begin{equation*}
\begin{split}
    \mathcal{J}_{2n} &= 2 \frac{(i \pi)^{2n-2}}{(2\pi)^{2n}} \int_{\mathcal{C}_\text{out}} dw_{2n} \int_{0 \xrightarrow{-} 1} dw_1  \int_{-1}^{1} d\xi \, f_{\xi}(w_1) \frac{1 - e^{-D|w_{2n} - w_{1}|}}{w_{2n}-w_1} \frac{i}{2 \pi} \frac{-\xi + 1}{1 - \xi^2} (-\xi-1)^{n-1} (-\xi+1)^{n-1}
    \\
    &\quad\quad\quad\quad\quad\quad\quad\quad\quad \times
        \frac{\exp\left[ \frac{1}{2 \pi i} \left(\log \frac{1-\xi}{1+\xi}\right) \left(\log \frac{w_{2n}-1}{w_{2n}}\right) \right]}{ \sqrt{ w_{2n} (w_{2n} - 1) } } \cdot \frac{\mathrm{sgn}(w_{2n})}{i} \sqrt{\frac{1+\xi}{1-\xi}}
    \\
    &= \frac{1}{\pi^3} 
    \int_{\mathcal{C}_\text{out}} dw_{2n} \int_{0}^{1} dw_1  \int_{-1}^{1} d\xi \, \frac{\exp\left[ \frac{1}{2 \pi i} \left(\log \frac{1+\xi}{1-\xi}\right) \left(\log \frac{1-w_{1}}{w_{1}} \frac{w_{2n}}{w_{2n}-1}\right) \right]}{ \sqrt{w_{1} (1-w_{1})} \sqrt{w_{2n} (w_{2n}-1)} }
    \frac{1 - e^{-D|w_{2n} - w_{1}|}}{w_{2n}-w_1}
    \\
    &\quad\quad\quad\quad\quad\quad\quad\quad\quad \times
    \left(\frac{1-\xi^2}{4}\right)^{n} 
    \frac{1}{(1-\xi^2)^{3/2}} \mathrm{sgn}(w_{2n})
\end{split}
\end{equation*}
Above, $\mathrm{sgn}(w_{2n})$ is the sign function, equalling $-1$ on the negative branch and $+1$ on the positive branch.
Now, we can set $\xi = \tanh(x/2)$ to get
\begin{equation*}
\begin{split}
    \mathcal{J}_{2n} 
    = 
    \frac{1}{2 \pi^3}
    \int_{\mathcal{C}_\text{out}} dw_{2n} \int_{0}^{1} dw_1  \int_{-\infty}^{\infty} dx \, \frac{\exp\left[ \frac{1}{2 \pi i} x \left(\log \frac{1-w_{1}}{w_{1}} \frac{w_{2n}}{w_{2n}-1}\right) \right] \mathrm{sgn}(w_{2n})}{ \sqrt{w_{1} (1-w_{1})} \sqrt{w_{2n} (w_{2n}-1)} }
    \frac{1 - e^{-D|w_{2n} - w_{1}|}}{w_{2n}-w_1} 
    \left( \frac{1}{(2\cosh(x/2))^2} \right)^{n} \cosh(x/2) 
\end{split}
\end{equation*}
Note that $\frac{1}{(2 \cosh(x/2))^2} < 1/4$ for $x \in \R - \{0\}$. This means that the series Eq.~\eqref{eq:log_expansion_A} converges so that 
\begin{equation}
    -\sum_{n=1}^{\infty} \frac{1}{2n} (2 \sin(\alpha/2))^{2n} \left(\frac{1}{(2 \cosh(x/2))^2}\right)^n 
    = 
    \frac{1}{2}
    \log \left[ 1 - \frac{\sin(\alpha/2)^2}{\cosh(x/2)^2} \right]
    = 
    \frac{1}{2}
    \log \left[ 
        \frac{\cos(\alpha)+\cosh(x)}
             {1+\cosh(x)}
    \right]
\end{equation}
where the logarithm's branch is well-defined and satisfies 
$
\log \left[ \frac{\cos(\alpha)+\cosh(x)}{1+\cosh(x)} \right] \xrightarrow[]{x \to \pm \infty}
0
$
.

From here, we can derive that the full sum over the integrals to be
\begin{equation}
\begin{split}
    &-\sum_{n=1}^{\infty} \frac{1}{2n}
    (2 \sin(\alpha/2))^2
    \mathcal{J}_{2n} \\
    &= 
    \frac{1}{4 \pi^3}
    \int_{\mathcal{C}_\text{out}} dw_{2n} \int_{0}^{1} dw_1  \int_{-\infty}^{\infty} dx \, \frac{\exp\left[ \frac{1}{2 \pi i} x \left(\log \frac{1-w_{1}}{w_{1}} \frac{w_{2n}}{w_{2n}-1}\right) \right]}{ \sqrt{w_{1} (1-w_{1})} \sqrt{w_{2n} (w_{2n}-1)} }
    \frac{1 - e^{-D|w_{2n} - w_{1}|}}{w_{2n}-w_1} \cosh\left(\frac{x}{2}\right)
    \log \left[ \frac{\cos(\alpha)+\cosh(x)}{1+\cosh(x)} \right].
\end{split}
\end{equation}
Note that the analytic continuation of the logarithm for complex $x$ can be chosen to have branch cuts along along all of the segments $x \in i (\pi - \alpha ) \cup i (\pi + \alpha) + 2 \pi i \Z$, each of which map the arguments
$\frac{\cos(\alpha)+\cosh(x)}{1+\cosh(x)}$ to paths $0 \to \infty \cup \infty \to 0$.
One can check that the branch cut is of the form
\begin{equation*}
\begin{split}
    \log \left[ \frac{\cos(\alpha)+\cosh(x+i0^+)}{1+\cosh(x+i0^+)} \right]
    -
    \log \left[ \frac{\cos(\alpha)+\cosh(x+i0^-)}{1+\cosh(x+i0^-)} \right]
    =
    \begin{cases}
        +2 \pi i
        , \quad &\text{ for } 
        x_0 \in (\pi - \alpha , \pi) + 2 \pi \Z 
        \\
        -2 \pi i
        , \quad &\text{ for } 
        x_0 \in (\pi , \pi + \alpha ) + 2 \pi \Z 
        \\
        0
        , \quad &\text{ for } 
        x_0 \in \R - \big( ( \pi-\alpha , \pi+\alpha ) + 2 \pi \Z \big)
    \end{cases}
\end{split}
\end{equation*}

At this point, we can perform the $x$ integral by closing the contour into the upper-half or lower-half plane, depending on if $\Lambda := \log \frac{1-w_{1}}{w_{1}} \frac{w_{2n}}{w_{2n}-1}$ is negative or positive. Let's first consider the case of $\Lambda < 0$, so we close into the upper-half plane. The contour can be deformed to surround the branch cuts. In particular, we end up getting
\begin{equation*}
\begin{split}
    &\int_{-\infty}^{\infty} dx \, 
    \exp\left[ \frac{\Lambda}{2 \pi i} \cdot x  \right]
    \cosh\left(\frac{x}{2}\right)
    \log \left[ \frac{\cos(\alpha)+\cosh(x)}{1+\cosh(x)} \right] 
    \\
    &= 
    + 2 \pi i \sum_{m = 0}^{\infty}
    \int_{i(\pi - \alpha + 2 \pi m)}^{i (\pi + 2 \pi m)}  dx 
    \exp\left[ \frac{\Lambda}{2 \pi i} \cdot x  \right]
    \cosh\left(\frac{x}{2}\right)
    - 2 \pi i \sum_{m = 0}^{\infty}
    \int_{i (\pi + 2 \pi m)}^{i (\pi + \alpha + 2 \pi m)}  dx 
    \exp\left[ \frac{\Lambda}{2 \pi i} \cdot x  \right]
    \cosh\left(\frac{x}{2}\right)
    \\
    &=
    -2\pi 
    \left(\sum_{m = 0}^{\infty} (-1)^m e^{m \Lambda} \right)
    \int_{\pi - \alpha}^{\pi}  dx \exp\left[ \frac{\Lambda}{2 \pi} \cdot x  \right]
    \cos\left(\frac{x}{2}\right) 
    +2\pi 
    \left(\sum_{m = 0}^{\infty} (-1)^m e^{m \Lambda} \right)
    \int_{\pi}^{\pi + \alpha}  dx \exp\left[ \frac{\Lambda}{2 \pi} \cdot x  \right]
    \cos\left(\frac{x}{2}\right) 
    \\
    &=
    -2\pi \frac{1}{1 + e^{\Lambda} } 
        \int_{0}^{\alpha} d\alpha' 
    \exp\left[ \frac{\Lambda}{2 \pi} \cdot (\pi - \alpha') \right] \cos\left(\frac{\pi - \alpha'}{2}\right) 
    +2\pi \frac{1}{1 + e^{\Lambda} } 
    \int_{-\alpha}^{0} d\alpha' \exp\left[ \frac{\Lambda}{2 \pi} \cdot (\pi - \alpha') \right] \cos\left(\frac{\pi - \alpha'}{2}\right) 
    \\
    &=
    -\pi \frac{1}{\cosh\left(\frac{\Lambda}{2}\right)}    
    \left\{
        \int_{0}^{\alpha} d\alpha' 
        -
        \int_{-\alpha}^{0} d\alpha' 
    \right\} 
    \sin\left(\frac{\alpha'}{2}\right) 
    \exp\left(-\frac{\alpha' \Lambda}{2 \pi}\right)
\end{split}
\end{equation*}
The analogous calculation for $\Lambda<0$ gives the same answer.
As such, we can write 
\begin{equation} \label{eq:leading_exact_plus_error}
\begin{split}
    &-\sum_{n=1}^{\infty} \frac{1}{2n}
    (2 \sin(\alpha/2))^2
    \mathcal{J}_{2n}
    \\
    &= 
    -\frac{1}{2 \pi^2} 
    \left\{
        \int_{0}^{\alpha} d\alpha' 
        -
        \int_{-\alpha}^{0} d\alpha' 
    \right\} 
    \sin\left(\frac{\alpha'}{2}\right)
    \int_{\mathcal{C}_\text{out}} dw_{2n} 
    \int_{0}^{1} dw_1  \, 
    \frac{\cosh\left(-\frac{\alpha' \Lambda}{2 \pi}\right)/\cosh\left(\frac{\Lambda}{2}\right)}
    { \sqrt{w_{1} (1-w_{1})} \sqrt{w_{2n} (w_{2n}-1)} }
    \frac{1 - e^{-D|w_{2n} - w_{1}|}}{w_{2n}-w_1}
    \\
    &= 
    -\frac{1}{\pi^2} 
    \left\{
        \int_{0}^{\alpha} d\alpha' 
        -
        \int_{-\alpha}^{0} d\alpha' 
    \right\} 
    \sin\left(\frac{\alpha'}{2}\right)
    \int^{1}_{\infty} dw_{2n} \int_{0}^{1} dw_1  \, 
    \left( \frac{1-w_{1}}{w_{1}} \frac{w_{2n}}{w_{2n}-1} \right)^{-\alpha' / 2 \pi} 
    \frac{1 - e^{-D(w_{2n} - w_{1})}}{(w_{2n}-w_1)^2}
    \\
    &= 
    -\frac{1}{\pi^2} 
    \left\{
        \int_{0}^{\alpha} d\alpha' 
        -
        \int_{-\alpha}^{0} d\alpha' 
    \right\} 
    \sin\left(\frac{\alpha'}{2}\right)
    \int_{0}^{1} d\Tilde{w}_{2n} \int_{0}^{\infty} d\Tilde{w}_1  \, 
    \left( \frac{\Tilde{w}_{1}}{\Tilde{w}_{2n}} \right)^{-\alpha' / 2 \pi} 
    \frac{1 - e^{-D(\frac{1}{1-\Tilde{w}_{2n}} - \frac{1}{1+\Tilde{w}_1})}}{(\Tilde{w}_{2n} + \Tilde{w}_1)^2} 
    \\
    &= 
    -\frac{1}{\pi^2} 
    \left\{
        \int_{0}^{\alpha} d\alpha' 
        -
        \int_{-\alpha}^{0} d\alpha' 
    \right\} 
    \sin\left(\frac{\alpha'}{2}\right)
    \int_{0}^{1} d\Tilde{w}_{2n} \int_{0}^{\infty} d\Tilde{w}_1  \, 
    \frac{ \left( \frac{\Tilde{w}_{1}}{\Tilde{w}_{2n}} \right)^{-\alpha' / 2\pi} }{(\Tilde{w}_{2n} + \Tilde{w}_1)^2} 
    \left\{
        \underbrace{
            \Big( 1 - e^{-D(\Tilde{w}_{2n} + \Tilde{w}_{1})} \Big)
        }_{\text{`Piece 1'}}
        -
        \underbrace{
            \Big( 
                e^{-D(\frac{1}{1-\Tilde{w}_{2n}} - \frac{1}{1+\Tilde{w}_1})} 
                - 
                e^{-D(\Tilde{w}_{2n} + \Tilde{w}_{1})} 
            \Big)
        }_{\text{`Piece 2'}}
    \right\}
\end{split}
\end{equation}
Above, the first equality was substituting in the result of the previous algebra. The second equality has two ingredients. One was the fact that $\frac{1}{\cosh\left(\frac{\Lambda}{2}\right)} \frac{1}{\sqrt{w_{1} (1-w_{1})} \sqrt{w_{2n} (w_{2n}-1)}} = \frac{2}{|w_{2n}-w_{1}|}$ over our integration range. The other is that the integral over $w_{2n}: \infty \to 1$ equals the integral over $w_{2n}: 0 \to -\infty$ upon the substitutions $\{w_1 \to 1-w_1 , w_{2n} \to 1-w_{2n}\}$. The third equality follows from substituting $\Tilde{w}_1 = \frac{1-w_1}{w_1}$, $\Tilde{w}_{2n} = \frac{w_{2n}-1}{w_{2n}}$. In the fourth equality, we split up the integrand and split it up into two pieces and denote their respective integrals as `Piece 1' and `Piece 2'.

We will argue now that the integral over `Piece 2' is negligible. First, we use the fact $\frac{1}{1-\tilde{w}_{2n}}-\frac{1}{1+\tilde{w}_{1}} = \frac{\tilde{w}_{2n}}{1-\tilde{w}_{2n}}+\frac{\tilde{w}_{1}}{1+\tilde{w}_{1}}$. 
Then, we note that for $D$ large enough, if $\tilde{w_1} > D^{-2/3}$ {\color{RoyalPurple}[resp. $\tilde{w_{2n}} > D^{-2/3}$]}, 
we have $\frac{\tilde{w}_{1}}{1+\tilde{w}_{1}} > 0.99 D^{-2/3}$ {\color{RoyalPurple}[resp. $\frac{\tilde{w}_{2n}}{1-\tilde{w}_{2n}} > D^{-2/3}$]}. 
In these regions, we would have $e^{-D(\frac{1}{1-\Tilde{w}_{2n}} - \frac{1}{1+\Tilde{w}_1})} < e^{-D^{1/3}}$ and $e^{-D(\Tilde{w}_{2n} + \Tilde{w}_{1})}  < e^{-0.99D^{1/3}}$
{\color{RoyalPurple}[resp. $e^{-D(\frac{1}{1-\Tilde{w}_{2n}} - \frac{1}{1+\Tilde{w}_1})} < e^{-D^{1/3}}$ and $e^{-D(\Tilde{w}_{2n} + \Tilde{w}_{1})} < e^{-D^{1/3}}$]}.
\begin{equation} \label{eq:leading_error_2}
\begin{split}
    |\text{`Piece 2'}|
    &=
    \left|
    \frac{1}{\pi^2}     
    \left\{
        \int_{0}^{\alpha} d\alpha' 
        -
        \int_{-\alpha}^{0} d\alpha' 
    \right\} 
    \sin\left(\frac{\alpha'}{2}\right)
    \int_{0}^{1} d\Tilde{w}_{2n} 
    \int_{0}^{\infty} d\Tilde{w}_1  
    \, 
    \frac{ \left( \frac{\Tilde{w}_{1}}{\Tilde{w}_{2n}} \right)^{-\alpha' / 2\pi} }{(\Tilde{w}_{2n} + \Tilde{w}_1)^2} 
    \Big( 
        e^{-D(\frac{1}{1-\Tilde{w}_{2n}} - \frac{1}{1+\Tilde{w}_1})} 
        - 
        e^{-D(\Tilde{w}_{2n} + \Tilde{w}_{1})} 
    \Big)
    \right|
    \\
    &<\quad
    2 e^{-D^{1/3}}
    \frac{1}{\pi^2}    
    \left\{
        \int_{0}^{\alpha} d\alpha' 
        -
        \int_{-\alpha}^{0} d\alpha' 
    \right\} 
    \sin\left(\frac{\alpha'}{2}\right)
    \int_{D^{-2/3}}^{1} d\Tilde{w}_{2n} 
    \int_{0}^{\infty} d\Tilde{w}_1  
    \, 
    \frac{ \left( \frac{\Tilde{w}_{1}}{\Tilde{w}_{2n}} \right)^{-\alpha' / 2\pi} }{(\Tilde{w}_{2n} + \Tilde{w}_1)^2} 
    \\
    &\quad+
    2 e^{-0.99 D^{1/3}}
    \frac{1}{\pi^2}    
    \left\{
        \int_{0}^{\alpha} d\alpha' 
        -
        \int_{-\alpha}^{0} d\alpha' 
    \right\} 
    \sin\left(\frac{\alpha'}{2}\right)
    \int_{D^{-2/3}}^{\infty} d\Tilde{w}_1 
    \int_{0}^{1} d\Tilde{w}_{2n} 
    \, 
    \frac{ \left( \frac{\Tilde{w}_{1}}{\Tilde{w}_{2n}} \right)^{-\alpha' / 2\pi} }{(\Tilde{w}_{2n} + \Tilde{w}_1)^2} 
    \\
    &\quad+
    \left|
        \frac{1}{\pi^2}   
        \left\{
            \int_{0}^{\alpha} d\alpha' 
            -
            \int_{-\alpha}^{0} d\alpha' 
        \right\} 
        \sin\left(\frac{\alpha'}{2}\right)
        \int_{0}^{D^{-2/3}} d\Tilde{w}_{2n} 
        \int_{0}^{D^{-2/3}} d\Tilde{w}_1  
        \, 
        \frac{ \left( \frac{\Tilde{w}_{1}}{\Tilde{w}_{2n}} \right)^{-\alpha' / 2\pi} }{(\Tilde{w}_{2n} + \Tilde{w}_1)^2} 
        \Big( 
            e^{-D(\frac{\Tilde{w}_{2n}}{1-\Tilde{w}_{2n}} + \frac{\Tilde{w}_1}{1+\Tilde{w}_1})} 
            - 
            e^{-D(\Tilde{w}_{2n} + \Tilde{w}_{1})} 
        \Big)
    \right|
    \\
    &=: \text{`Piece 2a'} + \text{`Piece 2b'} + \text{`Piece 2c'}
\end{split}
\end{equation}
Here on the last equality, $\text{`Piece 2a'}, \text{`Piece 2b'}, \text{`Piece 2c'}$ are respectively the first, second, third integrals after the previous inequality. We will analyze these terms separately.
First, we have 
\begin{equation} \label{eq:leading_error_2a}
\begin{split}
    \text{`Piece 2a'}
    &=
    2 e^{-D^{1/3}}
    \frac{1}{\pi^2} 
    \left\{
        \int_{0}^{\alpha} d\alpha' 
        -
        \int_{-\alpha}^{0} d\alpha' 
    \right\} 
    \sin\left(\frac{\alpha'}{2}\right)
    \int_{D^{-2/3}}^{1} d\Tilde{w}_{2n} 
    \int_{0}^{\infty} d\Tilde{w}_1  
    \, 
    \frac{ \left( \frac{\Tilde{w}_{1}}{\Tilde{w}_{2n}} \right)^{-\alpha' / 2\pi} }{(\Tilde{w}_{2n} + \Tilde{w}_1)^2} 
    \\
    &=
    2 e^{-D^{1/3}}
    \frac{1}{\pi^2} 
    \left\{
        \int_{0}^{\alpha} d\alpha' 
        -
        \int_{-\alpha}^{0} d\alpha' 
    \right\} 
    \int_{D^{-2/3}}^{1} d\Tilde{w}_{2n} 
    \, 
    \frac{\alpha'}{4} \frac{1}{w_{2n}}
    =
    \frac{\alpha^2}{2\pi^2}\log(D^{2/3})  e^{-D^{1/3}} < \frac{\log(D)e^{-D^{1/3}}}{3} = O(1/D).
\end{split}
\end{equation}
Next we have
\begin{equation} \label{eq:leading_error_2b}
\begin{split}
    \text{`Piece 2b'}
    &=
    2 e^{-0.99 D^{1/3}}
    \frac{1}{\pi^2} 
    \left\{
        \int_{0}^{\alpha} d\alpha' 
        -
        \int_{-\alpha}^{0} d\alpha' 
    \right\} 
    \sin\left(\frac{\alpha'}{2}\right)
    \int_{D^{-2/3}}^{\infty} d\Tilde{w}_1 
    \int_{0}^{1} d\Tilde{w}_{2n} 
    \, 
    \frac{ \left( \frac{\Tilde{w}_{1}}{\Tilde{w}_{2n}} \right)^{-\alpha' / 2\pi} }{(\Tilde{w}_{2n} + \Tilde{w}_1)^2} 
    \\
    &=
    2 e^{-0.99 D^{1/3}}
    \frac{1}{\pi^2}    
    \left\{
        \int_{0}^{\alpha} d\alpha' 
        -
        \int_{-\alpha}^{0} d\alpha' 
    \right\} 
    \sin\left(\frac{\alpha'}{2}\right)
    \int_{D^{-2/3}}^{\infty} d\Tilde{w}_1 
    \, 
    \Tilde{w}_1^{-2 - \frac{\alpha'}{2\pi}}
    \left( 
        \frac{\Tilde{w}_1}{1+\Tilde{w}_1} 
        - 
        {}_{2}F_{1} \, \left(1,1+\frac{\alpha'}{2\pi},2+\frac{\alpha'}{2\pi};-\frac{1}{\Tilde{w}_1}\right) 
        \cdot
        \frac{\alpha'}{2\pi+\alpha'} 
    \right)
    \\
    &<
    4 e^{-0.99 D^{1/3}}
    \frac{1}{\pi^2} 
    \left\{
        \int_{0}^{\alpha} d\alpha' 
        -
        \int_{-\alpha}^{0} d\alpha' 
    \right\} 
    \sin\left(\frac{\alpha'}{2}\right)
    \int_{D^{-2/3}}^{\infty} d\Tilde{w}_1 
    \, 
    \Tilde{w}_1^{-2 - \frac{\alpha'}{2\pi}}
    \\
    &= 
    4 e^{-0.99 D^{1/3}}
    \frac{1}{\pi^2} 
    \left\{
        \int_{0}^{\alpha} d\alpha' 
        -
        \int_{-\alpha}^{0} d\alpha' 
    \right\} 
    \sin\left(\frac{\alpha'}{2}\right)
    \frac{2\pi}{\alpha'+2\pi} D^{\left(2+\frac{\alpha'}{\pi}\right)/3}
    <
    \frac{8}{\pi} e^{-0.99 D^{1/3}} D = O(1/D).
\end{split}
\end{equation}
In the above, ${}_{2}F_{1}$ is the Gauss hypergeometric function~\cite{bailey1935generalized}. We used the following fact that for $\alpha' \in (0,\pi)$ and $\tilde{w}_1 \in (0,\infty)$ that
$ 0 < {}_{2}F_{1} \, \left(1,1+\frac{\alpha'}{2\pi},2+\frac{\alpha'}{2\pi};-\frac{1}{\Tilde{w}_1}\right) < 1 $. This follows from inputting negative values of $w$ into the next lemma and bounding ${}_{2}F_{1} \, \left(1,1+\frac{\alpha'}{2\pi},2+\frac{\alpha'}{2\pi};-\frac{1}{\Tilde{w}_1}\right)$ by $\int_{0}^{\infty} dt \, e^{-t} = 1$
\begin{lem} \label{lem:hypergeometric_functions_A}
The hypergeometric function ${}_{2}F_{1} \, \left(1,1+\frac{\alpha'}{2\pi},2+\frac{\alpha'}{2\pi};w\right)$ is an analytic continuation of the following integral.
\begin{equation}
    {}_{2}F_{1} \, \left(1,1+\frac{\alpha'}{2\pi},2+\frac{\alpha'}{2\pi};w\right)
    =
    \int_{0}^{\infty} dt \frac{ e^{-t} }{ 1- e^{-\frac{t}{1+\frac{\alpha'}{2\pi}}} w }
\end{equation}
\end{lem}
\begin{proof}
    By the definition of hypergeometric functions~\cite{bailey1935generalized}, and manipulations of the various Pochhammer symbols, one can show that ${}_{2}F_{1} \, \left(1,1+\frac{\alpha'}{2\pi},2+\frac{\alpha'}{2\pi};w\right)$ is defined by the series
    \begin{equation}
        {}_{2}F_{1} \, \left(1,1+\frac{\alpha'}{2\pi},2+\frac{\alpha'}{2\pi};w\right)
        =
        \sum_{n=0}^{\infty} \frac{z^n}{1+\frac{n}{1+\frac{\alpha'}{2\pi}}}
        =
        \int_{0}^{\infty} dt \, e^{-t} \sum_{n=0}^{\infty} z^n e^{-\left(1+\frac{n}{1+\frac{\alpha'}{2\pi}}\right) t}
        =
        \int_{0}^{\infty} dt \frac{ e^{-t} }{ 1- e^{-\frac{t}{1+\frac{\alpha'}{2\pi}}} w }.
    \end{equation}
\end{proof}

Now, we want to bound the last `Piece 2c'.
\begin{equation} \label{eq:leading_error_2c}
\begin{split}  
    |\text{`Piece 2c'}|
    &<
    \frac{1}{\pi^2} 
    \left\{
        \int_{0}^{\alpha} d\alpha' 
        -
        \int_{-\alpha}^{0} d\alpha' 
    \right\} 
    \sin\left(\frac{\alpha'}{2}\right)
    \int_{0}^{D^{-2/3}} d\Tilde{w}_{2n} 
    \int_{0}^{D^{-2/3}} d\Tilde{w}_1  
    \, 
    \left|
        \frac{ \left( \frac{\Tilde{w}_{1}}{\Tilde{w}_{2n}} \right)^{-\alpha' / 2\pi} }{(\Tilde{w}_{2n} + \Tilde{w}_1)^2} 
        e^{-D(\Tilde{w}_{2n} + \Tilde{w}_{1})}
        \left(
            1
            -
            e^{-D(\frac{\Tilde{w}_{2n}^2}{1-\Tilde{w}_{2n}} - \frac{\Tilde{w}_{1}^2}{1+\Tilde{w}_1})}
        \right)
    \right|
    \\
    &<
    \frac{1}{\pi^2} 
    \left\{
        \int_{0}^{\alpha} d\alpha' 
        -
        \int_{-\alpha}^{0} d\alpha' 
    \right\} 
    \sin\left(\frac{\alpha'}{2}\right)
    \int_{0}^{D^{-2/3}} d\Tilde{w}_{2n} 
    \int_{0}^{D^{-2/3}} d\Tilde{w}_1  
    \, 
    \left|
        \frac{ \left( \frac{\Tilde{w}_{1}}{\Tilde{w}_{2n}} \right)^{-\alpha' / 2\pi} }{(\Tilde{w}_{2n} + \Tilde{w}_1)^2} \cdot
        2 D 
        \left( 
            \frac{\Tilde{w}_{2n}^2}{1-\Tilde{w}_{2n}}
            - 
            \frac{\Tilde{w}_{1}^2}{1+\Tilde{w}_1} 
        \right)
    \right|
    \\
    &<
    \frac{2 D}{\pi^2} 
    \left\{
        \int_{0}^{\alpha} d\alpha' 
        -
        \int_{-\alpha}^{0} d\alpha' 
    \right\} 
    \sin\left(\frac{\alpha'}{2}\right)
    \int_{0}^{D^{-2/3}} d\Tilde{w}_{2n} 
    \int_{0}^{D^{-2/3}} d\Tilde{w}_1  
    \, 
    \left( \frac{\Tilde{w}_{1}}{\Tilde{w}_{2n}} \right)^{-\alpha' / 2\pi} \cdot
    \frac{3 D^{-2/3}}{\Tilde{w}_{2n} + \Tilde{w}_1}
    \\
    &=
    \frac{6}{\pi^2} D^{-1/3}
    \left\{
        \int_{0}^{\alpha} d\alpha' 
        -
        \int_{-\alpha}^{0} d\alpha' 
    \right\} 
    \left\{
        \pi
        -
        \frac{4\pi^2 \sin\left(\frac{\alpha'}{2}\right)}{2\pi\alpha'+\alpha'^2} \,
        {}_{2}F_{1} \, \left(1,\frac{\alpha'}{2\pi},2+\frac{\alpha'}{2\pi};-1\right)
    \right\}
    = O(D^{-1/3})
\end{split}
\end{equation}
The first inequality brought the absolute values inside the integral and rewrote the integrand. 
In the second inequality, we use two facts. First, we have $e^{-D(\tilde{w_1}+\tilde{w_{2n}})} < 1$. 
Then, we have that for $D$ (mildly) large enough, we have 
$D \left( \frac{\Tilde{w}_{2n}^2}{1-\Tilde{w}_{2n}} - \frac{\Tilde{w}_{1}^2}{1+\Tilde{w}_1} \right) \ll 1$
for 
$\Tilde{w}_{1} \in (0,D^{-2/3})$ and $\Tilde{w}_{2n} \in (0,D^{-2/3})$ 
so that
$
\left| 1 - e^{-D(\frac{\Tilde{w}_{2n}^2}{1-\Tilde{w}_{2n}} - \frac{\Tilde{w}_{1}^2}{1+\Tilde{w}_1})} \right| 
< 
\left|
    2 D 
    \left( 
        \frac{\Tilde{w}_{2n}^2}{1-\Tilde{w}_{2n}}
        - 
        \frac{\Tilde{w}_{1}^2}{1+\Tilde{w}_1} 
    \right)
\right|
$.
In the third inequality, we use
$
\left|
    \frac{
        \frac{\Tilde{w}_{2n}^2}{1-\Tilde{w}_{2n}}
        - 
        \frac{\Tilde{w}_{1}^2}{1+\Tilde{w}_1} 
    }{(\Tilde{w}_{2n} + \Tilde{w}_1)^2}
\right|
=
\left|
    \frac{
        \tilde{w}_1-(1+\tilde{w}_1)\tilde{w}_{2n}
    }{(1+\tilde{w_1})(1-\tilde{w_{2n}})(\Tilde{w}_{2n} + \Tilde{w}_1)}
\right|
<
\frac{3 D^{-2/3}}{\Tilde{w}_{2n} + \Tilde{w}_1}
$
within the integration region for $D$ large enough. 
In the fourth (in)equality, we perform the integrals over $\tilde{w}_1$ and $\tilde{w}_{2n}$.
In the final equality, we use the Lemma~\ref{lem:hypergeometric_functions_B} below to bound ${}_{2}F_{1} \, \left(1,\frac{\alpha'}{2\pi},2+\frac{\alpha'}{2\pi};-1\right) < \int_{0}^{\infty} dt_1 \int_{0}^{\infty} dt_2 \, e^{-t_1}e^{-t_2} = 1$. As such, one can see that the integrand is a finite convergent integral that can be uniformly bounded for $\alpha \in (0,\pi)$.

\begin{lem} \label{lem:hypergeometric_functions_B}
The hypergeometric function ${}_{2}F_{1} \, \left(1,\frac{\alpha'}{2\pi},2+\frac{\alpha'}{2\pi};w\right)$ is an analytic continuation of the following integral.
\begin{equation}
    {}_{2}F_{1} \, \left(1,\frac{\alpha'}{2\pi},2+\frac{\alpha'}{2\pi};w\right)
    =
    \int_{0}^{\infty} dt_1
    \int_{0}^{\infty} dt_2
    \frac{e^{-t_1}e^{-t_2} }{1-e^{-t_1\frac{2\pi}{\alpha'}} e^{-t_2\frac{2\pi}{2\pi+\alpha'}} w}
\end{equation}
\end{lem}
\begin{proof}
    By the definition of hypergeometric functions, and manipulations of the various Pochhammer symbols, one can show that ${}_{2}F_{1} \, \left(1,1+\frac{\alpha'}{2\pi},2+\frac{\alpha'}{2\pi};w\right)$ is defined by the series
    \begin{equation}
    \begin{split}
        &{}_{2}F_{1} \, \left(1,\frac{\alpha'}{2\pi},2+\frac{\alpha'}{2\pi};w\right)
        =
        \sum_{n=0}^{\infty} \frac{w^n}{(1+\frac{2\pi n}{\alpha'})(1+\frac{2\pi n}{2\pi+\alpha'})}
        \\
        &=
        \int_{0}^{\infty} dt_1
        \int_{0}^{\infty} dt_2
        \sum_{n=0}^{\infty} w^n e^{-t_1(1+\frac{2\pi n}{\alpha'})} e^{-t_2(1+\frac{2\pi n}{2\pi+\alpha'})} 
        =
        \int_{0}^{\infty} dt_1
        \int_{0}^{\infty} dt_2
        \frac{e^{-t_1}e^{-t_2} }{1-e^{-t_1\frac{2\pi}{\alpha'}} e^{-t_2\frac{2\pi}{2\pi+\alpha'}} w}
    \end{split}
    \end{equation}
\end{proof}

Combining Eqs.~(\ref{eq:leading_exact_plus_error},\ref{eq:leading_error_2},\ref{eq:leading_error_2a},\ref{eq:leading_error_2b},\ref{eq:leading_error_2c}), we get
\begin{equation}
    -\sum_{n=1}^{\infty} \frac{1}{2n}
    (2 \sin(\alpha/2))^2
    \mathcal{J}_{2n}
    =
    \frac{1}{\pi^2} 
    \left\{
        \int_{0}^{\alpha} d\alpha' 
        -
        \int_{-\alpha}^{0} d\alpha' 
    \right\} 
    \sin\left(\frac{\alpha'}{2}\right)
    \int_{0}^{1} d\Tilde{w}_{2n} 
    \int_{0}^{\infty} d\Tilde{w}_1  \, 
    \frac{ \left( \frac{\Tilde{w}_{1}}{\Tilde{w}_{2n}} \right)^{-\alpha' / 2\pi} }{(\Tilde{w}_{2n} + \Tilde{w}_1)^2}
    \Big( 1 - e^{-D(\Tilde{w}_{2n} + \Tilde{w}_{1})} \Big)
    +
    O(D^{-1/3})
\end{equation}
Now, we can perform the first two integrals over $\tilde{w}_1$ and $\tilde{w}_{2n}$ in Mathematica and arrive at
\begin{equation} \label{eq:leadingorder_w1_w2n_integrated}
\begin{split}
    &    
    -\sum_{n=1}^{\infty} \frac{1}{2n}
    (2 \sin(\alpha/2))^2
    \mathcal{J}_{2n}
    + O(D^{-1/3})
    \\
    &=
    -\frac{1}{2\pi^2} 
    \left\{
        \int_{0}^{\alpha} d\alpha' 
        -
        \int_{-\alpha}^{0} d\alpha' 
    \right\} 
    \Bigg\{
        \alpha'
        + 
        \alpha' D^{1+\frac{\alpha'}{2\pi}} 
        \frac{\Gamma\left(1+\frac{\alpha'}{2\pi}\right)}{\Gamma\left(2+\frac{\alpha'}{2\pi}\right)^2}
        \cdot
        {}_{2}F_{2} 
        \left[
            \begin{matrix} 
                1+\frac{\alpha'}{2\pi} \,\,,&\,\, 1+\frac{\alpha'}{2\pi} 
                \\ 
                2+\frac{\alpha'}{2\pi} \,\,,&\,\, 2+\frac{\alpha'}{2\pi} 
            \end{matrix}
            \,\,;\,\,
            -D
        \right]
    \\
    &\quad\quad\quad\quad\quad\quad\quad
        +
        2 \cdot D^{1+\frac{\alpha'}{2\pi}} \,
        \Gamma\left(-\frac{\alpha'}{2\pi}\right) \sin\left(\frac{\alpha'}{2}\right)
        \left(
            e^{-D}
            -
            \left(D - \frac{\alpha'}{2\pi}\right)
            \int_{1}^{\infty} dt \, {e^{-D t}} t^{\frac{\alpha}{2\pi}}
        \right)
    \Bigg\}
\end{split}
\end{equation}
In the above, $\Gamma$ is the Gamma function and 
$
{}_{2}F_{2} 
\left[
    \begin{matrix} 
        1+\frac{\alpha'}{2\pi} \,\,,&\,\, 1+\frac{\alpha'}{2\pi} 
        \\ 
        2+\frac{\alpha'}{2\pi} \,\,,&\,\, 2+\frac{\alpha'}{2\pi} 
    \end{matrix}
    \,\,;\,\,
    -D
\right]
$
is the generalized hypergeometric function~\cite{bailey1935generalized}.
Note that the second line above can be uniformly bounded so only contributes at most an $O(1/D)$ factor. We have the following lemma to analyze the other term.
\begin{lem} \label{lem:hypergeometric_functions_C}
    The hypergeometric function
    $
    {}_{2}F_{2} 
    \left[
        \begin{matrix} 
            1+\frac{\alpha'}{2\pi} \,\,,&\,\, 1+\frac{\alpha'}{2\pi} 
            \\ 
            2+\frac{\alpha'}{2\pi} \,\,,&\,\, 2+\frac{\alpha'}{2\pi} 
        \end{matrix}
        \,\,;\,\,
        w
    \right]
    $
    is an analytic continuation of the following integral.
    \begin{equation}
        {}_{2}F_{2} 
        \left[
            \begin{matrix} 
                1+\frac{\alpha'}{2\pi} \,\,,&\,\, 1+\frac{\alpha'}{2\pi} 
                \\ 
                2+\frac{\alpha'}{2\pi} \,\,,&\,\, 2+\frac{\alpha'}{2\pi} 
            \end{matrix}
            \,\,;\,\,
            w
        \right]
        =
        \left(1+\frac{\alpha'}{2\pi}\right)^2
        \int_{0}^{\infty} dt_1
        \int_{0}^{\infty} dt_2
        e^{-t_1\left(1+\frac{\alpha'}{2\pi}\right)}
        e^{-t_2\left(1+\frac{\alpha'}{2\pi}\right)} 
        \exp\left[w \, e^{-t_1}  e^{-t_2} \right].
    \end{equation}
\end{lem}
\begin{proof}
    By the definition of hypergeometric functions, and manipulations of the various Pochhammer symbols, one can show that $
    {}_{2}F_{2} 
    \left[
        \begin{matrix} 
            1+\frac{\alpha'}{2\pi} \,\,,&\,\, 1+\frac{\alpha'}{2\pi} 
            \\ 
            2+\frac{\alpha'}{2\pi} \,\,,&\,\, 2+\frac{\alpha'}{2\pi} 
        \end{matrix}
        \,\,;\,\,
        w
    \right]
    $ 
    is defined by the series
    \begin{equation}
    \begin{split}
        &
        {}_{2}F_{2} 
        \left[
            \begin{matrix} 
                1+\frac{\alpha'}{2\pi} \,\,,&\,\, 1+\frac{\alpha'}{2\pi} 
                \\ 
                2+\frac{\alpha'}{2\pi} \,\,,&\,\, 2+\frac{\alpha'}{2\pi} 
            \end{matrix}
            \,\,;\,\,
            w
        \right]
        =
        \left(1+\frac{\alpha'}{2\pi}\right)^2
        \sum_{n=0}^{\infty} 
        \frac{w^n}{n!}
        \frac{1}{\left(n+1+\frac{\alpha'}{2\pi}\right)^2}
        \\
        &=
        \left(1+\frac{\alpha'}{2\pi}\right)^2
        \int_{0}^{\infty} dt_1
        \int_{0}^{\infty} dt_2
        \sum_{n=0}^{\infty} \frac{w^n}{n!} 
        e^{-t_1\left(n+1+\frac{\alpha'}{2\pi}\right)} 
        e^{-t_2\left(n+1+\frac{\alpha'}{2\pi}\right)} 
        \\
        &=
        \left(1+\frac{\alpha'}{2\pi}\right)^2
        \int_{0}^{\infty} dt_1
        \int_{0}^{\infty} dt_2
        e^{-t_1\left(1+\frac{\alpha'}{2\pi}\right)}
        e^{-t_2\left(1+\frac{\alpha'}{2\pi}\right)} 
        \exp\left[w \, e^{-t_1}  e^{-t_2} \right].
    \end{split}
    \end{equation}
\end{proof}
This gives 
\begin{equation} \label{eq:hypergeom_2F2_term}
\begin{split}    
    &D^{1+\frac{\alpha'}{2\pi}} 
    \frac{\Gamma\left(1+\frac{\alpha'}{2\pi}\right)}{\Gamma\left(2+\frac{\alpha'}{2\pi}\right)^2}
    \cdot
    {}_{2}F_{2} 
    \left[
        \begin{matrix} 
            1+\frac{\alpha'}{2\pi} \,\,,&\,\, 1+\frac{\alpha'}{2\pi} 
            \\ 
            2+\frac{\alpha'}{2\pi} \,\,,&\,\, 2+\frac{\alpha'}{2\pi} 
        \end{matrix}
        \,\,;\,\,
        -D
    \right]
    \\
    &=
    \frac{1}{\Gamma\left(1+\frac{\alpha'}{2\pi}\right)}
    \int_{0}^{\infty} dt_1
    \int_{0}^{\infty} dt_2
    e^{-(t_1+t_2-\log(D))\left(1+\frac{\alpha'}{2\pi}\right)}
    \exp\left[-\, e^{-(t_1+t_2-\log(D))} \right]
    \\
    &=
    \frac{1}{\Gamma\left(1+\frac{\alpha'}{2\pi}\right)}
    \int_{-\log(D)}^{\infty} dt' (t'+\log(D)) e^{-t'\left(1+\frac{\alpha'}{2\pi}\right)} \exp\left[-e^{-t'}\right]
    \\
    &=
    \frac{1}{\Gamma\left(1+\frac{\alpha'}{2\pi}\right)}
    \int_{-\infty}^{\infty} dt' 
    (t'+\log(D)) 
    e^{-t'\left(1+\frac{\alpha'}{2\pi}\right)} 
    \exp\left[-e^{-t'}\right]
    +
    O(1/D)
    \\
    &=
    \frac{1}{\Gamma\left(1+\frac{\alpha'}{2\pi}\right)}
    \int_{0}^{\infty} d{t''}
    (-\log(t'')+\log(D)) 
    {t''}^{\left(1+\frac{\alpha'}{2\pi}\right)-1} 
    e^{-t''}
    +
    O(1/D)
    =
    \log(D) - \frac{\Gamma'\left(1+\frac{\alpha'}{2\pi}\right)}{\Gamma\left(1+\frac{\alpha'}{2\pi}\right)}
    +
    O(1/D).
\end{split}
\end{equation}
The first equality above uses Lemma~\ref{lem:hypergeometric_functions_C} and simplification. The second equality follows from setting $t'=t_1+t_2-\log(D)$ and doing the integral over $t_1-t_2$. In the third equality, the integral over $t': -\infty \to -\log(D)$ is an $O(1/D)$ term since the integrand is super-exponentially suppressed so we can add it in. In the fourth equality, we substitute $t'' = e^{-t'}$. In the last equality, $\Gamma'$ is the derivative of the Gamma function and note that the integration of the term corresponding to $\log(t'')$ gives $\Gamma'\left(1+\frac{\alpha'}{2\pi}\right)$ and the term corresponding to $\log(D)$ gives $\log(D) \Gamma\left(1+\frac{\alpha'}{2\pi}\right)$. 

Finally, using Eqs.~(\ref{eq:leadingorder_w1_w2n_integrated},\ref{eq:hypergeom_2F2_term}) we get
\begin{equation} \label{eq:final_sum_J2n}
\begin{split}
    &
    -\sum_{n=1}^{\infty} \frac{1}{2n}
    (2 \sin(\alpha/2))^2
    \mathcal{J}_{2n}
    + O(D^{-1/3})
    =
    -\frac{1}{2\pi^2} 
    \left\{
        \int_{0}^{\alpha} d\alpha' 
        -
        \int_{-\alpha}^{0} d\alpha' 
    \right\} 
    \alpha'
    \left(
        1 + \log(D) 
        - 
        \frac{\Gamma'\left(1+\frac{\alpha'}{2\pi}\right)}{\Gamma\left(1+\frac{\alpha'}{2\pi}\right)} 
    \right)
    \\
    &= 
    -\frac{\alpha^2}{2\pi^2} \left( 1 + \log(D) + \gamma_{\mathrm{Euler}} \right) 
    -
    4 \sum_{n=2}^{\infty} \frac{1}{2n} \left( \frac{\alpha}{2\pi} \right)^{2n} \zeta(2n-1)
\end{split}
\end{equation}
where $\zeta$ is the Riemann-zeta function.

\section{ 2D Free Fermion Determinants and Tau Functions } \label{app:diracAndTau}

We warn the reader that throughout this section the computations should be taken to be formal, and the results of this section aren't meant to be mathematically rigorous. However, these formal manipulations match the rigorous calculations of this paper for the dimer model, of Sec.~\ref{sec:multipleZipper_calculations}.

The tau function as introduced by~\cite{jmuTau1981}, is a function that's holomorphic in the variables $p_1, \cdots, p_s$, depending on a flat connection $\rho$ on the punctured plane $\C - \{p_1, \cdots, p_s\}$. A more closely related viewpoint that we take is related to that of Palmer~\cite{palmer1993} who interprets the tau function as some regularized determinant of a Cauchy-Riemann operator in the presence of a singular gauge field that's flat away from the punctures $p_1, \cdots, p_s$. 

The perspective we take in this paper is to formally consider the `determinant' of the two-dimensional Dirac operator $\slashed{\partial}(p_1,\cdots,p_s | \rho)$ in the presence of such a flat connection $\rho$ with respect to punctures $p_1, \cdots, p_s$. Normalizing with the determinant of the starnard Dirac operator $\slashed{\partial}$ we will get a function
\begin{equation}
    \tau(p_1, \cdots, p_s | \rho) := 
    \frac{\det\Big(\slashed{\partial}(p_1,\cdots,p_s | \rho)\Big)}{\det\big( \slashed{\partial} \big)}
\end{equation}
This function will differ from the original tau function of~\cite{jmuTau1981} in two ways, which is why we refer to our functions as `modified tau functions'. 

The first difference is that $\log \tau(p_1, \cdots, p_s | \rho)$ will be defined in terms of a series expansion of integrals, many of which are actually divergent integrals. Even though this series expansion may seem ill-defined, in our paper in the context of dimer models, the integrals in the series expansion are naturally regulated by the lattice scale and diverge logarithmically with this scale. 
The second is that our series expansion of $\log \tau$ is actually not holomorphic in the $p_i$, but rather is a sum of a holomorphic and antiholomorphic functions. The holomorphic part can be identified with the original Jimbo-Miwa-Ueno tau function while the antiholomorphic part can be identified with a similar `conjugate' version (see the end of Sec.~\ref{app:tau_zipper}).

In Sec.~\ref{app:tau_gaugeFields}, we will review the gauge-theoretic setup and the Jimbo-Miwa-Ueno definition of the tau function in their original vocabulary in terms of Fuschian systems, and then compare the Fuschian representations to the zipper representations of flat gauge fields. In Sec.~\ref{app:tau_dirac}, we review the calculus of computing determinants of Dirac operators and show how it's gauge-invariant. Then in Sec.~\ref{app:tau_fuschian}, we will show how the Fuschian representation of the gauge field gives a representation of $\log \tau$ that satisfies the defining Jimbo-Miwa-Ueno equations. In Sec.~\ref{app:tau_zipper}, we show how the alternate series expansion for $\tau$ encountered in dimer correlations (see Sec.~\ref{sec:multipleZipper_calculations}) can be derived using the zipper representation of the gauge fields. 

\subsection{Singular Gauge Fields and Isomonodromic Tau functions} \label{app:tau_gaugeFields}
Consider the punctured plane $\C - \{p_1 , \cdots , p_s\}$ and a connection $\rho$ that's flat away from the punctures. Furthermore, suppose that the monodromy of $\rho$ around infinity is zero, so that $\rho$ can be viewed as a flat connection on $\CP^1 - \{p_1 , \cdots , p_s\}$. 
In our setup, we will consider $GL(N,\C)$-connection, which can be viewed as an $\gl_N$-valued 1-form $A$ on $\CP^1$. We will mostly consider connections that satisfy the \textit{flatness condition} $dA - A \wedge A = 0$ away from the punctures. In local (real) coordinates $x,y$ on $\C$ or equivalently the (complex) coordinates $z = x + i y$ and $\overline{z} = x - i y$, we will have
\begin{equation}
\begin{split}
    &A = A_x(x,y) dx + A_y(x,y) dy = A_z(z,\overline{z}) dz + A_{\overline{z}}(z,\overline{z}) d\overline{z} 
    \\
    &A_x = A_{z} + A_{\overline{z}} \quad , \quad A_y = i A_{z} - i A_{\overline{z}}
    \\
    &\partial_y A_x - \partial_x A_y - [A_x,A_y] = 0 \quad \mathrm{iff} \quad \partial_{\overline{z}} A_z - \partial_z A_{\overline{z}} - [A_z,A_{\overline{z}}] = 0,
\end{split}
\end{equation}
where the last equation holds when we consider flat gauge fields. 
We say that a (local) $N$-component field $\Psi$ is a \textit{(local) flat section} of $A$ if it satisfies the equation $(d-A)\Psi = 0$, which gives a pair of differential equations that $\Psi$ needs to satisfy. 

Note for any $GL(N,\C)$-valued \textit{gauge transformation} $g(x,y)$, we have 
\begin{equation}
    (d-A)\Psi = 0 \quad \mathrm{iff} \quad (d - (g A g^{-1} - g d g^{-1})) (g \Psi)
\end{equation}
As such, the flatness of $\Psi$ with respect to $A$ is equivalent to the flatness of flatness of $g \Psi$ with respect to $g A g^{-1} - g d g^{-1}$. As such, we say that two gauge fields $A,A'$ are \textit{gauge equivalent} if there exists a function $g$ such that $A' = g A g^{-1} - g d g^{-1}$.

Via the path-ordered exponential, a \textit{basis} local flat sections $\{\Psi_1, \cdots, \Psi_N\}$ of any flat gauge field on $\CP^1 - \{p_1,\cdots,p_s\}$ exist on any simply connected subset of $\CP^1 - \{p_1 , \cdots , p_s\}$. However, solutions will often not be consistent in a loop going around a puncture, and will instead have a \textit{monodromy} going around a loop associated to crossings with branch cuts coming out from the punctures. For example, the solution $\Psi = \left(\frac{z}{z-1}\right)^{\alpha}$ is a holomorphic solution to $\left(\partial_z - \alpha\left(\frac{1}{z} - \frac{1}{z-1} \right) \right)\Psi = 0$, $\partial_{\overline{z}}\Psi = 0$, but for $\alpha \not\in \Z$, $\Psi$ has a branch cut along paths connecting $0 \to 1$. In general, given a basepoint $p$ in $\CP^1 - \{p_1 , \cdots , p_s\}$, the path-ordered exponential  associates flat connections to \textit{monodromy representations} $\pi_1(\CP^1 - \{p_1, \cdots, p_s\}) \to GL(N,\C)$. Under gauge transformation $g$, such representations are equivalent up to conjugation by $g(p) \in GL(N,\C)$. And, any two flat connections with the same monodromy representation are gauge equivalent.

Another way to think about the monodromy representation is as follows.
We can fix a simply connected subset by choosing a collection of $k$ non-intersecting segments, each of which start at one of the $p_1 , \cdots , p_s$ and end at some common point $p$, where the clockwise order of the zippers coming out from $p$ are $p_1 \to p_2 \to \cdots \to p_s$.
We often refer to such segments as `zippers' in analogy with the lattice case.
Cutting $\CP^1 - \{p_1,\cdots,p_s\}$ along a small neighborhood of these zippers will leave a simply connected region. 
And, we can construct a matrix of solutions ${\bf \Psi} = (\Psi_1 , \cdots , \Psi_N)$ with $(d-A){\bf \Psi}=0$ in this region and smooth it out near the zippers. 
Choosing a gauge transformation $g = {\bf \Psi}^{-1}$ will mean that in the region away from the cuts, we'll have that the transformed gauge field is zero, since $A' = -{\bf \Psi}^{-1} (d - A) {\bf \Psi} = 0$. 
After making this gauge transformation, the path-ordered exponential going from left-to-right across the zippers coming from $p_1, \cdots, p_s$ will be associated to monodromy matrices $M_1, \cdots, M_s$, which generate the flat connection $\pi_1(\CP^1 - \{p_1,\cdots,p_s\}) \to GL(N,\C)$. Sometimes, we refer to the $\{M_j\}$ as \textit{jump matrices}. Because there's no monodromy around the point $p$ where the zippers meet, there's an additional `flatness condition' that $M_s \cdots M_1 = \mathbbm{1}$. In general, the zippers need not meet at a single point, and any graph with endpoints at $p_1, \cdots, p_s$ will work. If $M_1, \cdots, M_s$ are associated to loops going around $p_1, \cdots, p_s$, the matrices going along the edges of any such graph are subject to a flatness condition around each vertex of the graph that's not $p_1, \cdots, p_s$. See Fig.~\ref{fig:flatGaugeFieldToZippers} for a depiction. We use these zipper representations in Sec.~\ref{app:tau_zipper}.

\begin{figure}[h!]
  \centering
  \includegraphics[width=\linewidth]{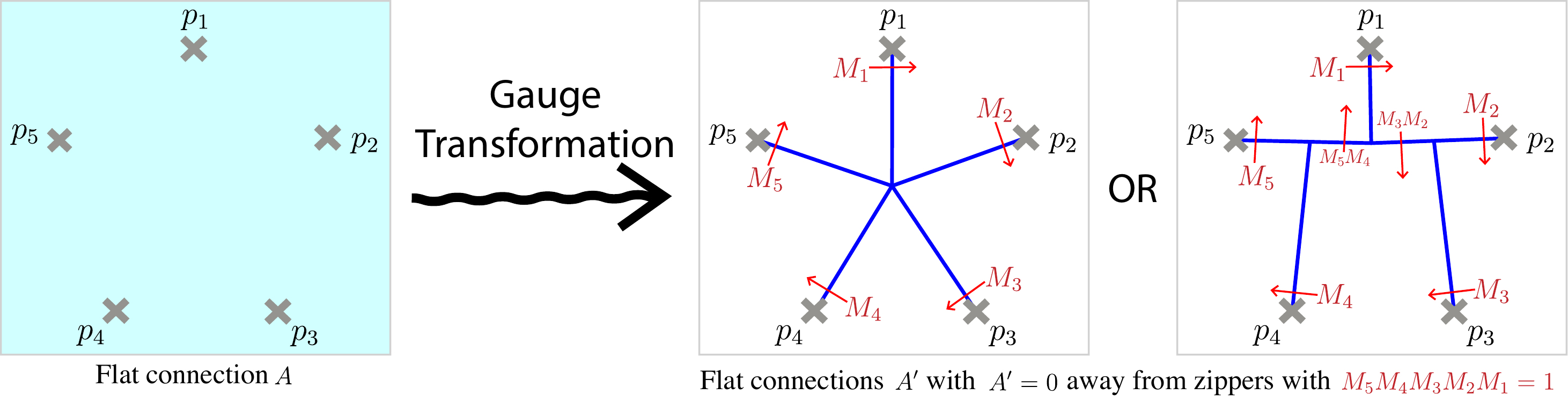}
  \caption{A gauge transformation can transform a flat gauge into one that's zero in a simply connected region away from some set of `zippers' which have matrices attached to them according to what the path-ordered exponential assigns going across the zipper. Around points that aren't $p_1, \cdots, p_s$ where the zippers meet, there is a flatness condition for the matrices.}
  \label{fig:flatGaugeFieldToZippers}
\end{figure}

A related problem of finding holomorphic solutions to so-called \textit{Fuschian systems} can be phrased this way. 
For some points $p_1 , \cdots , p_s \in \C$ and $n \times n$ matrices $B_1 , \cdots, B_s$ satisfying $\sum_{j} B_j = 0$, consider the differential equations
\begin{equation}
    \partial_z \Psi = \left( \sum_{j} \frac{B_j}{z-p_j} \right) \Psi \quad\quad \text{and} \quad\quad \partial_{\overline{z}} \Psi = 0.
\end{equation}
Local solutions to these equations are equivalent to finding local flat sections of a connection with $A_z = \sum_{j} \frac{B_j}{z-p_j}$, $A_{\overline{z}} = 0$. Note that away from the $p_j$, $dA - A \wedge A = 0$ so $A$ flat. The condition $\sum_{j} B_j = 0$ means that there's no pole of $A_z$ at infinity, so the gauge field is non-singular if extended to $\CP^1$. 

Since the Fuschian connection is flat, any points $p_1,\cdots,p_s$ and matrices $B_1,\cdots,B_s$ are associated to monodromy matrices $M_1, \cdots, M_s$ up to overall conjugation, as in the zipper construction above. 
\footnote{In fact, the inverse problem of constructing matrices construct matrices $B_1,\cdots,B_s$ given points $\{p_1, \cdots, p_s\}$ given monodromy data $\{M_1, \cdots, M_s\}$, known as Hilbert's 21st problem, can be solved for \textit{generic} monodromy matrices $\{M_1, \cdots, M_s\}$ (c.f.~\cite{bolibrukh199521st}).}
In general, small deformations the $\{p_j\}$ or of the matrices $\{B_j\}$ will lead to deformations in the monodromy matrices $\{M_j\}$. The problem of finding \textit{isomonodromic deformations} asks: given a small deformation of the $\{p_j\}$, what deformations of $\{B_j\}$ are necessary to keep the monodromy data fixed? The following \textit{Schlesinger Equations}
\begin{equation} \label{eq:schlesinger_1}
\begin{split}
    \partial_{p_k} B_i = \frac{[B_k,B_i]}{p_k - p_i}, \,\, \text{ if } k \neq i \quad\quad \text{and} \quad\quad
    \partial_{p_k} B_k = -\sum_{\ell \neq k} \frac{[B_k,B_{\ell}]}{p_k - p_{\ell}}
\end{split}
\end{equation}
give sufficient equations for finding isomonodromic deformations. Note that the second equation above follows from the condition that $\sum_{j} B_j = 0$. We will find it convenient to write these more compactly as
\begin{equation} \label{eq:schlesinger_2}
    \partial_{p_k} B_i = \frac{[B_k,B_i]}{p_k - p_i} - \delta_{k,i}\sum_{\ell} \frac{[B_k,B_{\ell}]}{p_k - p_{\ell}}
\end{equation}
with the understanding that the `singular' terms ``$\frac{[B_j,B_j]}{p_j-p_j}$'' are set to zero.  

In~\cite{jmuTau1981} Jimbo, Miwa, Ueno constructed the differential form $\omega$
\begin{equation}
\begin{split}
    \omega := \frac{1}{2} \sum_{i \neq j} \frac{\tr(B_i B_j)}{p_i - p_j} d(p_i - p_j),
\end{split}
\end{equation}
on the parameter space of points $\{p_j\}$. From Eq.~\eqref{eq:schlesinger_1}, one finds $d\omega = 0$, so that there are \textit{local} solutions $d \log \tau = \omega$, or equivalently of the equations
\begin{equation} \label{eq:Jimbo_Miwa_Ueno_eqs}
    \partial_{p_i} \log \tau = \sum_{j \neq i} \frac{\tr(B_i B_j)}{p_i - p_j}
\end{equation}
This local function $\tau$ (defined up to multiplication by an overall scalar or antiholomorphic function) is the isomodromic tau function of Jimbo-Miwa-Ueno.

\subsection{Determinants of 2D Dirac operators, Counterterms, and Gauge Invariance} \label{app:tau_dirac}

Here we will review general series expansions for determinants of 2D Dirac operators, also known as partition functions of 2D free Dirac fermions, in the presence of a background gauge fields. The lesson is that the naive series expansion for Dirac determinants in a background gauge field is not actually gauge invariant, but one can add a local counterterm making a gauge-invariant partition function. Our presentation is motivated by notes of Nair~\cite{nair2005}, for whom perspective is equating this fermion determinant with the action of the WZW theory~\cite{wittenNonabelian1983} (see also~\cite{polyakovWiegmann1983}). Our discussion sidesteps WZW theory and derives things more directly.

Let $A$ be any $\gl(N,\C)$ gauge field with components $A_z, A_{\overline{z}}$. Throughout this section, we'll assume that $A$ decays quickly, at least quadratically approaching infinity. In two dimensions, the 2D Dirac operator in the presence of this gauge field is the operator
\begin{equation}
    \slashed{\partial} - \slashed{A} = 
    \begin{pmatrix}
      0 & \partial_{z} \\ \partial_{\overline{z}} & 0
    \end{pmatrix}
    -
    \begin{pmatrix}
      0 & A_z \\ A_{\overline{z}} & 0
    \end{pmatrix}
    =
    \begin{pmatrix}
      0 & \partial_{z} - A_z \\ \partial_{\overline{z}} - A_{\overline{z}} & 0
    \end{pmatrix}
\end{equation}
formed from the operators $\slashed{\partial}$ and $\slashed{A}$.
acting on `spinors', which for our purposes are pairs of $\C^N$-valued functions $\psi_L,\psi_R$ arranged into $\C^{2N}$-valued ones as $\begin{pmatrix} \psi_L \\ \psi_R \end{pmatrix}$. The components $\psi_L$ and $\psi_R$ are the `left' and `right' components of the spinor. 

Our goal is to make sense of the expression $\det \left( \slashed{\partial} - \slashed{A} \right)$. The fact that $\slashed{\partial} - \slashed{A}$ is an infinite-dimensional linear operator indicates that such expressions are likely singular and ill-defined, even in a physical sense. We'll have more luck if instead we consider the ratio of the above determinant with $\det\left(\slashed{\partial}\right)$. Formally, we can expand the above quantity's logarithm as
\begin{equation}
    \log \frac{\det \left( \slashed{\partial} - \slashed{A} \right)}{\det\left(\slashed{\partial}\right)}
    \stackrel{\textbf{form}}{=} 
    \log \det \left[ 1 - \frac{1}{\slashed{\partial}} \slashed{A} \right] = \tr \log  \left[ 1 - \frac{1}{\slashed{\partial}} \slashed{A} \right] = - \sum_{n=1}^{\infty}  \frac{1}{n} \tr \left[ \left( \frac{1}{\slashed{\partial}} \slashed{A} \right)^n \right],
\end{equation}
where $\stackrel{\textbf{form}}{=}$ refers to an equality in the sense of formal series.
In the above, $\frac{1}{\slashed{\partial}}$ refers to the Green's function
\begin{equation}
    \frac{1}{\slashed{\partial}}(x_1,x_2)
    =
    \begin{pmatrix}
      0 & \frac{1}{\partial_{\overline{z}}}(x_1,x_2) \\ 
      \frac{1}{\partial_{z}}(x_1,x_2) & 0
    \end{pmatrix}
    =
    \begin{pmatrix}
      0 & \frac{1}{\pi}\frac{1}{z_1-z_2} \\ \frac{1}{\pi}\frac{1}{\overline{z}_1-\overline{z}_2} & 0
    \end{pmatrix}
\end{equation}
which is formed by the individual Green's functions $\frac{1}{\partial_{z}}(x_1,x_2) = \frac{1}{\pi}\frac{1}{\overline{z}_1-\overline{z}_2}$ and $\frac{1}{\partial_{\overline{z}}}(x_1,x_2) = \frac{1}{\pi}\frac{1}{z_1-z_2}$. 

In the above, for each coordinate $x_j = (x_{j,1},x_{j,2})$, we denote $z_j = x_{j,1} + i x_{j,2}$ and $\overline{z}_j = x_{j,1} - i x_{j,2}$. Throughout this section, we will use this notation, where `$x$' refers to the coordinate in $\R^2$ and $z$ amd $\overline{z}$ refer to the corresponding complex coordinate in $\C$ and its conjugate.

Now, we should make sense of the `matrix multiplication' and traces above. The space we're acting on is $\C^{2N}$ tensored with some suitable function space (e.g. $L^2(\C), C^{\infty}(\C)$). As such, linear operators on this space live in $M_{2N \times 2N}(\C)$ tensored with linear operators on the function space. Given two operators $\mathcal{O}_1, \mathcal{O}_2$, we'll have $(\mathcal{O}_1 \mathcal{O}_2)(x_1,x_2) = \int dz' \mathcal{O}_1(x_1,x') \mathcal{O}_2(x',x_2)$, and the quantity $\mathcal{O}_1(x_1,x') \mathcal{O}_2(x',x_2)$ refers to matrix multiplication in $M_{2N \times 2N}(\C)$. In this sense, since $\slashed{A}$ is a diagonal matrix, we would express
\begin{equation}
    \slashed{A}(x_1,x_2) = \slashed{A}(x_1) \delta^2(x_1-x_2)
\end{equation}
where $\delta^2(z_1,z_2)$ is a Dirac delta function. This means that 
\begin{equation}
    \frac{1}{\slashed{\partial}}\slashed{A}(x_1,x_2)
    =
    \begin{pmatrix}
      \frac{1}{\pi}\frac{A_{\overline{z}}(z_2)}{z_1-z_2} & 0 \\ 0& \frac{1}{\pi}\frac{A_z(z_2)}{\overline{z}_1-\overline{z}_2}
    \end{pmatrix}.
\end{equation}

Since this is block-diagonal with respect to the left and right spinor components, we can write
\begin{equation}
\begin{split}
    \tr &\left[\left( \frac{1}{\slashed{\partial}} \slashed{A} \right)^n\right]
    \stackrel{\textbf{form}}{=} 
    \tr \left[\left( \frac{1}{\partial_{\overline{z}}} A_{\overline{z}} \right)^n\right]
    +
    \tr \left[\left( \frac{1}{\partial_z} A_z \right)^n\right]
    \\
    &\stackrel{\textbf{form}}{=}
    \frac{1}{\pi^n} \int d^2 x_1 \cdots \int d^2 x_n \left\{ 
    \frac{\tr\left[A_{\overline{z}}(x_1) A_{\overline{z}}(x_2) \cdots A_{\overline{z}}(x_n) \right]}
         {(z_1 - z_2) \cdots (z_{n-1} - z_n) (z_n - z_1)} 
    +
    \frac{\tr\left[A_{z}(x_1) A_{z}(x_2) \cdots A_{z}(x_n) \right]}
         {(\overline{z}_1 - \overline{z}_2) \cdots (\overline{z}_{n-1} - \overline{z}_n) (\overline{z}_n - \overline{z}_1)}
    \right\}
\end{split}
\end{equation}
In the above, we integrate over all $x_1, \cdots, x_n \in \R^2$.
The first term in the brackets comes from the `left' half of the spinor field and the second term comes from the `right' half.

A degenerate case is $n=1$ where formally, the expression is 
\begin{equation}
    \tr \left[ \frac{1}{\slashed{\partial}} \slashed{A} \right] 
    \stackrel{\textbf{form}}{=} 
    \tr \left[\frac{1}{\partial_{\overline{z}}} A_{\overline{z}}\right] + \tr \left[\frac{1}{\partial_z} A_z\right]
    \stackrel{\textbf{form}}{=} 
    \frac{1}{\pi} \int d^2 x_1 \left\{ 
    \frac{\tr\left[A_{\overline{z}}(x_1)\right]}{z_1 - z_1} +
    \frac{\tr\left[A_{z}(x_1)\right]]}{\overline{z}_1 - \overline{z}_1}
    \right\}.
\end{equation}
A priori this seems ill-defined. However, a `point-splitting regularization' that smooths out the Green's functions $\frac{1}{\partial_z}$,$\frac{1}{\partial_{\overline{z}}}$ by bump functions $\mathcal{B}_{\epsilon}(x,x')$ of shrinking width $\epsilon$ will create `regulated' Green's functions. We will choose symmetric, normalized bump functions. For example, a Gaussian kernel $\mathcal{B}_{\epsilon}(x,x') = \frac{1}{\pi \epsilon} e^{-|x-x'|^2 / \epsilon}$ serves our purpose. This gives a regulated Green's function
\begin{equation}
    \left( \frac{1}{\partial_{\overline{z}}} \right)_{\textbf{reg}}(x_1,x_2) = \lim_{\epsilon \to 0^+} \int d^2 x' \frac{1}{\partial_{\overline{z}}}(x_1,x') \mathcal{B}_{\epsilon}(x',x_2) 
    = 
    \begin{cases}
        \frac{1}{\pi} \frac{1}{z_1-z_2}, &\,\,\text{ if } z_1 \neq z_2 \\
        0,                               &\,\,\text{ if } z_1 = z_2 \\
    \end{cases}
\end{equation}
Similarly, $\left( \frac{1}{\partial_z} \right)_{\textbf{reg}}$ agrees with the unregulated version away from the diagonal and returns zero on the diagonal. This gives that the $n=1$ term vanishes.
\begin{equation}
    \tr \left[ \frac{1}{\slashed{\partial}} \slashed{A} \right] \stackrel{\textbf{reg}}{=} 
    \int d^2 x
    \left\{
        \tr\left[ A_{\overline{z}}(x) \right]  \left( \frac{1}{\partial_{\overline{z}}} \right)_{\textbf{reg}}(x,x) +
        \tr\left[ A_{z}(x) \right]  \left( \frac{1}{\partial_{z}} \right)_{\textbf{reg}}(x,x)
    \right\}
    = 0,
\end{equation}
where $\stackrel{\textbf{reg}}{=}$ refers to this quantity vanishing in this regulated scheme.
In total, this will give a series expansion
\begin{equation}
\begin{split}
    \log \frac{\det \left( \slashed{\partial} - \slashed{A} \right)}{\det\left(\slashed{\partial}\right)}
    &\stackrel{\textbf{form}}{=} 
    \log \frac{\det \left( \partial_{\overline{z}} - A_{\overline{z}} \right)}{\det \left( \partial_{\overline{z}} \right)}
    +
    \log \frac{\det \left( \partial_{z} - A_{z} \right)}{\det \left( \partial_{z} \right)}
    \\
    &\stackrel{\textbf{reg}}{=}
    -\sum_{n=2}^{\infty}
    \frac{1}{n}
    \frac{1}{\pi^n} \int d^2 x_1 \cdots \int d^2 x_n \left\{ 
    \frac{\tr\left[A_{\overline{z}}(x_1) A_{\overline{z}}(x_2) \cdots A_{\overline{z}}(x_n) \right]}
         {(z_1 - z_2) \cdots (z_n - z_1)} 
    +
    \frac{\tr\left[A_{z}(x_1) A_{z}(x_2) \cdots A_{z}(x_n) \right]}
         {(\overline{z}_1 - \overline{z}_2) \cdots (\overline{z}_n - \overline{z}_1)} 
    \right\}
\end{split}
\end{equation}

Before, we said that the above series expansion is in fact not gauge invariant. First, we can't expect either of the parts $\log \frac{\det \left( \partial_{z} - A_{z} \right)}{\det \left( \partial_{z} \right)}$ and $\log \frac{\det \left( \partial_{\overline{z}} - A_{\overline{z}} \right)}{\det \left( \partial_{\overline{z}} \right)}$ to be gauge invariant. This is because general gauge fields $A_{z}$ and $A_{\overline{z}}$ can be written as
\begin{equation}
    A_{z} = -M \partial_{z} M^{-1} 
    \quad,\quad 
    A_{\overline{z}} = -\overline{M} \partial_{\overline{z}} \overline{M}^{-1} 
\end{equation}
for some independent matrix-valued fields $M, \overline{M}$.
\footnote{
This follows because for a gauge field $H$, the operators $\partial_z - H_z$ and $\partial_{\overline{z}} - H_{\overline{z}}$ are invertible. For example, we have a series 
\begin{equation}
\begin{split}
    \frac{1}{\partial_{\overline{z}} - H_{\overline{z}}}(x,x') 
    &= 
    \frac{1}{\partial_{\overline{z}}}(x,x') 
    + \left( \frac{1}{\partial_{\overline{z}}} H_{\overline{z}} \frac{1}{\partial_{\overline{z}}} \right)(x,x') 
    + \left( \frac{1}{\partial_{\overline{z}}} H_{\overline{z}} \frac{1}{\partial_{\overline{z}}} H_{\overline{z}} \frac{1}{\partial_{\overline{z}}} \right)(x,x')
    + \cdots
    \\
    &= 
    \frac{1}{\pi} \frac{1}{z-z'} 
    + \frac{1}{\pi^2} \int d^2 x_1 \frac{\tr\left[H_{\overline{z}}(x_1)\right]}{(z-z_1)(z_1-z')}
    + \frac{1}{\pi^3} \int d^2 x_1 \int d^2 x_2 \frac{\tr\left[H_{\overline{z}}(x_1)H_{\overline{z}}(x_2)\right]}{(z-z_1)(z_1-z_2)(z_2-z')}
    + \cdots
\end{split}
\end{equation}
which satisfies $\left(\left(\partial_{\overline{z}} - H_{\overline{z}}\right) \frac{1}{\partial_{\overline{z}} - H_{\overline{z}}}\right)(x,x') = \delta(x,x')$. 
As such, the flatness condition $\partial_{\overline{z}} \mathcal{A}_z -  [A_{\overline{z}}, \mathcal{A}_{z}] = \partial_z A_{\overline{z}}$ can be solved for $\mathcal{A}_z$, since $\partial_{\overline{z}} -  [A_{\overline{z}}, \boldsymbol{\cdot}]$ is a connection in the adjoint representation. Since the gauge field $(\mathcal{A}_z, A_{\overline{z}})$ is flat on $\CP^1$, $\overline{M}$ can be expressed as the path-ordered exponential with respect to the gauge field. Parallel manipulations work to find the matrix $M$.
}
So, each individual component $A_z, A_{\overline{z}}$ can be written as a gauge transformation of the trivial field, even though the full field may not. As such, if $\log \frac{\det \left( \partial_{z} - A_{z} \right)}{\det \left( \partial_{z} \right)}$ and $\log \frac{\det \left( \partial_{\overline{z}} - A_{\overline{z}} \right)}{\det \left( \partial_{\overline{z}} \right)}$ were gauge invariant, they'd both be trivial. 

At this point, we want to evaluate the variation of each of the summands $\log \frac{\det \left( \partial_{z} - A_{z} \right)}{\det \left( \partial_{z} \right)}$ and $\log \frac{\det \left( \partial_{\overline{z}} - A_{\overline{z}} \right)}{\det \left( \partial_{\overline{z}} \right)}$ with respect to an infinitesimal gauge transformation $e^{\delta u(x)}$, where $\delta u$ is small. Under this transformation $A \to e^{\delta u} A e^{-\delta u} - e^{\delta u}\,d\,e^{-\delta u}$, we have that the gauge field varies as
\begin{equation}
    \delta A_{z} = [\delta u, A_{z}] + \partial_{z} \delta u 
    \quad,\quad 
    \delta A_{\overline{z}} = [\delta u, A_{\overline{z}}] + \partial_{\overline{z}} \delta u 
\end{equation}
at first order in $\delta u$. Let's first examine the variation in the expansion $\log \frac{\det \left( \partial_{\overline{z}} - A_{\overline{z}} \right)}{\det \left( \partial_{\overline{z}} \right)}$. At first order, we'll have
\begin{equation}
\begin{split}
    \delta & \log \frac{\det \left( \partial_{\overline{z}} - A_{\overline{z}} \right)}{\det \left( \partial_{\overline{z}} \right)}
    \stackrel{\textbf{reg}}{=}
    -\sum_{n=2}^{\infty}
    \frac{1}{\pi^n} \int d^2 x_1 \cdots \int d^2 x_n 
    \frac{\tr\left[ \left([\delta u(x_1), A_{\overline{z}}(x_1)] + \partial_{\overline{z}_1} \delta u (x_1) \right) A_{\overline{z}}(x_2) \cdots A_{\overline{z}}(x_n) \right]}
    {(z_1 - z_2) \cdots (z_n - z_1)}
    \\
    &=
    -\sum_{n=2}^{\infty} \frac{1}{\pi^n} \int d^2 x_1 \cdots \int d^2 x_n 
    \tr \left[
      \delta u(x_1) \left\{
      \Big[ A_z(x_1) \,,\, \frac{ A_{\overline{z}}(x_2) \cdots A_{\overline{z}}(x_n)}{(z_1 - z_2) \cdots (z_n - z_1)} \Big] 
      \, - \, 
      \partial_{\overline{z}_1} \left(\frac{ A_{\overline{z}}(x_2) \cdots A_{\overline{z}}(x_n)}{(z_1 - z_2) \cdots (z_n - z_1)}\right)\right\} 
    \right]
    \\
    &=
    -\sum_{n=2}^{\infty} \frac{1}{\pi^n} \int d^2 x_1 \cdots \int d^2 x_n 
    \tr \left[
      \delta u(x_1) \left\{
      \frac{ A_{\overline{z}}(x_1) A_{\overline{z}}(x_2) \cdots A_{\overline{z}}(x_n)}{(z_1 - z_2) \cdots (z_n - z_1)}
      -
      \frac{ A_{\overline{z}}(x_2) \cdots A_{\overline{z}}(x_n) A_{\overline{z}}(x_1) }{(z_1 - z_2) \cdots (z_n - z_1)}
      \right\} 
    \right]
    \\
    &\quad-
    \frac{1}{\pi^2} \int d^2 x_1 \int d^2 x_2 \,\, 
    \tr \left[ \delta u (x_1) \, \partial_{\overline{z}_1} \frac{A_{\overline{z}}(x_2)}{(z_1-z_2)^2} \right]
    \\
    &\quad+
    \sum_{n'=3}^{\infty} \frac{1}{\pi^{n'}} \pi \int d^2 x_2 \cdots \int d^2 x_{n'}
    \tr \left[
      \delta u(x_2)
      \frac{ A_{\overline{z}}(x_2) \cdots A_{\overline{z}}(x_{n'})}{(z_2 - z_3) \cdots (z_{n'} - z_2)}
      -
      \delta u(x_{n'})
      \frac{ A_{\overline{z}}(x_2) \cdots A_{\overline{z}}(x_{n'}) }{(z_{n'} - z_2)(z_2 - z_3) \cdots (z_{n'-1} - z_{n'})}
    \right]
    \\
    &= - \frac{1}{\pi} \int d^2 x \, \tr \left[ \left( \partial_{z} \delta u \right) A_{\overline{z}} \right]
\end{split}
\end{equation}
The first equality follows from the product rule and a cyclic reordering of each of the $n$ terms, which gets rid of the $\frac{1}{n}$ factor. The second equality is two steps. First, we rearrange the commutator using the identity $\tr([\mathcal{A},\mathcal{B}] \mathcal{C}) = \tr(\mathcal{A}[\mathcal{B},\mathcal{C}])$. 
Next, we integrate the other summand by parts, noting that $\int_{\mathcal{R}} d^2 x \partial_z(\cdots) = -2 i \int_{\partial \mathcal{R}} d \overline{z} (\cdots)$ is a boundary term, which drops away if the integrand decays fast enough. 
The third equality does some rearranging uses a partial fractions decomposition together with the equality $\partial_{\overline{z}_1} \frac{1}{z_1 - z'} = \pi \delta^2(x_1 - x')$ before integrating over $x_1$. The $n'=2$ term is a special case and will be handled separately next. 
In the fourth equality, notice that in the infinite sums above, pairs of terms with index-labels $n,n'$ with $n'=n+1$ cancel each other exactly, so only the $n'=2$ term remains. Then, we use the fact that $\partial_{\overline{z}_1} \frac{1}{(z_1 - z_2)^2} = -\partial_{\overline{z}_1} {\partial}_{z_1} \frac{1}{(z_1-z_2)} = -\pi \partial_{z_1} \delta^2(x_1 - x_2)$. Then, the equality follows from integration-by-parts and integrating over $x_1$. Similarly, we can express
\begin{equation}
    \delta \log \frac{\det \left( \partial_{z} - A_{z} \right)}{\det \left( \partial_{z} \right)}
    \stackrel{\textbf{reg}}{=} -\frac{1}{\pi}\int d^2 x \, \tr \left[ \left( \partial_{\overline{z}} \delta u \right) A_{z} \right]
\end{equation}
so that the full variation of the regulated series expansion
\begin{equation}
    \delta \log \frac{\det \left( \slashed{\partial} - \slashed{A} \right)}{\det\left(\slashed{\partial}\right)} 
    \stackrel{\textbf{reg}}{=}
    -\frac{1}{\pi}\int d^2 x 
    \left\{
      \tr \left[ \left( \partial_{z} \delta u \right) A_{\overline{z}} \right] +
      \tr \left[ \left( \partial_{\overline{z}} \delta u \right) A_{z} \right]
    \right\}
    = -\frac{1}{\pi}\delta \int d^2 x \, \tr \left[ A_{\overline{z}} A_z \right].
\end{equation}
As such, the regulated series expansion is \textit{not} gauge invariant which can be said to be a \textit{gauge anomaly}. However, the addition of a \textit{local counterterm} $\frac{1}{\pi} \int d^2 x \, \tr \left[ A_{\overline{z}} A_z \right]$ to the series expansion will indeed give us a series for the determinant that's invariant under infinitesimal gauge transformations. 

As such, the series should be the same for `small gauge transformations' that are connected to the identity. However, for `large gauge transformations' that aren't connected to the identity, we can't expect that these two determinants are the same. This behavior is also the case for the Jimbo-Miwa-Ueno tau function. 
For example, the connections $A_z = \alpha \left(\frac{1}{z} - \frac{1}{z-1}\right), A_{\overline{z}} = 0$ and $A'_z = (\alpha + 1) \left(\frac{1}{z} - \frac{1}{z-1}\right), A'_{\overline{z}} = 0$ give the same monodromy representation. But, their corresponding Jimbo-Miwa-Ueno equations Eq.~\eqref{eq:Jimbo_Miwa_Ueno_eqs} don't match. Note that the gauge transformation connecting these two connections looks like $g(x) = \frac{z}{1-z}$ which is zero or singular or zero at $z=0,1$. So, these large gauge transformations cannot be extended to the punctures of the sphere.

In total, adding this counterterm gives us our final expression for the series expansion
\begin{equation} \label{eq:diracDeterminant_full}
\begin{split}
    \log \frac{\det \left( \slashed{\partial} - \slashed{A} \right)}{\det\left(\slashed{\partial}\right)} 
    \stackrel{\textbf{corr}}{=}& 
    -\sum_{m=2}^{\infty}
    \frac{1}{n}
    \frac{1}{\pi^n} \int d^2 x_1 \cdots \int d^2 x_n \left\{ 
    \frac{\tr\left[A_{\overline{z}}(x_1) \cdots A_{\overline{z}}(x_n) \right]}
         {(z_1 - z_2) \cdots (z_n - z_1)} 
    +
    \frac{\tr\left[A_{z}(x_1) \cdots A_{z}(x_n) \right]}
         {(\overline{z}_1 - \overline{z}_2) \cdots (\overline{z}_n - \overline{z}_1)} 
    \right\}
    \\
    &+ \frac{1}{\pi} \int d^2 x \, \tr \left[ A_{\overline{z}} A_z \right]
\end{split}
\end{equation}
where $\stackrel{\textbf{corr}}{=}$ refers to a `corrected' determinant, which we henceforth take to be the definition of the Dirac determinant's series expansion. 

We note that this quadratic term $A_z A_{\overline{z}}$ mixes the holomorphic and antiholomoprhic sectors, like a mass term would. Indeed, one can find that this term can be derived from doing this whole story with the massive Dirac operator $\slashed{\partial}+m$ and investigating the limit $m \to 0$. 

\subsubsection{Dirac Determinants of Abelian connections} \label{app:dirac_det_abelian}

As an example where we can evaluate the above series, let's take $A$ to be an abelian gauge field, so that $[A_{\overline{z}}(x),A_{\overline{z}}(x')]=[A_{z}(x),A_{z}(x')]=[A_{\overline{z}}(x),A_{z}(x')]=0$ for all $x,x'$.
Since $A$ is abelian, we can symmetrize over the order of the variables in the trace, or equivalently make the replacement
\begin{equation*}
    \frac{\tr\left[A_{\overline{z}}(x_1) A_{\overline{z}}(x_2) \cdots A_{\overline{z}}(x_n) \right]}
         {(z_1 - z_2) \cdots (z_n - z_1)} 
    \to \frac{1}{n!} \tr\left[A_{\overline{z}}(x_1) A_{\overline{z}}(x_2) \cdots A_{\overline{z}}(x_n) \right] \sum_{\sigma \in \mathrm{Sym}(n)} \frac{1}{(z_{\sigma(1)} - z_{\sigma(2)}) (z_{\sigma(2)} - z_{\sigma(3)}) \cdots (z_{\sigma(n)} - z_{\sigma(1)})}
\end{equation*}
(and similarly for the other term $z \leftrightarrow \overline{z}$) in the expression where $\mathrm{Sym}(n)$ is the symmetric group of order $n$.
Then by the next lemma, it turns out that all terms in the series for $n \ge 3$ vanish. 
\begin{lem}
    Let be $n \ge 3$. Then for $z_1, \cdots, z_n$ distinct,
    \begin{equation} \label{eq:permSumLem}
        \sum_{\sigma \in \mathrm{Sym}(n)} \frac{1}{(z_{\sigma(1)} - z_{\sigma(2)}) (z_{\sigma(2)} - z_{\sigma(3)}) \cdots (z_{\sigma(n)} - z_{\sigma(1)})} = 0
    \end{equation}
\end{lem}
A clean proof is done similarly to a similar computation in~\cite{boutillierThesis2005}.
\begin{proof}
    Let $f(z_1,\cdots,z_n)$ be the left-hand-side of Eq.~\eqref{eq:permSumLem}. 
    Let $V(z_1,\cdots,z_n) = \prod_{1 \le i < j \le n} (z_i - z_j)$ be the Vandermonde determinant.
    Note that $f$ is manifestly a symmetric function in $z_1,\cdots,z_n$, whereas $V$ is antisymmetric. This means that $f \cdot V$ is an antisymmetric function. Moreover, every term in $f \cdot V$ is a polynomial in $z_1, \cdots, z_n$ of lower degree than $V$. Since $V$ is the lowest degree non-constant antisymmmetric polynomial in $z_1, \cdots, z_n$, this means that $f \cdot V$ must be a constant. For $n>3$, every term has degree greater than zero, so $f \cdot V$ must be exactly zero, implying that $f$ vanishes. The vanishing for $n=3$ can be verified directly.
\end{proof}

So for abelian gauge fields, we have the expression
\begin{equation} \label{eq:diracDet_abelian}
\begin{split}
    &\log \frac{\det \left( \slashed{\partial} - \slashed{A} \right)}{\det\left(\slashed{\partial}\right)} 
    \stackrel{\textbf{corr}}{=}
    -\frac{1}{2}
    \frac{1}{\pi^2} \int d^2 x_1 \int d^2 x_2 \left\{ 
    \frac{\tr\left[ A_{\overline{z}}(x_1) A_{\overline{z}}(x_2) \right]}
         {(z_1 - z_2) (z_2 - z_1)} 
    +
    \frac{\tr\left[ A_{z}(x_1) A_{z}(x_2) \right]}
         {(\overline{z}_1 - \overline{z}_2) (\overline{z}_2 - \overline{z}_1)} 
    \right\} + \frac{1}{\pi} \int d^2 x \, \tr \left[ A_{\overline{z}} A_z \right]
    \\
    &=
    +\frac{1}{2}
    \frac{1}{\pi^2} \int d^2 x_1 \int d^2 x_2 \Big\{ 
    \tr\left[ A_{\overline{z}}(x_1) A_{\overline{z}}(x_2) \right] \partial_{z_1}\partial_{z_2} \log(|x_1-x_2|^2)
    +
    \tr\left[ A_{z}(x_1) A_{z}(x_2) \right] \partial_{\overline{z}_1}\partial_{\overline{z}_2} \log(|x_1-x_2|^2)
    \\ &\quad\quad\quad\quad\quad\quad\quad\quad\quad\quad
     - \tr \left[ A_{\overline{z}}(x_1) A_z(x_2) \right] \partial_{z_1} \partial_{\overline{z_2}} \log(|x_1-x_2|^2) 
     - \tr \left[ A_z(x_1) A_{\overline{z}}(x_2) \right] \partial_{\overline{z_1}} \partial_{z_2} \log(|x_1-x_2|^2) 
     \Big\}
    \\
    &=
    -\frac{1}{8 \pi^2} \int d^2 x_1 \int d^2 x_2 \big\{ \log(|x_1-x_2|^2) \, \tr[ F(x_1) F(x_2)] \big\}
\end{split}
\end{equation}
where $F = \partial_{y} A_{x} - \partial_{x} A_{y} = -2 i (\partial_{\overline{z}} A_{z} - \partial_{z} A_{\overline{z}}) $ is the field strength of the abelian connection. 
In the second equality, we use the facts $\partial_z \log(|x-x'|^2) = \frac{1}{z-z'}$ so $\partial_{\overline{z}'}\partial_{z} \log(|x-x'|^2) = \pi \delta^2(x-x')$ so that 
\begin{equation*}
    \int d^2 x A_{\overline{z}} A_{z} = \int d^2 x_1 \int d^2 x_2 A_{\overline{z}}(x_1) A_{z}(x_2) \delta^2(x_1-x_2) = \frac{1}{\pi} \int d^2 x_1 \int d^2 x_2 A_{\overline{z}}(x_1) A_{z}(x_2) \partial_{\overline{z}_1}\partial_{z_2} \log(|x_1 - x_2|^2).
\end{equation*}
In the third equality, we do various integrations-by-parts and reorganize terms. For $A$ a $\C^{\times}$ connection, this final result is the same result one would get from coupling the Gaussian Free Field (i.e. free boson) to the field strength $F$, which can be expressed as
\begin{equation}
\begin{split}
    &\log \frac{\det \left( \slashed{\partial} - \slashed{A} \right)}{\det\left(\slashed{\partial}\right)} 
    =
    -\frac{1}{8 \pi^2} \int d^2 x_1 \int d^2 x_2 \big\{ \log(|x_1-x_2|^2) \,  F(x_1) F(x_2) \big\}
    =
    \log \frac{Z_{GFF}(F)}{Z_{GFF}(0)}
    \\
    &\text{where } \quad
    Z_{GFF}(F) = \int [d \phi] e^{ - \int d^2 x \left\{ {\pi}(\partial \phi)^2 + F(x) \phi(x) \right\} }.
\end{split}
\end{equation}
Such a result was known as early as 1962 by Schwinger~\cite{schwinger1962} and is an example of \textit{bosonization}, where in 2D, certain bosonic and fermion correlators can be equated.

\subsection{`Fuschian' expansion of $\log\tau$ and the Jimbo-Miwa-Ueno equations} \label{app:tau_fuschian}

Now, we apply our formalism above to compute the determinant with respect to Fuschian connections and show how the expansion satisfies the Jimbo-Miwa-Ueno equations.
Above, we had to manipulate various infinities to justify a gauge-invariant series expansion. Here, we will encounter another set of infinities that come from the singular nature of the connection. To handle these singularities, we will instead consider a `smeared out' determinant where the singular connection is smoothed by a kernel of width $\epsilon$ to become a regular connection and consider the expansion around $\epsilon \to 0$.
This will give functions that diverge logarithmically with $\epsilon$. 
We conjecture that after appropriately subtracting off the divergence, the limiting series expansion indeed converges to the modified tau function.

It will be instructive to start out using our previous results above to handle the case of abelian connections, many of whose lessons with the divergences can be transferred onto the nonabelian case. 

Consider a singular abelian connection $A_z = \sum_{j} \frac{b_j}{z-p_j}$, $A_{\overline{z}} = 0$ with each $b_j \in \C$ satisfying $\sum_{j} b_j = 0$. This connection has monodromy $e^{2 \pi i b_j}$ around the puncture $p_j$. This will give us a field strength of $F = 2 \pi i \sum_{j} b_j \delta^2(x-x_j)$. Proceeding naively, this will give an expression
\begin{equation} \label{eq:abelian_singular_dirac_det}
    \log \frac{\det \left( \slashed{\partial} - \slashed{A} \right)}{\det\left(\slashed{\partial}\right)} 
    =
    -\sum_{i,j} \frac{1}{8 \pi^2} (2 \pi i b_i) (2 \pi i b_j) \log(|p_i - p_j|^2)
    = 
    +\sum_{i<j} b_i b_j \log(|p_i - p_j|^2)  + \sum_{j} \frac{1}{2} b_{j}^2 \log(|0|^2).
\end{equation}
The infiniteness of the $\log(|0|^2)$ terms in this expression can be traced back to the fact that our gauge fields were singular. However, note that if we treat this infinity as a constant, the tau function agrees exactly with the Jimbo-Miwa-Ueno equations Eq.~\eqref{eq:Jimbo_Miwa_Ueno_eqs}. 

To actually get a finite expression, one thing we could do is to `smooth out' the connection by convolving it with a Gaussian $\mathcal{B}_{\epsilon}(x,x') = \frac{1}{\pi \epsilon} e^{-|x-x'|^2/\epsilon}$ and considering the expansion as $\epsilon \to 0$. In this limit, the terms $\sum_{i<j} b_i b_j \log(|p_i - p_j|^2)$ will be preserved up to lower-order terms in $\epsilon$. However, the $\log(|0|^2)$ turns into
\begin{equation} \label{eq:logZero_reg}
\begin{split}
\log(|0|^2) \to
    \int d^2 x_1 \int d^2 x_2  \log(|x_1-x_2|^2) \frac{e^{-\frac{x_1^2 + x_2^2}{\epsilon}}}{(\pi \epsilon)^2}
    = \left( \int d^2 \tilde{x}_1 \log(|\tilde{x}_1|^2)  \frac{e^{-\frac{\tilde{x}_1^2}{2 \epsilon}}}{\pi \epsilon} \right)
      \left( \int d^2 \tilde{x}_2 \frac{e^{-2\frac{\tilde{x}_2^2}{\epsilon}}}{\pi \epsilon} \right) 
    = \log(2 \epsilon) - \gamma_{\mathrm{Euler}}
\end{split}
\end{equation}
where $\tilde{x}_1 = x_1 - x_2$, $\tilde{x}_2 = \frac{1}{2}(x_1 + x_2)$. This still a constant with respect to the $\{p_1 \cdots p_k\}$ that diverges logarithmically with $\epsilon$. This gives a smeared determinant
\begin{equation}
    \left[
    \log \frac{\det \left( \slashed{\partial} - \slashed{A} \right)}{\det\left(\slashed{\partial}\right)} 
    \right]_{\textbf{smeared},\epsilon}
    =
    \sum_{i<j} b_i b_j \log(|p_i - p_j|^2)  + \sum_{j} \frac{1}{2} b_{j}^2 \left( \log(2 \epsilon) - \gamma_{\mathrm{Euler}} \right) + O(\epsilon)
\end{equation}
for abelian gauge fields. As such,
\begin{equation*}
    \lim_{\epsilon \to 0} \left\{ \left[ \log \frac{\det \left( \slashed{\partial} - \slashed{A} \right)}{\det\left(\slashed{\partial}\right)} \right]_{\textbf{smeared},\epsilon} - \sum_{j} \frac{1}{2} b_{j}^2 \log( \epsilon) \right\}
\end{equation*}
is well-defined and can be viewed as the tau function agreeing with the Eqs.~\eqref{eq:Jimbo_Miwa_Ueno_eqs} for abelian fields. In the next section, we argue that the series expansion for nonabelian gauge fields corresponds to the tau function. 

Let's generalize the calculation above to non-abelian Fuschian systems.
Here, $A_{\overline{z}}=0$ and $A_z = \sum_{j} \frac{B_j}{z-p_j}$ for matrices $\{B_j\}$ satisfying $\sum_{j} B_j = 0$, where the $\{B_j\}$ depend on the $\{p_j\}$ via the Schlesinger Equations~\eqref{eq:schlesinger_2}. In this case, our series expansion becomes
\begin{equation}
\begin{split}
    \log \frac{\det \left( \slashed{\partial} - \slashed{A} \right)}{\det\left(\slashed{\partial}\right)} 
    \stackrel{\textbf{corr}}{=}& 
    -\sum_{n=2}^{\infty}
    \frac{1}{n}
    \frac{1}{\pi^n} 
    \int d^2 x_1 \cdots \int d^2 x_n \left\{
    \frac{\tr\left[ 
        \left( \sum_{i} \frac{B_i}{z_1 - p_i} \right)
        \cdots
        \left( \sum_{i} \frac{B_i}{z_n - p_i} \right)
    \right]}
         {(\overline{z}_1 - \overline{z}_2) \cdots (\overline{z}_n - \overline{z}_1)} 
    \right\} 
    =: \sum_{n=2}^{\infty} \mathcal{I}_n
\end{split}
\end{equation}
where $\mathcal{I}_n$ is the $n$th term in the expansion. To compare to the Jimbo-Miwa-Ueno equations Eq.~\eqref{eq:Jimbo_Miwa_Ueno_eqs}, we take the derivative with respect the $p_j$. First, we look at the $n=2$ term. This calculation ends up working the same as the abelian case. In particular after a Gaussian smoothing by $\mathcal{B}_{\epsilon}$, we get
\begin{equation}
    \mathcal{I}_2 = \sum_{i<j} \tr\left[B_i B_j\right] \log(|p_i - p_j|^2)  + \sum_{j} \frac{1}{2} \tr\left[B_{j}^2\right] \left( \log(2 \epsilon) - \gamma_{\mathrm{Euler}} \right) + O(\epsilon)
\end{equation}
so that as $\epsilon \to 0$
\begin{equation}
\begin{split}
    \partial_{p_k} \mathcal{I}_2 
    &= 
    \sum_{j \neq k} \frac{\tr(B_k B_j)}{p_k - p_j}
    +
    \sum_{i \neq j} \tr\left[ \left( \frac{[B_k,B_i]}{p_k - p_i} - \delta_{k,i}\sum_{\ell} \frac{[B_k,B_{\ell}]}{p_k - p_{\ell}} \right) B_j\right] \log(|p_i - p_j|^2)
    \\
    &=
    \sum_{j \neq k} \frac{\tr(B_k B_j)}{p_k - p_j}
    -
    \int d^2 x_1 \int d^2 x_2 \,\,
    \frac{\tr\left[ 
        \left(
            \sum_{\ell} \frac{[B_k,B_{\ell}]}{p_k - p_{\ell}} \left(\frac{1}{z_1 - p_\ell} - \frac{1}{z_1 - p_k}\right)
        \right)
        \left( \sum_{i} \frac{B_i}{z_2 - p_i} \right)
    \right]}
         {(\overline{z}_1 - \overline{z}_2)(\overline{z}_2 - \overline{z}_1)}
    =:
    \sum_{j \neq k} \frac{\tr(B_k B_j)}{p_k - p_j} + \delta_{k}^{(2)} \mathcal{I}_{2}.
\end{split}
\end{equation}
In the first equality, the first term comes from the derivative of the `$\log(|p_i - p_j|^2)$' terms and the second comes from the Schlesinger equations and the variation of $\{B_j\}$, noting that $\tr\left[B_j^2\right]$ remains invariant. The second equality just rewrites the second term, which we define as $\delta_{k}^{(2)} \mathcal{I}_{2}$ in parallel with future notation. (See also the following derivation for an explanation for this rewriting of $\delta_{k}^{(2)} \mathcal{I}_{2}$.)

Note that the first term in the equality is what we want for the Jimbo-Miwa-Ueno equations. We will see that $\delta_{k}^{(2)} \mathcal{I}_{2}$ gets cancelled out by higher terms. Now, we write the derivative of $\mathcal{I}_n$ for $n>2$:
\begin{equation}
\begin{split}
    \partial_{p_k} \mathcal{I}_n
    &=
    -\frac{1}{\pi^n} 
    \int d^2 x_1 \cdots \int d^2 x_n
    \frac{\tr\left[ 
        \left( 
            \frac{B_k}{(z_1 - p_k)^2}
            +
            \sum_{i} \left( \frac{[B_k,B_i]}{p_k - p_i} - \delta_{k,i}\sum_{\ell} \frac{[B_k,B_{\ell}]}{p_k - p_{\ell}} \right) \frac{1}{z_1 - p_i}
        \right)
        \left( \sum_{i} \frac{B_i}{z_2 - p_i} \right)
        \cdots
        \left( \sum_{i} \frac{B_i}{z_n - p_i} \right)
    \right]}
         {(\overline{z}_1 - \overline{z}_2)(\overline{z}_2 - \overline{z}_3) \cdots (\overline{z}_n - \overline{z}_1)} 
    \\
    &=
    -\frac{1}{\pi^n} 
    \int d^2 x_1 \cdots \,\,
    \frac{\tr\left[ 
        \left( 
            -\partial_{z_1} \frac{B_k}{z_1 - p_k}
        \right)
        \left( \sum_{i} \frac{B_i}{z_2 - p_i} \right)
        \cdots
        \left( \sum_{i} \frac{B_i}{z_n - p_i} \right)
    \right]}
         {(\overline{z}_1 - \overline{z}_2)(\overline{z}_2 - \overline{z}_3) \cdots (\overline{z}_n - \overline{z}_1)} 
    \\
    &\quad\quad -
    \frac{1}{\pi^n} 
    \int d^2 x_1 \cdots \,\,
    \frac{\tr\left[ 
        \left(
            \sum_{i} \frac{[B_k,B_i]}{p_k - p_i} \frac{1}{z_1 - p_i} - \sum_{\ell} \frac{[B_k,B_{\ell}]}{p_k - p_{\ell}}  \frac{1}{z_1 - p_k}
        \right)
        \left( \sum_{i} \frac{B_i}{z_2 - p_i} \right)
        \cdots
        \left( \sum_{i} \frac{B_i}{z_n - p_i} \right)
    \right]}
         {(\overline{z}_1 - \overline{z}_2)(\overline{z}_2 - \overline{z}_3) \cdots (\overline{z}_n - \overline{z}_1)} 
    \\
    &=
    -\frac{1}{\pi^{n-1}} 
    \int d^2 x_2 \cdots \int d^2 x_n \,\,
    \left\{
    \frac{\tr\left[ 
        \left( 
            \frac{B_k}{z_2 - p_k}
        \right)
        \left( \sum_{i} \frac{B_i}{z_2 - p_i} \right)
        \cdots
        \left( \sum_{i} \frac{B_i}{z_n - p_i} \right)
    \right]}
         {(\overline{z}_2 - \overline{z}_3) \cdots (\overline{z}_{n-1} - \overline{z}_n)(\overline{z}_n - \overline{z}_2)} 
    -
    \frac{\tr\left[ 
        \left( 
            \frac{B_k}{z_n - p_k}
        \right)
        \left( \sum_{i} \frac{B_i}{z_2 - p_i} \right)
        \cdots
        \left( \sum_{i} \frac{B_i}{z_n - p_i} \right)
    \right]}
         {(\overline{z}_n - \overline{z}_2)(\overline{z}_2 - \overline{z}_3) \cdots (\overline{z}_{n-1} - \overline{z}_n)}
    \right\} 
    \\
    &\quad\quad -
    \frac{1}{\pi^{n}} 
    \int d^2 x_1 \cdots \int d^2 x_n \,\,
    \frac{\tr\left[ 
        \left(
            \sum_{\ell} \frac{[B_k,B_{\ell}]}{p_k - p_{\ell}} \left(\frac{1}{z_1 - p_\ell} - \frac{1}{z_1 - p_k}\right)
        \right)
        \left( \sum_{i} \frac{B_i}{z_2 - p_i} \right)
        \cdots
        \left( \sum_{i} \frac{B_i}{z_n - p_i} \right)
    \right]}
         {(\overline{z}_1 - \overline{z}_2)(\overline{z}_2 - \overline{z}_3) \cdots (\overline{z}_n - \overline{z}_1)}
    \\
    &=
    +\frac{1}{\pi^{n-1}} 
    \int d^2 x_2 \cdots \int d^2 x_n \,\,
    \frac{\tr\left[
        \left(
            \sum_{\ell} \frac{[B_k,B_{\ell}]}{p_k - p_{\ell}} \left(\frac{1}{z_2 - p_\ell} - \frac{1}{z_2 - p_k}\right)
        \right)
        \left( \sum_{i} \frac{B_i}{z_3 - p_i} \right)
        \cdots
        \left( \sum_{i} \frac{B_i}{z_n - p_i} \right)
    \right]}
         {(\overline{z}_2 - \overline{z}_3) \cdots (\overline{z}_{n-1} - \overline{z}_n)(\overline{z}_n - \overline{z}_2)} 
    \\
    &\quad\quad  -
    \frac{1}{\pi^{n}} 
    \int d^2 x_1 \cdots \int d^2 x_{n} \,\,
    \frac{\tr\left[ 
        \left(
            \sum_{\ell} \frac{[B_k,B_{\ell}]}{p_k - p_{\ell}} \left(\frac{1}{z_1 - p_\ell} - \frac{1}{z_1 - p_k}\right)
        \right)
        \left( \sum_{i} \frac{B_i}{z_2 - p_i} \right)
        \cdots
        \left( \sum_{i} \frac{B_i}{z_n - p_i} \right)
    \right]}
         {(\overline{z}_1 - \overline{z}_2)(\overline{z}_2 - \overline{z}_3) \cdots (\overline{z}_n - \overline{z}_1)}
    \\
    &=: \delta_{k}^{(1)} \mathcal{I}_{n} + \delta_{k}^{(2)} \mathcal{I}_{n},
\end{split}
\end{equation}
where $\delta_{k}^{(1)} \mathcal{I}_n$ and $\delta_{k}^{(2)} \mathcal{I}_n$ are the first and second terms of the second-to-last line above. In the first equality, the cyclicity of the trace gets rid of the $1/n$ factor and lets us put the derivative completely in the first $A_{\overline{z}}(x_1)$ factor. In the second equality, the two terms from the derivatives of the $\left\{\frac{1}{z_1 - p_i}\right\}$ factors and of the $\{B_j\}$ factors respectively are split up, and the first term is rewritten as a derivative with respect $z_1$. The third equality is a few steps. First, the second term is rewritten whose form matches for $n=2$ as well as $n>2$. Then for the first term, we integrate the $\partial_{z_1}$ by parts giving $\partial_{z_1} \frac{1}{(\overline{z}_1 - \overline{z}_2)(\overline{z}_2 - \overline{z}_3) \cdots (\overline{z}_n - \overline{z}_1)} = \pi \left( \delta^2(x_1-x_2) - \delta^2(x_1-x_n) \right)\frac{1}{(\overline{z}_2 - \overline{z}_3) \cdots (\overline{z}_n - \overline{z}_2)}$, which requires $n>2$. Then, integration over $x_1$ gives what is written. In the fourth line, the first term is rewritten by changing variables for one half of the integrand, expanding terms in partial fractions, and recollecting terms.

Note that for these terms, we have $\delta_k^{(2)} \mathcal{I}_n+\delta_k^{(1)} \mathcal{I}_{n+1}=0$ for all $n \ge 2$. As such, the derivative of the series expansion will have all terms telescope away except for $\sum_{j \neq k} \frac{\tr(B_k B_j)}{p_k - p_j}$. This is what we wanted to show.

Above, we didn't deal with the smearing procedure for the terms with $n>2$. In general, we conjecture that the only divergence comes from the $n=2$ term and that after subtracting away the divergence $\sum_{j} \frac{1}{2} \tr\left[B_{j}^2\right] \log(\epsilon)$, the series expansion above (potentially after some resummation procedure) has a well defined limit as $\epsilon \to 0$.

\subsection{`Zipper' expansion of $\log\tau$} \label{app:tau_zipper}

Now, let us use zipper representations of flat gauge fields to expand the Dirac determinant. In this section, we consider a gauge field localized on $s$ non-intersecting path segments $\gamma_1, \cdots, \gamma_s$ which may share endpoints and that's flat away from these endpoints. Then, after orienting the path segments, we assign a matrix of $U_j = e^{u_j} \in GL(N,\C)$ that gives the path-ordered exponential going right-to-left across $\gamma_j$. See Fig.~\ref{fig:zipp_collection} for an example.

\begin{figure}[h!]
  \centering
  \includegraphics[width=0.7\linewidth]{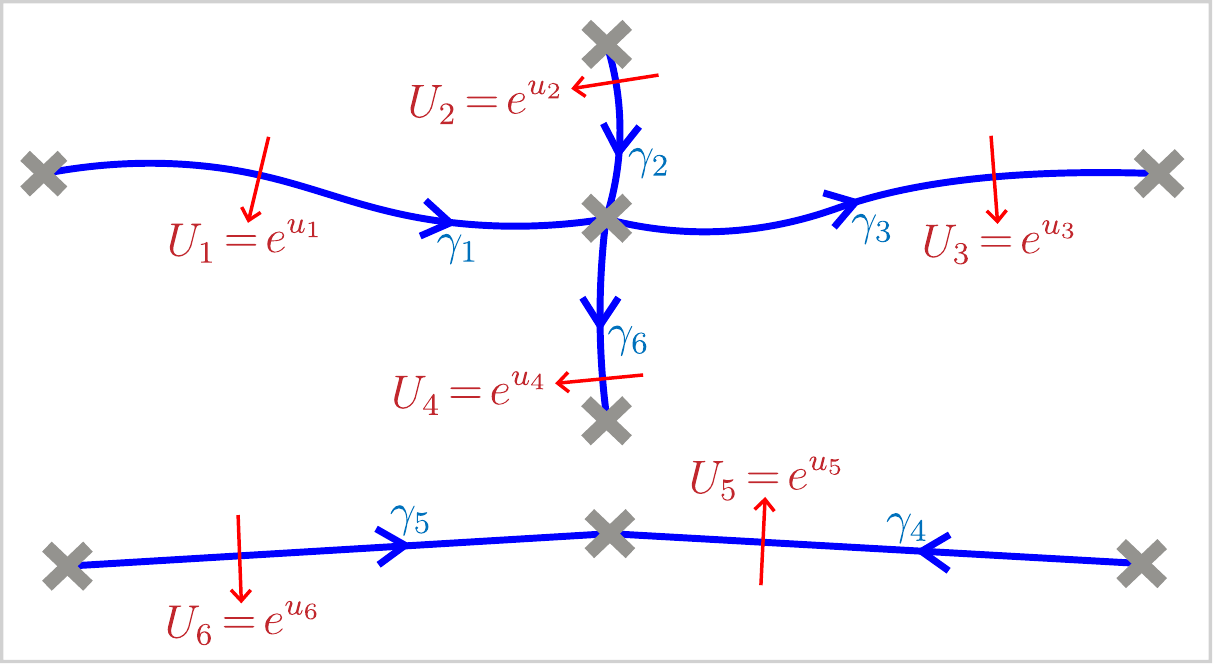}
  \caption{Example zipper representation of flat gauge field.}
  \label{fig:zipp_collection}
\end{figure}

This gauge field corresponds to a connection $d - A$ for a $\gl(N,\C)$-valued 1-form $A$ which we write as
\begin{equation} \label{eq:zipper_gaugeField}
    A = \sum_{j=1}^{s} A^{(j)} 
    \quad\text{where}\quad
    A^{(j)}:= u_j \, \delta^{\perp}(\gamma_j).
\end{equation}
Here, we define the singular 1-form-valued distribution $\delta^{\perp}(\gamma_j)$ whose fundamental property is 
\begin{equation} \label{eq:deltaPerp_defn}
    \int_{\R^2} \delta^{\perp}(\gamma_j) \wedge v := \int_{\gamma_j} v
\end{equation}
for 1-forms $v$. Soon, we'll give a concrete construction of this form.

Say that $\gamma_j$ goes between $q_j \to q'_j$. Then, we claim that the differential of this form acts as a difference of delta functions, that $(d \delta^{\perp}(\gamma_j))(x) = (\delta^2(x - q'_j) - \delta^2(x - q_j)) d^2 x$, where $d^2 x = d x \wedge dy$ is the volume form on $\R^2$. To see this, let $g \in C^1(\C)$ be a smooth test function. Then,
\begin{equation}
    \int_{\R^2} g \, d\delta^{\perp}(\gamma_j) 
    = 
    \int_{\R^2} d(g \, \delta^{\perp}(\gamma_j)) 
    +
    \int_{\R^2} \delta^{\perp}(\gamma_j) \wedge dg
    = 
    \int_{\gamma_j} dg 
    = 
    g(q'_j) - g(q_j) 
    = 
    \int g(x) (\delta^2(x - q'_j) - \delta^2(x - q_j))d^2 x.
\end{equation}
The first equality uses a differential forms identity. The second one uses that $g \, \delta^{\perp}(\gamma_j)$ vanishes at infinity and the definition of $\delta^{\perp}(\gamma_j)$. So, the monodromy of the restricted connection $d -  u_j \, \delta^{\perp}(\gamma_j)$ going counterclockwise around $q_j$ or clockwise around $p'_j$ is $e^{u_j}$. Since the connection is zero away from $\gamma_j$, the path-ordered exponential across $\gamma_j$ is also $e^{u_j}$. 

Now, let's describe the distribution $\delta^{\perp}(\gamma_j)$ as the $\epsilon \to 0$ limit of a collection of `regulated' 1-forms $\delta^{\perp}(\gamma_j)^{\epsilon}$. Let $\gamma_j^{L,\epsilon}$,$\gamma_j^{R,\epsilon}$ be two curves between $p_j \to p'_j$, such that $\gamma_j^{L,\epsilon}$ {\color{RoyalPurple} [resp. $\gamma_j^{R,\epsilon}$]} is on the left {\color{RoyalPurple} [resp. right]} side of $\gamma_j$, and that they're within a distance $\epsilon$ of $\gamma_j$. Then, $\mathcal{R}_j^{\epsilon}$ is a region bounded by $\overline{\gamma_j^{L,\epsilon}} \bigsqcup \gamma_j^{R,\epsilon}$ containing $\gamma_j$. Now, let $f^{\epsilon}$ be a smooth function on $\mathcal{R}^{\epsilon}_j - \{p_j, p'_j\}$ that takes $0 < f(x) < 1$ for all $x \in \mathrm{int}(\mathcal{R}_j^{\epsilon})$ and so that $f(x) = 0$ {\color{RoyalPurple} [resp. $f(x) = 1$]} on $\gamma_j^{L,\epsilon}$ {\color{RoyalPurple} [resp. $\gamma_j^{R,\epsilon}$]}. Then, we define our regulated 1-forms as 
\begin{equation} \label{eq:deltaPerp_reg_def}
    (\delta^{\perp}(\gamma_j)^{\epsilon})(x)
    =
    \begin{cases}
    (d f^{\epsilon})(x) & \text{ if } x \in \mathrm{int}(\mathcal{R}_j^{\epsilon}) \\
    0 & \text{ otherwise}
    \end{cases}.
\end{equation}
See Fig.~\ref{fig:smoothed_deltaPerp_form} for a depiction.

\begin{figure}[h!]
  \centering
  \includegraphics[width=0.9\linewidth]{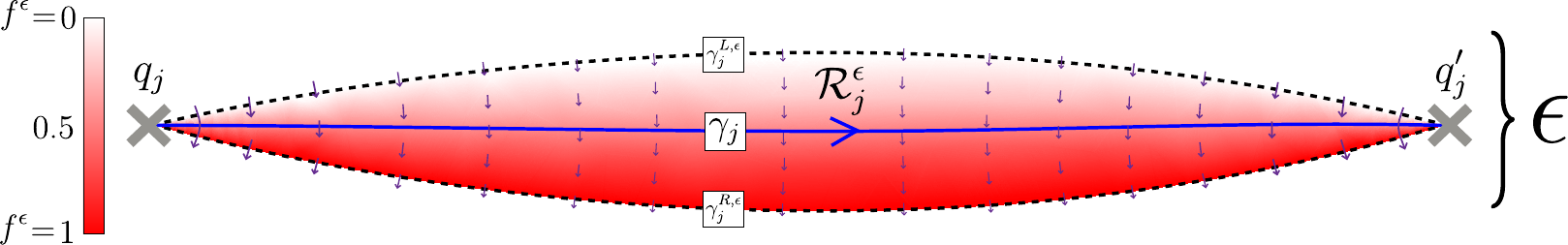}
  \caption{A depiction of $\delta^{\perp}(\gamma_j)^{\epsilon}$ (purple vector field) which is the gradient of the function $f^{\epsilon}$ (the white-red color gradient) as well as the auxiliary objects in the construction. The curves $\gamma_j^{L,\epsilon}$,$\gamma_j^{R,\epsilon}$ are within a distance $\epsilon$ of $\gamma_j$.}
  \label{fig:smoothed_deltaPerp_form}
\end{figure}

Now, we argue that the $\epsilon \to 0$ limit of these forms satisfies the defining equation Eq.~\eqref{eq:deltaPerp_defn}. Let $v$ be a sufficiently regular 1-form. Then,
\begin{equation} \label{eq:deltaPerp_reg_limit}
    \int_{\R^2} \delta^{\perp}(\gamma_j)^{\epsilon} \wedge v 
    = 
    \int_{\mathcal{R}^{\epsilon}_j} d f^{\epsilon} \wedge v
    = 
    -\int_{\mathcal{R}^{\epsilon}_j} f^{\epsilon} d v 
    +\int_{\partial \mathcal{R}^{\epsilon}_j} f \, v \xrightarrow{\epsilon \to 0} \int_{\gamma_j} v.
\end{equation}
The first equality uses that $\delta^{\perp}(\gamma_j)^{\epsilon} = 0$ outside of $\mathcal{R}^{\epsilon}_j$. The second equality uses a differential forms identity and Stokes' theorem. In last equality, first we use $\int_{\partial \mathcal{R}^{\epsilon}_j} f \, v = \int_{\gamma^{R, \epsilon}_j} v$ which tends to $\int_{\gamma_j} v$. Then, with mild regularity conditions on $v$, the term $\int_{\mathcal{R}^{\epsilon}_j} f^{\epsilon} d v$ vanishes as $\epsilon \to 0$. For our cases of interest, we will choose $v$ to be closed which would give an equality to $\int_{\gamma^{R, \epsilon}_j} v$ even for finite $\epsilon$.

To compare to the zipper representation as in Fig.~\ref{fig:flatGaugeFieldToZippers}, note that some of the endpoints $\{q_j, q'_j\}$ will correspond to punctures and an associated monodromy, while others won't be punctures and will just be places where the zippers meet. A subtle point is that there may be a puncture at a point $p$ even though the monodromy around $p$ is trivial. This will happen when the loop going around $p$ traces out a nontrivial loop in $\pi_1(GL(N,\C)) = \Z$. For example, this happens with the abelian connection $A'_{\overline{z}} = 0$ , $A'_{z} = \frac{1}{z} - \frac{1}{z-1}$ for which the monodromies around $p = 0,1$ are trivial but trace out nontrivial paths in $\C - \{0\}$.

Now we want to expand the Dirac determinant Eq.~\eqref{eq:diracDeterminant_full} with the correction Eq.~\eqref{eq:zipper_gaugeField}. First, recall that we can expand the determinant as
\begin{equation} \label{eq:diracExpansion_zipper}
\begin{split}
    &\log \frac{\det \left( \slashed{\partial} - \slashed{A} \right)}{\det\left(\slashed{\partial}\right)}
    \stackrel{\textbf{corr}}{=}
    \tr \log \left( \mathbbm{1} - \frac{1}{\slashed{\partial}} \sum_{j=1}^{s} \slashed{A^{(j)}} \right) 
    + \frac{1}{\pi} \int d^2 x  \left( \sum_{j=1}^{s} A^{(j)}_{\overline{z}} \right)  \left( \sum_{j=1}^{s} A^{(j)}_{z} \right)
    \\&=
    \sum_{j=1}^{s} \left(
        -\frac{1}{2}\frac{1}{\pi^2} \int d^2 x_1 \int d^2 x_2 \left\{ 
            \frac{\tr\left[ A^{(j)}_{\overline{z}}(x_1) A^{(j)}_{\overline{z}}(x_2) \right]}{(z_1 - z_2) (z_2 - z_1)} 
            +
            \frac{\tr\left[ A^{(j)}_{z}(x_1) A^{(j)}_{z}(x_2) \right]}{(\overline{z}_1 - \overline{z}_2) (\overline{z}_2 - \overline{z}_1)} 
        \right\}
        + \frac{1}{\pi} \int d^2 x \, A^{(j)}_{\overline{z}}(x) A^{(j)}_{z}(x)
    \right)
    \\
    &\quad -
    \sum_{n=2}^{\infty} \frac{1}{n} \sum_{\substack{j_1 \cdots j_n \in \{1 \cdots s\} \\ !(j_1 = \cdots = j_n)}}
    \frac{1}{\pi^n} \int d^2 x_1 \cdots \int d^2 x_n\left\{ 
    \frac{\tr\left[A^{(j_1)}_{\overline{z}}(x_1) \cdots A^{(j_n)}_{\overline{z}}(x_n) \right]}
         {(z_1 - z_2) \cdots (z_n - z_1)} 
    +
    \frac{\tr\left[A^{(j_1)}_{z}(x_1) \cdots A^{(j_n}_{z}(x_n) \right]}
         {(\overline{z}_1 - \overline{z}_2) \cdots (\overline{z}_n - \overline{z}_1)} 
    \right\}
\end{split}
\end{equation}
In the first equality we just split up $A$ into its zipper components. In the second equality, we use a few facts. 
First, $A^{(j)}_{\overline{z}}(x) A^{(j)}_{z}(x) = 0$ for $j \neq j'$ since the zippers are disjoint. Then from the same reasoning as the abelian connections, the terms in the expansion $-\frac{1}{n}\tr\left[\left(\frac{1}{\slashed{\partial}} \slashed{A^{(j)}}\right)^n\right]$ for $n>2$ disappear. The other terms are the expressions for $-\frac{1}{n}\tr\left[\frac{1}{\slashed{\partial}} \slashed{A^{(j_1)}} \cdots \frac{1}{\slashed{\partial}} \slashed{A^{(j_n)}}\right]$ for $j_1, \cdots, j_n$ not all equal. 

First, using Eq.~\eqref{eq:diracDet_abelian} and the fact that $d A^{(j)} = u_j (\delta^2(x - q_j) - \delta^2(x - q'_j))d^2 x$
\begin{equation} \label{eq:diracDet_sameZipTerms}
\begin{split}
    &-\frac{1}{2}\frac{1}{\pi^2} \int d^2 x_1 \int d^2 x_2 \left\{ 
            \frac{\tr\left[ A^{(j)}_{\overline{z}}(x_1) A^{(j)}_{\overline{z}}(x_2) \right]}{(z_1 - z_2) (z_2 - z_1)} 
            +
            \frac{\tr\left[ A^{(j)}_{z}(x_1) A^{(j)}_{z}(x_2) \right]}{(\overline{z}_1 - \overline{z}_2) (\overline{z}_2 - \overline{z}_1)} 
        \right\}
        + \frac{1}{\pi} \int d^2 x A^{(j)}_{\overline{z}} A^{(j)}_{z}
    \\
    &= \frac{1}{4 \pi^2} \tr\left[u_j^2\right] \left(\log(|q'_j-q_j|^2) - \log(|0|^2)\right)
\end{split}
\end{equation}
where $\log(|0|^2)$ is an infinite constant as in Eq.~\eqref{eq:logZero_reg}.

Now, we will expand and re-sum the rest of the terms of Eq.~\eqref{eq:diracExpansion_zipper}. First, we will consider the case of the partial sum over the terms of the form 
$\tr \left[ 
       \left( \frac{1}{\slashed{\partial}} \slashed{A^{(1)}} \right)^{k}
       \frac{1}{\slashed{\partial}} \slashed{A^{(2)}}
     \right]$
for $k \ge 1$ which will add up to a sum
\begin{equation*}
\begin{split}
    -\sum_{k=1}^{\infty}
    \frac{1}{\pi^{k + 1}} \int d^2 x_1 \cdots \int d^2 x_{k + 1} 
    \left\{ 
    \frac{\tr\left[A^{(1)}_{\overline{z}}(x_1) \cdots A^{(1)}_{\overline{z}}(x_{k}) A^{(2)}_{\overline{z}}(x_{k+1})\right]}
         {(z_1 - z_2) \cdots (z_{k+1} - z_1)} 
    +
    \frac{\tr\left[A^{(1)}_{z}(x_1) \cdots A^{(1)}_{z}(x_{k}) A^{(2)}_{z}(x_{k+1})\right]}
         {(\overline{z}_1 - \overline{z}_2) \cdots (\overline{z}_{k+1} - \overline{z}_1)} 
    \right\}.
\end{split}
\end{equation*}
Note we don't include the pre-factor $\frac{1}{k+1}$ because the cyclicity of the trace will give $k+1$ equal terms corresponding to $k$ above. 

The main point in analyzing the above integrals are the observations that 
\begin{equation}
    A^{(j)}_{\overline{z}} d^2 x = \frac{1}{2 i} A^{(j)} \wedge dz
    \quad,\quad
    A^{(j)}_{z} d^2 x = \frac{1}{-2 i} A^{(j)} \wedge d\overline{z}.
\end{equation}
This observation combined with the defining property Eq.~\eqref{eq:deltaPerp_defn} means that we can replace the integrals over $d^2 x$ with holomorphic and antiholomorphic integrals over the curves $\{\gamma_j\}$. 

For example, the $k=1$ term above looks like
\begin{equation}
\begin{split}
    &\tr \left[ \frac{1}{\slashed{\partial}} \slashed{A^{(1)}} \frac{1}{\slashed{\partial}} \slashed{A^{(2)}} \right]
    =
    -\frac{\tr[u_1 u_2]}{\pi^2} 
    \left\{
        \int_{\gamma_1} \frac{dz_1}{2i} \int_{\gamma_2} \frac{dz_2}{2i}  \frac{1}{(z_1-z_2)(z_2-z_1)} 
        +
        \int_{\gamma_1} \frac{d\overline{z}_1}{-2i} \int_{\gamma_2} \frac{d\overline{z}_2}{-2i} \frac{1}{(\overline{z}_1-\overline{z}_2)(\overline{z}_2-\overline{z}_1)} 
    \right\}
    \\
    &= +\frac{\tr[u_1 u_2]}{(2\pi)^2} \log \left| \frac{(q'_2-q'_1)(q_2-q_1)}{(q'_2-q_1)(q_2-q'_1)} \right|^2.
\end{split}
\end{equation}
This is not singular if $\gamma_1: q_1 \to q'_1$ and $\gamma_2: q_2 \to q'_2$ don't share any endpoints, but has a logarithmic divergence otherwise.

Applying this prescription naively with other terms $k > 1$ leads to the expression
\begin{equation*}
    -\tr \left[ 
       \left( \frac{1}{\slashed{\partial}} \slashed{A^{(1)}} \right)^{k}
       \frac{1}{\slashed{\partial}} \slashed{A^{(2)}}
    \right]
    =
    \text{``}
    -\tr[u_1^k u_2]
    \left\{
        \frac{1}{(-2 \pi i)^{k+1}} \int_{\gamma_1} dz_1 \cdots \int_{\gamma_1} dz_k \int_{\gamma_2} dz_{k+1}  \frac{1}{(z_1-z_2)\cdots(z_{k+1}-z_1)} 
        +
        \mathrm{c.c.}
    \right\}
    \text{''}
\end{equation*}
where $\mathrm{c.c.}$ refers to complex conjugation. This expression is singular for along the given domains. However, by the discussion around Eqs.~(\ref{eq:deltaPerp_reg_def},\ref{eq:deltaPerp_reg_limit}) about the smeared out $\delta^{\perp}(\gamma_1)^{\epsilon}$, together with the holomorphicity of the integrands $\frac{1}{(z_1-z_2)\cdots(z_{k+1}-z_1)}$ implies that the contour $\gamma_j$ gets split up in the following way. Let $\gamma_j^{(1)}, \cdots, \gamma_j^{(k)}$ be non-intersecting curves $q_j \to q'_j$ close to $\gamma_j$ so that $\gamma_j^{(k)} \to \cdots \to \gamma_j^{(1)}$ is the left-to-right order relative to the orientation of $\gamma_j$ (see Fig.~\ref{fig:splitContours_appendix}). Then, the correct expression is
\begin{equation}
    -\tr \left[ 
       \left( \frac{1}{\slashed{\partial}} \slashed{A^{(1)}} \right)^{k}
       \frac{1}{\slashed{\partial}} \slashed{A^{(2)}}
    \right]
    = 
    -\frac{\tr[u_1^k u_2]}{k!}
    \sum_{\sigma \in \mathrm{Sym}(k)}
    \left\{
        \frac{1}{(2 \pi i)^{k+1}} \int_{\gamma_1^{(\sigma(1))}} dz_1 \cdots \int_{\gamma_1^{(\sigma(k))}} dz_k \int_{\gamma_2} dz_{k+1}  \frac{1}{(z_1-z_2)\cdots(z_{k+1}-z_1)} 
        +
        \mathrm{c.c.}
    \right\}
\end{equation}
so that we split up the contours for each variable $z_1 \cdots z_k$ into disjoint non-intersecting contours in all possible ways with equal weight. 

\begin{figure}[h!]
  \centering
  \includegraphics[width=0.7\linewidth]{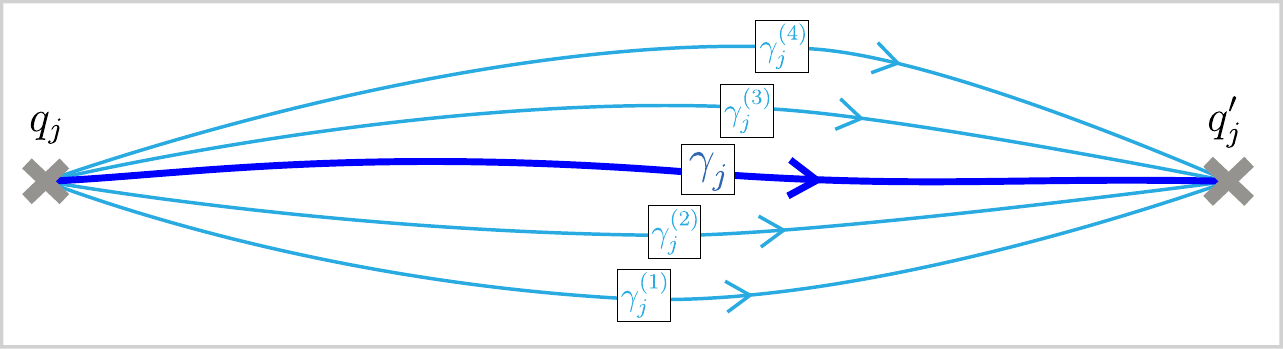}
  \caption{Splitting the contour $\gamma_j$ into a series of disjoint curves $\gamma_j^{(1)}, \cdots, \gamma_j^{(k)}$ with the left-to-right order being $\gamma_j^{(k)} \to \cdots \to \gamma_j^{(1)}$.}
  \label{fig:splitContours_appendix}
\end{figure}

Now for a permutation $\sigma \in \mathrm{Sym}(\ell)$ for any $\ell \ge 1$, we'll define
\begin{equation}
    \overline{[\sigma]} = \overline{[\sigma(1) \cdots \sigma(\ell)]} 
    := 
    \frac{1}{(2 \pi i)^{\ell+1}} \int_{\gamma_1^{(\sigma(1))}} dz_1 \cdots \int_{\gamma_1^{(\sigma(\ell))}} dz_{\ell} \int_{\gamma_2} dz_{\ell+1}  \frac{1}{(z_1-z_2)\cdots(z_{\ell+1}-z_1)} 
    +
    \mathrm{c.c.}
\end{equation}
These quantities have many relations between them. For example, by the residue theorem and some convenient factors of $\frac{1}{2 \pi i}$, we have $\overline{[2 1]} = \overline{[1 2]} - \overline{[1]}$. So, we can change the orders of the contours, sometimes at the cost of a residue which is a lower-order integral. In general, we have the relation
\begin{equation} \label{eq:contourInt_relation_app}
    \overline{[\sigma(1) \cdots \sigma(m),\sigma(m+1) \cdots \sigma(\ell)]} = 
    \begin{cases}
        \overline{[\sigma(1) \cdots \sigma(m+1),\sigma(m) \cdots \sigma(\ell)]}
        \mp
        \overline{[\sigma'(1) \cdots \sigma'(m) \cdots \sigma'(\ell-1)]}
        \quad &\text{if } \sigma(m) = \sigma(m+1) \pm 1 \\
        \overline{[\sigma(1) \cdots \sigma(m+1),\sigma(m) \cdots \sigma(\ell)]} \quad &\text{otherwise}
    \end{cases}
\end{equation}
Above, $\sigma' \in \mathrm{Sym}(\ell-1)$ is the permutation induced from $\sigma: \{1 \cdots \ell\} \backslash \{m\} \to \{1 \cdots \ell\} \backslash \{\sigma(m)\}$ by preserving the ordering. Define (c.f.~\cite{stanley2000enumerative})
\begin{equation}
    \mathrm{desc}(\sigma) := \# \{m \,|\, \sigma(m) > \sigma(m+1) \} 
    \quad,\quad 
    A(\ell, m) = \# \{\sigma \in \mathrm{Sym}(\ell) \,|\, \mathrm{desc}(\sigma) = m \}
\end{equation}
as the number of \textit{descents} of a permutation and the number of permutations with a fixed number of descents. These $A(k,m)$ are called \textit{Eulerian numbers}. Note that $A(k,m) = 0$ for $m > k-1$ and define $A(0,m) := 0$ \footnote{WARNING: This choice doesn't agree with the standard convention that $A(0,m) = \delta_{m,0}$.}. They satisfy the generating relation
\begin{equation} \label{eq:eulerian_gen_func}
    \sum_{k=0}^{\infty}\sum_{m=0}^{\infty} \frac{\eta^k}{k!} A(k,m) \xi^{m} = \frac{\xi-1}{\xi - e^{(\xi-1)\eta}}-1.
\end{equation}

For $\sigma \in \mathrm{Sym}(\ell)$, repeated applications of Eq.~\eqref{eq:contourInt_relation_app} gives
\begin{equation}
    \overline{[\sigma]} = \sum_{i=0}^{\ell} (-1)^i {\mathrm{desc}(\sigma) \choose i} \overline{[1, 2, \cdots, \ell - i]}
\end{equation}
so that our full sum becomes
\begin{equation}
    -\sum_{k=1}^{\infty} \tr \left[ \left( \frac{1}{\slashed{\partial}} \slashed{A^{(1)}} \right)^{k} \frac{1}{\slashed{\partial}} \slashed{A^{(2)}} \right]
    =
    -\sum_{k=1}^{\infty} \frac{\tr[u_1^k u_2]}{k!} \sum_{m=0}^{k} A(k, m) \sum_{i=0}^{m} (-1)^i {m \choose i} \overline{[1, 2, \cdots, k - i]}
\end{equation}
Now we want to ask what the coefficient of $\overline{[1, 2, \cdots, k - i]}$ is in the sum above. To do this, we use generating functions. 
For the function $C(t)$ formed by replaced $\overline{[1, 2, \cdots, \ell]} \mapsto t^{\ell}$ in the sum above, the coefficient of $t^{\ell}$ will give the coefficient of $\overline{[1, 2, \cdots, \ell]}$. We have
\begin{equation}
\begin{split}
    C(t) :=&
    -\sum_{k=1}^{\infty} \sum_{m=0}^{k} \frac{\tr[u_1^k u_2]}{k!} A(k, m) 
    \sum_{i=0}^{m} (-1)^i {m \choose i} t^{k-i} 
    =
    -\sum_{k=0}^{\infty} \sum_{m=0}^{\infty} \frac{\tr[u_1^k u_2]}{k!} A(k, m) t^{k}\left(1-\frac{1}{t}\right)^{m}
    \\
    =&
    -\tr \left[ 
        \left\{ \sum_{k=1}^{\infty} \sum_{m=0}^{\infty} \frac{(u_1 t)^k}{k!} A(k, m) \left(1-\frac{1}{t}\right)^{m} \right\} 
    u_2 \right]
    =
    -\tr \left[ 
        \left(\frac{1}{1 - t (1-e^{-u_1})} - 1\right) u_2 
    \right]
    = 
    -\sum_{\ell = 1}^{\infty} t^{\ell} \tr \left[ (1-e^{-u_1})^\ell u_2 \right].
\end{split}
\end{equation}
In the second-to-last equality, we simplify after using the generating relation Eq.~\eqref{eq:eulerian_gen_func}. This implies 
\begin{equation}
    -\sum_{k=1}^{\infty} \tr \left[ \left( \frac{1}{\slashed{\partial}} \slashed{A^{(1)}} \right)^{k} \frac{1}{\slashed{\partial}} \slashed{A^{(2)}} \right]
    =
    -\sum_{\ell=1}^{\infty} \tr \left[ (1-e^{-u_1})^\ell u_2 \right]
    \left\{
        \frac{1}{(2 \pi i)^{\ell+1}} \int_{\gamma_1^{(1)}} dz_1 \cdots \int_{\gamma_1^{(\ell)}} dz_{\ell} \int_{\gamma_2} dz_{\ell+1}  \frac{1}{(z_1-z_2)\cdots(z_{\ell+1}-z_1)}
        +
        \mathrm{c.c.}
    \right\}
\end{equation}

In general, the same exact steps above can re-sum many groups of terms in the full series $\log \frac{\det \left( \slashed{\partial} - \slashed{A} \right)}{\det\left(\slashed{\partial}\right)}$. In particular, for some $n > 1$, let $j_1 \neq j_2$, $j_2 \neq j_3$, $\cdots$, $j_{n} \neq j_1$, with each $j_{\cdots} \in \{1,\cdots,s\}$. Then, we have the equality
\begin{equation} \label{eq:dirac_app_traceSumGrouped}
\begin{split}
    &\sum_{k_1=1}^{\infty} \cdots \sum_{k_{n}=1}^{\infty} 
    \tr \left[ \left(\frac{1}{\slashed{\partial}} \slashed{A^{(j_1)}}\right)^{k_1} \cdots \left(\frac{1}{\slashed{\partial}} \slashed{A^{(j_{n})}}\right)^{k_{n}} \right]
    =
    \sum_{\ell_1=1}^{\infty} \cdots \sum_{\ell_{n}=1}^{\infty}
    \tr \left[ (1-e^{-u_{j_1}})^{\ell_1} \cdots (1-e^{-u_{j_n}})^{\ell_{n}} \right]
    \\
    &
    \times
    \Bigg\{ 
        \frac{1}{(2 \pi i)^{\ell_1 + \cdots + \ell_{n}}} 
        \int_{\gamma_{j_1}^{(1)}} dz_1 \cdots \int_{\gamma_{j_1}^{(\ell_1)}} dz_{\ell} 
        \cdots\cdots
        \int_{\gamma_{j_{n}}^{(1)}} dz_{\ell_{1}+\cdots+\ell_{n-1}+1} \cdots \int_{\gamma_{j_n}^{(\ell_n)}} dz_{\ell_{1}+\cdots+\ell_n}
        \frac{1}{(z_1-z_2)\cdots(z_{\ell_1+\cdots+\ell_n}-z_1)}
        +
        \mathrm{c.c.}
    \Bigg\}.
\end{split}
\end{equation}
Note that now, the series expansion is in terms of the $\{e^{-u_j} = U_j^{-1}\}$.

Now, we can rewrite the full partition function using some combinatorics. First, we have
\begin{equation} \label{eq:dirac_app_traceSum_ungrouped_to_grouped}
    \sum_{n'=2}^{\infty} \frac{1}{n'} \sum_{\substack{j'_1 \cdots j'_{n'} \in \{1 \cdots s\} \\ !(j'_1 = \cdots = j'_{n'})}}
    \tr \left[ \frac{1}{\slashed{\partial}} \slashed{A^{(j'_1)}} \cdots \frac{1}{\slashed{\partial}} \slashed{A^{(j'_{n'})}} \right]
    =
    \sum_{n=2}^{\infty} \frac{1}{n} 
    \sum_{\substack{j_1 \cdots j_n \in \{1 \cdots s\} \\ j_1 \neq j_2, \cdots, j_n \neq j_1}}
    \sum_{k_1=1}^{\infty} \cdots \sum_{k_n=1}^{\infty} 
    \tr \left[ \left(\frac{1}{\slashed{\partial}} \slashed{A^{(j_1)}}\right)^{k_1} \cdots \left(\frac{1}{\slashed{\partial}} \slashed{A^{(j_n)}}\right)^{k_n} \right].
\end{equation}
This follows from cyclicity of the trace and noting that the total weight of an equivalence class of the terms modulo cyclic symmetry is the same on both sides. 
Then we will combine Eqs.~(\ref{eq:diracExpansion_zipper},\ref{eq:diracDet_sameZipTerms},\ref{eq:dirac_app_traceSum_ungrouped_to_grouped},\ref{eq:dirac_app_traceSumGrouped}), which gives us (after an additional step of combinatorics which is the same as the reverse of the Eq.~\eqref{eq:dirac_app_traceSum_ungrouped_to_grouped})
\begin{equation} \label{eq:zipper_expansion_final}
\begin{split}
    &\log \frac{\det \left( \slashed{\partial} - \slashed{A} \right)}{\det\left(\slashed{\partial}\right)}
    \\
    &\stackrel{\textbf{corr}}{=}
    \sum_{j=1}^{N} \frac{1}{4 \pi^2} \tr\left[u_j^2\right] \log\left(\frac{|q'_j-q_j|^2}{|0|^2}\right)
    \\
    &\quad
    -\sum_{n=2}^{\infty} \frac{1}{n}
    \sum_{\substack{j_1 \cdots j_n \in \{1 \cdots s\} \\ !(j_1 = \cdots = j_n)}} 
    {\tr \left[ (1 - U_{j_1}^{-1}) \cdots (1 - U_{j_n}^{-1}) \right]}
    \left\{
    \frac{1}{(2 \pi i)^{n}}
    \int_{\gamma_{j_1}^{(c(1))}} dz_1 \cdots \int_{\gamma_{j_n}^{(c(n))}} dz_{\ell_{1}+\cdots+\ell_n}
    \frac{1}{(z_1-z_2)\cdots(z_n-z_1)}
    + \mathrm{c.c.}
    \right\}
\end{split}
\end{equation}
where the choice of $c(1) \cdots c(n)$ satisfy $c(k) = 1, c(k+1) = 2, \cdots, c_{k+\ell} = \ell+1$ for $j_{k-1} \neq j_{k} = \cdots = j_{k+\ell} \neq j_{k+\ell+1}$, the indices $k+n \sim k$ being defined cyclically. This is just a cyclic reordering of the variables in the integration variables of Eq.~\eqref{eq:dirac_app_traceSumGrouped}.

Note that this series is entirely in terms of the monodromy representation $\{U_j^{-1}\}$. 
However as we noted before, the $\tau$ function depends not only on the monodromy data but changes under an action of $\pi_1(GL(N,\C)) = \Z$ acting on a puncture that twists the flat connection. (See for example the dicussion before Eq.~\eqref{eq:diracDeterminant_full} about large gauge transformations and the result for the abelian case Eq.~\eqref{eq:diracDet_abelian}.)
As such, the series above should not be expected to converge to the desired result (even after subtracting off appropriate logarithmic divergences) for all choices of $\{U_j\}$, but should be viewed as a series whose analytic continuation corresponds to the $\tau$ function. 

Also, note that the the holomorphic and antiholomorphic integrals above evaluate to holomorphic and antiholomorphic functions in the punctures. After splitting up the logarithms into their holmorphic and antiholomorphic components, we find that the entire regularized determinant splits up into a holomorphic and antiholomorphic piece. 

The fact that the Fuschian representation of Sec.~\ref{app:tau_fuschian} satisfies the Jimbo-Miwa-Ueno equations leads us to identify the holomorphic piece with the original tau function of~\cite{jmuTau1981}.
The antiholomorphic piece should be the conjugate of the tau function corresponding to the conjugate monodromies $\overline{U}_1,\cdots,\overline{U}_s$. 

On this note, we conjecture that the regulation of divergent integrals as appeared in Sec.~\ref{sec:nonSeparatedZippers} for the dimer model asymptotically gives the Jimbo-Miwa-Ueno tau function.

\bibliography{bibliography.bib}

\end{document}